\documentclass[3p,number,times,fleqn]{elsarticle-arxive}

\usepackage{amssymb,amsmath,amsthm}

\usepackage{bbm}
\usepackage{url}
\usepackage{xspace}
\usepackage{enumerate}
\usepackage{paralist}
\usepackage{multirow}
\usepackage{hyperref}
\usepackage{centernot}

\usepackage{tikz}
\usetikzlibrary{shapes,decorations,decorations.pathmorphing,arrows,fit}
\usetikzlibrary{decorations.markings,calc,patterns}
\usetikzlibrary{decorations.pathreplacing,backgrounds}

\usepackage{misc}

\usepackage{pifont}

\newcommand{\finofex}{\hfill\ding{115}}

\usepackage{color}
\usepackage{enumitem}

\usepackage{thmtools,thm-restate}
\usepackage{xspace}

\newtheorem{theorem}{Theorem} 
\newtheorem{lemma}[theorem]{Lemma} 
\newtheorem{definition}[theorem]{Definition}
\newtheorem{corollary}[theorem]{Corollary}
\newtheorem{exmp}[theorem]{Example}
\newtheorem{observation}[theorem]{Observation}
\newtheorem{proposition}[theorem]{Proposition}

\usepackage{thmtools,thm-restate}

\journal{Artificial Intelligence}

\begin{document}

\begin{frontmatter}



\title{Logical Separability of Labeled Data Examples under Ontologies}


\author[hi]{Jean Christoph Jung\corref{cor1}}
\ead{jungj@uni-hildesheim.de}

\author[br]{Carsten Lutz}
\ead{clu@informatik.uni-leipzig.de}

\author[liv]{Hadrien Pulcini}
\ead{h.pulcini@liverpool.ac.uk}

\author[liv]{Frank Wolter}
\ead{wolter@liverpool.ac.uk}

\address[hi]{Institute of Computer Science, University of Hildesheim,
Germany}
 
\address[br]{Institute for Computer Science,
  Leipzig University, Germany}
\address[liv]{Department of Computer Science,
  University of Liverpool, UK}

\begin{abstract}
Finding a logical formula that separates positive and
negative examples given in the form of labeled data items is
fundamental in applications such as concept learning,
reverse engineering of database queries, generating
referring expressions, and entity comparison in knowledge graphs. 
In this paper, we investigate the
existence of a separating formula for data in the
presence of an ontology. Both for the ontology language and
the separation language, we concentrate on first-order logic
and the following important fragments thereof: the description logic
$\mathcal{ALCI}$, the guarded fragment, the two-variable
fragment, and the guarded negation fragment. For separation, we also consider (unions of) conjunctive queries. We consider several forms of separability that
differ in the treatment of negative examples and in
whether or not they admit the use of additional
helper symbols to achieve separation.  
Our main results are model-theoretic characterizations of (all variants
of) separability, the comparison of the separating power of different
languages, and the investigation of the computational complexity of
deciding separability. 
\end{abstract}


\end{frontmatter}

\section{Introduction}

There are many scenarios 
in which the aim is to find some kind of logical expression that
separates positive from negative examples given in the form of labeled
data items within a data set.  For instance, in \emph{entity
comparison} the positive and negative examples are single entities
within a knowledge graph and one aims to explore the relationship
between them by searching for relevant logical features that
distinguish them from each other. In another application scenario, the
positive and negative examples have been derived using a classifier
whose behaviour one aims to understand and explain by means of a
logical expression that applies to the positive examples, but not to the
negative ones. Even more ambitiously, one might be in a supervised
learning scenario and aim at a logical expression that generalizes the
positive examples but does not apply to any negative example and that
can potentially serve as a classifier for future prediction tasks.
Indeed, in \emph{concept learning in description logic}, the aim is to
automatically construct a concept description from examples that can
then be used for various purposes including classifier explanation,
prediction, and as a building block in ontology engineering. 
In yet another application area called \emph{reverse engineering of
database queries} or \emph{query by example}, the examples are answers
and non-answers to a query that a user who is not familiar with the underpinning
query language aims to construct. While the user is unable to
construct the query, they may be able to provide such examples, and
a separating query is then a reconstruction, or at least
approximation, of the original query. Finally, in \emph{generating
referring expressions} in computational linguistics and data
management, the aim is to find a meaningful logical description of a
real world object whose name is  meaningless to the typical user. In
this case the positive examples consist of a single individual and the
negative examples of all remaining individuals in the domain.
In all
these settings it is often the case that, in addition to the data set,
some background knowledge in the form of an ontology is given.  We return
to these application scenarios and research areas below, discuss
them in more detail, and also provide links to
related work.

In this article, we investigate the 
separation of 
positive and negative
examples in the presence of an ontology. 
As usual when data and ontologies are combined, we assume that the data is
incomplete and adopt an open world semantics.  More precisely, we assume that a
labeled knowledge base (KB) $(\Kmc,P,N)$ is given with $\Kmc=(\Omc,\Dmc)$,
where \Omc is an ontology defined in a fragment of first-order logic (FO) or
a description logic, \Dmc is a set of facts, $P$ is a set of positive examples,
and $N$ is a set of negative examples. All examples are tuples of constants
from $\Dmc$ of the same length. As usual, we write $\Kmc\models \vp(\vec{a})$
for a tuple $\vec{a}$ in $\Dmc$ and FO formula $\varphi$ in case
that $\varphi(\vec{a})$ is true in every model of $\Omc\cup \Dmc$.

Due to the open
world semantics, different choices
are possible regarding the definition of a formula $\vp$ that
separates $(\Kmc,P,N)$. While it is uncontroversial to demand that
$\Kmc \models \vp(\vec{a})$ for all $\vec{a} \in P$, for negative
examples $\vec{b} \in N$ it makes sense to demand that
$\Kmc \not\models \vp(\vec{b})$, but also that
$\Kmc \models \neg \vp(\vec{b})$. When $\vp$ is formulated in logic
$\Lmc_S$, we refer to the former as \emph{weak $\Lmc_S$-separability} and to
the latter because \emph{strong $\Lmc_S$-separability}. The following example illustrates the difference between the two choices.

\begin{exmp}
{\em	Consider the KB $\Kmc=(\Omc,\Dmc)$ with $\Omc$ empty and $\Dmc$ stating that $c_{1},c_{2}$ are football clubs competing in the Bundesliga and that $c_{3},c_{4}$ compete in the Premier League: 
\begin{eqnarray*}
	\Dmc & = & \{{\sf compete\_in}(c_{1},b), {\sf compete\_in}(c_{2},b), 
	{\sf compete\_in}(c_{3},p), {\sf compete\_in}(c_{4},p),\\
	 & & {\sf Bundesliga}(b), {\sf PremierLeague}(p), {\sf Footballclub}(c_{1}),
	{\sf Footballclub}(c_{2})\}
\end{eqnarray*}
Let $P=\{c_{1},c_{2}\}$ and $N=\{c_{3},c_{4}\}$. Then ${\sf Footballclub}(x)$ weakly separates
$P$ from $N$ but it does not strongly separate $P$ from $N$. Indeed, as we adopt an open world semantics, no formula in FO strongly separates $P$ from $N$ as $\Kmc$ does not contain any negative information. In this example, ${\sf Footballclub}(x)$ is a rather misleading separator because clubs that compete in the PremierLeague are also football clubs and so it only weakly separates because the KB is incomplete. To exclude ${\sf Footballclub}(x)$ as a weak separator, we add this information to the ontology, obtaining
$$
\Omc'= \{ \forall x (\exists y ({\sf compete\_in}(x.y) \wedge {\sf PremierLeague}(y))\rightarrow {\sf Footballclub}(x))\},
$$
and consider the KB $\Kmc'=(\Omc',\Dmc)$. Then ${\sf Footballclub}(x)$ no longer weakly separates $P$ and $N$. In fact, weak separability is now only witnessed by more meaningful separators such as 
$$
\varphi(x) = \exists y ({\sf compete\_in}(x,y) \wedge {\sf Bundeliga}(y)).
$$
As $\Kmc'$ still lacks negative information, $P$ and $N$ are still not strongly separable in FO. They become strongly separable by $\varphi(x)$ after adding the disjointness axiom
$$
\forall x((\exists y {\sf compete\_in}(x,y) \wedge {\sf Bundesliga}(y))
                  \rightarrow \neg \exists y ({\sf compete\_in}(x,y) \wedge {\sf Premierleague}(y)))
$$
to $\Omc'$. We next consider separability in FO of the singleton $P_{1}=\{c_1\}$ from the remaining constants $N_{1}= \{c_{2},c_{3},c_{4},b,p\}$. As we know already how to separate $c_{1}$ from 
$c_{3}$ and $c_{4}$, we can focus on $b$, $p$ and $c_{2}$. Weak separation from $b$ and $p$ is straightforward by taking $\exists y\;{\sf compete\_in}(x,y)$. To obtain strong separation one might add the disjointness axioms ${\sf Bundesliga}(x) \rightarrow \neg \exists y\;{\sf compete\_in}(x,y)$ and ${\sf Premierleague}(x) \rightarrow \neg \exists y\;{\sf compete\_in}(x,y)$. The constants $c_{1}$ and $c_{2}$, however, are neither weakly nor strongly separable in any fragment of FO as one can easily see that swapping $c_{1}$ and $c_{2}$ does not change the KB (except for the naming of $c_{1}$ and $c_{2}$).
Hence $P_{1}$ and $N_{1}$ are not weakly nor strongly separable in FO either.\finofex
}
\end{exmp} 

In this article we study both weak and strong separability in
labeled KBs. Both notions depend on
the languages used for separation and to formulate the ontology. Hence, we
aim to investigate the role of ontologies and the impact of the
ontology language, to compare the separating power
of different separation languages and, finally, to determine the
decidability and complexity of separability as a decision problem.
Our main tool are model-theoretic characterizations of separability
tailored towards the considered ontology and separation languages. 
Given two logical languages \Lmc and $\Lmc_S$,  with
\emph{$(\Lmc,\Lmc_S)$-separability} we mean $\Lmc_S$-separability of labeled
\Lmc-knowledge bases. 
 
%
%
%

FO is the most powerful language for formulating ontologies and separating
examples that we consider in this article. This choice reflects the fact that
the large majority of ontologies deployed in the real world are formulated in
(fragments of) FO and that fragments of FO also cover many natural languages
for separation. While one cannot expect decidability results for full FO as
the ontology language, it turns out that one can obtain very powerful
characterization results for labeled KBs with unrestricted FO-ontologies
that are useful also for many fragments of FO. In
practice, most ontologies are formulated in description logics and the second
main ontology and separation language that we consider is the expressive
description logic \ALCI, 
a rather representative and frequently used description logic. Separation using
$\ALCI$-concepts is of direct practical interest for concept learning in
description logic, entity comparison, and generating referring expressions. We
also discuss the behavior of a few other description logics and comment on
interesting description logics left for future investigation. We further
investigate in full detail two extensions of $\ALCI$: the guarded fragment, GF, and
the two-variable fragment, FO$^2$, of FO.  Both languages are generally
regarded as fundamental decidable fragments of FO that still enjoy many of the
desirable properties of description logics, and their investigation has led to a
better understanding of the reasons for their good computational behavior. We
consider them here to gain a better understanding of the computational and
semantic properties of separability in general, depending on the expressive
power of the separation language. As separating formulas, we further
consider conjunctive queries, CQs, defined as FO-formulas constructed from
atoms using conjunction and existential quantification, and unions of
conjunctive queries, UCQs, defined as disjunctions of CQs. Separation using CQs
and UCQs is of direct interest for query by example, but also enables us to
characterize separation in the languages discussed above. Finally, we formulate
a few results about the guarded negation fragment of FO, GNFO, which is a more
recent decidable fragment of FO containing both GF and UCQ that still enjoys
many of the desirable properties of GF.

In connection with the application scenarious discussed at the beginning of the
introduction, we define the following special cases of $(\Lmc,\Lmc_{S})$-separability:
\begin{itemize}
	\item $(\Lmc,\Lmc_{S})$-\emph{definability} is 
	$(\Lmc,\Lmc_{S})$-separability for labeled KBs where $P$ and $N$ partition the
	example space; that is, inputs are labeled $\Lmc$-KBs $(\Kmc,P,N)$ such that
	$N$ is defined as the set of $n$-tuples in $\Dmc$ that are not in $P$, with $n$ the length of example tuples. 
	\item $(\Lmc,\Lmc_{S})$-\emph{referring expression existence}  is $(\Lmc,\Lmc_{S})$-definability for labeled KBs where $P$
	is a singleton set.
	\item $(\Lmc,\Lmc_{S})$-\emph{entity distinguishability}
	is $(\Lmc,\Lmc_{S})$-separability for labeled KBs where $P$ and $N$ are both singleton sets.
\end{itemize}
We show that, with only very few exceptions, the special cases behave in exactly the same way as general separability and therefore mostly focus on separability in the remainder of the introduction. We start the discussion of our results with weak separability. 
Our first main result provides a
characterization of weak $(\text{FO},\text{FO})$-separability in
terms of homomorphisms. It implies that
\begin{itemize}

  \item[{\sf (wA)}] weak $(\text{FO},\Lmc_S)$-separability coincides
    for all FO-fragments $\Lmc_S$ situated between UCQ and FO.

\end{itemize}
Thus, somewhat surprisingly, UCQ, GNFO, and FO all have the same separating power.
This result shows that separability under the open world semantics
adopted in this article is significantly more cautious than separability under the
closed world semantics, which is typically adopted in the database literature. Under the latter semantics, every database can be described in FO up to isomorphisms, hence separability using FO-formulas corresponds to ``being non-isomorphic'' and is clearly much more powerful than separability using UCQs. Our characterization also implies a close link between separability and the evaluation of rooted UCQs (unions of CQs in which every variable is reachable from an answer variable) on FO-KBs. In fact, we show that
\begin{itemize}

  \item[{\sf (wB)}] there are mutual polynomial time Turing reductions
    between weak $(\Lmc,\Lmc_{S})$-separability and the complement of
    rooted UCQ-evaluation on $\Lmc$-KBs for all FO-fragments $\Lmc_S$
    between UCQ and FO and all FO-fragments~$\Lmc$.

\end{itemize}
For the special cases
of referring expression existence and entity distinguishability in which
$P$ is a singleton set, UCQ can be replaced by CQ. In particular, in
the mutual reductions, rooted UCQ-evaluation can be replaced by rooted CQ-evaluation.
As a first application of ${\sf (wB)}$, it follows from the fact that rooted
UCQ-evaluation on GNFO-KBs is \TwoExpTime-complete that
$(\text{GNFO},\text{GNFO})$-separability and, equivalently,
$(\text{GNFO},\text{UCQ})$-separability, are decidable and \TwoExpTime-complete
in combined complexity, where the KB $\Kmc=(\Omc,\Dmc)$ and the sets $P$ and
$N$ of examples are regarded as the input.

We then proceed to study weak $(\Lmc,\Lmc)$-separability for the fragments
$\Lmc \in \{ \ALCI, \text{GF},\text{FO}^2 \}$. Note that these
fragments do not contain UCQ nor CQ and thus the above results do not apply.
For these fragments the following distinction turns out to be of fundamental importance when defining weak separability: one might or
might not admit the use of \emph{helper symbols}, that is, symbols that do not occur
in the KB, in the separating formula. 
\begin{exmp}\label{exm:new}
	{\em 
Consider the KB $\Kmc=(\Omc,\Dmc)$ with $\Omc= \{ \forall x (\exists y \;{\sf shaves}(x,y) \wedge \exists y \;{\sf shaves}(y,x)\}$
stating that everybody shaves somebody and is shaved by somebody and 
$
\mathcal{D}=\{{\sf shaves}(p,p),{\sf shaves}(n,n')\}$. Let $P=\{p\}$ and $N=\{n\}$. Then ${\sf shaves}(x,x)$ weakly separates $P$ from $N$ in FO but there is no $\mathcal{ALCI}$-concept using only the symbol ${\sf shaves}$ that weakly separates $P$ from $N$ as under $\Kmc$ every such concept is either valid or unsatisfiable. However, by admitting a fresh unary relation symbol, say ${\sf Barber}$, we obtain the formula $\varphi(x)={\sf Barber}(x) \rightarrow \exists y ({\sf shaves}(x,y)\wedge {\sf Barber}(y))$ which weakly separates $P$ and $N$ and is equivalent to an $\mathcal{ALCI}$-concept. Observe that indeed $\Kmc\models \varphi(p)$ since if $p$ is a barber, then, by ${\sf shaves}(p,p)$, $p$ shaves a barber.\finofex}
\end{exmp}
Motivated by Example~\ref{exm:new} we distinguish between \emph{projective}
(with helper symbols) and \emph{non-projective} (without helper symbols)
versions of separability. While we show that projective and non-projective
weak $(\text{FO},\Lmc_S)$-separability coincide for all FO-fragments $\Lmc_S$
situated between UCQ and FO, and therefore for the cases considered in {\sf
  (wA)} and {\sf (wB)}, Example~\ref{exm:new} shows that the projective and
non-projective version do not coincide for $\ALCI$, and similar examples will
be given for $\text{GF}$ and $\text{FO}^2$.  We then start the investigation
with $\ALCI$ and show the rather unexpected result that
\begin{itemize}
  \item[{\sf (wC)}] projective weak $(\ALCI,\ALCI)$-separability is the same as
(both projective and non-projective) weak $(\ALCI,$ $\text{UCQ})$-separability and
thus, by {\sf (wA)}, also as (projective and non-projective) weak $(\ALCI,\text{FO})$-separability. 
\end{itemize}
It then follows from {\sf (wB)} and the known result that rooted UCQ-evaluation on $\ALCI$-KBs is co\NExpTime-complete in combined
complexity that projective weak $(\ALCI,\ALCI)$-separability is \NExpTime-complete in combined complexity. 

By Example~\ref{exm:new} the equivalence stated in {\sf (wC)} depends on
admitting helper symbols in separating $\ALCI$-concepts and its proof relies on
a characterization of projective weak $(\ALCI,\ALCI)$-inseparability using
functional bisimulations. We next turn to the technically rather intricate case
of non-projective weak $(\ALCI,\ALCI)$-separability and characterize it using a mix of homomorphisms, bisimulations, and types. Intuitively, one has to understand when a helper symbol in a separating concept can be equivalently replaced by a (possibly compound) concept in the language of the KB. 
The characterization allows us to show that non-projective weak $(\ALCI,\ALCI)$-separability is also \NExpTime-complete in
combined complexity.

For projective and non-projective
weak $(\text{GF},\text{GF})$-separability, we establish characterizations
that parallel those for $\ALCI$ except that bisimulations are replaced
with guarded bisimulations. As in the \ALCI-case, projective
$(\text{GF},\text{GF})$-separability coincides with
(projective and non-projective) $(\text{GF},\text{UCQ})$-separability,
and thus also with $(\text{GF},\text{FO})$-separability. For referring expression existence and entity distinguishability, UCQ can be replaced by CQ. We additionally observe that projective
and non-projective $(\text{GF},\text{GF})$-separability
also coincide with projective and, respectively, non-projective,
$(\text{GF},\text{openGF})$-separability, where openGF is a `local'
version of GF that only speaks about the neighbourhoods of the positive
examples. While openGF has the same separating power as GF and enforces
arguably more natural separating expressions than GF, we also show that
separating expressions are sometimes less succinct. Finally, our main
complexity result concerning GF is that 
\begin{itemize}
  \item[{\sf (wD)}] projective and non-projective weak $(\text{GF},\text{GF})$-separability are
both \TwoExpTime-complete in combined complexity.
\end{itemize} 
The proofs use the same approach as those given for $\ALCI$ and the complexity result relies on the fact that UCQ-evaluation of rooted UCQs on GF-KBs is \TwoExpTime-complete. The proofs are technically more challenging, however, as one
has to work with guarded bisimulations rather than bisimulations for $\ALCI$.
We next show that, in contrast, 
\begin{itemize}
  \item[{\sf (wE)}] projective and non-projective weak $(\text{FO}^2, \text{FO}^2)$-separability and
$(\text{FO}^2, \text{FO})$-separability are both undecidable.
\end{itemize}
Moreover, they coincide neither in the projective nor in the
non-projective case. The latter result is proved, under mild conditions, for all fragments $\Lmc$ of FO that enjoy the finite model property (every satisfiable ontology is satisfied in a finite model) but are not finitely
controllable for rooted UCQs (evaluating rooted UCQs on finite models
does not coincide with evaluating them on arbitrary models). 

As ontologies are typically small compared with databases, we also
consider the data complexity of deciding separability where only the
database and the sets of positive and negative examples are regarded
as the input, but the ontology is fixed. Some care is needed when
using the aforementionted link with rooted UCQ-evaluation to analyze
the data complexity of separability since, in the mutual reductions,
the input database of the separability problem is used to construct
the query of the evaluation problem, and vice versa. We show that,
under mild conditions on the FO-fragment $\Lmc$, one can construct for
any $\Lmc$-ontology $\Omc$ an $\Lmc$-ontology $\Omc'$ such that one
obtains a mutual polynomial time Turing reduction between rooted
UCQ-evaluation on KBs with ontology $\Omc$ and the complement of
UCQ-separability of labeled KBs with ontology $\Omc'$, and vice
versa. We then construct an \ALCI-ontology $\Omc$ such that rooted
UCQ-evaluation on KBs with ontology $\Omc$ is
\coNExpTime-complete. Hence, 
\begin{itemize}
  \item[{\sf (wF)}] projective weak $(\ALCI,\ALCI)$-separability
is \NExpTime-complete in data complexity.   
\end{itemize}
Thus, remarkably, in this case data and combined
complexity coincide. This lower bound applies to
$(\text{GNFO},\text{GNFO})$ and $(\text{GF},\text{GF})$-separability,
but we conjecture that the latter problems are actually
\TwoExpTime-hard in data complexity so that 
also in these cases data and combined complexity
coincide.  While the combined complexity of
special cases such as referring expression existence coincides with
that of separability for the languages considered, this remains open
for data complexity.

We then switch to strong separability, first observing that in marked
contrast to the weak case, projective strong
$(\Lmc,\Lmc_S)$-separability coincides with non-projective strong
$(\Lmc,\Lmc_S)$-separability for all choices of \Lmc and $\Lmc_S$
relevant to this paper. We establish a characterization of strong
$(\text{FO}, \text{FO})$-separability in terms of KB unsatisfiability
and show that 
\begin{itemize}
  \item[{\sf (sA)}] strong $(\text{FO}, \text{FO})$-separability coincides
with strong $(\text{FO}, \text{UCQ})$-separability and consequently
also with strong $(\text{FO}, \Lmc_S)$-separability for all $\Lmc_S$
situated between UCQ and FO. 
\end{itemize}
For strong referring expression existence and entity distinguishability UCQs can be replaced by CQs. We next consider the same FO-fragments $\ALCI,
\text{GF}, \text{FO}^2$ as before and show that for each of these
fragments \Lmc, strong $(\Lmc,\Lmc)$-separability coincides with
strong $(\Lmc,\text{FO})$-separability and thus the connection to KB
unsatisfiability applies. This allows us to derive tight complexity
bounds for strong $(\Lmc,\Lmc)$-separability: 
\begin{itemize}
  \item[{\sf (sB)}] for $\ALCI$, $\text{GF}$, and
	FO$^2$ strong separability is \ExpTime, \TwoExpTime, and, respectively, \NExpTime-complete in combined complexity. It is \coNPclass-complete in data complexity in all three cases. 
\end{itemize}
Note that strong $(\text{FO}^2,
\text{FO}^2)$-separability thus turns out to be decidable, in contrast
to the weak case. On the other hand, we show that the relationship between GF and openGF is the same as in the weak case: strong $(\text{GF},\text{GF})$-separability coincides with strong $(\text{GF},\text{openGF})$-separability but separating expressions can be more succinct. Finally, we show that in contrast to weak separability one can bound the size of strongly separating expressions independently from the size of the underlying database in terms of the size of the ontology, for $\mathcal{ALCI}$, GF, and FO$^{2}$.

The following table gives an overview of our most important results,
focussing on weak projective $(\Lmc,\Lmc)$-separability and strong
$(\Lmc,\Lmc)$-separability, for $\Lmc\in
\{\ALCI,\text{GF},\text{FO}^{2}\}$. For both weak and strong
separability, we list (in columns 2 and 4) settings with the same
separating power and (in columns 3 and 5) the combined complexity of
deciding the respective version of 
separability. All complexity results are completeness results. 

\begin{center}
	\begin{tabular}{|c|c|c|c|c|}
		\hline
		& \multicolumn{2}{c|}{Weak Separability} &
	      \multicolumn{2}{c|}{Strong Separability} \\
		\hline
		KB language $\mathcal{L}$ & same separating power &
		complexity & same separating power & complexity\\
		\hline                     
		$\mathcal{ALCI}$ & proj $\mathcal{ALCI}$, UCQ, FO & \textsc{NExpTime} & $\mathcal{ALCI}$, UCQ, FO & $\ExpTime$\\
		GF & proj GF, proj openGF, UCQ, FO & \textsc{2ExpTime} & GF, openGF, UCQ, FO & \textsc{2ExpTime}\\ 
		FO$^{2}$ & proj FO$^{2}$ & undecidable &
		FO$^{2}$, UCQ, FO & co$\NExpTime$\\
		\hline
	\end{tabular}
\end{center}
Our results rely on three important assumptions: we assume that different constants can denote the same element and so we do not make the unique name assumption, we do not admit constants (nor function symbols) in ontologies and separating formulas, and we do not admit any restrictions on the relations symbols from the KB that can be used in separating formulas. In all three cases it is of interest to explore what happens if one drops or relaxes the assumption. In the final section of this article we provide some results in this direction.
  
\paragraph{Structure of the Paper} We discuss related work
in Section~\ref{sec:rel}. In Section~\ref{sec:prelim}, we introduce the
basic notions that are used in the rest of the paper. In
Section~\ref{sec:weak}, we introduce weak separability and make
fundamental observations in the context of full first-order logic (as
ontology and separating language). We then investigate in
Section~\ref{sec:weakdecidable} weak separability for decidable
fragments of first-order logic. In Sections~\ref{sec:strong}
and~\ref{sec:strongdecidable}, we study strong separability starting
again with fundamental results. Finally, in
Section~\ref{sec:discussion}, we discuss some variations and
extensions of the problems investigated in this article and conclude
with possible future work. Some proofs are deferred to an appendix. 

\section{Related Work and Applications}\label{sec:rel}
This article extends the conference paper~\cite{KR}. It also generalizes, corrects, and puts into context some results first presented in~\cite{DBLP:conf/ijcai/FunkJLPW19}. For example, after establishing the link between projective weak $(\ALCI,\ALCI)$-separability and rooted UCQ-evaluation on $\mathcal{ALCI}$-KBs, it is proved in \cite{DBLP:conf/ijcai/FunkJLPW19} already that projective weak $(\ALCI,\ALCI)$-separability is \NExpTime-complete in combined complexity. It is also claimed to be $\Pi_{2}^{p}$-complete in data complexity. In this article (and~\cite{KR} already) this error is fixed. As mentioned above, in this article (and~\cite{KR}), we adopt the standard semantics of first-order logic and description logic and hence do not make the unique name assumption (UNA), but it
is made in~\cite{DBLP:conf/ijcai/FunkJLPW19}. If $\mathcal{ALCI}$ or a fragment is used as the ontology and separation language, the UNA does not affect weak separability nor strong
separability, but for FO and extensions of $\mathcal{ALCI}$ with
number restrictions such as $\mathcal{ALCQI}$ it does. We refer the
reader to Section~\ref{sec:discussion} for a detailed discussion of
the role of the UNA in the context of separability. We also note
that in~\cite{DBLP:conf/ijcai/FunkJLPW19} only the projective case of
weak separability is considered (without using the term projective).
The results obtained in~\cite{DBLP:conf/ijcai/FunkJLPW19} and also the
recent~\cite{kr2021restr} (which considers a more flexible version of
separarability in which the symbols that are used in separating
expressions can be restricted) are also discussed in more detail in Section~\ref{sec:discussion}.

We next discuss related work and applications of our
results in concept learning in DL, query by example, generating referring expressions, and entity comparison. We start with concept learning in DL as first proposed
in~\cite{DBLP:conf/ilp/BadeaN00}. Inspired by inductive logic
programming, refinement operators are used to construct a concept that
generalizes positive examples while not encompassing any negative
ones. An ontology may or may not be present. 
There has been significant interest in this approach, both for weak
separation~\cite{DBLP:conf/ilp/LehmannH09,DBLP:journals/ml/LehmannH10,Lisi15,DBLP:conf/aaai/SarkerH19}
and strong
separation~\cite{DBLP:conf/ilp/FanizzidE08,DBLP:conf/dlog/Lisi12,DBLP:journals/fgcs/RizzoFd20}.
Prominent systems include the {\sc DL Learner}~\cite{DBLP:conf/www/BuhmannLWB18,DBLP:journals/ws/BuhmannLW16}, {\sc
DL-Foil}~\cite{DBLP:conf/ekaw/Fanizzi0dE18} and its extension {\sc DL-Focl} \cite{DBLP:conf/ekaw/0001FdE18}, SPaCEL~\cite{DBLP:journals/jmlr/TranDGM17}, {\sc
YinYang}~\cite{DBLP:journals/apin/IannonePF07}, and {\sc pFOIL-DL}~\cite{Straccia15}.
A method for generating strongly separating concepts based on
bisimulations has been developed
in~\cite{DBLP:conf/soict/HaHNNST12,DBLP:journals/vjcs/TranNH15,DBLP:journals/vjcs/TranNH15b}
and an approach based on answer set programming was proposed
in~\cite{DBLP:conf/aiia/Lisi16}. Algorithms for DL concept learning
typically aim to be complete, that is, to find a separating concept
whenever there is one. Complexity lower bounds for separability as
studied in this paper then point to an inherent complexity that no
such algorithm can avoid. Undecidability even means that there can be
no learning algorithm that is both terminating and complete. 
We note that computing least
common subsumers (LCS) and most specific concepts (MSC) can be viewed
as DL concept learning in the case that only positive, but no negative
example are
available~\cite{DBLP:conf/aaai/CohenBH92,Nebel91,DBLP:conf/ijcai/BaaderKM99,DBLP:conf/ijcai/ZarriessT13}.
A recent study of LCS and MSC from a separability angle is
in~\cite{aaaithis}.

Query by example is an active topic in database research for
many years,
see e.g.~\cite{DBLP:conf/sigmod/TranCP09,DBLP:conf/sigmod/ZhangEPS13,DBLP:conf/pods/WeissC17,kalashnikov2018fastqre,deutch2019reverse,DBLP:conf/icdt/StaworkoW12}
%
and~\cite{martins2019reverse} for a recent survey. In this context,
separability has also received attention
\cite{DBLP:journals/tods/ArenasD16,DBLP:conf/icdt/Barcelo017,DBLP:journals/tods/KimelfeldR18,DBLP:conf/icdt/CateD21}. A
crucial difference to the present paper is that QBE in classical
databases uses a closed world semantics under which there
is a unique natural way to treat negative examples: simply demand
that the separating formula evaluates to false there. Thus, the
distinction between weak and strong separability, and also between
projective and non-projective separability does not arise. 
QBE for ontology-mediated
querying \cite{GuJuSa-IJCAI18,DBLP:conf/gcai/Ortiz19,DBLP:journals/corr/abs-2108-10021} and for SPARQL
queries~\cite{DBLP:conf/www/ArenasDK16}, in contrast, use an open
world semantics. The former is captured by the framework studied in
the current article. In fact, our results imply that the existence of
a separating UCQ is decidable for ontology languages such as \ALCI and
the guarded fragment. The corresponding problem for
CQs is undecidable even when the ontology is formulated in the
inexpressive description logic \ELI
\cite{DBLP:conf/ijcai/FunkJLPW19,aaaithis}, we refer the reader to Section~\ref{sec:discussion} for further discussion.

Generating referring expressions (GRE) has originated from linguistics
\cite{DBLP:journals/coling/KrahmerD12} and has also become an important challenge 
in data management. In fact, very often in applications the individual names in ontologies, knowledge graphs, or data sets are insufficient ``to allow humans to figure 
out what real objects they refer to''~\cite{DBLP:conf/ijcai/BorgidaTW17}. 
In logic-based approaches to GRE, a referring expression for an individual is a formula that distinguishes that individual from all other relevant individuals. 
GRE fits into the framework used in this paper since a formula that
separates a single data item from all other items in a KB can serve
as a referring expression for the former.
Both weak and strong
separability are conceivable: weak separability means that the
positive data item is the only one that we are certain to satisfy the
separating formula and strong separability means that in addition we
are certain that the other data items do not satisfy the
formula. Approaches to GRE such as the ones in
\cite{DBLP:conf/kr/BorgidaTW16,DBLP:journals/corr/abs-2106-15513} aim for stronger guarantees, for instance by demanding that $\Kmc\models \forall x\; ((a=x)\leftrightarrow \varphi(x))$ for any referring expression $\varphi$ for $a$ under $\Kmc$. 
We note that in a closed world context description logic concepts 
have also been proposed for singling out 
a domain element in an interpretation~\cite{DBLP:conf/inlg/ArecesKS08}.
The computation of referring expressions has recently also received 
interest in the context of ontology-mediated querying
\cite{DBLP:conf/kr/BorgidaTW16,DBLP:conf/ausai/TomanW19}. 

In \emph{entity comparison}, one aims to compare two selected data
items, highlighting both their similarities and their differences. An
approach to entity comparison in RDF graphs is presented
in~\cite{DBLP:conf/semweb/PetrovaSGH17,DBLP:conf/semweb/PetrovaKGH19}.
There, SPARQL queries are used to describe both similarities and
differences, under an open world semantics. The `computing
similarities' part of this approach is closely related to the LCS and
MSC mentioned above. The `computing differences' is closely related to
QBE and fits into the framework studied in this paper. In fact, it
corresponds to entity distinguishability with an empty ontology.

\section{Preliminaries}
\label{sec:prelim}
Let $\Sigma_{\text{full}}$ be a set of \emph{relation symbols} 
that contains countably many symbols of every arity $n\geq 1$ and let
$\text{Const}$ be a countably infinite set of \emph{constants}.  A
\emph{signature} is a set of relation symbols $\Sigma \subseteq
\Sigma_{\text{full}}$.  We write $\vec{a}$ for a tuple
$(a_{1},\ldots,a_{n})\in \text{Const}^{n}$ and set
$[\vec{a}]=\{a_{1},\ldots,a_{n}\}$.  A \emph{database} $\Dmc$ is a
finite set of \emph{ground atoms} $R(\vec{a})$, where $R\in
\Sigma_{\text{full}}$ and $\vec{a}\in \text{Const}^{n}$ with $n$ the arity of $R$. We use $\text{cons}(\Dmc)$ 
to denote the set of constant symbols in $\Dmc$.

Denote by FO the set of first-order (FO) formulas constructed from
constant-free atomic formulas $x=y$ and $R(\vec{x})$,
$R\in \Sigma_{\text{full}}$, 
using conjunction, disjunction, negation, and existential and
universal quantification.
As usual, we write $\vp(\vec{x})$ to indicate that the free variables
in the FO-formula $\vp$ are all from $\vec{x}$ and call a formula
\emph{open} if it has at least one free variable and a \emph{sentence}
otherwise. Note that we do not admit constants in FO-formulas. While
many results presented in this paper should lift to the case
with constants, dealing with constants introduces significant
technical issues that are outside the scope of this article.
We refer the reader to Section~\ref{sec:discussion} for a discussion
of separability for languages with constants.

We consider fragments of FO that are closed under conjunction in the sense that if $\varphi_{1}$ and $\varphi_{2}$ are in the fragment, then a formula logically equivalent to $\varphi_{1} \wedge \varphi_{2}$ is in the fragment.
We consider various such fragments. A \emph{conjunctive
	query (CQ)} takes the form $q(\vec{x})=\exists \vec{y}\, \varphi$
where $\varphi$ is a conjunction of atomic formulas $x=y$ and
$R(\vec{y})$. We assume w.l.o.g.\ that if a CQ contains an equality
$x=y$, then $x$ and $y$ are free variables. 
A \emph{union of conjunctive queries (UCQ)} is a disjunction of CQs
that all have
the same free variables. In the context of CQs and UCQs, we speak of
\emph{answer variables} rather than of free variables. Observe that
UCQ is closed under conjunction. A \emph{unary} UCQ is a UCQ with a single answer variable. CQs and UCQs play an important role
in database theory~\cite{DBLP:books/aw/AbiteboulHV95}.

In the \emph{guarded fragment (GF)} of
FO~\cite{ANvB98,DBLP:journals/jsyml/Gradel99}, formulas are
built from atomic formulas $R(\vec{x})$ and $x=y$ by applying the Boolean connectives and \emph{guarded quantifiers} of
the form
$$
\forall \vec{y}(\alpha(\vec{x},\vec{y})\rightarrow \varphi(\vec{x},\vec{y}))
\text{ and }
\exists \vec{y}(\alpha(\vec{x},\vec{y})\wedge \varphi(\vec{x},\vec{y}))
$$
where $\varphi(\vec{x},\vec{y})$ is a guarded formula 
and $\alpha(\vec{x},\vec{y})$ is an
atomic formula or an equality $x=y$ that contains all variables in
$[\vec{x}] \cup [\vec{y}]$. The formula $\alpha$ is called the \emph{guard of the quantifier}. GF generalizes many of the fundamental properties of modal logic and of description logics, such as decidability, the finite model property, and the tree-model property~\cite{DBLP:journals/eatcs/Gradel99}.
An even more expressive extension of GF with those properties is the \emph{guarded negation
	fragment} GNFO of FO which contains both GF and UCQ. GNFO is
obtained by imposing a guardedness condition on negation instead of on
quantifiers, details can be found
in~\cite{DBLP:journals/jacm/BaranyCS15}. Another fragment of FO that generalizes modal and description logic is the \emph{two-variable
	fragment} FO$^{2}$ of FO that contains every formula in FO that uses 
only two fixed variables, say $x$ and $y$~\cite{DBLP:journals/bsl/GradelKV97}. 
We assume that FO$^{2}$-formulas use unary and binary relation symbols
only. FO$^2$ is also decidable and enjoys the finite model property,
but it is generally regarded as a less robust generalization than the
guarded fragment as it does not enjoy a natural generalization of the tree-model property~\cite{DBLP:journals/eatcs/Gradel99}. We show in this article that this has significant repercussions also for separability.  

For \Lmc an FO-fragment, an \emph{$\Lmc$-ontology} is a finite set of
$\Lmc$-sentences. An $\Lmc$-\emph{knowledge base (KB)} is a pair
$(\Omc,\Dmc)$, where $\Omc$ is an $\Lmc$-ontology and $\Dmc$ a
database. For any syntactic object $O$ such as a formula, an ontology,
and a KB, we use $\mn{sig}(O)$ to denote the set of relation symbols
that occur in $O$ and $||O||$ to denote the \emph{size} of $O$, that
is, the number of symbols needed to write it with names of relations,
variables, and constants counting as a single symbol.

As usual, KBs $\Kmc=(\Omc,\Dmc)$ are interpreted in
\emph{structures}
$
\Amf=(\text{dom}(\Amf),(R^{\Amf})_{R\in \Sigma_{\text{full}}},(c^{\Amf})_{c \in \text{Const}})
$
where $\text{dom}(\Amf)$ is the non-empty \emph{domain} of $\Amf$,
each $R^{\Amf}$ is a relation over $\text{dom}(\Amf)$ whose arity
matches that of $R$, and $c^{\Amf} \in \text{dom}(\Amf)$ for all $c\in
\text{Const}$. Note that we do not make the \emph{unique name
	assumption (UNA)}, that is $c_1^\Amf=c_2^\Amf$ might hold even when
$c_1 \neq c_2$. This assumption is essential for several of our results and we discuss its role in detail in Section~\ref{sec:discussion}.
For a structure $\Amf$, FO formula $\varphi(\vec{x})$, and tuple $\vec{a}\in \text{dom}(\Amf)^{|\vec{x}|}$ we write $\Amf \models \varphi(\vec{a})$ if
$\varphi(\vec{x})$ is satisfied in $\Amf$ under the assignment $\vec{x}\mapsto \vec{a}$. A structure $\Amf$ is a \emph{model of a KB} $\Kmc=(\Omc,\Dmc)$ if it satisfies all sentences
in $\Omc$ and all ground atoms in $\Dmc$. A KB $\Kmc$ is
\emph{satisfiable} if there exists a model of~$\Kmc$. For a KB $\Kmc=(\Omc,\Dmc)$, formula $\varphi(\vec{x})$, and tuple $\vec{a}\in \text{cons}(\Dmc)^{|\vec{x}|}$, we say that $\varphi(\vec{a})$ \emph{is entailed by $\Kmc$}, in symbols $\Kmc\models \varphi(\vec{a})$, if $\Amf\models \varphi(\vec{a}^{\Amf})$ for every model $\Amf$ of $\Kmc$, where   $\vec{a}^{\Amf}$ stands for $(a_{1}^{\Amf},\ldots,a_{n}^{\Amf})$ if
$\vec{a}=(a_{1},\ldots,a_{n})$. 

We next introduce notation for query evaluation on KBs. Let $\Qmc$ (the query language) and $\Lmc$ (the ontology language) be fragments of FO. For example, $\mathcal{Q}$ could be the set of UCQs and $\Lmc$ could be GF. Then \emph{$\mathcal{Q}$-evaluation on $\Lmc$-KBs} is the problem to decide, given a formula (also called query) $q(\vec{x})\in \mathcal{Q}$, an $\Lmc$-KB
$\Kmc=(\Omc,\Dmc)$, and a tuple $\vec{a}\in \text{cons}(\Dmc)^{|\vec{x}|}$, whether $\Kmc \models q(\vec{a})$. The complexity of $\mathcal{Q}$-evaluation on $\Lmc$-KBs has been investigated extensively, for various query language $\mathcal{Q}$ and ontology languages $\Lmc$. 

Description logics 
are fragments of FO that only support relation symbols of arities one
and two, called concept names and role names.  DLs come with their own
syntax, which we introduce
next~\cite{handbook,DL-Textbook}. 
A \emph{role} is a role name or an \emph{inverse role} $R^-$ with $R$
a role name. For uniformity, we set $(R^{-})^{-}=R$. 
\emph{$\mathcal{ALCI}$-concepts} are defined by the
grammar
$$
C,D ~::=~\bot \mid A \mid 
\neg C \mid 
C \sqcap D \mid \exists R.C 
$$
where $A$ ranges over concept names and $R$ over roles.  As usual, we
write $\top$ for $\neg \bot$, $C \sqcup D$ for $\neg (\neg C
\sqcap \neg D)$, $C \rightarrow D$ for $\neg C \sqcup D$, and $\forall
R . C$ for $\neg \exists R. \neg C$.
An $\mathcal{ALCI}$-\emph{concept inclusion (CI)} takes the form
$C\sqsubseteq D$ where $C$ and $D$ are $\mathcal{ALCI}$-concepts. An
$\mathcal{ALCI}$-\emph{ontology} is a finite set of
$\mathcal{ALCI}$-CIs. An \emph{$\ALCI$-KB} $\Kmc=(\Omc,\Dmc)$ consists
of an $\ALCI$-ontology $\Omc$ and a database \Dmc that uses only unary
and binary relation symbols. We sometimes also mention the fragment
\ALC of \ALCI in which inverse roles are not
available. 

If $\Sigma$ is a signature of concept and role names, then an \emph{$\mathcal{ALCI}(\Sigma)$-concept} is an $\mathcal{ALCI}$-concept 
using only symbols in $\Sigma$. A \emph{$\Sigma$-role} is a role
that is either a role name in $\Sigma$ or of the form $R=S^{-}$ with
$S\in \Sigma$.
  
To obtain a semantics, every \ALCI-concept $C$ can be translated into
a GF-formula $C^\dag$ with one free variable $x$: 
see \cite{DL-Textbook}.
$$ \begin{array}{rcl} \bot^{\dag} & = & \neg(x=x)\\ A^{\dag}  &=& A(x)
  \\ (\neg \vp)^{\dag}  & = & \neg \vp^{\dag}\\ (C \sqcap D)^\dag &=&
  C^\dag \wedge D^\dag \\ (\exists R.C)^\dag &=& \exists y \, (R(x,y)
  \land C^\dag[y/x]) \\ (\exists R^-.C)^\dag &=& \exists y \, (R(y,x)
  \land C^\dag[y/x]).  \end{array} $$
%
%
A CI $C\sqsubseteq D$ translates into the GF-sentence $\forall x \,
(C^\dag(x) \to D^\dag(x))$.  By reusing variables, we can even obtain
formulas and ontologies from $\text{GF} \cap \text{FO}^2$.  It follows
that every \ALCI-concept, CI, and ontology can be viewed as a
$\text{GF}\cap \text{FO}^2$-formula, sentence, and ontology,
respectively.
The \emph{extension} $C^{\Amf}$ of a concept $C$ in a structure $\Amf$
is \mbox{defined as}
$ C^{\Amf} = \{ a\in \text{dom}(\Amf) \mid \Amf\models C^{\dag}(a)\}.
$
We write $\Omc\models C \sqsubseteq D$ if $C^\Amf
\subseteq D^\Amf$ holds in every model $\Amf$ of $\Omc$.  Concepts $C$
and $D$ are \emph{equivalent} w.r.t.\ an ontology \Omc if $\Omc
\models C\sqsubseteq D$ and $\Omc \models D\sqsubseteq C$. Let
$\Kmc=(\Omc,\Dmc)$ be an $\ALCI$-KB, $C$ an $\ALCI$-concept, and $a\in
\text{cons}(\Dmc)$. Then $C(a)$ \emph{is entailed by $\Kmc$}, in
symbols $\Kmc\models C(a)$, if $a^{\Amf}\in C^{\Amf}$ for all models
$\Amf$ of $\Kmc$. Note that $\ALCI$ enjoys the finite model property
in the sense that $\Kmc\models C(a)$ iff $a^{\Amf}\in C^{\Amf}$ for
all finite models $\Amf$ of $\Kmc$~\cite{DL-Textbook}. 

A summary of the inclusion relationships between the languages introduced above is given in Figure~\ref{fig:logics}. The figure also contains a few description logics that have not yet been introduced but which are part of the discussion of related and future work. We refer the reader to~\cite{DL-Textbook} for a detailed introduction.


	\begin{figure*}\label{fig:logics}
		\begin{center}

	\tikzset{every picture/.style={line width=0.75pt}} 
	
	\begin{tikzpicture}[x=0.75pt,y=0.75pt,yscale=-1,xscale=1]
		
		\draw    (246.97,247.15) -- (246.97,271.53) ;
		\draw    (246.97,189.09) -- (246.97,213.47) ;
		\draw    (410.81,192.1) -- (410.81,216.48) ;
		\draw    (322.15,248.61) -- (246.97,271.53) ;
		\draw    (246.97,189.09) -- (322.15,214.92) ;
		\draw    (322.15,146.49) -- (322.15,154.86) ;
		\draw    (322.15,88.43) -- (322.15,112.8) ;
		\draw    (246.97,305.21) -- (322.15,345.04) ;
		\draw    (322.15,88.43) -- (246.97,155.41) ;
		\draw    (351.24,129.64) -- (410.81,158.42) ;
		\draw    (156.15,189.43) -- (217.88,230.31) ;
		\draw    (322.15,188.55) -- (246.97,213.47) ;
		\draw    (293.06,71.58) -- (156.15,155.74) ;
		\draw    (246.97,247.15) -- (322.15,271.98) ;
		\draw   (293.06,71.58) .. controls (293.06,62.28) and (306.08,54.74) .. (322.15,54.74) .. controls (338.22,54.74) and (351.24,62.28) .. (351.24,71.58) .. controls (351.24,80.88) and (338.22,88.43) .. (322.15,88.43) .. controls (306.08,88.43) and (293.06,80.88) .. (293.06,71.58) -- cycle ;
		\draw   (293.06,129.64) .. controls (293.06,120.34) and (306.08,112.8) .. (322.15,112.8) .. controls (338.22,112.8) and (351.24,120.34) .. (351.24,129.64) .. controls (351.24,138.95) and (338.22,146.49) .. (322.15,146.49) .. controls (306.08,146.49) and (293.06,138.95) .. (293.06,129.64) -- cycle ;
		\draw   (293.06,171.7) .. controls (293.06,162.4) and (306.08,154.86) .. (322.15,154.86) .. controls (338.22,154.86) and (351.24,162.4) .. (351.24,171.7) .. controls (351.24,181.01) and (338.22,188.55) .. (322.15,188.55) .. controls (306.08,188.55) and (293.06,181.01) .. (293.06,171.7) -- cycle ;
		\draw   (217.88,230.31) .. controls (217.88,221.01) and (230.91,213.47) .. (246.97,213.47) .. controls (263.04,213.47) and (276.07,221.01) .. (276.07,230.31) .. controls (276.07,239.61) and (263.04,247.15) .. (246.97,247.15) .. controls (230.91,247.15) and (217.88,239.61) .. (217.88,230.31) -- cycle ;
		\draw   (203.73,172.25) .. controls (203.73,162.95) and (223.09,155.41) .. (246.97,155.41) .. controls (270.86,155.41) and (290.22,162.95) .. (290.22,172.25) .. controls (290.22,181.55) and (270.86,189.09) .. (246.97,189.09) .. controls (223.09,189.09) and (203.73,181.55) .. (203.73,172.25) -- cycle ;
		\draw   (381.72,175.26) .. controls (381.72,165.96) and (394.74,158.42) .. (410.81,158.42) .. controls (426.88,158.42) and (439.9,165.96) .. (439.9,175.26) .. controls (439.9,184.56) and (426.88,192.1) .. (410.81,192.1) .. controls (394.74,192.1) and (381.72,184.56) .. (381.72,175.26) -- cycle ;
		\draw   (217.88,288.37) .. controls (217.88,279.07) and (230.91,271.53) .. (246.97,271.53) .. controls (263.04,271.53) and (276.07,279.07) .. (276.07,288.37) .. controls (276.07,297.67) and (263.04,305.21) .. (246.97,305.21) .. controls (230.91,305.21) and (217.88,297.67) .. (217.88,288.37) -- cycle ;
		\draw   (293.06,361.89) .. controls (293.06,352.59) and (306.08,345.04) .. (322.15,345.04) .. controls (338.22,345.04) and (351.24,352.59) .. (351.24,361.89) .. controls (351.24,371.19) and (338.22,378.73) .. (322.15,378.73) .. controls (306.08,378.73) and (293.06,371.19) .. (293.06,361.89) -- cycle ;
		\draw   (293.06,288.83) .. controls (293.06,279.52) and (306.08,271.98) .. (322.15,271.98) .. controls (338.22,271.98) and (351.24,279.52) .. (351.24,288.83) .. controls (351.24,298.13) and (338.22,305.67) .. (322.15,305.67) .. controls (306.08,305.67) and (293.06,298.13) .. (293.06,288.83) -- cycle ;
		\draw   (293.06,231.77) .. controls (293.06,222.46) and (306.08,214.92) .. (322.15,214.92) .. controls (338.22,214.92) and (351.24,222.46) .. (351.24,231.77) .. controls (351.24,241.07) and (338.22,248.61) .. (322.15,248.61) .. controls (306.08,248.61) and (293.06,241.07) .. (293.06,231.77) -- cycle ;
		\draw   (381.72,233.32) .. controls (381.72,224.02) and (394.74,216.48) .. (410.81,216.48) .. controls (426.88,216.48) and (439.9,224.02) .. (439.9,233.32) .. controls (439.9,242.62) and (426.88,250.17) .. (410.81,250.17) .. controls (394.74,250.17) and (381.72,242.62) .. (381.72,233.32) -- cycle ;
		\draw    (322.15,305.67) -- (322.15,345.04) ;
		\draw   (127.06,172.58) .. controls (127.06,163.28) and (140.08,155.74) .. (156.15,155.74) .. controls (172.22,155.74) and (185.24,163.28) .. (185.24,172.58) .. controls (185.24,181.88) and (172.22,189.43) .. (156.15,189.43) .. controls (140.08,189.43) and (127.06,181.88) .. (127.06,172.58) -- cycle ;
		\draw    (410.81,250.17) -- (322.15,271.98) ;
		
		\draw (310.9,64.64) node [anchor=north west][inner sep=0.75pt]  [font=\normalsize] [align=left] {$\displaystyle \text{FO}$};
		\draw (143.41,162.27) node [anchor=north west][inner sep=0.75pt]  [font=\normalsize] [align=left] {$\displaystyle \text{FO}^{2}$};
		\draw (311.99,165.92) node [anchor=north west][inner sep=0.75pt]  [font=\normalsize] [align=left] {$\displaystyle \text{GF}$};
		\draw (304.72,121.4) node [anchor=north west][inner sep=0.75pt]  [font=\normalsize] [align=left] {$\displaystyle \text{GNF}$};
		\draw (400.3,227.25) node [anchor=north west][inner sep=0.75pt]  [font=\normalsize] [align=left] {$\displaystyle \text{CQ}$};
		\draw (224.1,224.09) node [anchor=north west][inner sep=0.75pt]  [font=\normalsize] [align=left] {$\displaystyle \mathcal{ALCI}$};
		\draw (307.04,281.59) node [anchor=north west][inner sep=0.75pt]  [font=\normalsize] [align=left] {$\displaystyle \mathcal{ELI}$};
		\draw (298.81,224.53) node [anchor=north west][inner sep=0.75pt]  [font=\normalsize] [align=left] {$\displaystyle \mathcal{ALCQ}$};
		\draw (218.99,165.15) node [anchor=north west][inner sep=0.75pt]  [font=\normalsize] [align=left] {$\displaystyle \mathcal{ALCQI}$};
		\draw (393.76,168.04) node [anchor=north west][inner sep=0.75pt]  [font=\normalsize] [align=left] {$\displaystyle \text{UCQ}$};
		\draw (310.76,353.33) node [anchor=north west][inner sep=0.75pt]  [font=\normalsize] [align=left] {$\displaystyle \mathcal{EL}$};
		\draw (229.82,281.21) node [anchor=north west][inner sep=0.75pt]  [font=\normalsize] [align=left] {$\displaystyle \mathcal{ALC}$};

	\end{tikzpicture}
	\caption{Relationship between languages.} 
\end{center}
\end{figure*}

We next introduce the Gaifman graph of a structure~\cite{Libkin2004} and notation for homomorphisms. Both play a prominent role throughout the paper. Let $\Amf$ be a structure. The \emph{Gaifman graph} $G_{\Amf}$ of $\Amf$
has the set of vertices $\text{dom}(\Amf)$ and an edge $\{ d,e \}$ whenever there exists $\vec{a}\in R^{\Amf}$
containing $d,e$ for some relation $R$. We often apply graph-theoretic terminology to structures, in the obvious way. For example, a node $b$ is \emph{reachable} from a node $a$ in $\Amf$ if there is a path from $a$ to $b$ in the Gaifman graph of $\Amf$. The \emph{distance} $\text{dist}_{\Amf}(a,b)$ between $a$ and $b$ is the length of the shortest path from $a$ to $b$, if there exists one.
The \emph{outdegree} of $\Amf$ is defined as the degree of its Gaifman graph.
A \emph{homomorphism} $h$ from a
structure $\Amf$ to a structure $\Bmf$ is a function
$h:\text{dom}(\Amf)
\rightarrow \text{dom}(\Bmf)$ such that $\vec{a}\in
R^{\Amf}$ implies $h(\vec{a})\in R^{\Bmf}$ for all relation symbols
$R$ and tuples $\vec{a}\in \text{dom}(\Amf)^{n}$ with $n$ the arity of $\vec{a}$ and $h(\vec{a})$ being defined
component-wise in the expected way. 
Note that homomorphisms need not preserve 
constant symbols.  Every database $\Dmc$ gives rise to the finite
structure $\Amf_{\Dmc}$ with $\text{dom}(\Amf_\Dmc)=\text{cons}(\Dmc)$
and $\vec{a}\in R^{\Amf_{\Dmc}}$ iff $R(\vec{a})\in \Dmc$. A
homomorphism from database $\Dmc$ to structure $\Amf$ is a
homomorphism from $\Amf_{\Dmc}$ to $\Amf$.  A \emph{pointed structure}
takes the form $\Amf,\vec{a}$ with \Amf a structure and $\vec{a}$ a
tuple of elements of $\text{dom}(\Amf)$.  
%
A homomorphism
from $\Amf,\vec{a}$ to pointed structure
$\Bmf,\vec{b}$ is a homomorphism $h$ from $\Amf$ to $\Bmf$ with
$h(\vec{a})=\vec{b}$.  We write
$\Amf,\vec{a} \rightarrow \Bmf,\vec{b}$ if
such a homomorphism exists. We use the same notation for databases.
For example, a \emph{pointed
database} is a pair $\Dmc,\vec{a}$ with $\vec{a}$ a tuple in $\Dmc$ and a homomorphism from $\Dmc,\vec{a}$ to pointed structure $\Amf,\vec{b}$ is a homomorphism from
$\Amf_{\Dmc},\vec{a}$ to $\Amf,\vec{b}$. We write $\Dmc,\vec{a}\not\rightarrow \Amf,\vec{b}$ if no such homomorphism exists. 




We introduce a few fundamental properties of logics. A fragment $\Lmc$ of FO enjoys the \emph{relativization property}~\cite{modeltheory} if for every \Lmc-sentence $\varphi$ and unary relation symbol $A \notin
\mn{sig}(\varphi)$, there exists a sentence $\varphi'$ 
such that for every structure $\Amf$ with $A^{\Amf}\not=\emptyset$, $\Amf\models \varphi'$ iff
$\Amf_{|A}\models \varphi$, where $\Amf_{|A}$ is the restriction of $\Amf$ to domain~$A^\Amf$. To illustrate, given an $\ALCI$-CI $C \sqsubseteq D$,
then $C_{|A}\sqsubseteq D_{|A}$ is as required, where $C_{|A}$ is defined inductively by setting $\bot_{|A}=\bot$, $B_{|A}=B \sqcap A$, $(\neg C)_{|A}=A\sqcap \neg C_{|A}$, $(C\sqcap D)_{|A} = C_{|A} \sqcap D_{|A}$,
and $(\exists R.C)_{|A}= A \sqcap \exists R.C_{|A}$. All languages depicted in
Figure~\ref{fig:logics} enjoy the relativization property.

A fragment $\Lmc$ of FO enjoys the \emph{finite model property} if for every $\Lmc$-KB $\Kmc$, $\Lmc$-formula $\varphi(\vec{x})$, and tuple $\vec{a}\in \text{cons}(\Dmc)^{|\vec{x}|}$,  $\Kmc\models \varphi(\vec{a})$ iff $\Amf\models \varphi(\vec{a}^{\Amf})$ for every finite model $\Amf$ of $\Kmc$. All logics in Figure~\ref{fig:logics} with the exception of $\mathcal{ALCQI}$ and FO enjoy the finite model property~\cite{DBLP:journals/bsl/GradelKV97,DBLP:journals/jsyml/Gradel99,DBLP:journals/jsyml/BaranyBC18,DBLP:journals/mlq/Mortimer75}.

Finally, we also require a version of the finite model property for query evaluation. Let \Lmc be a fragment of FO. Evaluating queries from a query
language $\mathcal{Q}$ contained in $\text{FO}$ is \emph{finitely controllable} on
\Lmc-KBs if for every $\Lmc$-ontology $\Omc$, database $\Dmc$,
formula $\varphi(\vec{x})$ in $\mathcal{Q}$, tuple of constants $\vec{a}\in \text{cons}(\Dmc)^{|\vec{x}|}$, if $(\Omc,\Dmc)\not\models \varphi(\vec{a})$, then there is a finite model \Amf of $\Kmc$ such that $\Amf\not\models\varphi(\vec{a}^{\Amf})$~\cite{DBLP:journals/jcss/JohnsonK84,DBLP:journals/jcss/Rosati11}. 
Note that \Lmc has the finite model property if evaluating
queries from \Lmc is finitely controllable on \Lmc-KBs. Evaluating CQs and UCQs is finitely controllable on $\Lmc$-KBs for all $\Lmc$ depicted in Figure~\ref{fig:logics} with the exception of $\mathcal{ALCQI}$, FO$^{2}$, and FO~\cite{DL-Textbook,DBLP:journals/corr/BaranyGO13,DBLP:journals/jcss/Rosati11}.
It follows that FO$^{2}$ is the only example of a language in Figure~\ref{fig:logics}
that enjoys the finite model property but for which evaluating CQs and UCQs is not finitely controllable.

\section{Fundamental Results for Weak Separability}\label{sec:weak}

\newcommand{\mLOKB}{labeled $\Lmc$-KB}

We introduce the problem of (weak) separability in its
projective and non-projective version. We also discuss special cases 
of separability such as definability, referring expression existence, and entity distinguishability. We then give a fundamental
characterization of
$(\text{FO},\text{FO})$-separability which 
has the consequence that UCQs have the same separating
power as FO. This allows us to settle the complexity of
deciding separability in GNFO.
%
%
\newcommand{\LmcO}{\Lmc}
\begin{definition}
	\label{def:separa}
	Let $\LmcO$ be a fragment of FO. A \emph{\mLOKB} takes the form
	$(\Kmc,P,N)$ with $\Kmc=(\Omc,\Dmc)$ an $\LmcO$-KB and
	$P,N\subseteq \text{cons}(\Dmc)^{n}$ non-empty sets of \emph{positive}
	and \emph{negative examples}, all of them tuples of the same length
	$n$. 
	
	%
	An FO-formula $\varphi(\vec{x})$ with $n$ free variables
	\emph{(weakly) separates $(\Kmc,P,N)$} if
	\begin{enumerate}
		
		\item 
		$\Kmc\models \varphi(\vec{a})$ for all $\vec{a}\in P$ and 
		
		\item $\Kmc\not\models \varphi(\vec{a})$ for all $\vec{a}\in N$.
		
	\end{enumerate}
	%
	%
	Let $\Lmc_S$ be a fragment of FO. We say that $(\Kmc,P,N)$ is
	\emph{projectively $\Lmc_S$-separable} if there is an
	$\Lmc_S$-formula $\varphi(\vec{x})$ that separates $(\Kmc,P,N)$ and
	\emph{(non-projectively) $\Lmc_S$-separable} if there is such a
	$\varphi(\vec{x})$ with $\mn{sig}(\varphi) \subseteq
	\mn{sig}(\Kmc)$.
	%
\end{definition}
The following example illustrates the definition.
\begin{exmp}\label{exmp:44}
	{\em Let $\Kmc_1=(\emptyset,\Dmc_{1})$ where 
$$
		\Dmc_{1}   =  \{{\sf born\_in}(a,c), \ {\sf citizen\_of}(a,c), \ {\sf born\_in}(b,c_{1}), \ {\sf citizen\_of}(b,c_{2}), \ {\sf Person}(a)\}.
$$
		Then ${\sf Person}(x)$ separates $(\Kmc_{1},\{a\},\{b\})$. As any
		citizen is a person,
		however,
		this
		separating formula is 
		not natural and it only separates because of incomplete information about~$b$. 
		This may change with knowledge from the ontology. Let  
		$$
		\Omc_{2}=\{ \forall x (\forall y ({\sf citizen\_of}(x,y) \rightarrow {\sf
			Person}(x)))\}
		$$
		and $\Kmc_2=(\Omc_{2},\Dmc_{1})$. (It will later be
		useful to note that the ontology $\Omc_{2}$ can also
		be regarded as an $\ALCI$-ontology since $\forall x
		(\forall y ({\sf citizen\_of}(x,y) \rightarrow {\sf
			Person}(x)))$ is equivalent to the CI $\exists {\sf citizen\_of}.\top \sqsubseteq {\sf Person}$.)  We have $\Kmc_2\models {\sf Person}(b)$ and so
		${\sf Person}(x)$ no longer separates. However, the more natural formula 
		$$
		\varphi(x)= \exists y({\sf born\_in}(x,y) \wedge {\sf citizen\_of}(x,y)),
		$$
		separates $(\Kmc_{2},\{a\},\{b\})$.  Thus
		$(\Kmc_{2},\{a\},\{b\})$ is non-projectively
		\Lmc-separable for
		$\Lmc\in\{\text{CQ},\text{GF},\text{FO}^2\}$.
	
In contrast, it will follow from the model-theoretic characterization given below (Example~\ref{ex:mmm}) that the labeled KB $(\Kmc_{2},\{b\},\{a\})$ in which $b$ is the positive example and $a$ the negative example is not FO-separable.
Note that this observation rests on admitting models $\Amf$ of $\Kmc_{2}$ with $c_{1}^{\Amf}=c_{2}^{\Amf}$ (no UNA) because otherwise 
$$
\exists y_{1}\exists y_{2} (y_{1}\not=y_{2} \wedge {\sf born\_in}(x,y_{1})\wedge {\sf citizen\_of}(x,y_{2}))
$$
would be a separating formula.\finofex}
\end{exmp}
In the projective case, one admits symbols that are not from $\mn{sig}(\Kmc)$ as 
helper symbols in separating formulas.
Their availability sometimes makes inseparable KBs separable, for some separation languages. The following
example illustrates the role of helper symbols. 
Note
that in \cite{DBLP:conf/ijcai/FunkJLPW19}, helper symbols are generally admitted
and the results depend on this assumption.

\begin{exmp}\label{ex:ex2}
	{\em The separating formula $\varphi(x)$ in Example~\ref{exmp:44}
		cannot be expressed as an $\ALCI$-concept. Using a helper concept name $A$, however, we obtain the separating concept
		$$
		C= \forall {\sf born\_in}.A \rightarrow \exists {\sf citizen\_of}.A.
		$$
		and thus $(\Kmc_2,\{a\},\{b\})$ is projectively \ALCI-separable. Note that $\Kmc_{2}\models C(a)$ because whenever $a^{\Amf}\in (\forall {\sf born\_in}.A)^{\Amf}$ for a model $\Amf$ of $\Kmc_{2}$, then $c^{\Amf}\in A^{\Amf}$ and so $a^{\Amf}\in (\exists {\sf citizen\_of}.A)^{\Amf}$. On the other hand, $\Kmc_{2}\not\models C(b)$ because the structure $\Amf$ obtained by adding ${\sf Person}(b)$ and $A(c_{1})$ to $\Dmc_{1}$ is a model of $\Kmc_{2}$ with $b^{\Amf}\not\in C^{\Amf}$.
		For separation, it is thus important
		that $A$ is not constrained by $\Omc_{2}$. For a natural interpretation of the separating concept $C$ one should view $A$ as a universally quantified second-order variable. Then $\Kmc_{2}\models C(a)$ can be understood as: `for any class $A$, it follows from $\Kmc_{2}$ that $C(a)$'. 
		Finally note that ${\sf Person}$ is
		a concept name that, despite being in $\text{sig}(\Kmc_{2})$, is also sufficiently
		unconstrained by $\Omc_{2}$ to act as a helper symbol: by replacing $A$
		by ${\sf Person}$ in $C$, one obtains a (rather unnatural) concept
		that witnesses also non-projective
		\ALCI-separability of $(\Kmc_2,\{a\},\{b\})$. An example of a labeled KB that is projectively $\ALCI$-separable but not non-projectively $\ALCI$-separable was given in Example~\ref{exm:new} and will be discussed further in Example~\ref{exm:11} below. \finofex }
\end{exmp}
%
Each choice of an ontology language $\LmcO$ and a separation
language $\Lmc_{S}$ give rise to a separability problem and a
projective separability problem, defined as follows.

\vspace*{-2mm}
\begin{center}
	\fbox{\begin{tabular}{@{\;}ll@{\;}}
			\small{PROBLEM} : & (Projective) $(\LmcO,\Lmc_{S})$-separability\\
			{\small INPUT} : &  A \mLOKB\xspace $(\Kmc,P,N)$ \\
			{\small QUESTION} : & Is $(\Kmc,P,N)$ (projectively) 
			$\Lmc_S$-separable?  \end{tabular}}
\end{center}

\smallskip
\noindent
We study both the \emph{combined complexity} and the \emph{data
	complexity} of separability. In the former, the full labeled KB $(\Kmc,P,N)$ with $\Kmc=(\Omc,\Dmc)$ is taken as the input. In
the latter, only $\Dmc$ and the examples $P,N$ are regarded as 
the input while $\Omc$ is assumed to be fixed.  
We next introduce important special cases of $(\LmcO,\Lmc_{S})$-separability:
\begin{itemize}
	\item $(\LmcO,\Lmc_{S})$-\emph{definability} is 
	$(\LmcO,\Lmc_{S})$-separability for labeled KBs where $P$ and $N$ partition the
example space, that is, inputs are {\mLOKB}s $(\Kmc,P,N)$ such that
$N=\text{cons}(\Dmc)^n \setminus P$, $n$ the length of example tuples. 
     \item $(\LmcO,\Lmc_{S})$-\emph{referring expression existence (RE-existence)}  is $(\LmcO,\Lmc_{S})$-definability for labeled KBs where $P$
     is a singleton set.
     \item $(\LmcO,\Lmc_{S})$-\emph{entity distinguishability}
     is $(\LmcO,\Lmc_{S})$-separability for labeled KBs where $P$ and $N$ are both singleton sets.
     \end{itemize}
Note that for definability and referring expression existence we could
remove the set $N$ of negative examples from the input as it is
uniquely determined by $\Dmc$ and $P$. As the length of example tuples
in $P$ is not bounded (and so the size of $\text{cons}(\Dmc)^n
\setminus P$ is not polynomial in the size of $\text{cons}(\Dmc)$ and
$P$) this could affect the complexity of deciding definability and/or referring expression existence. In our proofs below we make sure that we show the complexity bounds independently from the representation.

Note also that as we only study FO-fragments $\Lmc_{S}$ that are closed under
conjunction, a labeled KB $(\Kmc,P,N)$ is (projectively)
$\Lmc_{S}$-separable if and only if all $(\Kmc,P,\{\vec{b}\})$,
$\vec{b}\in N$, are (projectively) $\Lmc_{S}$-separable. In fact, a
formula that separates $(\Kmc,P,N)$ can be obtained by taking the
conjunction of formulas that separate $(\Kmc,P,\{\vec{b}\})$,
$\vec{b}\in N$. When formulating semantic characterizations we thus often consider labeled KBs with single
negative examples only. In terms of computational complexity, we have the following reduction.
\begin{observation}\label{obs:red}
  Let $\Lmc$ and $\Lmc_{S}$ be fragments of FO with $\Lmc_{S}$ closed under conjunction. Then there is a polynomial time Turing reduction of (projective) $(\Lmc,\Lmc_S)$-separability to (projective) $(\Lmc,\Lmc_S)$-separability with a single negative example.
\end{observation}
Our first result provides a characterization of
$(\text{FO},\text{FO})$-separability 
in terms of homomorphisms, linking it to UCQ-separability and in fact
to UCQ-evaluation on KBs.
We first give some preliminaries.  With every pointed
database $\Dmc,\vec{a}$, where $\vec{a}=(a_1,\dots,a_n)$, we associate
a CQ $\varphi_{\Dmc,\vec{a}}(\vec x)$ with free variables
$\vec{x}=(x_1,\dots,x_n)$ that is obtained from $\Dmc,\vec{a}$ as
follows: view each $R(c_{1},\ldots,c_{m})\in\Dmc$ as an atom
$R(x_{c_{1}},\ldots,x_{c_{m}})$, existentially quantify all variables
$x_c$ with $c \in \text{cons}(\Dmc) \setminus [\vec{a}]$, replace
every variable $x_c$ such that $a_{i}=c$ for some $i$ with the
variable $x_i$ such that $i$ is minimal with $a_i=c$, and finally add
$x_i=x_j$ whenever $a_i=a_j$. Then $\varphi_{\Dmc,\vec{a}}(\vec x)$ is the logically strongest CQ such that $\Dmc\models \varphi_{\Dmc,\vec{a}}(\vec a)$. This link between conjunctive queries and pointed databases is well known~\cite{DBLP:conf/stoc/ChandraM77}.
 For a pointed database $\Dmc,\vec{a}$,
we write $\Dmc_{\text{con}(\vec{a})}$ to denote the restriction of
\Dmc to those constants that are reachable from some constant in
$\vec{a}$ in the Gaifman graph of~$\Dmc$. Equivalently, $R(\vec{b})\in \Dmc_{\text{con}(\vec{a})}$ iff $R(\vec{b})\in \Dmc$ and there exists $a\in [\vec{a}]$ such that there exists a path from $a$ to some $b\in [\vec{b}]$ in $\Dmc$.
\begin{exmp}\label{ex:m}
	{\em Consider the database $\Dmc_{1}$ introduced in Example~\ref{exmp:44}.
		Then 
		$$
		\varphi_{{\Dmc_{1}}_{\text{con}(a)},a}(x) = \exists y({\sf born\_in}(x,y) \wedge {\sf citizen\_of}(x,y) \wedge {\sf Person}(x))
		$$ 
		extends the formula we identified as a separating formula for $(\Kmc_{2},\{a\},\{b\})$ in the example with the atom ${\sf Person}(x)$. 
	
    Note that $\varphi_{{\Dmc_{1}}_{\text{con}(b)},b}(x) = \exists y_{1}\exists y_{2}({\sf born\_in}(x,y_{1}) \wedge {\sf citizen\_of}(x,y_{2}))$.\finofex
}
\end{exmp}
We are now in a position to formulate our first result. 
\begin{theorem}\label{critFOwithoutUNA}
	Let $(\Kmc,P,\{\vec{b}\})$ be a labeled FO-KB,
	$\Kmc=(\Omc,\Dmc)$. 
	Then the following
	conditions are equivalent:
	\begin{enumerate}
		\item $(\Kmc,P,\{\vec{b}\})$ is projectively UCQ-separable;
		\item $(\Kmc,P,\{\vec{b}\})$ is projectively FO-separable;
		\item there exists a model $\Amf$ of $\Kmc$ such that for all $\vec{a}\in P$:
		$\Dmc_{\text{con}(\vec{a})},\vec{a}\not\rightarrow \Amf,\vec{b}^{\Amf}$;
		\item the UCQ $\bigvee_{\vec{a}\in P}\varphi_{\Dmc_{\text{con}(\vec{a})},\vec{a}}$ separates
		$(\Kmc,P,\{\vec{b}\})$.
	\end{enumerate} 
\end{theorem}
\begin{proof} \ ``1. $\Rightarrow$ 2.'' and ``4. $\Rightarrow$
	1.'' are trivial and ``$3. \Rightarrow$ 4.'' is
	straightforward. We thus concentrate on ``2. $\Rightarrow$
	3.'' Assume that $(\Kmc,P,\{\vec{b}\})$ is separated by an
	FO-formula $\varphi(\vec{x})$. Then there is a model $\Amf$
	of $\Kmc$ such that $\Amf \not\models \varphi(\vec{b}^{\Amf})$. Let
	$\vec{a}\in P$. Since $\Kmc \models \varphi(\vec{a})$, there
	is no model $\Bmf$ of $\Kmc$ and such that
	$\Bmf,\vec{a}^{\Bmf}$ and $\Amf,\vec{b}^{\Amf}$ are
	isomorphic, meaning that there is an isomorphism $\tau$ from
	\Bmf to \Amf with $\tau(\vec{a}^\Bmf)=\vec{b}^\Amf$.  We show
	that $\Amf$
	satisfies Condition~3. Assume to the contrary that there is
	a homomorphism $h$ from $\Dmc_{\text{con}(\vec{a})},\vec{a}$
	to $\Amf,\vec{b}^{\Amf}$ for some $\vec{a}\in P$.  Let the
	structure $\Bmf$ be obtained from \Amf by setting
	$c^{\Bmf}=h(c)$ for all $c\in
	\text{cons}(\Dmc_{\text{con}(\vec{a})})$ and
	$c^{\Bmf}=c^{\Amf}$ for all remaining constants $c$.  This
	construction relies on not making the UNA.  \Bmf is a model
	of \Kmc since \Omc does not contain constants.  It is easy
	to verify that $\Bmf,\vec{a}^{\Bmf}$ and
	$\Amf,\vec{b}^{\Amf}$ are isomorphic and thus we have
	obtained a contradiction.
\end{proof}
Note that the UCQ in Point~4 of Theorem~\ref{critFOwithoutUNA}
is a concrete separating formula. It is only of size
polynomial in the size of the KB, but often not very
illuminating.  It also contains no helper
symbols (in fact, it even contains only relation
	symbols that occur in $\Dmc$ while symbols that only occur
	in \Omc are not used) and thus we obtain the following.
\begin{corollary}\label{cor:weak1}
	$(\text{FO},\text{FO})$-separability, 
	$(\text{FO},\Lmc_S)$-separability, and projective $(\text{FO},\Lmc_{S})$-separability coincide for all
	FO-fragments $\Lmc_S \supseteq \text{UCQ}$. The same is true for definability.
\end{corollary}
The UCQ in Point~4 of Theorem~\ref{critFOwithoutUNA} is a CQ if $P$ contains a single example. Thus we obtain the following result for referring expression existence and entity distinguishability.
\begin{corollary}\label{cor:weak2}
	 $(\text{FO},\text{FO})$-entity distinguishability, 
	 $(\text{FO},\Lmc_S)$-entity distinguishability, and
	 projective $(\text{FO},\Lmc_S)$-entity distinguishability coincide for all
	 FO-fragments $\Lmc_S \supseteq \text{CQ}$. The same is true for RE-existence.
\end{corollary}
The following example applies the characterization given in Theorem~\ref{critFOwithoutUNA} to
Examples~\ref{exmp:44} and \ref{ex:m}. 
\begin{exmp}\label{ex:mmm}
	{\em Consider the labeled KB $(\Kmc_{2},\{a\},\{b\})$ introduced in Example~\ref{exmp:44}. Point~4 of Theorem~\ref{critFOwithoutUNA} shows that in Example~\ref{exmp:44} we have chosen the canonical UCQ $\varphi_{{\Dmc_{1}}_{\text{con}(a)},a}(x)$ (which is actually a CQ since we have only one positive example) that separates $(\Kmc_{2},\{a\},\{b\})$ if there is an FO-formula at all that separates $(\Kmc_{2},\{a\},\{b\})$.
	
	It follows from Point~3 of Theorem~\ref{critFOwithoutUNA} that
	the labeled KB $(\Kmc_{2},\{b\},\{a\})$ introduced in Example~\ref{exmp:44}
	is not FO-separable: for any model $\Amf$ of $\Kmc_{2}$ there is a homomorphism $h$ from ${\Dmc_{1}}_{\text{con}(b)}$ to $\Amf$ mapping $b$ to $a^{\Amf}$ as one can
	always set $h(c_{1})=h(c_{2})=c^{\Amf}$. \finofex
}
\end{exmp}
We next observe a general result about the impact of strengthening the ontology on the separability of labeled KBs. Say that $(\Lmc,\Lmc_{S})$-separability is \emph{anti-monotone in the ontology} if for all labeled $\Lmc$-KBs $(\Kmc_{i},P,N)$ with $\Kmc_{i}=(\Omc_{i},\Dmc)$ for $i=1,2$, if $\Omc_{1} \subseteq \Omc_{2}$ and $(\Kmc_{2},P,N)$ is $\Lmc_{S}$-separable, then $(\Kmc_{1},P,N)$ is $\Lmc_{S}$-separable. Then the following result follows from Point~3 of Theorem~\ref{critFOwithoutUNA}.
\begin{corollary}\label{cor:antimon}
	Projective and non-projective $(\text{FO},\Lmc_S)$-separability are anti-monotone in the ontology, for all
	FO-fragments $\Lmc_S \supseteq \text{UCQ}$.
\end{corollary}	
Note that Corollary~\ref{cor:antimon} \emph{never} holds in the non-projective
case if that case differs from the projective case. The argument is straightforward. If $(\Kmc,P,N)$ is $\Lmc_{S}$-separable using helper symbols but not without helper symbols, then by adding tautologies using the helper symbols used in the separating formula to the ontology of $\Kmc$ makes $(\Kmc,P,N)$ $\Lmc_{S}$-separable. Note also that Corollary~\ref{cor:antimon} has no counterpart for the database of the labeled KB as adding ground atoms to the database can clearly both make non-separable labeled KBs separable but also separable labeled KBs non-separable. 
 
Theorem~\ref{critFOwithoutUNA} shows that there is a very close link between separability and query evaluation on KBs. In particular, it implies that for all
$(\LmcO,\Lmc_S)$ with $\LmcO$ and $\Lmc_{S}$ fragments of FO such that $\Lmc_S
\supseteq \text{UCQ}$, $(\LmcO,\Lmc_S)$-separability can be
polynomially reduced to the complement of UCQ-evaluation on $\LmcO$-KBs. 
We can actually
do better than this: call a CQ \emph{rooted} if every variable is reachable from an answer variable. A UCQ is called \emph{rooted} if all its CQs are rooted. Then the CQs $\varphi_{\Dmc_{\text{con}(\vec{a}),\vec{a}}}$
are rooted and so are the UCQs 
$\bigvee_{\vec{a}\in P}\varphi_{\Dmc_{\text{con}(\vec{a}),\vec{a}}}$.
Theorem~\ref{critFOwithoutUNA} now implies the reductions from left to right of the following corollary. For the converse direction, we require that the fragment $\Lmc$ enjoys the relativization property introduced above.
\begin{corollary}\label{cor:rel}
Let $\Lmc$ be a fragment of FO enjoying the relativization property.
\begin{enumerate}
	\item For all fragments $\Lmc_S \supseteq \text{UCQ}$ of FO, $(\LmcO,\Lmc_S)$-separability for labeled KBs with a single negative example can be
mutually polynomial time reduced with the complement of
rooted UCQ-evaluation on $\LmcO$-KBs.  
\item For all fragments $\Lmc_S\supseteq \text{CQ}$ of FO,
$(\Lmc,\Lmc_S)$-entity distinguishability can be
mutually polynomial time reduced with the complement of
rooted CQ-evaluation on $\LmcO$-KBs.
\end{enumerate}
For the reduction from left to right the relativization property can be dropped.  
\end{corollary}
\begin{proof} \
We show for Point~1 the reduction from rooted UCQ-evaluation to separability. The reduction from rooted CQ-evaluation to entity distinguishability in Point~2 follows directly from the proof. Assume an \Lmc-KB $\Kmc=(\Omc,\Dmc)$, 
	a rooted UCQ $q(\vec{x})=\bigvee_{i\in
		I}q_{i}(\vec{x})$, $\vec{x}=(x_{1},\ldots,x_{n})$, and a tuple $\vec{a}$ in $\Dmc$ are given. Consider the relativization $\Omc_{|A}$ of the sentences of $\Omc$ to $A$
	    and $\Dmc^{+A}=\Dmc \cup \{A(c) \mid c\in \text{cons}(\Dmc)\}$, for a
	    fresh unary relation $A$.
	Regard each CQ $q_{i}(x_{1},\ldots,x_{n})$ as a pointed
	database $\Dmc_{i},([x_{1}],\ldots,[x_{n}])$ as follows. Let
	$\sim$ be the smallest (in terms of containment) equivalence
	relation that contains $(x,y)$ for every conjunct $(x=y)$ of
	$q_i$.
%
	Then regard the equivalence classes $[x]$ as constants
	and set $R([y_{1}],\ldots,[y_{m}])\in \Dmc_{i}$ iff there are $y_1'\in [y_{1}],\ldots,y_{m}'\in [y_{m}]$
	such that $R(y_{1}',\ldots,y_{m}')$ is a conjunct of $q_{i}$. We assume the 
	pointed databases $\Dmc_{i},([x_{1}],\ldots,[x_{n}])$, $i\in I$,
	are mutually disjoint and also disjoint from $\Dmc$. The copy of $([x_{1}],\ldots,[x_{n}])$ in $\Dmc_{i}$ is denoted 
	$([x_{1}]^{i},\ldots,[x_{n}]^{i})$.
	Let $\Dmc'=\Dmc^{+A} \cup \bigcup_{i\in I}\Dmc_{i}$ and set
	\[P=
	\{([x_{1}]^{i},\ldots,[x_{n}]^{i})\mid i \in I\}, \quad N = \{\vec{a}\}.\]
	Consider the labeled KB $(\Kmc',P,N)$ for $\Kmc'=(\Omc_{|A},\Dmc')$. (Intuitively, the reason for replacing $\Omc$ by $\Omc_{|A}$ is that we do not want the new database assertions in $\bigcup_{i\in I}\Dmc_{i}$ to interfere with what is entailed at $\vec{a}$. For example, $(\Omc,\Dmc')$ could be unsatisfiable.) We show that $(\Kmc',P,N)$ is UCQ-separable iff $\Kmc \not\models q(\vec{a})$. By Theorem~\ref{critFOwithoutUNA}, $(\Kmc',P,N)$ is UCQ-separable iff $\Kmc'\not\models q(\vec{a})$. Hence we have to show that $\Kmc \models q(\vec{a})$ iff
	$\Kmc' \models q(\vec{a})$. The direction from left to right follows from the condition that $\Omc_{|A}$ is a relativization of $\Omc$. Conversely, assume that $\Kmc \not\models q(\vec{a})$. Take a model $\Amf$ of $\Kmc$ with $\Amf\not\models q(\vec{a})$. We may assume that $A^{\Amf}= \text{dom}(\Amf)$. We expand $\Amf$ to a model $\Amf'$ of $\Kmc'$ by taking the disjoint union of $\Amf$ with $\Dmc'$ (regarded as a structure). Using the condition that $q$ is rooted it now follows directly that $\Amf'\not\models q(\vec{a})$. Hence $\Kmc' \not\models q(\vec{a})$.
\end{proof}
Since rooted UCQ-evaluation on FO-KBs is undecidable, so is
$(\text{FO},\text{FO})$-separability.  However, for GNFO we obtain the following result from $\text{GNFO}
\supseteq \text{UCQ}$. 
\begin{corollary}
	\label{thm:gnfo}
	For all FO-fragments $\Lmc_{S}\supseteq \text{UCQ}$:
	\begin{enumerate}
		\item $(\text{GNFO},\text{GNFO})$-separability, $(\text{GNFO},\Lmc_S)$-separability, and projective $(\text{GNFO},\Lmc_S)$-separability coincide. The same is true for definability.
		\item $(\text{GNFO},\Lmc_S)$-separability and $(\text{GNFO},\Lmc_S)$-definability are 
	\TwoExpTime-complete in combined complexity.
	\end{enumerate}
	Points~1 and~2 also hold for RE-existence and entity distinguishability, for
	all FO-fragments $\Lmc_{S}\supseteq \text{CQ}$.
\end{corollary}
\begin{proof} \
It is known that UCQ-evaluation on GNFO-KBs is decidable in \TwoExpTime in combined complexity~\cite{DBLP:journals/jacm/BaranyCS15}. This implies the 
\TwoExpTime upper bound in combined complexity.
\TwoExpTime-hardness in combined complexity can be shown by a polynomial time reduction of satisfiability (in a model $\Amf$ with $|\text{dom}(\Amf)|\geq 2$) of GNFO-sentences which is also \TwoExpTime-hard \cite{DBLP:journals/jacm/BaranyCS15}. To show the reduction observe that
a GNFO-sentence $\vp$ is satisfiable in a model $\Amf$ with $|\text{dom}(\Amf)|\geq 2$ iff the labeled KB
$(\Kmc,\{a\},\{b\})$ is (projectively and, equivalently, non-projectively) separable
in any FO-fragment $\Lmc_{S}\supseteq \text{CQ}$, where $\Kmc=(\Omc,\Dmc)$, $\Omc = \{ \vp \}$ and
$\Dmc = \{ P(a), N(b) \}$ with $P,N$ fresh unary relation symbols. Observe that the labeled KB $(\Kmc,\{a\},\{b\})$ can be used to also prove the lower bound for definability, RE-existence, and entity distinguishability.
\end{proof}
If follows from Theorems~\ref{thm:correctedresults} and \ref{thm:correctedresults2} below that $(\ALCI,\Lmc_{S})$-separability is
\NExpTime-complete in data complexity, for all FO-fragments
$\Lmc_{S}\supseteq \text{UCQ}$. Hence,
$(\text{GNFO},\Lmc_{S})$-separability is \NExpTime-hard in data
complexity, for all such $\Lmc_{S}$. We conjecture that it is even
\TwoExpTime-hard in data complexity, see~\ref{app:hardness} for
further discussion. The data complexity of RE-existence and
entity distinguishability remains open.
%

We briefly discuss the case of FO-separability of labeled KBs in which
the ontology is empty. From the connection to rooted UCQ-evaluation,
it is immediate that this problem is in \coNPclass since evaluating
CQs on databases is \NPclass-complete~\cite{DBLP:books/aw/AbiteboulHV95}.
\begin{theorem}\label{thm:empty} 
	$(\text{FO},\text{FO})$-separability on labeled KBs with empty ontology
	is \coNPclass-complete. This is also true for definability, RE-existence, and entity distinguishability.
\end{theorem}
\begin{proof} \
	It remains to prove the lower bounds. For separability and entity distinguishability the lower bound follows directly from Corollary~\ref{cor:rel} and the fact that rooted CQ-evaluation on databases is \NPclass-hard~\cite{DBLP:books/aw/AbiteboulHV95}.
    It remains to prove the \coNPclass-lower bound for RE-existence (the lower bound for definability follows). The proof is by a polynomial time reduction of the problem whether there does not exist a homomorphism from a connected database $\Dmc_{1}$ to a connected database $\Dmc_{2}$ (in symbols, $\Dmc_{1}\not\rightarrow \Dmc_{2}$) which is \coNPclass-hard. Assume $\Dmc_{1}$ and $\Dmc_{2}$ are given. We may assume that $\text{cons}(\Dmc_{1})\cap \text{cons}(\Dmc_{2})=\emptyset$. Take a constant $a$ in $\Dmc_{1}$ and a constant $b$ in $\Dmc_{2}$ and let 
    $\Dmc_{1}'=\Dmc_{1}\cup \{A(a)\}$ and $\Dmc_{2}'=\Dmc_{2}\cup\{A(c) \mid c\in \text{cons}(\Dmc_{2})\}$ for a fresh unary relation $A$. Then $\Dmc_{1} \not\rightarrow \Dmc_{2}$ iff
    $\Dmc_{1}',a \not\rightarrow \Dmc_{1}'\cup \Dmc_{2}',b$ for any $b\in N:= \text{cons}(\Dmc_{1}\cup \Dmc_{2})\setminus \{a\}$. Let $\Kmc=(\emptyset,\Dmc_{1}'\cup \Dmc_{2}')$. Then the latter problem
    is equivalent to FO-separability of $(\Kmc,\{a\},N)$,
    by Theorem~\ref{critFOwithoutUNA}.    
\end{proof}
This result is in 
contrast to GI-completeness of the FO-definability problem on closed
world structures in which one asks for a finite structure $\Amf$ and set $P$ of tuples in $\text{dom}(\Amf)^{n}$  whether there exists an FO-formula $\varphi(\vec{x})$ that defines $P$ in the sense that $P=\{ \vec{a}\in\text{dom}(\Amf)^{n}\mid \Amf\models \varphi(\vec{a})\}$~\cite{DBLP:journals/tods/ArenasD16}. 

\section{Weak Separability for Decidable Fragments of
FO}\label{sec:weakdecidable}

We study $(\Lmc,\Lmc)$-separability for
$\Lmc \in \{ \ALCI, \text{GF}, \text{FO}^2 \}$.
None of these 
fragments \Lmc contains UCQ (or even CQ), and thus we cannot use
Theorem~\ref{critFOwithoutUNA} in the same way as for GNFO above. We also investigate the special cases of definability, referring expression existence,
and entity distinguishability. 

\subsection{Separability of Labeled $\mathcal{ALCI}$-KBs}
\label{section:ALCI}


We are interested in separating labeled \ALCI-KBs $(\Kmc,P,N)$ in
terms of \ALCI-concepts which is relevant for concept learning, for
generating referring expressions, and for entity comparison. Note that
since \ALCI-concepts are FO-formulas with one free variable, positive
and negative examples are single constants rather than proper tuples.

We start by considering projective $(\ALCI,\ALCI)$-separability and show that
Point~3 of Theorem~\ref{critFOwithoutUNA} also characterizes projective $\ALCI$-separability of labeled $\ALCI$-KBs. This has profound consequences.
For example, it follows that $(\ALCI,\text{FO})$-separability coincides with projective $(\ALCI,\mathcal{ALCI})$-separability and that projective $(\ALCI,\mathcal{ALCI})$-separability can be decided using rooted 
UCQ-evaluation on $\ALCI$-KBs. To prove the characterization we require an intermediate step
in which we characterize projective $(\ALCI,\mathcal{ALCI})$-separability using a suitable version of functional bisimulations. 
We start by introducing bisimulations for $\mathcal{ALCI}$~\cite{TBoxpaper,goranko20075}. Let
$\Amf$ and $\Bmf$ be structures and $\Sigma$ a signature of concept and role names. A relation
$S \subseteq \text{dom}(\Amf) \times \text{dom}(\Bmf)$ is an
\emph{$\mathcal{ALCI}(\Sigma)$-bisimulation between $\Amf$ and $\Bmf$}
if the following conditions hold:
\begin{enumerate}
	\item if $(d,e)\in S$ and $A\in \Sigma$, then
	$d \in A^{\Amf}$ iff $e\in A^{\Bmf}$; 
	\item if $(d,e)\in S$, $R$ is a $\Sigma$-role, and $(d,d')\in R^{\Amf}$, 
	then there is an $e'$ with $(e,e')\in R^{\Bmf}$ and
	$(d',e')\in S$;
	\item if $(d,e)\in S$, $R$ is a $\Sigma$-role, and $(e,e')\in R^{\Bmf}$, 
	then there is a $d'$ with $(d,d')\in R^{\Amf}$ and
	$(d',e')\in S$.
\end{enumerate}	
We write $\Amf,d\sim_{\ALCI,\Sigma}\Bmf,e$ and call pointed structures $\Amf,d$ and $\Bmf,e$
\emph{$\ALCI(\Sigma)$-bisimilar} if there exists an
$\ALCI(\Sigma)$-bisimulation $S$ between $\Amf$ and $\Bmf$ such that $(d,e)\in S$. 
We say that $\Amf,d$ and $\Bmf,e$ are \emph{$\mathcal{ALCI}(\Sigma)$-equivalent},
in symbols $\Amf,d\equiv_{\ALCI,\Sigma}\Bmf,e$, if $d\in C^{\Amf}$ iff
$e\in C^{\Bmf}$ for all $\ALCI(\Sigma)$-concepts $C$.
$\mathcal{ALCI}(\Sigma)$-bisimilarity implies $\mathcal{ALCI}(\Sigma)$-equivalence. The converse direction does not always hold, but it holds if at least one structure has finite outdegree~\cite{TBoxpaper}. We apply the following lemma mainly to finite structures (which trivially have finite outdegree).
\begin{lemma}\label{lem:equivalence}
	Let $\Amf,d$ and $\Bmf,e$ be pointed structures,
	assume that $\Amf$ has finite outdegree, and let $\Sigma$ be a signature. 
	Then
	$$
	\Amf,d \equiv_{\ALCI,\Sigma} \Bmf,e \quad \text{ iff } \quad
	\Amf,d \sim_{\ALCI,\Sigma}\Bmf,e.
	$$
	For the ``if'' direction, the condition ``$\Amf$ has finite outdegree'' can be dropped.
\end{lemma}
To characterize projective separability we require a functional version of 
$\ALCI(\Sigma)$-bisimilarity. In detail, we write $\Amf,d \sim_{\ALCI,\Sigma}^{f} \Bmf,e$ if there exists an
$\ALCI(\Sigma)$-bisimulation $S$ between $\Amf$ and $\Bmf$ that
contains $(d,e)$ and is \emph{functional}, that is,
$(a,b_{1}),(a,b_{2})\in S$ implies $b_{1}=b_{2}$ for all $a\in \text{dom}(\Amf)$ and $b_{1},b_{2}\in \text{dom}(\Bmf)$. Note that
$\Amf,d \sim_{\ALCI,\Sigma}^{f} \Bmf,e$ implies that there is a
homomorphism from $\Amf_{\text{con}(d)},d$ to $\Bmf,e$, where
$\Amf_{\text{con}(d)}$ is the restriction of $\Amf$ to all nodes
reachable from $d$ (here we assume that $X^{\Amf}=\emptyset$ for all concept and role names $X\not\in\Sigma$ because otherwise one might reach nodes from $d$ in $\Amf$ that are not relevant to $\ALCI(\Sigma)$-bisimilarity of $d$ and $e$).

%
\begin{restatable}{theorem}{thmLmodeltheorynull}
	\label{thm:L-modeltheory0}
	Assume that $(\Kmc,P,\{b\})$ is a labeled $\ALCI$-KB with $\Kmc=(\Omc,\Dmc)$
	and $\Sigma=\text{sig}(\Kmc)$. Then the following
	conditions are equivalent:
	\begin{enumerate}
		
		\item $(\Kmc,P,\{b\})$ is projectively
		$\ALCI$-separable;
		
		\item there exists a finite model $\Amf$ of $\Kmc$ and a finite set $\Sigma_{\text{help}}$ of concept names disjoint from $\text{sig}(\Kmc)$ such that for all models $\Bmf$ of
		$\Kmc$ and all $a\in P$: $\Bmf,a^{\Bmf}
		\not\sim_{\ALCI,\Sigma\cup \Sigma_{\text{help}}} \Amf, b^{\Amf}$;
		
		\item there exists a finite model $\Amf$ of $\Kmc$ such that for all models $\Bmf$ of $\Kmc$ and all
		$a\in P$: $\Bmf,a^{\Bmf} \not\sim_{\ALCI,\Sigma}^{f} \Amf,
		b^{\Amf}$; 
		
		\item there exists a finite model $\Amf$ of $\Kmc$ such that for all $a\in P$:
		$\Dmc_{\text{con}(a)},a\not\rightarrow \Amf,b^{\Amf}$;
	
        \item there exists a model $\Amf$ of $\Kmc$ such that for all $a\in P$:
        $\Dmc_{\text{con}(a)},a\not\rightarrow \Amf,b^{\Amf}$.
       
\end{enumerate}
\end{restatable}
\begin{proof} \
``1. $\Rightarrow$ 5.'' If $(\Kmc,P,\{b\})$ is projectively $\ALCI$-separable,
then it is projectively FO-separable. It follows that Condition~3 of Theorem~\ref{critFOwithoutUNA} holds which is identical to Condition~5.


``2. $\Rightarrow$ 1.'' Assume Condition~2 holds for $\Amf$ and $\Sigma_{\text{help}}$. Let 
$
t_{\Amf}(b) = \{C \in \ALCI(\Sigma\cup \Sigma_{\text{help}}) \mid b^{\Amf}\in C^{\Amf}\}$.
It follows from Lemma~\ref{lem:equivalence} that 
$
\Kmc \cup \{C(a) \mid C \in t_{\Amf}(b)\}
$
is not satisfiable, for any $a\in P$. (This can be seen as follows: assume there exists $a\in P$ and a model $\Bmf$ of $\Kmc$ such that $a^{\Bmf}\in C^{\Bmf}$ for all $C\in t_{\Amf}(b)$. Then $\Bmf,a^{\Bmf} \equiv_{\ALCI,\Sigma\cup \Sigma_{\text{help}}} \Amf,b^{\Amf}$. By Lemma~\ref{lem:equivalence}, $\Bmf,a^{\Bmf} \sim_{\ALCI,\Sigma\cup\Sigma_{\text{help}}} \Amf,b^{\Amf}$ which contradicts Condition~2.)
By compactness (and closure under conjunction of $t_{\Amf}(b)$) we find for every $a\in P$ a concept $C_{a}\in t_{\Amf}(b)$ such that $\Kmc\models \neg C_{a}(a)$. Thus, the concept $\neg (\bigsqcap_{a\in P}C_{a})$ separates $(\Kmc,P,\{b\})$, as required.


``3. $\Rightarrow$ 2.'' Take a model $\Amf$ of $\Kmc$ such that Condition~3 holds. We may assume that $\Amf$ only interprets the symbols in $\text{sig}(\Kmc)$.  
Define $\Amf'$ by expanding $\Amf$ as follows. Take for any $d\in \text{dom}(\Amf)$ a fresh concept name $A_{d}$ and set
$A_{d}^{\Amf'} =\{d\}$. 
Then Condition~2 holds for $\Amf'$ and $\Sigma_{\text{help}} = \{A_{d}\mid d\in \text{dom}(\Amf)\}$.

``4. $\Rightarrow$ 3.'' Take a model $\Amf$ of $\Kmc$ such that Condition~4 holds. As every functional bisimulation between any model $\Bmf$ of $\Kmc$ and $\Amf$ witnessing $\Bmf,a^{\Bmf} \sim_{\ALCI,\Sigma}^{f} \Amf,
b^{\Amf}$ induces a homomorphism from $\Dmc_{\text{con}(a)}$ to $\Amf$ mapping $a$ to $b^{\Amf}$, it follows that $\Amf$ satisfies Condition~3. 	

``5. $\Rightarrow$ 4.'' This follows from the fact that rooted UCQ-evaluation on $\ALCI$-KBs is finitely controllable.
\end{proof}
The following example illustrates the characterization given in Theorem~\ref{thm:L-modeltheory0} by applying it to Examples~\ref{exmp:44} and \ref{ex:ex2}.
\begin{exmp}
{\em Consider the KB $\Kmc_{2}=(\Omc_{2},\Dmc_{1})$ introduced in Examples~\ref{exmp:44} and \ref{ex:ex2}. We know already that $(\Kmc_{2},\{a\},\{b\})$ is $\ALCI$-separable. The structure $\Amf$ of Condition~2 in Theorem~\ref{thm:L-modeltheory0} can be used 
to construct a separating concept following the proof of ``2. $\Rightarrow$ 1.''.
For example, obtain a model $\Amf$ of $\Kmc_{2}$ from $\Dmc_{1}$ by adding ${\sf Person}(b)$. Then there does not exist a model $\Bmf$ of $\Kmc_{2}$ such that $\Bmf,a^{\Bmf}\sim_{\ALCI,\text{sig}(\Kmc_{2})}\Amf,b^{\Amf}$. Thus we have found a structure witnessing Condition~2 with $\Sigma_{\text{help}}=\emptyset$.
Now $t_{\Amf}(b)$ contains a concept $C$ such that $\neg C$ separates $(\Kmc_{2},\{a\},\{b\})$. An example of such a concept in $t_{\Amf}(b)$ is $C=\forall {\sf born\_in}.\neg \exists {\sf citizen\_of}^{-}.\top$. Then $\neg C$ separates $(\Kmc_{2},\{a\},\{b\})$ since $\Kmc_{2}\models \neg C(a)$.\finofex
}
\end{exmp}
It follows from Theorem~\ref{thm:L-modeltheory0} that Point~3 of Theorem~\ref{critFOwithoutUNA} also characterizes projective $(\ALCI,\ALCI)$-separability.  
This leads to the following result.
\begin{restatable}{theorem}{thmdatacomplexity}
	\label{thm:correctedresults}
For any FO-fragment $\Lmc_{S}\supseteq \text{UCQ}$:
\begin{enumerate}
	\item projective $(\ALCI,\ALCI)$-separability coincides with
    $(\ALCI,\Lmc_S)$-separability, and the same is true for definability;
    \item projective $(\ALCI,\ALCI)$-separability and projective $(\ALCI,\ALCI)$-definability are \NExpTime-complete.
\end{enumerate}
This holds also for RE-existence and entity distinguishability, for
all FO-fragments $\Lmc_{S}\supseteq \text{CQ}$.
\end{restatable}
\begin{proof} \ It remains to prove the complexity result. The
  \NExpTime upper bound follows from the fact that rooted
  UCQ-evaluation on $\ALCI$-KBs is in
  \coNExpTime~\cite{LutzDL07,Lutz-IJCAR08}. For the lower bound we
  first consider entity distinguishability. It is shown in
  \cite{LutzDL07,Lutz-IJCAR08} that unary rooted CQ-evaluation on
  $\ALCI$-KBs with a single constant is co\NExpTime-hard in combined
  complexity. Then the proof of Corollary~\ref{cor:rel} shows that
  projective $(\ALCI,\ALCI)$-entity distinguishability is
  \NExpTime-hard. It remains to prove \NExpTime-hardness of projective
  $(\ALCI,\ALCI)$-RE existence, as then hardness of projective
  $(\ALCI,\ALCI)$-definability follows. We show this by a polynomial
  time reduction of projective $(\ALCI,\ALCI)$-entity
  distinguishability. Assume a labeled KB $(\Kmc,\{a\},\{b\})$ with
  $\Kmc=(\Omc,\Dmc)$ is given. We may assume that
  $\text{cons}(\Dmc_{\text{con}(a)})\cap
  \text{cons}(\Dmc_{\text{con}(b)})=\emptyset$. (To show this, let $\Dmc'$ be a copy of $\Dmc$ using a set of constants disjoint from $\text{cons}(\Dmc)$. 
  We show that $(\Omc,\Dmc) \models C(a)$ iff $(\Omc,\Dmc\cup \Dmc')\models C(a)$, for every $\mathcal{ALCI}$-concept $C$. The implication from left to right follows from $\Dmc\subseteq \Dmc\cup \Dmc'$ and the converse implication can be proved by extending any model $\Amf$ of $\Kmc$ to a model of $(\Omc,\Dmc\cup \Dmc')$ 
  by setting ${c'}^{\Amf}=c^{\Amf}$ for the copy $c'$ of $c\in \text{cons}(\Dmc)$ in $\Dmc'$. Now the claim follows by replacing $\Dmc$ by $\Dmc\cup \Dmc'$, with $b$ and $b'$ swapped.) Let
  $\Dmc'=\Dmc\cup \{A(a),A(b)\}$ for a fresh concept name $A$ and let
  $N=\text{cons}(\Dmc)\setminus \{a\}$. Then $(\Kmc,\{a\},\{b\})$ is
  projectively $(\ALCI,\ALCI)$-separable iff $(\Kmc,\{a\},\{b\})$ is
  UCQ-separable iff $((\Omc,\Dmc'),\{a\},N)$ is UCQ-separable iff
  $((\Omc,\Dmc'),\{a\},N)$ is projectively $(\ALCI,\ALCI)$-separable,
  as required for the reduction.
\end{proof}
We next determine the data complexity of projective
$(\ALCI,\ALCI)$-separability and projective
$(\ALCI,\ALCI)$-definability. Quite remarkably, it turns out to be
identical to the combined complexity, namely \NExpTime-complete.
\begin{restatable}{theorem}{thmdatacomplexity2}
	\label{thm:correctedresults2}
	Projective $(\ALCI,\ALCI)$-separability and projective $(\ALCI,\ALCI)$-definability are \NExpTime-complete in data complexity.
\end{restatable}	
\begin{proof} \
	The upper bound follows from the \NExpTime upper bound for the
        same problems in combined complexity. For the lower bound for
        separability, it follows from the proof of
        Corollary~\ref{cor:rel} that is suffices to construct an
        $\ALCI$-ontology~$\Omc$ such that it is \coNExpTime-hard to
        decide unary rooted UCQ-evaluation on KBs with
        ontology~$\Omc$. The construction of such an ontology $\Omc$
       and in particular of the rooted UCQs that
       demonstrate \coNExpTime-hardness
       is quite tedious, we present details
        in the appendix. To show the lower bound for definability, we show that for a fixed $\ALCI$-ontology~$\Omc$, projective $\ALCI$-separability of labeled KBs with ontology $\Omc$
can be reduced in polynomial time to projective $\ALCI$-definability
for KBs with ontology $\Omc$. We use the same technique as in the reduction of entity distinguishability in the proof of Theorem~\ref{thm:correctedresults} above. Assume a labeled $\ALCI$-KB $(\Kmc,P,N)$ with $\Kmc=(\Omc,\Dmc)$ is given. We may assume that $\text{cons}(\Dmc_{\text{con}(a)})\cap \text{cons}(\Dmc_{\text{con}(b)})=\emptyset$, for any distinct $a,b\in P\cup N$ (see the proof of Theorem~\ref{thm:correctedresults}). Let $\Dmc'=\Dmc\cup \{A(a) \mid a\in P\cup N\}$ for a fresh concept name $A$ and let $N'=\text{cons}(\Dmc)\setminus P$. Then $(\Kmc,P,N)$ is projectively $(\ALCI,\ALCI)$-separable iff $(\Kmc,P,N)$ is UCQ-separable iff $((\Omc,\Dmc'),P,N')$ is UCQ-separable iff $((\Omc,\Dmc'),P,N')$ is projectively $(\ALCI,\ALCI)$-separable, as required for the reduction.
\end{proof}
It remains open whether the \NExpTime lower bound of Theorem~\ref{thm:correctedresults2} also holds for RE-existence and entity distinguishability.

\bigskip

We now turn to non-projective
separability.  We first observe that projective and non-projective
separability are indeed different. We use a simplified variant of 
Example~\ref{exm:new} we are going to revisit throughout this article.
\begin{exmp}\label{exm:11}
{\em Let $\Kmc_{3}=(\Omc_{3},\Dmc_{3})$ be the \ALCI-KB where $\Omc_{3} = \{ \top \sqsubseteq \exists R.\top \sqcap
	\exists R^{-}.\top \}$ and $\Dmc_{3}= \{R(a,a), R(b,c),R(c,b)\}$.
%
%
%
%
	Further let $P=\{a\}$ and $N=\{b\}$. Then the \ALCI-concept $A
	\rightarrow \exists R.A$ separates $(\Kmc_{3},P,N)$ using the concept name
	$A$ as a helper symbol. Thus $(\Kmc_{3},P,N)$ is projectively
	\ALCI-separable. Projective \ALCI-separability can also be shown 
	using Theorem~\ref{thm:L-modeltheory0} by observing that the structure $\Amf$ corresponding to $\Dmc_{3}$ is a model of $\Omc_{3}$ and that ${\Dmc_{3}}_{\text{con}(a)},a\not\rightarrow \Amf,b^{\Amf}$ as $(b,b)\not\in R^{\Amf}$.
	
	In contrast, $(\Kmc_{3},P,N)$ is not non-projectively
	\ALCI-separable. In fact, every \ALCI-concept $C$ with
	$\mn{sig}(C)=\{R\}$ is equivalent to $\top$ or to $\bot$
	w.r.t. \Omc. Thus if $\Kmc_{3}\models C(a)$, then $\Omc_{3}\models C \equiv
	\top$, and so $\Kmc_{3}\models C(b)$.
	
	It is also instructive to consider an ontology that is weaker
	than $\Omc_{3}$. For example, let
	$\Kmc_{4}=(\emptyset,\Dmc_{3})$. Then $(\Kmc_{4},P,N)$ is
	non-projectively $\ALCI$-separable. This is witnessed, for
	example, by the concept $C=\exists R.\forall R.\bot
	\rightarrow \exists R.\exists R.\forall R.\bot$ (which is
	obtained from $A \rightarrow \exists R.A$ by replacing $A$ by
	$\exists R.\forall R.\bot$) since $\Kmc_{4}\models C(a)$ but
	$\Kmc_{4}\not\models C(b)$. \finofex}
	%
\end{exmp}
Of course, Example~\ref{exm:11} implies that an analogue of Point~1 of
Theorem~\ref{thm:correctedresults} fails for non-projective
separability.  In fact, the labeled \ALCI-KB $(\Kmc_{3},P,N)$ in
Example~\ref{exm:11} is not non-projectively \ALCI-separable but is
separated by the CQ $R(x,x)$.


We next aim to characterize non-projective $(\ALCI,\ALCI)$-separability in the style
of Point~3 of Theorem~\ref{critFOwithoutUNA}. We start with noting
that the ontology $\Omc_{3}$ used in Example~\ref{exm:11} is very strong 
and enforces that all elements of all models of $\Omc_{3}$ are
$\ALCI(\mn{sig}(\Kmc_{3}))$-equivalent to each other. For ontologies that make
such strong statements, symbols from outside of $\mn{sig}(\Kmc)$ might
be required to construct a separating concept. It turns out that this
is in fact the only effect that distinguishes non-projective from projective
separability.
We next make this
precise.

For a KB $\Kmc$, we use $\mn{cl}(\Kmc)$
to denote the set of concepts in $\Kmc$ and the concepts $\exists
R.\top$ and $\exists R^- . \top$ for all role names $R\in {\sf
	sig}(\Kmc)$, closed under subconcepts and single negation.  A
\emph{$\Kmc$-type} is a set $t\subseteq \mn{cl}(\Kmc)$ such that there
exists a model $\Amf$ of $\Kmc$ and an $a\in \text{dom}(\Amf)$ with
$\text{tp}_{\Kmc}(\Amf,a)=t$ where
$$
\text{tp}_{\Kmc}(\Amf,a)= \{ C\in \mn{cl}(\Kmc) \mid a\in C^{\Amf}\}
$$
is the \emph{$\Kmc$-type of $a$ in $\Amf$}. Conversely, a \Kmc-type $t$ is \emph{realizable in $\Kmc,b$},
where $\Kmc=(\Omc,\Dmc)$ and $b \in \text{cons}(\Dmc)$, if there
exists a model $\Amf$ of $\Kmc$ such that
$\text{tp}_{\Kmc}(\Amf,b^{\Amf})=t$. We often identify a $\Kmc$-type with the conjunction over all its concepts and thus write, for example, $\exists R.t$ for the concept $\exists R.(\bigsqcap_{C\in t}C)$. 

Finally, we need the notion of \emph{complete} types which is similar
in spirit to the notion of a complete theory in classical logic
\cite{modeltheory}.
\begin{definition}
  A $\Kmc$-type $t$ is \emph{$\ALCI$-complete for $\Kmc$} if for any
  two pointed models $\Amf_{1},b_{1}$ and $\Amf_{2},b_{2}$ of $\Kmc$,
  $t=\text{tp}_{\Kmc}(\Amf_{1},b_{1})=
  \text{tp}_{\Kmc}(\Amf_{2},b_{2})$ implies $ \Amf_{1},b_{1}
  \equiv_{\ALCI,\mn{sig}(\Kmc)} \Amf_{2},b_{2}$.  
\end{definition}
We next illustrate the basic notions. 
\begin{exmp}\label{exm:completion}
{\em 	(1) In the KB $\Kmc_{3}$ from Example~\ref{exm:11}, there is only a single $\Kmc_{3}$-type, $t_{0}=\{\top, \exists R.\top,\exists R^{-}.\top,\exists R.\top\sqcap\exists R^{-}.\top\}$, and
	this type is \ALCI-complete for $\Kmc_{3}$.
	
	(2) In the KB $\Kmc_{4}$ from Example~\ref{exm:11}, there are
	four different $\Kmc_{4}$-types, determined by any combination of (negated) $\exists R.\top$ and $\exists R^{-}.\top$. The type $t_{1}=\{\top, \neg \exists R.\top, \neg \exists R^{-}.\top,\neg(\exists R.\top\sqcap\exists R^{-}.\top)\}$ is the only type that is  $\ALCI$-complete for $\Kmc_{4}$.
	
	(3) Let $\Kmc$ be any KB not using any role names. Then every $\Kmc$-type is $\ALCI$-complete for $\Kmc$.
	
	(4) Let $\Dmc$ be a database and $\Sigma=\text{sig}(\Dmc)$. Let $\Omc_\Dmc$ be the ontology that
	contains all $\ALCI(\Sigma)$-CIs that are
	true in the structure $\Amf_{\Dmc}$ defined by $\Dmc$. This ontology is infinite, but
	easily seen to be logically equivalent to a finite
	ontology: define an equivalence relation $\sim$ on $\text{cons}(\Dmc)$ by setting $a\sim b$ if $\Amf_{\Dmc},a \sim_{\ALCI,\Sigma}\Amf_{\Dmc},b$, and denote by $[a]$ the equivalence class of $a$. Then we find $\ALCI$-concepts $C_{[a]}$, $a\in \text{cons}(\Amf)$, such that $a\in C_{[a]}^{\Amf_{\Dmc}}$ but $b\not\in C_{[a]}^{\Amf_{\Dmc}}$ for any $b\not\sim a$. We may assume that $A$ or its negation are a conjunct of $C_{[a]}$, for any concept name $A\in \Sigma$. Then $\Omc_{\Dmc}$ is axiomatized by taking the CI $\top \sqsubseteq \bigsqcup_{a\in \text{cons}(\Dmc)}C_{[a]}$, every CI $C_{[a]} \sqsubseteq \exists R.C_{[b]}$ with $b\in R[a]$, and every CI $C_{[a]}\sqsubseteq \forall R.\bigsqcup_{b\in R[a]} C_{[b]}$, where $R[a]$ is the set of $b$ such that $R(c,b)\in \Dmc$ for some $c\in [a]$. Let
	$\Kmc_{\Dmc}=(\Omc_{\Dmc},\Dmc)$. Then every $\Kmc_{\Dmc}$-type is \ALCI-complete for $\Kmc_{\Dmc}$. Intuitively, if we assume that the database $\Dmc$
	is complete in the sense that any ground atom $R(\vec{a})$ with $R\in \Sigma$ that is not in $\Dmc$ is false then $\Omc_{\Dmc}$ is the logically strongest $\ALCI$-ontology representing this assumption. \finofex
}
\end{exmp}
Observe that the types under Points~(2) and~(3) are $\ALCI$-complete simply because nodes that satisfy them cannot be connected to any other node via roles in $\text{sig}(\Kmc)$ in any model of the respective KBs. We say that a $\Kmc$-type
$t$ is \emph{disconnected} if $\neg \exists R . \top \in t$ for all $\text{sig}(\Kmc)$-roles $R$ and \emph{connected} otherwise. 
The $\Kmc$-types under Points~(2) and~(3) are disconnected.
We are now in a position to formulate the characterization of
non-projective $(\ALCI,\ALCI)$-separability. 
\begin{restatable}{theorem}{thmchartwelve}
	\label{thm:char12}
	Assume that $(\Kmc,P,\{b\})$ is a labeled $\ALCI$-KB with 
	$\Kmc=(\Omc,\Dmc)$
	and $\Sigma=\text{sig}(\Kmc)$. Then the following
	conditions are equivalent:
	\begin{enumerate}
		
		\item $(\Kmc,P,\{b\})$ is non-projectively
		$\ALCI$-separable;
		\item there exists a finite model $\Amf$ of $\Kmc$ of such that for all models $\Bmf$ of
		$\Kmc$ and all $a\in P$: $\Bmf,a^{\Bmf}
		\not\sim_{\ALCI,\Sigma} \Amf, b^{\Amf}$;
		\item there exists a finite model $\Amf$ of $\Kmc$ such that
	for all $a\in P$:
	\begin{enumerate}
		\item $\Dmc_{\text{con}(a)},a\not\rightarrow \Amf,b^{\Amf}$ and
		\item if $\text{tp}_{\Kmc}(\Amf,b^{\Amf})$ is connected and
		$\ALCI$-complete for $\Kmc$, then $\text{tp}_{\Kmc}(\Amf,b^{\Amf})$ is not
		realizable in $\Kmc,a$.
	\end{enumerate}
\end{enumerate}
\end{restatable}
Before proving Theorem~\ref{thm:char12} we make a few observations and illustrate the result by examples. Note that Point~2 of Theorem~\ref{thm:char12} coincides with Point~2 of Theorem~\ref{thm:L-modeltheory0} characterizing projective $(\ALCI,\ALCI)$-separability except that $\Sigma_{\text{help}}=\emptyset$. This should be intuitively clear and the
proof of ``1. $\Leftrightarrow$ 2.'' here is essentially the same as the proof in Theorem~\ref{thm:L-modeltheory0}. Point~3 (a) of Theorem~\ref{thm:char12} coincides with Point~4 of Theorem~\ref{thm:L-modeltheory0}. It follows that Point~3 strengthens Point~4 of Theorem~\ref{thm:L-modeltheory0} by (b).
The following example illustrates Point~3 (b).
\begin{exmp}
 {\em Consider the KB $\Kmc_{3}$ from Example~\ref{exm:11}. We know that $(\Kmc_{3},\{a\},\{b\})$ is not \ALCI-separable and aim to confirm this using Point~3 of Theorem~\ref{thm:char12}. Let $\Amf$ be any finite model of $\Kmc_{3}$ and assume that ${\Dmc_{3}}_{\text{con}(a)},a\not\rightarrow \Amf,b^{\Amf}$.
  As $t_{0}= \{\top, \exists R.\top,\exists R^{-}.\top,\exists R.\top\sqcap\exists R^{-}.\top\}$ is the only $\Kmc_{3}$-type we have that $\text{tp}_{\Kmc}(\Amf,b^{\Amf})=t_{0}$.
  Thus, $\text{tp}_{\Kmc_{3}}(\Amf,b^{\Amf})$ is connected and
  $\ALCI$-complete for $\Kmc_{3}$. But then $\text{tp}_{\Kmc_{3}}(\Amf,b^{\Amf})$ is also realizable in $\Kmc,a$, and so (b) is refuted for $\Amf$.
  
  Consider the KB $\Kmc_{4}$ from Example~\ref{exm:11}. We know that $(\Kmc_{4},\{a\},\{b\})$ is \ALCI-separable and confirm this using Point~3 of Theorem~\ref{thm:char12}. Let $\Amf$ be the structure defined by $\Dmc_{3}$. Then (a) holds since ${\Dmc_{3}}_{\text{con}(a)},a\not\rightarrow \Amf,b^{\Amf}$ which follows from $\Amf\not\models R(b,b)$. Condition~(b) also holds since
  $\text{tp}_{\Kmc_{4}}(\Amf,b^{\Amf})= t_{0}$ is not $\ALCI$-complete for $\Kmc_{4}$.\finofex
}
\end{exmp}

{\bf Proof of Theorem~\ref{thm:char12}.}
``1. $\Rightarrow$ 2.'' Let
$C$ be an $\mathcal{ALCI}(\Sigma)$-concept such that $\Kmc\models C(a)$ for all $a\in P$ but $\Kmc\not\models C(b)$. By the finite model property of $\mathcal{ALCI}$ (Section~\ref{sec:prelim}) there exists a finite model $\Amf$ of $\Kmc$ such that $b^{\Amf}\not\in C^{\Amf}$. Then for all models $\Bmf$ of
$\Kmc$ and all $a\in P$: $\Bmf,a^{\Bmf}
\not\equiv_{\ALCI,\Sigma} \Amf, b^{\Amf}$ since otherwise $\Kmc\not\models C(a)$ for some $a\in P$. But then, by Lemma~\ref{lem:equivalence}, $\Bmf,a^{\Bmf}
\not\sim_{\ALCI,\Sigma} \Amf, b^{\Amf}$ for all $a\in P$, as required.

``2. $\Rightarrow$ 1.'' This is a special case of the proof of ``2. $\Rightarrow$ 1.'' for Theorem~\ref{thm:L-modeltheory0}. 

``2. $\Rightarrow$ 3.'' Assume that $\Amf$ is a finite model of $\Kmc$ such that
for all models $\Bmf$ of
$\Kmc$ and all $a\in P$: $\Bmf,a^{\Bmf}
\not\sim_{\ALCI,\Sigma} \Amf, b^{\Amf}$. We know already from the proof of
Theorem~\ref{thm:L-modeltheory0} that $\Dmc_{\text{con}(a)},a\not\rightarrow \Amf,b^{\Amf}$ for all $a\in P$. We show that (b) holds as well.
Assume for a proof by contradiction that there is $a\in P$ such that $\text{tp}_{\Kmc}(\Amf,b^{\Amf})$ is connected,
$\ALCI$-complete for $\Kmc$, and realizable in $\Kmc,a$. Then take a model $\Bmf$ of $\Kmc$ such that
$\text{tp}_{\Kmc}(\Bmf,a^{\Bmf})=\text{tp}_{\Kmc}(\Amf,b^{\Amf})$.
Then $\Bmf,a^{\Bmf}\equiv_{\ALCI,\Sigma} \Amf,b^{\Amf}$ since $\text{tp}_{\Kmc}(\Amf,b^{\Amf})$ is $\ALCI$-complete for $\Kmc$.
But then $\Bmf,a^{\Bmf}\sim_{\ALCI,\Sigma} \Amf,b^{\Amf}$, by Lemma~\ref{lem:equivalence},
and we have derived a contradiction. 

``3. $\Rightarrow$ 2.'' 
Assume that Condition~3 holds for $\Amf$. 
If $\text{tp}_{\Kmc}(\Amf,b^{\Amf})$ is connected and $\ALCI$-complete
for $\Kmc$, then by (b) $\text{tp}_{\Kmc}(\Amf,b^{\Amf})$ is not
realizable in $\Kmc,a$ for any $a\in P$. Thus, there is no model $\Bmf$ of $\Kmc$ with 
$\Bmf,a^{\Bmf}\equiv_{\ALCI,\Sigma} \Amf,b^{\Amf}$, for any $a\in P$.
By Lemma~\ref{lem:equivalence}, there is no model $\Bmf$ of $\Kmc$ such that 
$\Bmf,a^{\Bmf}\sim_{\ALCI,\Sigma} \Amf,b^{\Amf}$, for any $a\in P$, as required.
Note that $\neg \text{tp}_{\Kmc}(\Amf,b^{\Amf})$ separates
$(\Kmc,P,\{b\})$ in this case.

If $\text{tp}_{\Kmc}(\Amf,b^{\Amf})$ is disconnected, then it follows from
$\Dmc_{\text{con}(a)},a \not\rightarrow \Amf,b^{\Amf}$ that either there exists $A$ with $A(a)\in \Dmc$
and $b^{\Amf}\not\in A^{\Amf}$ or there exists $R$ with $R(a,c)\in \Dmc$ for some $c$. In both cases
$\text{tp}_{\Kmc}(\Amf,b^{\Amf})$ is not realizable in $\Kmc,a$, and the proof continues as above.

Assume now that $\text{tp}_{\Kmc}(\Amf,b^{\Amf})$ is connected and not $\ALCI$-complete for $\Kmc$. For a model $\Cmf$ of $\Kmc$ and $\ell\geq 0$ we denote by 
$\Cmf_{\Dmc,b}^{\leq \ell}$ the substructure of $\Cmf$ induced by all nodes reachable 
from some $c^{\Cmf}$, $c\in \text{cons}(\Dmc_{\text{con}(b)})$, in at most $\ell$ steps.
We construct for any $\ell\geq 0$ a finite model $\Cmf$ of $\Kmc$ such that
\begin{itemize}
	\item[(i)] $\Cmf_{\Dmc,b}^{\leq \ell},b^{\Cmf} \rightarrow \Amf,b^{\Amf}$;
	\item[(ii)] for any two distinct $d_{1},d_{2}\in \text{dom}(\Cmf_{\Dmc,b}^{\leq \ell})$: 
	$\Cmf,d_{1}\not\sim_{\ALCI,\Sigma} \Cmf,d_{2}$.
\end{itemize}
We first show that the implication ``3. $\Rightarrow$ 2.'' is proved if such a $\Cmf$ can be constructed.

\medskip
\noindent
\emph{Claim 1}. If (i) and (ii) hold for some $\ell \geq |\Dmc|$ for $\Cmf$ and $\Dmc_{\text{con}(a)},a \not\rightarrow \Amf,b^{\Amf}$ for all $a\in P$,
then Condition~2 holds for $\Cmf$.

\medskip
\textit{Proof of Claim~1.} The proof of Claim~1 is indirect. Assume that Condition~2 does not hold for $\Cmf$, that (i) holds for some $\ell \geq |\Dmc|$, and $\Dmc_{\text{con}(a)},a \not\rightarrow \Amf,b^{\Amf}$ for all $a\in P$. We show that (ii) does not hold for $\ell$ and $\Cmf$. As we assume that Condition~2 does not hold for $\Cmf$, there exists
$a\in P$ and a model $\Bmf$ of $\Kmc$ and a bisimulation $S$ witnessing 
$\Bmf,a^{\Bmf}\sim_{\ALCI,\Sigma}\Cmf,b^{\Cmf}$. 
As there is no homomorphism from $\Dmc_{\text{con}(a)}$ to $\Amf$ 
mapping $a$ to $b^{\Amf}$, by Condition~(i) there is no homomorphism from 
$\Dmc_{\text{con}(a)}$ to $\Cmf_{\Dmc,b}^{\leq \ell}$ mapping $a$ to $b^{\Cmf}$ (since the composition of homomorphisms is a homomorphism). As the bisimulation $S$ induces a relation between $\text{cons}(\Dmc_{\text{con}(a)})$ and $\text{dom}(\Cmf_{\Dmc,b}^{\leq \ell})$ with domain $\text{cons}(\Dmc_{\text{con}(a)})$ that is a homomorphism if it is functional, there exist $e\in\text{cons}(\Dmc_{\text{con}(a)})$ and $d_{1},d_{2}\in \text{dom}(\Cmf_{\Dmc,b}^{\leq \ell})$ with 
$d_{1}\not=d_{2}$ and $(e,d_{1}),(e,d_{2})\in S$. Then $\Cmf,d_{1}
\sim_{\ALCI,\Sigma} \Cmf,d_{2}$ (since the composition  of
bisimulations is a bisimulation). Hence Condition~(ii) does not hold,
as required. This finishes the proof of Claim~1.

\medskip
We come to the construction of the model $\Cmf$ of $\Kmc$ satisfying (i) and (ii). It is illustrated in Figure~\ref{fig:ALCIsigma}.
\begin{figure*}
	\begin{center}
		

		\caption{Construction of $\mathfrak{C}$.} \label{fig:ALCIsigma}
	\end{center}
\end{figure*}
To construct $\Cmf$ we first provide a more constructive characterization of when a $\Kmc$-type $t$ is not $\ALCI$-complete for $\Kmc$. By definition, this means that there are non-$\mathcal{ALCI}(\Sigma)$-bisimilar pointed models of $\Kmc$ realizing $t$. We equivalently rephrase this in terms of the existence of certain paths realizing types. Let $R$ be a role. 
We say that $\Kmc$-types $t_{1}$ and $t_{2}$ are \emph{$R$-coherent} if there exists a
model $\Amf$ of $\Kmc$ and nodes $d_{1}$ and $d_{2}$ realizing $t_{1}$
and $t_{2}$, respectively, such that $(d_{1},d_{2})\in R^{\Amf}$. We
write $t_1\rightsquigarrow_R t_2$ in this case.
A sequence 
\begin{equation*}
	\sigma=t_{0}R_{0}\ldots R_{n}t_{n+1}
\end{equation*}
of $\Kmc$-types $t_{0},\ldots,t_{n+1}$ and $\Sigma$-roles 
$R_{0},\ldots, R_{n}$ \emph{witnesses $\ALCI$-incompleteness of a $\Kmc$-type $t$} 
if $t=t_{0}$, $n\geq 1$, and 
\begin{itemize}
	\item $t_i\rightsquigarrow_{R_{i}} t_{i+1}$ for $i\leq n$;
	\item there exists a model $\Amf$ of $\Kmc$ and nodes 
	$d_{n-1},d_{n}\in \text{dom}(\Amf)$ with $(d_{n-1},d_{n})\in R_{n-1}^{\Amf}$ such that $d_{n-1}$ and $d_{n}$ realize $t_{n-1}$ and $t_{n}$ in $\Amf$,
	respectively, and there does not exist $d_{n+1}$ in $\Amf$ realizing $t_{n+1}$ with $(d_{n},d_{n+1})\in R_{n}^{\Amf}$.
\end{itemize}
\emph{Claim 2}. 
	The following conditions are equivalent, for any $\Kmc$-type $t$:
	\begin{enumerate}
		\item $t$ is not $\ALCI$-complete for $\Kmc$;	
		\item there is a sequence of length not exceeding $2^{||\Omc||}+2$ witnessing $\ALCI$-incompleteness of $t$.
	\end{enumerate}
It is in \ExpTime to decide whether a $\Kmc$-type $t$ is $\ALCI$-complete for $\Kmc$.

\medskip
Claim~2 is proved in the appendix. The intuition behind its proof is as follows. One can define a `maximal model' of $\Kmc$ realizing $t$ by taking a node `realizing' $t$ and then, inductively, for every role $R$ and $\Kmc$-type $t'$ with $t\rightsquigarrow_{R} t'$ an $R$-successor `realizing' $t'$. Then one continues with those $R$-successors in the same way, and so on. This model is maximal in the sense that for every node realizing a $\Kmc$-type every  $\Kmc$-type for which this is logically possible is realized in an $R$-successor. As not every model of $\Kmc$ realizing $t$ is bisimilar to the maximal one just constructed, one obtains a sequence witnessing $\ALCI$-incompleteness of $t$ for $\Kmc$. A pumping argument bounds its length.

We return to the construction of $\Cmf$. By Claim~2 we can take a sequence $\sigma=t_{0}^{\sigma}R_{0}^{\sigma}\ldots R_{m_{\sigma}}^{\sigma}t_{m_{\sigma}+1}^{\sigma}$ 
	that witnesses $\ALCI$-incompleteness of
	$t_{0}^{\sigma}:=\text{tp}_{\Kmc}(\Amf,b^{\Amf})$ for $\Kmc$, where $1\leq m_{\sigma}\leq L_{\Omc}:=2^{||\Omc||}+1$. 
	Note that there exists $d\in \text{dom}(\Amf)$ such that $(b^{\Amf},d)\in (R_{0}^{\sigma})^{\Amf}$, 
	since $\exists R^\sigma_0.\top \in t^\sigma_0$.
	
	We unfold $\Amf$ to a model $\Amf_{\Dmc,b}^{\leq \ell}$ as follows: the domain of $\Amf_{\Dmc,b}^{\leq \ell}$ is the
	set of sequences $c^{\Amf}R_{1}d_{1}\cdots R_{n}d_{n}$, $0\leq n \leq \ell$, with $c\in \text{cons}(\Dmc_{\text{con}(b)})$, $(c^{\Amf},d_{1})\in R_{1}^{\Amf}$ and $(d_{i},d_{i+1})\in R_{i+1}^{\Amf}$ 
        for all $i<n$, together with a disjoint copy $\Bmf$ of $\Amf$ satisfying $(\Omc,\Dmc\setminus \Dmc_{\text{con}(b)})$. On that disjoint copy all symbols except the constants in 
        $\Dmc_{\text{con}(b)}$ are defined as before. This part of
    $\Amf_{\Dmc,b}^{\leq \ell}$ is not relevant at all in what follows; it is only needed to obtain a model of $(\Omc,\Dmc\setminus \Dmc_{\text{con}(b)})$. Next we set $w d \in A^{\Amf_{\Dmc,b}^{\leq \ell}}$ if $d\in A^{\Amf}$ for all $w d\in \text{dom}(\Amf_{\Dmc,b}^{\leq \ell})$, $(w,wRd)\in R^{\Amf_{\Dmc,b}^{\leq \ell}}$ for all $wRd\in \text{dom}(\Amf_{\Dmc,b}^{\leq \ell})$, $(c_{1}^{\Amf},c_{2}^{\Amf})\in R^{\Amf_{\Dmc,b}^{\leq \ell}}$ if 
		$(c_{1}^{\Amf},c_{2}^{\Amf})\in R^{\Amf}$ for all $c_{1},c_{2}\in \text{cons}(\Dmc_{\text{con}(b)})$, and $c^{\Amf_{\Dmc,b}^{\leq \ell}}=c^{\Amf}$ for all $c\in \text{cons}(\Dmc_{\text{con}(b)})$.
		Note that in the tree-shaped models $\Amf_{c}$ hooked to $c^{\Amf}$, $c\in \text{cons}(\Dmc_{\text{con}(b)})$, all nodes of any depth $k< \ell$ 
	    have an $R$-successor in $\Amf_{c}$ of depth $k+1$, for some $R\in \mn{sig}(\Kmc)$. In Figure~\ref{fig:ALCIsigma}, $\Amf_{\Dmc,b}^{\leq \ell}$
	    is represented by the circle in the middle.
    	Denote by $L$ the set of all leaf nodes in $\Amf_{\Dmc,b}^{\leq \ell}$, that is to say, all nodes that have depth exactly $\ell$ in some $\Amf_{c}$, $c\in \text{cons}(\Dmc_{\text{con}(b)})$. In Figure~\ref{fig:ALCIsigma}, these are $d,d',d''$.
	
	We obtain $\Cmf$ by attaching to every $d\in L$ a tree-shaped model $\mathfrak{F}_{d}$ such that in the resulting model
	no node in $L$ is $\ALCI(\Sigma)$-bisimilar to any other node
	in $\Amf_{\Dmc,b}^{\leq \ell}$. It then directly follows that $\Cmf$
	satisfies Conditions~(i) and (ii). 

    Take for any $d\in L$ a number $N_{d}>|\Dmc|+2\ell+2(L_{\Omc} +1)$ 
	such that $|N_{d}-N_{d'}|>2(L_{\Omc}+1)$ for $d\not=d'$. Now fix $d\in L$ and let $t_{0}=\text{tp}_{\Kmc}(\Amf,d)$.
	By first walking from $d$ to $b^{\Amf}$ in $\Amf_{\Dmc,b}^{\leq \ell}$ and then using that $\text{tp}_{\Kmc}(\Amf,b^{\Amf})$
        is not $\ALCI$-complete for $\Kmc$ we find a sequence
	$$
        t_{0}R_{0}\cdots R_{n_{d}}t_{n_{d}+1}
        $$
        that witnesses $\mathcal{ALCI}$-incompleteness of $t_{0}$ for $\Kmc$ and ends with the tail 
        $t_{m_{\sigma}-1}^{\sigma}R_{m_{\sigma}-1}^{\sigma}t_{m_{\sigma}}^{\sigma}R_{m_{\sigma}}^{\sigma}t_{m_{\sigma}+1}^{\sigma}$
        of the sequence $\sigma$ witnessing $\ALCI$-incompleteness of  
        $\text{tp}_{\Kmc}(\Amf,b^{\Amf})$ for $\Kmc$. Using a straightforward pumping argument as in the proof of Claim~2 
	we may assume that $n_{d}\leq L_{\Omc}$. We define the tree-shaped model $\mathfrak{F}_{d}$ in such a way that we have a node $c_{0}\in \text{dom}(\mathfrak{F}_{d})$ 
        such that in the model $\Cmf$ we have 
        $$
        c_{0} \in (\forall \Sigma^{N_{d}}.D)^{\Cmf} \subseteq \text{dom}(\mathfrak{F}_{d}), \quad  D=\exists \Sigma^{L_{\Omc}}.(t_{n_{d}} \sqcap \neg \exists R_{n_{d}}.t_{n_{d}+1})
	$$
        where $\exists \Sigma^{k}.C$ stands for the disjunction of all $\exists \rho.C$
	with $\rho$ a path $R_{1}\cdots R_{m}$ of $\Sigma$-roles $R_{1},\ldots,R_{m}$ and $m\leq k$, and $\forall \Sigma^{k}.C=\neg \exists \Sigma^{k}.\neg C$.
        Thus, we aim to achieve that $D$ holds in the $N_{d}$-neighbourhood of $c_{0}$, and that this does not hold
        for any other node in $\Cmf$ outside $\mathfrak{F}_{d}$. Clearly then we are done as Condition~(ii) now follows from the fact that there exists a path from $d$ to a node
	satisfying $\forall \Sigma^{N_{d}}.D$ that is shorter than any such path in $\Cmf$ from any other node in $\Amf_{\Dmc,b}^{\leq \ell}$.
	        Observe that $(t_{n_{d}} \sqcap \neg \exists R_{n_{d}}.t_{n_{d}+1})$ 
        is satisfiable in a model of $\Kmc$ because $t_{n_{d}} R_{n_{d}}t_{n_{d}+1}$
        is the tail of a sequence witnessing $\ALCI$-incompleteness for $\Kmc$.   
	
        To construct $\Fmf_{d}$ consider the `almost maximal' model (introduced in the argument for Claim 2) $\Umf_{c_{0}}$ of
	$\Omc$ whose root $c_{0}$ realizes $t_{n_{d}}$ such that if a node $e\in \text{dom}(\Umf_{c_{0}})$ 
	realizes any $\Kmc$-type $t$ and is of depth $k\geq 0$, then for every $\Kmc$-type $t'$
	with $t\rightsquigarrow_{R} t'$ for some $\Sigma$-role $R$ there exists $e'$ realizing $t'$ of depth $k+1$ with 
	$(e,e')\in R^{\Umf_{c_0}}$, \emph{except if} $k\leq N_{d}+L_{\Omc}+1$, 
	$t=t_{n_{d}}$, $R=R_{n_{d}+1}$, and $t'=t_{n_{d}+1}$. $\Umf_{c_{0}}$ with root $c_{0}$ is depicted on the left hand side of Figure~\ref{fig:ALCIsigma}.
	Observe that $\Umf_{c_{0}}$ is a model of $\Kmc$ because $t_{n_{d}} R_{n_{d}}t_{n_{d}+1}$ is the tail of a sequence witnessing $\ALCI$-incompleteness for $\Kmc$.
        Also observe that $D$ is true in the $N_{d}$-neighbourhood of $c_{0}$ and false outside the $N_{d}+2(L_{\Omc}+1)$-neighbourhood of $c_{0}$, formally:
	\begin{itemize}
		\item $e\in D^{\Umf_{c_0}}$ for all $e$ with $\text{dist}_{\Umf_{c_{0}}}(c_{0},e)\leq N_{d}$;
		\item $e\not\in D^{\Umf_{c_0}}$ for all $e$ with $\text{dist}_{\Umf_{c_{0}}}(c_{0},e)> N_{d}+2(L_{\Omc}+1)$.
	\end{itemize}
	
	%
	We next need to attach $\Umf_{c_{0}}$ to $\Amf_{\Dmc,b}^{\leq \ell}$ at $d$ (and rename it to $\mathfrak{F}_{d}$). For this again the sequence 
        $t_{0}R_{0}\cdots R_{n_{d}}t_{n_{d}+1}$ is crucial. Now, however, we are not interested in its tail but in the initial part $t_{0}R_{0}\cdots t_{n_{d}-1}R_{n_{d}-1}t_{n_{d}}$.
        Starting from $t_{n_{d}}$ realized in $c_{0}$ in $\Umf_{c_{0}}$ we can follow $R_{n_{d}-1}^{-}$ to $t_{n_{d}-1}$, then follow $R_{n_{d}-1}$ to $t_{n_{d}}$, and then follow            
        again $R_{n_{d}-1}^{-}$ to $t_{n_{d}-1}$, and so on. In this way, we can go back and force $N_{d}$ times between $t_{n_{d}}$ and $t_{n_{d}-1}$ 
        and then follow the inverse of the sequence $t_{0}R_{0}\cdots t_{n_{d}-1}R_{n_{d}-1}t_{n_{d}}$ after $2N_{d}$ steps. In other words, 
        $\Umf_{c_0}$ contains a path 
	$
	e_{0},\ldots,e_{n_{d}}\ldots,e_{n_{d}+2N_{d}}=c_{0}
	$ (see Figure~\ref{fig:ALCIsigma})
	such that $t_{0}$ is realized in $e_{0}$ and
	\begin{itemize}
		\item $(e_{i},e_{i+1})\in R_{i}^{\Umf_{c_0}}$ for all $i<n_{d}$, and $(e_{n_{d}+2k+1},e_{n_{d}+2k}),(e_{n_{d}+2k+1},e_{n_{d}+2k+2})\in R_{n_{d}-1}^{\Amf_{c_0}}$ for $0\leq k <N_{d}$; 
		\item $e_{n_{d}+2k}\in t_{n_{d}}^{\Umf_{c_0}}$, for all $k\leq N_{d}$, and
		$e_{n_{d}+2k+1}\in t_{n_{d}-1}^{\Amf_{c_0}}$, for all $k<N_{d}$.
	\end{itemize}
	Now $\mathfrak{F}_{d}$ is obtained from $\Umf_{c_0}$ by renaming $e_{0}$ to $d$. Finally $\Cmf$ is obtained from
	$\Amf_{\Dmc,b}^{\leq \ell}$ by hooking $\mathfrak{F}_{d}$ at $d$ to $\Amf_{\Dmc,b}^{\leq \ell}$ for all $d\in L$, see Figure~\ref{fig:ALCIsigma}.
	$\Cmf$ is a model of $\Kmc$ since $t_{0}$ is realized in $e_{0}$ in $\Umf_{c_{0}}$ and in $d$ in $\Amf_{\Dmc,b}^{\leq \ell}$. Moreover, clearly $\Cmf$ satisfies Condition~(i).
	For Condition~(ii) assume $d\in L$ is as above. Let
	$C_{d}=\forall \Sigma^{N_{d}}.D$. 
	Then $e_{n_{d}+2N_{d}}\in C_{d}^{\Cmf}$ and by construction $C_{d}^{\Cmf}\subseteq \text{dom}(\mathfrak{F}_{d})$.
    This finishes the construction of the structure $\Cmf$ satisfying (i) and (ii) except that $\Cmf$ fails to be finite since the structures $\mathfrak{F}_{d}$ are not finite.
    This is straightforward to repair, however: instead of constructing the $\mathfrak{F}_{d}$ as tree-shaped almost maximal structures in which for every $\Kmc$-type $t$ realized in a node $e$ of depth $k> N_{d}+L_{\Omc}+1$ and $\Kmc$-type $t'$ with $t\rightsquigarrow_{R} t'$ there exists an $R$-successor of $e$ of depth $k+1$ realizing $t'$ we do the following for a sufficiently large $k$ (for instance, $k= 2(N_{d}+L_{\Omc}+1))$. We take a node $e'$ realizing $t'$ with depth between $N_{d}+L_{\Omc}+1$ and $2(N_{d}+L_{\Omc}+1)$ and connect $e$ and $e'$ with $R$.
     The resulting structures $\mathfrak{F}_{d}$ are finite, and so $\mathfrak{C}$ is finite and behaves in exactly the same way as the original structure $\Cmf$.
     This finishes the proof of Theorem~\ref{thm:char12}.\qed

\medskip

In practice, one would expect that KBs $\Kmc$ are such that no
connected \Kmc-type is \ALCI-complete for $\Kmc$ (while every disconnected
\Kmc-type is necessarily \ALCI-complete for $\Kmc$). It thus makes sense to
consider the following special case. A labeled \ALCI-KB $(\Kmc,P,N)$
is \emph{strongly incomplete} if no connected $\Kmc$-type that is
realizable in some $\Kmc,b$, with $b\in N$, is $\ALCI$-complete.  For
\ALCI-KBs that are strongly incomplete, we can drop Point~3 (b)
from Theorem~\ref{thm:char12} and obtain the following from
Theorem~\ref{critFOwithoutUNA}.
%
\begin{corollary}
	\label{corr:stronglyincompl}
	For labeled \ALCI-KBs that are strongly incomplete, non-projective
	$\ALCI$-separability coincides with non-projective and projective
	$\Lmc_S$-separability for all FO-fragments $\Lmc_S \supseteq
	\text{UCQ}$.
\end{corollary}
It follows from Theorem~\ref{thm:char12} that there is a polynomial time reduction of projective
$(\ALCI,\ALCI)$-separability to non-projective
$(\ALCI,\ALCI)$-separability. Let
$(\Kmc,P,\{b\})$, $\Kmc=(\Omc,\Dmc)$, be a labeled \ALCI-KB.  Then
\Kmc is projectively \ALCI-separable if and only if $(\Kmc',P,\{b\})$
is non-projectively \ALCI-separable where $\Kmc'=(\Omc',\Dmc)$ and
$\Omc'=\Omc \cup \{ A \sqsubseteq A \}$, $A$ a fresh concept name. In
fact, $\Kmc$ is clearly projectively \ALCI-separable iff $\Kmc'$ is,
and $\Kmc'$ is projectively \ALCI-separable iff it is non-projectively
\ALCI-separable because no connected $\Kmc'$-type is \ALCI-complete
and thus Point~3 (b) of Theorem~\ref{thm:char12} is vacuously true for
$\Kmc'$. This also implies that whenever a labeled \ALCI-KB is
projectively separable, then a single fresh concept name 
suffices for separation.

We now also have everything in place to clarify
the complexity of non-projective
$(\mathcal{ALCI},\mathcal{ALCI})$-separability.
\begin{theorem}
	\label{thm:alcinexp}
	Non-projective $(\ALCI,\ALCI)$-separability is $\NExpTime$-complete in combined complexity. This also holds for definability, RE-existence, and entity distinguishability. For separability and definability the \NExpTime-lower bound holds already in data complexity.
\end{theorem}
\begin{proof}\
	The lower bounds are a consequence of
	Theorem~\ref{thm:correctedresults}, Theorem~\ref{thm:correctedresults2},
	and the mentioned reduction of
	projective separability to non-projective separability. For the
	upper bound, we first observe that according to Claim~2 in the proof of Theorem~\ref{thm:char12} it is in \ExpTime to decide whether a given $\Kmc$-type $t$ is
	$\ALCI$-complete. Let $(\Kmc,P,\{b\})$ be a labeled \ALCI-KB.  For
	any $\Kmc$-type $t$, let $\Kmc_{t}=(\Omc_{t},\Dmc_{t})$ where
	$\Omc_{t}= \Omc \cup \{A \sqsubseteq \bigsqcap_{C\in t}C\}$ and
	$\Dmc_{t}=\Dmc\cup \{A(b)\}$ for a fresh concept name $A$. By
	Theorem~\ref{thm:char12}, $(\Kmc,P,\{b\})$ is non-projectively $\ALCI$-separable iff
	there exists a $\Kmc$-type $t$ that is realizable in $\Kmc,b$ such
	that (i)~$\Kmc_{t}\not\models \bigvee_{a\in
		P}\varphi_{\Dmc_{\text{con}(a),a}}(b)$ and~(ii) if $t$ is
	connected and $\ALCI$-complete for $\Kmc$, then $t$ is not realizable in
	$\Kmc,a$ for any $a\in P$. The \NExpTime upper bound now follows
	from the fact that rooted UCQ-evaluation on \ALCI-KBs is in
	\coNExpTime and that \ALCI-completeness of $t$
	and realizability of $t$ in $\Kmc,a$ can be checked in \ExpTime.
	%
\end{proof}


When the ontology in \Kmc is 
empty, then no connected \Kmc-type is \ALCI-complete and thus 
Point~3 (b) of Theorem~\ref{thm:char12} is vacuously true. It follows that 
non-projective and projective \ALCI-separability of KBs 
$(\emptyset,\Dmc)$ coincides with FO-separability and is 
\coNPclass-complete (Theorem~\ref{thm:empty}). The same is true for
definability, RE-existence, and entity distinguishability.

\subsection{Separability of Labeled GF-KBs}
\label{section:GF}

We study projective and non-projective
$(\text{GF},\text{GF})$-separability which turns out to behave similarly
to the \ALCI case in many ways. The non-projective case is, however, significantly more difficult to analyse. A new aspect we consider for GF is a comparison with separability in the fragment openGF of GF in which one can only speak locally about neighbourhoods of tuples and not about disconnected parts. It turns out that separability is not affected by this restriction, but the length of separating formulas is.


We start with an example which shows that projective and
non-projective $(\text{GF},\text{GF})$-separability do not coincide.
Note that Example~\ref{exm:11} does not serve this purpose since the
labeled KB given there is separable by the GF-formula $R(x,x)$.  We
use the more succinct \ALCI-syntax for GF-formulas and ontologies
whenever possible.
\begin{exmp}\label{ex:GF} {\em Define $\Kmc=(\Omc,\Dmc)$ where
$$
	\Omc  = \{\top \sqsubseteq \exists R.\top \sqcap \exists
	R^{-}.\top,\ 
	\forall x \forall y(R(x,y) \rightarrow \neg R(y,x))\} 
$$
and $\Dmc  = \{R(a,c), R(c,d), R(d,a), R(b,e)\}$ is depicted below:
	\begin{center}
		\tikzset{every picture/.style={line width=0.55pt}} 
		
		\begin{tikzpicture}[x=0.75pt,y=0.75pt,yscale=-0.6,xscale=0.6]
		
		\draw    (83.75,114.5) -- (115.97,74.86) ;
		\draw [shift={(116.5,74.17)}, rotate = 487.38] [fill={rgb, 255:red, 0; green, 0; blue, 0 }  ][line width=0.08]  [draw opacity=0] (8.93,-4.29) -- (0,0) -- (8.93,4.29) -- cycle    ;
		
		\draw    (150.25,119.5) -- (91.25,119.98) ;
		\draw [shift={(88.25,120)}, rotate = 359.53999999999996] [fill={rgb, 255:red, 0; green, 0; blue, 0 }  ][line width=0.08]  [draw opacity=0] (8.93,-4.29) -- (0,0) -- (8.93,4.29) -- cycle    ;
		
		\draw    (125.25,77) -- (156.28,112.73) ;
		\draw [shift={(158.25,115)}, rotate = 229.03] [fill={rgb, 255:red, 0; green, 0; blue, 0 }  ][line width=0.08]  [draw opacity=0] (8.93,-4.29) -- (0,0) -- (8.93,4.29) -- cycle    ;
		
		\draw  [fill={rgb, 255:red, 0; green, 0; blue, 0 }  ,fill opacity=1 ] (116.6,70.23) .. controls (116.6,68.41) and (118.08,66.93) .. (119.9,66.93) .. controls (121.72,66.93) and (123.2,68.41) .. (123.2,70.23) .. controls (123.2,72.06) and (121.72,73.53) .. (119.9,73.53) .. controls (118.08,73.53) and (116.6,72.06) .. (116.6,70.23) -- cycle ;
		\draw  [fill={rgb, 255:red, 0; green, 0; blue, 0 }  ,fill opacity=1 ] (157.4,119.83) .. controls (157.4,118.01) and (158.88,116.53) .. (160.7,116.53) .. controls (162.52,116.53) and (164,118.01) .. (164,119.83) .. controls (164,121.66) and (162.52,123.13) .. (160.7,123.13) .. controls (158.88,123.13) and (157.4,121.66) .. (157.4,119.83) -- cycle ;
		\draw  [fill={rgb, 255:red, 0; green, 0; blue, 0 }  ,fill opacity=1 ] (77,120.23) .. controls (77,118.41) and (78.48,116.93) .. (80.3,116.93) .. controls (82.12,116.93) and (83.6,118.41) .. (83.6,120.23) .. controls (83.6,122.06) and (82.12,123.53) .. (80.3,123.53) .. controls (78.48,123.53) and (77,122.06) .. (77,120.23) -- cycle ;
		\draw  [fill={rgb, 255:red, 0; green, 0; blue, 0 }  ,fill opacity=1 ] (237,90.23) .. controls (237,88.41) and (238.48,86.93) .. (240.3,86.93) .. controls (242.12,86.93) and (243.6,88.41) .. (243.6,90.23) .. controls (243.6,92.06) and (242.12,93.53) .. (240.3,93.53) .. controls (238.48,93.53) and (237,92.06) .. (237,90.23) -- cycle ;
		\draw    (243.6,90.23) -- (296.55,89.25) ;
		\draw [shift={(297.42,89.23)}, rotate = 538.87] [fill={rgb, 255:red, 0; green, 0; blue, 0 }  ][line width=0.08]  [draw opacity=0] (8.93,-4.29) -- (0,0) -- (8.93,4.29) -- cycle    ;
		
		\draw  [fill={rgb, 255:red, 0; green, 0; blue, 0 }  ,fill opacity=1 ] (299.42,89.23) .. controls (299.42,87.41) and (300.9,85.93) .. (302.72,85.93) .. controls (304.55,85.93) and (306.02,87.41) .. (306.02,89.23) .. controls (306.02,91.06) and (304.55,92.53) .. (302.72,92.53) .. controls (300.9,92.53) and (299.42,91.06) .. (299.42,89.23) -- cycle ;
		
		\draw (66.9,131) node  [font=\small] [align=left] {$\displaystyle a$};
		\draw (229.9,103) node  [font=\small] [align=left] {$\displaystyle b$};
		\draw (121.9,55) node  [font=\small] [align=left] {$\displaystyle c$};
		\draw (176.9,125) node  [font=\small] [align=left] {$\displaystyle d$};
		\draw (315.4,104.5) node  [font=\small] [align=left] {$\displaystyle e$};

		\end{tikzpicture}
	\end{center}
	\vspace*{-2mm}
	%
	The labeled GF-KB $(\Kmc,\{a\},\{b\})$ is separated by the
	\ALCI-concept $ C=A \rightarrow \exists R.\exists R.\exists R.A $
	that uses the concept name $A$ as a helper symbol. In contrast, the
	KB is not non-projectively GF-separable since every GF-formula
	$\varphi(x)$ with $\mn{sig}(\varphi) = \{ R\}$ is either valid (equivalent to
	$x=x$) or unsatisfiable (equivalent to $\neg(x=x)$) w.r.t.~\Omc.
    It follows that if $\Kmc\models \varphi(a)$, then $\varphi$ is valid w.r.t. $\Omc$ and so $\Kmc\models \varphi(b)$.
	
	To illustrate the role of the second sentence in \Omc, let $\Omc^-$
	be $\Omc$ without that sentence. Then $\Kmc^-=(\Omc^-,\Dmc)$ is
	separated by the GF-formula obtained from the separating
	\ALCI-concept $C$ above by replacing each occurrence of $A(x)$ in
	$C^\dagger$ by the formula $\chi(x)=\exists y (R(x,y) \wedge R(y,y))$ (the resulting formulas is equivalent to $\chi(x) \rightarrow \exists y(R^{3}(x,z) \wedge \chi(z))$ and we therefore have non-projective GF-separability.\finofex
}
\end{exmp}
Let \emph{openGF} be the fragment of GF that consists of all open 
formulas in GF whose subformulas are all open and in which equality is not 
used as a guard. For example, $A(x) \wedge \exists x B(x)$ is in GF but not in openGF because it contains a closed subformula. OpenGF was first considered in 
\cite{tocl2020} where it is also observed that an open GF formula $\varphi(\vec{x})$ is equivalent to a formula in openGF if and only if
for all structures $\Amf$ and $\vec{a}\in \text{dom}(\Amf)^{|\vec{x}|}$ $\Amf\models \varphi(\vec{a})$ iff $\Amf_{\text{con}(\vec{a})}\models \varphi(\vec{a})$, where $\Amf_{\text{con}(\vec{a})}$ is the restriction
of $\Amf$ to all $b\in \text{dom}(\Amf)$ that are reachable from some
$a\in [\vec{a}]$ in the Gaifman graph of $\Amf$.
Thus, openGF can only speak about the neighourhood of $\vec{a}$ and its truth does not depend on any disconnected parts of $\Amf$. Informally, openGF relates to GF in the same way as \ALCI relates to the extension of \ALCI with the 
universal role \cite{DL-Textbook}. We start our 
investigation with observing the following. 
\begin{theorem}\label{thm:GFopen}
	$(\text{GF},\text{GF})$-separability coincides with 
	$(\text{GF},\text{openGF})$-separability, both in the projective 
	and in the non-projective case. 
\end{theorem}
Theorem~\ref{thm:GFopen} will follow from model-theoretic characterizations of
separability provided below. Arguably, openGF formulas are more natural for
separation purposes than unrestricted GF formulas as they speak only about the neighbourhood of the examples. The next example shows that this is at the expense of larger
separating formulas (a slightly modified example shows the same behaviour
for $\mathcal{ALCI}$ and its extension with the universal role).
\begin{exmp}\label{exmp:openGFGF}
	{\em Let 
	$$\Omc=\{A \sqsubseteq \forall R.A, \ \forall xy(R(x,y)\rightarrow
	\neg R(y,x))\}
	$$ and let  $\Dmc_{n}$ contain two $R$-paths of length $n$,
	$a_0 R a_1 R \dots R a_n$ and $b_0 R b_1 R \dots R b_n$ with $a_n$ labeled with $E$: 
	\begin{center}
		\tikzset{every picture/.style={line width=0.75pt}} 

		\tikzset{every picture/.style={line width=0.75pt}} 
		
		\begin{tikzpicture}[x=0.75pt,y=0.75pt,yscale=-1,xscale=1]
			
			\draw    (132,124) -- (150.5,124.22) ;
			\draw [shift={(153.5,124.25)}, rotate = 180.67] [fill={rgb, 255:red, 0; green, 0; blue, 0 }  ][line width=0.08]  [draw opacity=0] (4.93,-2.29) -- (0,0) -- (4.93,2.29) -- cycle    ;
			\draw    (176,124.5) -- (194.5,124.72) ;
			\draw [shift={(197.5,124.75)}, rotate = 180.67] [fill={rgb, 255:red, 0; green, 0; blue, 0 }  ][line width=0.08]  [draw opacity=0] (4.93,-2.29) -- (0,0) -- (4.93,2.29) -- cycle    ;
			\draw    (222,124) -- (240.5,124.22) ;
			\draw [shift={(243.5,124.25)}, rotate = 180.67] [fill={rgb, 255:red, 0; green, 0; blue, 0 }  ][line width=0.08]  [draw opacity=0] (4.93,-2.29) -- (0,0) -- (4.93,2.29) -- cycle    ;
			\draw    (132,156) -- (150.5,156.22) ;
			\draw [shift={(153.5,156.25)}, rotate = 180.67] [fill={rgb, 255:red, 0; green, 0; blue, 0 }  ][line width=0.08]  [draw opacity=0] (4.93,-2.29) -- (0,0) -- (4.93,2.29) -- cycle    ;
			\draw    (176,156.5) -- (194.5,156.72) ;
			\draw [shift={(197.5,156.75)}, rotate = 180.67] [fill={rgb, 255:red, 0; green, 0; blue, 0 }  ][line width=0.08]  [draw opacity=0] (4.93,-2.29) -- (0,0) -- (4.93,2.29) -- cycle    ;
			\draw    (222,156) -- (240.5,156.22) ;
			\draw [shift={(243.5,156.25)}, rotate = 180.67] [fill={rgb, 255:red, 0; green, 0; blue, 0 }  ][line width=0.08]  [draw opacity=0] (4.93,-2.29) -- (0,0) -- (4.93,2.29) -- cycle    ;
			\draw  [dash pattern={on 0.84pt off 2.51pt}] (104.25,122.63) .. controls (104.25,114.55) and (111.69,108) .. (120.88,108) .. controls (130.06,108) and (137.5,114.55) .. (137.5,122.63) .. controls (137.5,130.7) and (130.06,137.25) .. (120.88,137.25) .. controls (111.69,137.25) and (104.25,130.7) .. (104.25,122.63) -- cycle ;
			\draw  [dash pattern={on 0.84pt off 2.51pt}] (104.25,155.63) .. controls (104.25,147.55) and (111.69,141) .. (120.88,141) .. controls (130.06,141) and (137.5,147.55) .. (137.5,155.63) .. controls (137.5,163.7) and (130.06,170.25) .. (120.88,170.25) .. controls (111.69,170.25) and (104.25,163.7) .. (104.25,155.63) -- cycle ;
			\draw  [dash pattern={on 0.84pt off 2.51pt}]  (203.5,124.25) -- (217.5,124.25) ;
			\draw  [dash pattern={on 0.84pt off 2.51pt}]  (202.5,156.25) -- (216.5,156.25) ;
			
			\draw (120,122) node   [align=left] {$\displaystyle a_{0}$};
			\draw (166,123.5) node   [align=left] {$\displaystyle a_{1}$};
			\draw (255,125) node   [align=left] {$\displaystyle a_{n}$};
			\draw (120,154) node   [align=left] {$\displaystyle b_{0}$};
			\draw (166,155.5) node   [align=left] {$\displaystyle b_{1}$};
			\draw (255,155) node   [align=left] {$\displaystyle b_{n}$};
			\draw (255,110) node   [align=left] {$\displaystyle E$};
			\draw (90,120.5) node   [align=left] {$\displaystyle +$};
			\draw (90,154.5) node   [align=left] {$\displaystyle -$};

		\end{tikzpicture}
		
	\end{center}
	Consider the labeled GF-KB $(\Kmc_{n},\{ a_0 \},\{b_0\})$ with
	$\Kmc_{n}=(\Omc,\Dmc_{n})$. Then the GF-formula $A(x)\rightarrow \exists
	y(A(y) \wedge E(y))$ separates $(\Kmc_{n},\{ a_0\},\{b_0\})$. An openGF formula separating $(\Kmc_{n},\{a_{0}\},\{b_{0}\})$ is given by the $\mathcal{ALCI}$-concept $A \rightarrow \exists R^{n}.(A \sqcap E)$ and we show in Proposition~\ref{prop:example1} below that the shortest separating openGF-formula has 
	length at least $n$ (even if it uses helper symbols).\finofex
}
\end{exmp}
To characterize separability in GF and openGF we define guarded 
bisimulations, a standard tool 
for proving that two
structures satisfy the same guarded formulas~\cite{DBLP:books/daglib/p/Gradel014,tocl2020}. Guarded bisimulations generalize $\mathcal{ALCI}$-bisimulations.
Let $\Amf$ be structure. 
A set $G \subseteq \text{dom}(\Amf)$ is \emph{guarded} in $\mathfrak{A}$ if $G$ is a singleton or 
there exists $R$ with $\Amf\models R(\vec{a})$ such that $G = [\vec{a}]$. 
A tuple $\vec{a}$ in $\text{dom}(\Amf)$ is
\emph{guarded} in $\mathfrak{A}$ if $[\vec{a}]$ is a 
subset of some guarded set in~$\mathfrak{A}$. 

Let $\Sigma$ be a signature. The restriction of a structure $\Amf$ 
to a nonempty subset $A$ of $\text{dom}(\Amf)$ is denoted $\Amf_{|A}$. 
The \emph{$\Sigma$-reduct} of $\Amf$ coincides with $\Amf$ except that all symbols not in $\Sigma$ are interpreted by the empty set.
For tuples $\vec{a}=(a_{1},\ldots,a_{n})$ in $\mathfrak{A}$
and $\vec{b}=(b_{1},\ldots,b_{n})$ in $\mathfrak{B}$ we call a mapping $p$ from $[\vec{a}]$ to
$[\vec{b}]$ with $p(a_{i})=b_{i}$ for $1\leq i \leq n$ (written $p:\vec{a}\mapsto \vec{b}$)
a \emph{partial $\Sigma$-homomorphism} if $p$ is a homomorphism from the $\Sigma$-reduct of $\mathfrak{A}_{|[\vec{a}]}$ 
to $\mathfrak{B}_{|[\vec{b}]}$. We call $p$ a \emph{partial $\Sigma$-isomorphism} if, in addition,
the inverse of $p$ is a partial $\Sigma$-homomorphism from  $\mathfrak{B}_{|[\vec{b}]}$ to $\mathfrak{A}_{|[\vec{a}]}$. 

A set $I$ of partial $\Sigma$-isomorphisms $p: \vec{a} \mapsto \vec{b}$ 
from guarded tuples $\vec{a}$ in $\mathfrak{A}$ to guarded tuples $\vec{b}$ in $\Bmf$ is called a 
\emph{connected guarded $\Sigma$-bisimulation} if the following hold for all 
$p: \vec{a} \mapsto \vec{b}\in I$:
\begin{enumerate}
	\item[(i)] for every guarded tuple $\vec{a}'$ in $\Amf$ with $[\vec{a}]\cap [\vec{a}']\not=\emptyset$
	there exists a guarded tuple $\vec{b}'$ in $\Bmf$ and $p': \vec{a}'\mapsto \vec{b}'\in I$ such that $p'$ 
	and $p$ coincide on $[\vec{a}]\cap [\vec{a}']$.
	\item[(ii)] for every guarded tuple $\vec{b}'$ in $\Bmf$ with $[\vec{b}]\cap [\vec{b}']\not=\emptyset$
	there exists a guarded tuple $\vec{a}'$ in $\Amf$ and $p': \vec{a}'\mapsto \vec{b}'\in I$ such that
	$p'^{-1}$ and $p^{-1}$ coincide on $[\vec{b}]\cap [\vec{b}']$.
\end{enumerate}
Assume that $\vec{a}$ and $\vec{b}$ are (possibly not guarded) tuples in $\Amf$
and $\Bmf$. Then we say that $\Amf,\vec{a}$ and $\Bmf,\vec{b}$ are \emph{connected guarded $\Sigma$-bisimilar},
in symbols $\Amf,\vec{a} \sim_{\text{openGF},\Sigma} \Bmf,\vec{b}$,
if there exists a partial $\Sigma$-isomorphism $p: \vec{a}\mapsto \vec{b}$ and a 
connected guarded $\Sigma$-bisimulation $I$ such that Conditions (i) and (ii) hold for $p$~\cite{tocl2020}. 

Connected guarded $\Sigma$-bisimulations differ from the standard guarded
$\Sigma$-bisimulations \cite{DBLP:books/daglib/p/Gradel014} in requiring $[\vec{a}]\cap
[\vec{a}']\not=\emptyset$ in Condition~(i) and $[\vec{b}]\cap
[\vec{b}']\not=\emptyset$ in Condition~(ii). If we drop these conditions then
we talk about \emph{guarded $\Sigma$-bisimulations} and \emph{guarded $\Sigma$-bisimilarity},
in symbols $\Amf,\vec{a} \sim_{\text{GF},\Sigma} \Bmf,\vec{b}$.

We also use finitary versions of guarded bisimulations.
In these versions of (connected) guarded bisimulations the 
Conditions~(i) and (ii)
are required to hold a finite number $\ell \geq 0$ of rounds only. Thus, one considers sets
$I_{\ell},\ldots,I_{0}$ of partial $\Sigma$-isomorphisms such that $I_{\ell}$ 
contains the partial $\Sigma$-isomorphism $p:\vec{a}\mapsto \vec{b}$ and 
for any $p\in I_{i}$ there
exist $p'\in I_{i-1}$ such that (i) and, respectively, (ii) hold, for $0<i\leq \ell$. 
If such sets exist then we say that $\Amf,\vec{a}$ and $\Bmf,\vec{b}$ are 
\emph{(connected) guarded $\Sigma$ $\ell$-bisimilar}
and write $\Amf,\vec{a} \sim_{\text{openGF},\Sigma}^{\ell} \Bmf,\vec{b}$ and 
$\Amf,\vec{a} \sim_{\text{GF},\Sigma}^{\ell} \Bmf,\vec{b}$, respectively. 

We say that $\Amf,\vec{a}$ and $\Bmf,\vec{b}$ are \emph{GF$(\Sigma)$-equivalent}, in symbols
$\Amf,\vec{a} \equiv_{\text{GF},\Sigma} \Bmf,\vec{b}$, if $\Amf\models \varphi(\vec{a})$
iff $\Bmf\models \varphi(\vec{b})$ for all formulas $\varphi(\vec{x})$ in GF$(\Sigma)$.
The \emph{guarded quantifier rank} $\text{gr}(\varphi)$ of a formula $\varphi$ in GF is
the number of nestings of guarded quantifiers in it.
We say that $\Amf,\vec{a}$ and $\Bmf,\vec{b}$ are \emph{GF$^{\ell}(\Sigma)$-equivalent}, in symbols
$\Amf,\vec{a} \equiv_{\text{GF},\Sigma}^{\ell} \Bmf,\vec{b}$, if $\Amf\models \varphi(\vec{a})$
iff $\Bmf\models \varphi(\vec{b})$ for all formulas $\varphi(\vec{x})$ in GF$(\Sigma)$
of guarded quantifier rank at most $\ell$. The same notation is applied for openGF. For example, we say that $\Amf,\vec{a}$ and $\Bmf,\vec{b}$ are \emph{openGF$(\Sigma)$-equivalent}, in symbols
$\Amf,\vec{a} \equiv_{\text{openGF},\Sigma} \Bmf,\vec{b}$, if $\Amf\models \varphi(\vec{a})$
iff $\Bmf\models \varphi(\vec{b})$ for all formulas $\varphi(\vec{x})$ in openGF$(\Sigma)$.
We say that $\Amf,\vec{a}$ and $\Bmf,\vec{b}$ are \emph{openGF$^{\ell}(\Sigma)$-equivalent}, in symbols
$\Amf,\vec{a} \equiv_{\text{openGF},\Sigma}^{\ell} \Bmf,\vec{b}$, if $\Amf\models \varphi(\vec{a})$
iff $\Bmf\models \varphi(\vec{b})$ for all formulas $\varphi(\vec{x})$ in openGF$(\Sigma)$
of guarded quantifier rank at most $\ell$. The following is shown
in~\cite{ANvB98,DBLP:books/daglib/p/Gradel014,tocl2020} and the notion
of $\omega$-saturated structures is discussed, for example, in~\cite{modeltheory}. 
\begin{lemma}\label{lem:guardedbisim}
	Let $\Amf,\vec{a}$ and $\Bmf,\vec{b}$ be pointed structures and $\Sigma$ a signature.
	Then for $\Lmc\in \{\text{GF},\text{openGF}\}$ and all $\ell\geq 0$:
	$$
	\Amf,\vec{a} \equiv_{\Lmc,\Sigma}^{\ell} \Bmf,\vec{b} \quad \text{ iff } \quad 
	\Amf,\vec{a} \sim_{\Lmc,\Sigma}^{\ell} \Bmf,\vec{b}.
	$$
	Moreover,
	$$
	\Amf,\vec{a} \sim_{\Lmc,\Sigma} \Bmf,\vec{b} \quad \text{ implies } \quad
	\Amf,\vec{a} \equiv_{\Lmc,\Sigma} \Bmf,\vec{b}
	$$
	and, conversely, if either $\Amf$ is finite or both $\Amf$ and $\Bmf$ are $\omega$-saturated, then
	$$
	\Amf,\vec{a} \equiv_{\Lmc,\Sigma} \Bmf,\vec{b} \quad \text{ implies } \quad 
	\Amf,\vec{a} \sim_{\Lmc,\Sigma} \Bmf,\vec{b}
	$$
\end{lemma}
We illustrate connected guarded bisimulations by using them to prove the claim made in Example~\ref{exmp:openGFGF} above. 
\begin{proposition}\label{prop:example1}
Let $(\Kmc_{n},\{a_{0}\},\{b_{0}\})$ be as in Example~\ref{exmp:openGFGF}. Then
any openGF-formula separating $(\Kmc_{n},\{a_{0}\},\{b_{0}\})$ has guarded quantifier rank at least $n$ (and so length at least $n$).
\end{proposition}
\begin{proof} \
Let $\Sigma=\{A,E,R\}=\text{sig}(\mathcal{K}_{n})$. To prove that no openGF-formula of guarded quantifier rank $m < n$ separates
$(\mathcal{K}_{n}, \{a_{0}\},\{b_{0}\})$, it is sufficient to prove that for all models $\Amf$ of $\Kmc_{n}$ and signatures $\Sigma_{\text{help}}$ of helper symbols there exists 
a model $\mathfrak{B}$ of $\Kmc_{n}$ such that $\mathfrak{A},b_{0}^{\Amf} \sim^m_{\text{openGF},\Sigma\cup\Sigma_{\text{help}}} \mathfrak{B}, a_{0}^{\Bmf}$. (To see this assume
that there exists a formula $\varphi(x)$ in openGF of guarded quantifier rank $m<n$ using symbols in $\Sigma\cup\Sigma_{\text{help}}$ such that $\Kmc_{n} \models \varphi(a_{0})$
and $\Kmc_{n}\not\models \varphi(b_{0})$. Take any model $\Amf$ of $\Kmc_{n}$ such that $\Amf\not\models \varphi(b_{0})$.
By Lemma~\ref{lem:guardedbisim}, $\Amf\not\models \varphi(a_{0})$, and we have derived a contradiction.) Now
let $\mathfrak{A}$ be a model of $\mathcal{K}_{n}$. Observe that  $b_{i}^{\Amf}\not=b_{i+1}^{\Amf}$ since $\forall xy(R(x,y)\rightarrow \neg R(y,x))\in \Omc$. We can unfold $\Amf$ into a guarded tree-decomposable model $\Amf^{\ast}$ of $\Omc$ (see the appendix and~\cite{tocl2020})
such that there is a sequence $d_{0},\ldots,d_{n}\in \text{dom}(\Amf^{\ast})$
with 
\begin{itemize}
	\item $\Amf,b_{0}^{\Amf} \sim_{\text{openGF},\Sigma\cup\Sigma_{\text{help}}} \Amf^{\ast}, d_{0}$;
	\item $\Amf,(b_{i}^{\Amf},b_{i+1}^{\Amf}) \sim_{\text{openGF},\Sigma\cup\Sigma_{\text{help}}} \Amf^{\ast}, (d_{i},d_{i+1})$, for all $i< n$;
    \item $d_{0}$ and $d_{i}$ have distance exactly $i$ in $\Amf^{\ast}$, for 
     all $i\leq n$.
\end{itemize}
Define a structure $\Bmf'$ as follows: the domain of $\Bmf'$ and the interpretation of symbols in $\Sigma\cup \Sigma_{\text{help}}$ is given by taking the
disjoint union of $\Amf^{\ast}$ and $\Amf$. For the interpretation of the constants $a_{0},\ldots,a_{n}$ take the path $d_{0},\ldots,d_{n}\in
\text{dom}(\Amf^{\ast})$ and set $a_{i}^{\Bmf'}:= d_{i}$ for $i\leq n$. 
The constants $b_{0},\ldots,b_{n}$ are interpreted by setting $b_{i}^{\Bmf'}:=b_{i}^{\Amf}$
for $i\leq n$. Note that $\Amf,b_{0}^{\Amf}\sim_{\text{openGF},\Sigma\cup\Sigma_{\text{help}}}\Bmf',a_{0}^{\Bmf'}$ follows from 
$\Amf,b_{0}^{\Amf} \sim_{\text{openGF},\Sigma\cup\Sigma_{\text{help}}} \Amf^{\ast}, d_{0}$. However, $\Bmf'$ might not yet be a model of $\Dmc_{n}$
since $d_{n}\in E^{\Bmf'}$ is not guaranteed. We therefore add $d_{n}$ to the interpretation of $E$ and denote the resulting structure by $\Bmf$.
Then $\Bmf$ is a model of $\Kmc_{n}$ and, as $d_{0}$ and $d_{n}$ have distance exactly $n$ in $\Amf^{\ast}$ the addition of $E$ only becomes 
relevant after $n-1$ steps of the connected guarded bisimulation witnessing $\Amf,b_{0}^{\Amf}\sim_{\text{openGF},\Sigma\cup\Sigma_{\text{help}}}\Bmf',a_{0}^{\Bmf'}$.
We thus have that $\Amf,b_{0}^{\Amf}\sim_{\text{openGF},\Sigma\cup\Sigma_{\text{help}}}^{n-1}\Bmf,a_{0}^{\Bmf}$, as required.
\end{proof}
We next extend the notion of a functional $\mathcal{ALCI}(\Sigma)$-bisimulation to the guarded case. We write
$\Amf,\vec{a} \sim_{\text{openGF},\Sigma}^{f} \Bmf,\vec{b}$ if there exists a
set $I$ of partial $\Sigma$-isomorphisms witnessing that 
$\Amf,\vec{a} \sim_{\text{openGF},\Sigma} \Bmf,\vec{b}$ such that
$p(a) = p'(a)$ for all $p,p'\in I$ and $a\in \text{dom}(\Amf)$ such that
$a$ is in the domain of both $p$ and $p'$. Then $I$ is called a \emph{connected
functional guarded $\Sigma$-bisimulation}. We are now in a position to formulate a counterpart of Theorem~\ref{thm:L-modeltheory0} for $(\text{GF},\text{openGF})$-separability.
\begin{restatable}{theorem}{thmLmodeltheorygfnull}
	\label{thm:L-modeltheorygf0}
	Assume that $(\Kmc,P,\{\vec{b}\})$ is a labeled GF-KB with $\Kmc=(\Omc,\Dmc)$
	and $\Sigma=\text{sig}(\Kmc)$. Then the following
	conditions are equivalent:
	\begin{enumerate}
		
		\item $(\Kmc,P,\{\vec{b}\})$ is projectively
		openGF-separable;
		
		\item there exists a finite model $\Amf$ of $\Kmc$ and a finite set $\Sigma_{\text{help}}$ of unary relation symbols disjoint from $\text{sig}(\Kmc)$ such that for all models $\Bmf$ of
		$\Kmc$ and all $\vec{a}\in P$: $\Bmf,\vec{a}^{\Bmf}
		\not\sim_{\text{openGF},\Sigma\cup \Sigma_{\text{help}}} \Amf, \vec{b}^{\Amf}$;
		
		\item there exists a finite model $\Amf$ of $\Kmc$ such that for all models $\Bmf$ of $\Kmc$ and all
		$\vec{a}\in P$: $\Bmf,\vec{a}^{\Bmf} \not\sim_{\text{openGF},\Sigma}^{f} \Amf,
		\vec{b}^{\Amf}$; 
		
		\item there exists a finite model $\Amf$ of $\Kmc$
		  such that for all $\vec{a} \in P$:
		$\Dmc_{\text{con}(\vec{a})},\vec{a}\not\rightarrow \Amf,\vec{b}^{\Amf}$;
		
		\item there exists a model $\Amf$ of $\Kmc$ such that for all $\vec{a}\in P$:
		$\Dmc_{\text{con}(\vec{a})},\vec{a}\not\rightarrow \Amf,\vec{b}^{\Amf}$.		
	\end{enumerate}
\end{restatable}
\begin{proof} \
	The proof is essentially the same as the proof of Theorem~\ref{thm:L-modeltheory0} with (functional) $\mathcal{ALCI}(\Sigma)$-bisimulations replaced by (functional) connected guarded bisimulations. For the implication ``2. $\Rightarrow$ 1.'' one uses Lemma~\ref{lem:guardedbisim} instead of Lemma~\ref{lem:equivalence}. For
	``4. $\Rightarrow$ 3.'' consider a model $\Amf$ of $\Kmc$ such that Condition~4 holds but assume there exists a model $\Bmf$ of $\Kmc$ and a connected functional guarded bisimulation $I$ refuting Condition~3 for some $\vec{a}\in P$ for $\Amf$. Then $f=\bigcup_{p\in I}p$ is a function and thus, as every $p\in I$ is a $\Sigma$-isomorphism, it is a $\Sigma$-homomorphism from
	$\Bmf_{\text{con}(\vec{a}^{\Bmf})},\vec{a}^{\Bmf}$ to $\Amf,\vec{b}^{\Amf}$. Let $h$ be the mapping from $\text{cons}(\Dmc_{\text{con}(\vec{a})})$ to $\text{dom}(\Bmf)$ defined by setting $h(c)=c^{\Bmf}$. Then $h$ is a homomorphism from $\Dmc_{\text{con}(\vec{a})}$ to $\Bmf$ and so the composition $f\circ h$ is a homomorphism witnessing $\Dmc_{\text{con}(\vec{a})},\vec{a}\rightarrow \Amf,\vec{b}^{\Amf}$. We have derived a contradiction. For ``5. $\Rightarrow$ 4.'' observe that UCQ-evaluation on GF-KBs is finitely controllable~\cite{DBLP:journals/corr/BaranyGO13}.
\end{proof}
Observe that the projective case of Theorem~\ref{thm:GFopen} follows directly from Theorem~\ref{thm:L-modeltheorygf0} since even projective (GF,openGF)-separability and (GF,FO)-separability coincide by Point~4 and
$\text{openGF} \subseteq \text{GF} \subseteq \text{FO}$. 

UCQ-evaluation is decidable in \TwoExpTime on GF-KBs~\cite{DBLP:journals/corr/BaranyGO13}
and satisfiability of GF-sentences is 2\ExpTime-hard~\cite{DBLP:journals/jsyml/Gradel99}.
We therefore obtain the following result in the same way as Corollary~\ref{thm:gnfo}.
\begin{theorem}
	\label{thm:gfo}
	For any FO-fragment $\Lmc_{S}\supseteq \text{UCQ}$:
	\begin{enumerate}
		\item projective $(\text{GF},\text{GF})$-separability,
		projective $(\text{GF},\text{openGF})$-separability, and
		$(\text{GF},\Lmc_S)$-separability coincide, and the same is true for definability;
		\item projective $(\text{GF},\text{GF})$-separability and projective $(\text{GF},\text{GF})$-definability are \TwoExpTime-complete in combined complexity.
	\end{enumerate}
	This holds also for RE-existence and entity distinguishability, for
	all FO-fragments $\Lmc_{S}\supseteq \text{CQ}$.	
\end{theorem}
Regarding data complexity, it follows from the \NExpTime-hardness of
projective $(\ALCI,\ALCI)$-separability that projective
$(\text{GF},\text{GF})$-separability is also \NExpTime-hard in data
complexity. We conjecture that it is actually \TwoExpTime-hard in data
complexity. Some justification for this conjecture is provide
in~\ref{app:hardness}, following the proof of \NExpTime-hardness of
projective $(\ALCI,\ALCI)$-separability.

%
We now consider non-projective separability. Let $\Kmc=(\Omc,\Dmc)$ be a GF-KB. The main result characterizing non-projective separability is very similar to non-projective $\mathcal{ALCI}$: again, in addition to UCQ-evaluation one has to understand complete types that distinguish the non-projective from the projective case. For each $n \geq 1$, fix a tuple
of distinct variables $\vec{x}_n$ of length~$n$.  We use $\cl(\Kmc)$
to denote the smallest set of GF-formulas that is closed under
subformulas and single negation and contains: all formulas from \Omc;
$x=y$ for distinct variables $x,y$; for all $R\in \mn{sig}(\Kmc)$ of
arity $n\geq 2$ and all distinct $x,y \in [\vec{x}_n]$, the formulas
$R(\vec{x}_n)$, $\exists \vec{y}_1 \, (R(\vec{x}_n) \wedge
x\not=y)$ where $\vec{y}_1$ is $\vec{x}_n$ without $x$, and
$\exists \vec{y}_2 \, R(\vec{x}_n)$ for all $\vec{y}_2$ with
$[\vec{y}_2] \subseteq [\vec{x}_n] \setminus \{ x,y \}$ (for $n\geq 3$).
%
Let $\Amf$ be a model of $\Kmc$ and $\vec{a}$ a tuple in $\Amf$.
The \emph{$\Kmc$-type of $\vec{a}$ in $\Amf$} is defined as
$$
\text{tp}_{\Kmc}(\Amf,\vec{a})= \{ \varphi\mid \Amf\models \varphi(\vec{a}), \varphi\in \cl(\Kmc)[\vec{x}]\},
$$
where $\cl(\Kmc)[\vec{x}]$ is obtained from $\cl(\Kmc)$ by
substituting in any formula $\varphi\in \cl(\Kmc)$ the free variables of $\varphi$ 
by variables in $\vec{x}$ in all possible ways, $\vec{x}$ a tuple of distinct variables 
of the same length as $\vec{a}$. Any such $\Kmc$-type of some $\vec{a}$ in
a model $\Amf$ of $\Kmc$ is called a \emph{$\Kmc$-type} and denoted $\Phi(\vec{x})$. $\Phi(\vec{x})$ is
\emph{realizable in a pair $\Kmc,\vec{b}$} with $\vec{b}$ a tuple in $\text{cons}(\Dmc)$ if there exists a model $\Amf$ of
$\Kmc$ with $\text{tp}_{\Kmc}(\Amf,\vec{b}^{\Amf})=\Phi(\vec{x})$.
Recall that in $\ALCI$, a type is \emph{connected} if it contains a concept of the form $\exists R.\top$. We generalize this notion to GF by demanding that the type $\Phi(\vec{x})$ contains a formula of the form 
$\exists
\vec{y}_1 \, (R(\vec{x}_{n}) \wedge x \not=y)$ from above, for some relation $R$ of arity $n\geq 2$ and with $x$ replaced by some $x_{i}$ from $\vec{x}$. 
Otherwise $\Phi(\vec{x})$ is called \emph{disconnected}.  The definition of a complete type is the same as for $\mathcal{ALCI}$.
\begin{definition}
	Let $\Kmc$ be a GF-KB.  A $\Kmc$-type $\Phi(\vec{x})$ is
	\emph{openGF-complete for $\Kmc$} if for any two pointed models
	$\Amf_{1},\vec{b}_{1}$ and $\Amf_{2},\vec{b}_{2}$ of $\Kmc$,
	$\Phi(\vec{x})=\text{tp}_{\Kmc}(\Amf_{1},\vec{b}_{1})=\text{tp}_{\Kmc}(\Amf_{2},\vec{b}_{2})$
	implies
	$
	\Amf_{1},\vec{b}_{1}\equiv_{\text{openGF},\Sigma} \Amf_{2},\vec{b}_{2}.
	$
\end{definition}
Observe that every disconnected $\Kmc$-type is openGF-complete. The following example illustrates the notion of $\Kmc$-types in GF.
\begin{exmp}
	{\em Consider the KB $\Kmc^{-}=(\Omc^{-},\Dmc)$ from Example~\ref{ex:GF}.
		In contrast to $\ALCI$ (where there exists only a single $\Kmc^{-}$-type $\{\top,\exists R.\top,\exists R^{-}.\top,\exists R.\top\sqcap \exists R^{-}.\top\}$ which is connected and $\ALCI$-complete for $\Kmc^{-}$), in GF there are multiple $\Kmc^{-}$-types with a single free variable $x$:
		\begin{itemize}
			\item there exists a uniquely determined such $\Kmc^{-}$-type containing 
		$
		\{R(x,x),\neg\exists y(R(x,y)\wedge x\not=y),\neg\exists y(R(y,x)\wedge x\not=y)\}
		$. This type is disconnected (and therefore openGF-complete for $\Kmc^{-}$). 
		    \item There exists a uniquely determined such $\Kmc^{-}$-type containing $\forall xy(R(x,y)\rightarrow
		    \neg R(y,x))$. This type is connected (as it contains $\exists y (R(x,y) \wedge x\not=y)$ and openGF-complete for $\Kmc^{-}$.
		    \item There are multiple connected and not openGF-complete such  $\Kmc^{-}$-types. For example, there exist such types containing $\{R(x,x),\exists y(R(x,y)\wedge x\not=y)\}$, $\{R(x,x),\neg\exists y(R(x,y)\wedge x\not=y)\}$, or $\{\neg R(x,x),\exists y(R(x,y)\wedge x\not=y)\}$. 
		 \end{itemize}
		It follows from Point~2 above that for the KB $\Kmc=(\Omc,\Dmc)$ from
		Example~\ref{ex:GF} there is only a single $\Kmc$-type
		$\Phi(x)$. This type is connected and openGF-complete for $\Kmc$.
		There are, up to variable renaming, two $\Kmc$-types
		$\Phi(x,y)$ containing $x\not=y$: the $\Kmc$-type
		containing $R(x,y),\neg R(y,x)$ and the $\Kmc$-type
		containing $\neg R(x,y),\neg R(y,x)$. Both are
		openGF-complete for $\Kmc$ and connected.\finofex
	}
\end{exmp}
We could now characterize non-projective
$(\text{GF},\text{GF})$-separability in a way that is completely
analogous to Theorem~\ref{thm:char12}, replacing \ALCI-completeness of
types with openGF-completeness.  However, this works only for labeled
KBs $(\Kmc,P,\{\vec{b}\})$, $\Kmc = (\Omc,\Dmc)$, such that all
constants in $[\vec{b}]$ can reach one another in the Gaifman graph
of~$\Dmc$. To formulate a condition for the general case,
%
%
for a tuple $\vec{a}=(a_{1},\ldots,a_{n})$ and $I
\subseteq \{1,\ldots,n\}$ let $\vec{a}_{I}=(a_{i}\mid i\in I)$. 
\begin{restatable}{theorem}{thmcritGFtwelve}
	\label{thm:critGF1}
	Assume that $(\Kmc,P,\{\vec{b}\})$ is a labeled GF-KB with $\Kmc=(\Omc,\Dmc)$, $\Sigma=\text{sig}(\Kmc)$, and 
	$\vec{b}=(b_{1},\ldots,b_{n})$. Then the following conditions are equivalent:
	\begin{enumerate}
		\item $(\Kmc,P,\{\vec{b}\})$ is non-projectively GF-separable;
		\item $(\Kmc,P,\{\vec{b}\})$ is non-projectively openGF-separable; 
		\item there exists a finite model $\Amf$ of $\Kmc$ such that for all models $\Bmf$ of $\Kmc$ and all $\vec{a}\in P$: $\Bmf,\vec{a}^{\Bmf}\not\sim_{\text{GF},\Sigma} \Amf,\vec{b}^{\Amf}$;
		\item there exists a finite model $\Amf$ of $\Kmc$ such that for all models $\Bmf$ of $\Kmc$ and all $\vec{a}\in P$: $\Bmf,\vec{a}^{\Bmf}\not\sim_{\text{openGF},\Sigma} \Amf,\vec{b}^{\Amf}$;
		\item there exists a model $\Amf$ of $\Kmc$ such that for all $\vec{a}\in P$:
	%
	\begin{enumerate}
		\item $\Dmc_{\text{con}(\vec{a})},\vec{a}\not\rightarrow \Amf,\vec{b}^{\Amf}$ and
		\item if the set $I$ of all $i$ such that
		$\text{tp}_{\Kmc}(\Amf,b_{i}^{\Amf})$ is connected and openGF-complete is not empty,
		then 
		\begin{enumerate}
			
			\item $J=\{1,\dots,n\}\setminus I \neq \emptyset$
			and
			$\Dmc_{\text{con}(\vec{a}_{J})},\vec{a}_{J}\not\rightarrow \Amf,\vec{b}_{J}^{\Amf}$ or
			
			\item $\text{tp}_{\Kmc}(\Amf,\vec{b}^{\Amf})$ is not realizable in 
			$\Kmc,\vec{a}$.
			
		\end{enumerate}
	\end{enumerate}
\end{enumerate}
\end{restatable}  
\begin{proof}\
(sketch)
The equivalence of Points~1 and 3 and of Points~2 and 4 can be proved
using Lemma~\ref{lem:guardedbisim} and the finite model property of GF (see Section~\ref{sec:prelim}).
These proofs are exactly the same as the proofs we have given for $\mathcal{ALCI}$ already and therefore omitted. The implication ``2. $\Rightarrow$ 1.'' holds since openGF is contained in GF. We next show ``3. $\Rightarrow$ 4.'' Assume that there exists a model
$\Amf$ such that Point~3 holds. We show that Point~4 also holds for $\Amf$.
For a proof by contradiction, assume that there exist a model $\Bmf$ of $\Kmc$
and $\vec{a}\in P$ such that $\Bmf,\vec{a}^{\Bmf} \sim_{\text{openGF},\Sigma} \Amf,\vec{b}^{\Amf}$. Let $\Amf_{\vec{b}}$ and $\Bmf_{\vec{a}}$ be the restrictions of $\Amf$ and $\Bmf$ to all nodes reachable from $\vec{b}^{\Amf}$ in $\Amf$ and from $\vec{a}^{\Bmf}$ in $\Bmf$, respectively. We interpret
all constants from $\Dmc_{\text{con}(\vec{a})}$ in $\Bmf_{\vec{a}}$ as before.
Note that $\Bmf_{\vec{a}},\vec{a}^{\Bmf} \sim_{\text{GF},\Sigma} \Amf_{\vec{b}},\vec{b}^{\Amf}$. Let $\Bmf'$ be an isomorphic copy of $\Amf$ disjoint from $\Bmf_{\vec{a}}$. Interpret the constants in $\Dmc\setminus \Dmc_{\text{con}(\vec{a})}$ in $\Bmf'$ by the copies of their interpretations in $\Amf$. We have $\Bmf_{\vec{a}}\cup \Bmf',\vec{a}^{\Bmf} \sim_{\text{GF},\Sigma} \Amf, \vec{b}^{\Amf}$ since we can extend the guarded bisimulation witnessing $\Bmf_{\vec{a}},\vec{a}^{\Bmf} \sim_{\text{GF},\Sigma} \Amf_{\vec{b}},\vec{b}^{\Amf}$ by the identity mappings between $\Bmf'$ and $\Amf$. Then $\Bmf_{\vec{a}}\cup \Bmf'$
is a model of $\Omc$ since $\Amf$ is a model of $\Omc$ and
$\Bmf_{\vec{a}}\cup \Bmf',\vec{a}^{\Bmf} \sim_{\text{GF},\Sigma} \Amf,
\vec{b}^{\Amf}$. Moreover, $\Bmf_{\vec{a}}\cup \Bmf'$ is a model of
$\Dmc$ by definition. Hence,
$\Bmf_{\vec{a}}\cup \Bmf'$ is a model $\Kmc$ and we have derived a contradiction.

We have shown that Points~1 to 4 are equivalent. The proof of the equivalence of Point~5 and Points~1 to 4 is based on the same ideas as its counterpart for $\mathcal{ALCI}$ but has to deal with a few additional problems. (1) We admit proper tuples and not only singletons as examples. This leads to various case distinctions and technical issues. For example, we require a lemma that states that a type $\Phi(\vec{x})$ is openGF-complete iff all its restrictions to a single variable are openGF-complete. (2) The unfolding from $\Amf$ into $\Amf_{\Dmc,b}^{\leq \ell}$ cannot be replicated by an unfolding into a guarded tree-decomposable structures, as one might have hoped. The issue is that the leafs $L$ of the unfolding $\Amf_{\Dmc,b}^{\leq \ell}$ should have distance $\ell$ from the interpretation of the database $\Dmc$. This is not the case for unfoldings into guarded tree-decomposable structures. Instead, a rather different unfolding is required that takes copies of the whole input structure. (3) Finally, in contrast to $\ALCI$ the structure $\Cmf$ one constructs can be of infinite outdegree and cannot be easily transformed into a finite model (or model of finite outdegree). To nevertheless use the machinery of guarded bisimulations, we work with bounded guarded bisimulations instead of infinitary ones. 
A detailed proof dealing with these issues is given in the appendix of this article.
\end{proof}
Replicating the case of \ALCI, we could now define a notion of
strongly incomplete GF-KBs and observe a counterpart of
Corollary~\ref{corr:stronglyincompl}. We refrain from giving the
details. Also as for \ALCI, we can reduce projective
$(\text{GF},\text{GF})$-separability to non-projective
$(\text{GF},\text{GF})$-separability in polynomial time and show that
a single unary helper symbol always suffices to separate a labeled GF-KB that
is projectively GF-separable. 
We obtain
the following in a similar way as 
Theorem~\ref{thm:alcinexp}.
\begin{theorem} 
	\label{thm:gftwoexp}
	Non-projective $(\text{GF},\text{GF})$-separability is \TwoExpTime-complete in combined complexity. The same is true for definability, RE-existence, and
	entity indistinguishability. 
\end{theorem}
%


%
In the special case where the ontology is empty, Point~5(b) of
Theorem~\ref{thm:critGF1} is vacuously true and thus projective and
non-projective GF-separability coincide with FO-separability.
This also follows from the observation in the previous section 
that already non-projective $\ALCI$-separability coincides with FO-separability
if the ontology is empty and the fact that $\ALCI$ is contained in GF.

\subsection{Separability of Labeled FO$^2$-KBs}
\label{section:twovariable}

We show that $(\text{FO}^2, \text{FO}^2)$- and
$(\text{FO}^2, \text{FO})$-separability are undecidable both in the
projective and in the non-projective case. We also show that
these separation problems do not coincide even in
the projective case, in contrast to our results for \ALCI and GF in the
previous sections. The latter result is a consequence of a general theorem
that states that for all fragments
$\Lmc$ of FO enjoying the relativization property and the finite model property but for which rooted UCQ-evaluation is not finitely controllable, projective $(\Lmc,\Lmc)$-separability does not coincide with projective $(\Lmc,\text{FO})$-separability.
In the context of FO$^2$,
we generally assume that examples are tuples of length one or two and that only unary and binary relation symbols are used.

UCQ-evaluation on FO$^{2}$-KBs is
undecidable~\cite{DBLP:conf/icdt/Rosati07} and the proof easily adapts
to rooted UCQs. Together with Theorem~\ref{critFOwithoutUNA}, we
obtain undecidability of $(\text{FO}^{2},\text{FO})$-separability both
in the projective and non-projective case (which coincide, due to that
theorem). We adapt the mentioned undecidability proof to show that  
projective and non-projective
$(\text{FO}^{2},\Lmc_S)$-separability is undecidable for every FO-fragment
$\Lmc_S$ that contains $\text{FO}^2$. The same is true 
for definability, RE-existence, and entity distinguishability.
\begin{restatable}{theorem}{thmfotwofo}\label{thm:fo2fo}
  For all FO-fragments $\Lmc_S\supseteq\text{FO}^2$, projective and
  non-projective $(\text{FO}^{2},\Lmc_S)$-separability are undecidable.
  This is also true for definability, RE-existence, and entity distinguishability.
\end{restatable}

The proof of Theorem~\ref{thm:fo2fo} is by reduction from tiling
problems which we introduce next. A \emph{tiling
system} is a triple $(T,H,V)$ with $T$ a finite set and
$H,V\subseteq T\times T$. A \emph{solution} to $(T,H,V)$ is a function
$\tau:\mathbbm{N}\times \mathbbm{N}$ such that, for all $i,j\geq 0$:
\begin{enumerate}[label=(\roman*)]

  \item $(\tau(i,j),\tau(i+1,j))\in H$, and 

  \item $(\tau(i,j),\tau(i,j+1))\in V$.

\end{enumerate} 
A solution is \emph{periodic} if there are periods $h,v\geq 1$ such
that $\tau(i,j)=\tau(i+h,j)=\tau(i,j+v)$, for all $i,j\geq 0$. Note
that a periodic solution can be thought of as a torus, labeled with
elements from $T$ consistent with $H$ and $V$. We say that a tiling
system \emph{admits a (periodic) solution}, if there is a (periodic)
solution to it. It is well-known that the problem of deciding whether
a given tiling system admits a (resp., periodic) solution is
undecidable~\cite{Berger1966}. However, we are going to exploit a
stronger undecidability result due to Gurevich and
Koryakov~\cite{gurevichTilings}; see
also~\cite[Theorem~3.1.7]{DBLP:books/sp/BorgerGG1997} for a new proof.
Recall that two sets $A,B$ are \emph{recursively inseparable} if there
is no recursive (that is, decidable) set that contains $A$ and is
disjoint from $B$.
\begin{theorem}[\cite{gurevichTilings,DBLP:books/sp/BorgerGG1997}]
  \label{thm:tilings} The set of tiling systems that admit no solution
  is recursively inseparable from the set of tiling systems that admit
  a periodic solution. 
\end{theorem}
Hence, one can show undecidability of a language $L\subseteq \Sigma^*$
by giving a computable function $f$ from tiling systems to $\Sigma^*$
such that: \textit{(i)} tiling systems which admit a periodic solution
are mapped to $L$, and~\textit{(ii)} tiling systems without a solution
are mapped to $\Sigma^*\setminus L$. We employ this strategy in the following
proof of Theorem~\ref{thm:fo2fo}. 

\medskip\noindent\begin{proof}
%
%
%
%
%
\ By Observation~\ref{obs:red}, it suffices to show undecidability for 
RE-existence and undecidability of separability, definability, and entity distinguishability follow. 

Given a tiling system $(T,H,V)$, we construct a labeled FO$^2$-KB
$(\Kmc,\{a\},N)$ with $\Kmc = (\Omc,\Dmc)$ as follows: 
\begin{align}
\Omc = \{\, & B\sqsubseteq \exists R_h.B\sqcap \exists R_v.B
\label{eq:tree}\\
%
%
& \forall xy\, (B(x)\wedge B(y)\rightarrow U(x,y)), \\
& \forall xy\, (\neg R_v(x,y)\rightarrow \overline
R_v(x,y)), \label{eq:roleinclusion}\\
%
%
& B \sqsubseteq \bigsqcup_{t\in T} \big(A_t\sqcap
\bigsqcap_{t'\in T\setminus\{t\}}\neg A_{t'}\big), \label{eq:tile1}\\
%
%
& A_t \sqsubseteq \forall R_v.\bigsqcup_{(t,t')\in V} A_{t'}, &
\text{for all $t\in T$}
\label{eq:tile2}\\
%
& A_t \sqsubseteq \forall R_h.\bigsqcup_{(t,t')\in H} A_{t'}, &
\text{for all $t\in T$} 
\label{eq:tile3}\\
&\! \} \notag\\
\Dmc & = \{\ U(a,a_1),\ R_v(a_1,a_2),\ R_h(a_2,a_3),\ R_h(a_1,a_4),\
\overline R_v(a_4,a_3),\ B(b)\ \}\notag \\
N & = \{\ a_1,\ a_2,\ a_3,\ a_4,\ b\ \}\notag
\end{align}
\tikzset{every picture/.style={line width=0.75pt}} 
The database $\Dmc_{\text{con}(a)}$ can be depicted as follows:
\begin{center}
\begin{tikzpicture}[x=0.75pt,y=0.75pt,yscale=-1,xscale=1]
	
	\draw  [fill={rgb, 255:red, 0; green, 0; blue, 0 }  ,fill opacity=1 ] (262.1,270.25) .. controls (262.1,269.28) and (261.32,268.5) .. (260.35,268.5) .. controls (259.38,268.5) and (258.6,269.28) .. (258.6,270.25) .. controls (258.6,271.22) and (259.38,272) .. (260.35,272) .. controls (261.32,272) and (262.1,271.22) .. (262.1,270.25) -- cycle ;
	\draw  [fill={rgb, 255:red, 0; green, 0; blue, 0 }  ,fill opacity=1 ] (261.9,238.05) .. controls (261.9,237.08) and (261.12,236.3) .. (260.15,236.3) .. controls (259.18,236.3) and (258.4,237.08) .. (258.4,238.05) .. controls (258.4,239.02) and (259.18,239.8) .. (260.15,239.8) .. controls (261.12,239.8) and (261.9,239.02) .. (261.9,238.05) -- cycle ;
	\draw    (260.38,265.52) -- (260.28,246.05) ;
	\draw [shift={(260.27,243.05)}, rotate = 449.7] [fill={rgb, 255:red, 0; green, 0; blue, 0 }  ][line width=0.08]  [draw opacity=0] (5,-2.5) -- (0,0) -- (5,2.5) -- (3.5,0) -- cycle    ;
	\draw    (265.18,270.72) -- (286,270.27) ;
	\draw [shift={(289,270.2)}, rotate = 538.76] [fill={rgb, 255:red, 0; green, 0; blue, 0 }  ][line width=0.08]  [draw opacity=0] (5,-2.5) -- (0,0) -- (5,2.5) -- (3.5,0) -- cycle    ;
	\draw  [fill={rgb, 255:red, 0; green, 0; blue, 0 }  ,fill opacity=1 ] (295.3,270.25) .. controls (295.3,269.28) and (294.52,268.5) .. (293.55,268.5) .. controls (292.58,268.5) and (291.8,269.28) .. (291.8,270.25) .. controls (291.8,271.22) and (292.58,272) .. (293.55,272) .. controls (294.52,272) and (295.3,271.22) .. (295.3,270.25) -- cycle ;
	\draw    (265.18,237.72) -- (286,237.27) ;
	\draw [shift={(289,237.2)}, rotate = 538.76] [fill={rgb, 255:red, 0; green, 0; blue, 0 }  ][line width=0.08]  [draw opacity=0] (5,-2.5) -- (0,0) -- (5,2.5) -- (3.5,0) -- cycle    ;
	\draw  [fill={rgb, 255:red, 0; green, 0; blue, 0 }  ,fill opacity=1 ] (295.3,237.25) .. controls (295.3,236.28) and (294.52,235.5) .. (293.55,235.5) .. controls (292.58,235.5) and (291.8,236.28) .. (291.8,237.25) .. controls (291.8,238.22) and (292.58,239) .. (293.55,239) .. controls (294.52,239) and (295.3,238.22) .. (295.3,237.25) -- cycle ;
	\draw [dash pattern={on 0.84pt off 2.51pt}]   (293.58,265.05) -- (293.48,245.58) ;
	\draw [shift={(293.47,242.58)}, rotate = 449.7] [fill={rgb, 255:red, 0; green, 0; blue, 0 }  ][line width=0.08]  [draw opacity=0] (5,-2.5) -- (0,0) -- (5,2.5) -- (3.5,0) -- cycle    ;
	\draw    (211,250) .. controls (200.92,257.15) and (213.03,271.59) .. (248.86,270.45) ;
	\draw [shift={(251.67,270.33)}, rotate = 536.99] [fill={rgb, 255:red, 0; green, 0; blue, 0 }  ][line width=0.08]  [draw opacity=0] (5,-2.5) -- (0,0) -- (5,2.5) -- (3.5,0) -- cycle    ;
	
	\draw (210.47,239.93) node [anchor=north west][inner sep=0.75pt]  [font=\footnotesize] [align=left] {$\displaystyle a$};
	\draw (269.18,274.72) node [anchor=north west][inner sep=0.75pt]  [font=\footnotesize] [align=left] {$\displaystyle R_{h}{}$};
	\draw (267.18,218.72) node [anchor=north west][inner sep=0.75pt]  [font=\footnotesize] [align=left] {$\displaystyle R_{h}{}$};
	\draw (239.18,246.72) node [anchor=north west][inner sep=0.75pt]  [font=\footnotesize] [align=left] {$\displaystyle R_{v}$};
	\draw (302.18,246.72) node [anchor=north west][inner sep=0.75pt]  [font=\footnotesize] [align=left] {$\displaystyle \overline{R_{v}}$};
	\draw (199.18,265.72) node [anchor=north west][inner sep=0.75pt]  [font=\footnotesize] [align=left] {$\displaystyle U$};

\end{tikzpicture}
\end{center}

By (the comment after) Theorem~\ref{thm:tilings}, it suffices to show the following Claim.

\smallskip\noindent\textit{Claim.} Let $\Lmc_S$ be an FO-fragment with 
$\Lmc_S\supseteq \text{FO}^2$. Then, we have: 

\begin{enumerate}

  \item If $(\Kmc,\{a\},N)$ is projectively or non-projectively
    $\Lmc_S$-separable, then $(T,H,V)$ admits a solution.

  \item If $(T,H,V)$ admits a periodic solution, then 
    $(\Kmc,\{a\},N)$ is non-projectively $\Lmc_S$-separable.

\end{enumerate}

\smallskip\noindent\textit{Proof of the Claim.} For Point~1, suppose
that $(\Kmc,\{a\},N)$ is projectively or non-projectively
$\Lmc_S$-separable. Thus, $(\Kmc,\{a\},\{b\})$ is FO-separable. 
By
Theorem~\ref{critFOwithoutUNA} and Corollary~\ref{cor:weak2}, 
$\varphi_{\Dmc_{\text{con}(a)},a}(x)$ separates
$(\Kmc,\{a\},\{b\})$. Let $\Amf$ be a structure witnessing
$\Kmc\not\models \varphi_{\Dmc_{\text{con}(a)},a}(b)$. Since \Amf is a model of~\eqref{eq:tree}--\eqref{eq:roleinclusion} and
$\Amf\not\models\varphi_{\Dmc_{\text{con}(a)},a}(b)$, it contains an infinite grid
formed by relations $R_v$ and $R_h$. Since \Amf is a model
of~\eqref{eq:tile1} every element in the grid is labeled with 
$A_t$ for a unique element $t\in T$. Finally, since \Amf is a model
of~\eqref{eq:tile2} and~\eqref{eq:tile3}, the relations
$H$ and $V$ are respected along $R_h$ and $R_v$, respectively.

\smallskip For Point~2, suppose that $(T,H,V)$ admits a periodic
solution $\tau$ with periods $h,v\geq 1$. We show that $(\Kmc,\{a\},N)$ is
non-projectively FO$^2$-separable under the assumption that $\Kmc$
mentions a binary relation symbol $S$. This is without loss of
generality, as we can include $\forall xy\, S(x,y)\rightarrow
S(x,y)$ in $\Omc$. 

Let $\pi$ be a bijection from $[h]\times[v]$ to $[hv]$ and
let $C_{ij}$ be the \ALCI-concept (corresponding to an FO$^2$-formula)
expressing that there is an $S$-path of length $\pi(i,j)$.  We
construct the following FO$^2$-formula $\varphi_{hv}(x)$, written as
an \ALCI-concept: 
\begin{align}
  \exists U.\bigsqcap_{i\in[h]}\bigsqcap_{j\in [v]}(\forall R_v.\forall R_h.C_{ij}\rightarrow
  \exists R_h.\exists \overline R_v. C_{ij} ).\label{eq:sepconcept}
\end{align}
It should be clear that $\Kmc\models \varphi_{hv}(a)$ since already
$\varphi_{\Dmc_{\text{con}(a)},a}(x)\models \varphi_{hv}(x)$. To see that
$\Kmc\not\models\varphi_{hv}(d)$, for all $d\in N$, we construct a (finite) model \Amf
witnessing that. Informally, $\Amf$ consists of two disconnected
parts. One part is $\Dmc$ viewed as a structure; the other is an $h\times
v$-torus over binary symbols $R_v,R_h$ in which each element has an
outgoing $S$-path. More precisely, the torus has domain $[h]\times [v]$
and each element $(i,j)$ is labeled with the unary symbol
$A_{\tau(i,j)}$ and has an outgoing $S$-path of length $\pi(i,j)$.
Formally, we have: 
\begin{align*}
  B^\Amf & = [h]\times [v] \\
  R_v^{\Amf} & = \{(a_1,a_2)\} \cup \{( (i,j), (i,j\oplus_v 1)\mid
  i\in [h], j\in [v]\}
  \\
  R_h^{\Amf} & = \{(a_2,a_3),(a_1,a_4)\} \cup \{( (i,j), (i\oplus_h 1,j)\mid i\in [h], j\in [v]\}
  \\
  A_t^{\Amf} & = \{(i,j)\in [h]\times[v]\mid \tau(i,j)=t\} \hspace{5cm}\text{for all
  $t\in T$}\\
  U^{\Amf} & = \{(a,a_1)\}\cup
  ([h]\times[v])\times([h]\times[v]) \\
  a^{\Amf} & = a\quad \quad a_i^\Amf = a_i,\text{ for
  }i\in\{1,2,3,4\}\quad \quad b^\Amf = (0,0)
\end{align*}
and $S^\Amf$ is as described above and $\overline R_v^\Amf$ is the
complement of $R_v^\Amf$. 

It is readily checked that $\Amf$ is a model of \Kmc. Moreover, we
have
$\Amf\not\models\varphi_{hv}(a_i)$, for $i\in\{1,\ldots,4\}$, since
the $a_i$ do not have $U$-successors in $\Amf$. Suppose finally that 
$\Amf\models\varphi_{hv}(b)$, let $(i_0,j_0)$ be the $U$-successor of
$b^\Amf$ that witnesses the big conjunction in~\eqref{eq:sepconcept},
and let $i=i_0\oplus_h 1$ and $j=j_0\oplus_v 1$. By construction of
$\Amf$, we have $\Amf\models \forall R_v.\forall R_h.C_{ij}(i_0,j_0)$, but 
$\Amf\not\models \exists R_h.\exists \overline R_v.C_{ij}(i_0,j_0)$, in
contradiction to the implication in~\eqref{eq:sepconcept}. Hence, 
$\Amf\not \models\varphi_{hv}(b)$. This finishes the proof of the
Claim and in fact of the Theorem.
\end{proof}

We next discuss the relationship between FO-separability and FO$^2$-separability in labeled FO$^2$-KBs. Example~\ref{ex:GF} shows that
$(\text{FO}^{2},\text{FO})$-separability and
$(\text{FO}^{2},\text{FO}^2)$-separability do not coincide in the
non-projective case, since every FO$^2$-formula $\vp(x)$ with
$\mn{sig}(\vp) = \{ R \}$ is equivalent to $x=x$ or to $\neg(x=x)$
w.r.t.\ the ontology \Omc used there. The example also yields that
projective and non-projective $(\text{FO}^{2},\text{FO}^2)$-separability do not coincide. We next observe that $(\text{FO}^{2},\text{FO})$-separability and
$(\text{FO}^{2},\text{FO}^2)$-separability do not coincide also in the
projective case, in a more general
setting. 
Note that FO$^{2}$ has the finite model property and the relativization property~\cite{DBLP:journals/mlq/Mortimer75}. The following example shows that evaluating rooted CQs on FO$^{2}$-KBs is not finitely controllable using the ontology constructed in the proof of Theorem~\ref{thm:fo2fo}.
\begin{exmp}\label{ex:exmp}
	{\em 
Let $\Omc$ be the ontology constructed in the proof of Theorem~\ref{thm:fo2fo}
for a tiling system that admits a solution but does not admit any periodic solution. Let $\Dmc$ be the database used in that proof, $\Dmc'=\{B(b)\}$, $\Kmc'=(\Omc,\Dmc')$, and consider the rooted CQ $\varphi_{\Dmc_{\text{con}(a),a}}(x)$.
Then $\Kmc'\not\models \varphi_{\Dmc_{\text{con}(a),a}}(b)$ but
$\Amf\models \varphi_{\Dmc_{\text{con}(a),a}}(b)$ for every finite
model $\Amf$ of $\Kmc'$.\finofex
}
\end{exmp}

\begin{restatable}{theorem}{thmsepfcp}\label{thm:sepfcp}
	Let $\LmcO$ be a fragment of FO that has the relativization property
	and the finite model property. Then the following holds:
	\begin{itemize}
		\item If projective $(\LmcO,\text{FO})$-separability
	coincides with projective $(\LmcO,\LmcO)$-separability for labeled KBs with a single negative example tuple of length $n$, then
	evaluating rooted UCQs of arity $n$ on $\LmcO$-KBs is finitely controllable;
	    \item if projective $(\LmcO,\text{FO})$-entity distinguishability  coincides with projective $(\LmcO,\LmcO)$-entity distinguishability for example tuples of length $n$, then
	    evaluating rooted CQs of arity $n$ on $\LmcO$-KBs is finitely controllable.
	 \end{itemize}
\end{restatable}
\begin{proof} \ We prove Point~1. Point~2 is a direct consequence of the proof. Assume that evaluating rooted UCQs on \Lmc-KBs is not
finitely controllable, that is, there is an \Lmc-KB $\Kmc=(\Omc,\Dmc)$, 
a rooted UCQ $q(\vec{x})=\bigvee_{i\in
	I}q_{i}(\vec{x})$, $\vec{x}=(x_{1},\ldots,x_{n})$, and a tuple $\vec{a}$ in $\Dmc$ such that
$\Kmc\not\models q(\vec{a})$, but $\Bmf\models q(\vec{a})$
for all finite models $\Bmf$ of $(\Omc,\Dmc)$. The proof is now very similar to the proof of Corollary~\ref{cor:rel}.
Consider the relativization $\Omc_{|A}$ of the sentences of $\Omc$ to $A$
and $\Dmc^{+A}=\Dmc \cup \{A(a) \mid a\in \text{cons}(\Dmc)\}$, for a
fresh unary relation symbol $A$.
Define the labeled KB $(\Kmc',P,N)$ as in the proof of Corollary~\ref{cor:rel}. Then the UCQ $q(\vec x)$ separates $(\Kmc',P,N)$ but we show that
$(\Kmc',P,N)$ is not $\Lmc$-separable.
Suppose there is an \Lmc-formula $\varphi(\vec{x})$ that separates
$(\Kmc',P,N)$. Since \Lmc has the finite model property, there exists a finite model $\Amf_{f}$ of
$\Kmc'$ such that $\Amf_{f}\models \neg\varphi(\vec{a})$.
As $\Bmf\models q(\vec{a})$ for all finite models $\Bmf$ of $(\Omc,\Dmc)$, there exists $i\in I$ with
$\Amf_{f}\models q_{i}(\vec{a})$. Then there is a homomorphism $h$ from $\Dmc_{i},([x_{1}]^{i},\ldots,[x_{n}]^{i})$ 
to $\Amf_{f},\vec{a}$ witnessing this. We modify $\Amf_{f}$ to
obtain a new structure $\Amf_{f}'$ which coincides with $\Amf_{f}$
except that the constants $c$ in $\Dmc_{i}$ are interpreted as $h(c)$.
Then $\Amf_{f}'$ is a model of $\Kmc'$ with $\Amf_{f}'\models
\neg\varphi([x_{1}]^{i},\ldots,[x_{n}]^{i})$ which contradicts the assumption
that $\varphi(\vec{x})$ separates $(\Kmc',P,N)$.  
\end{proof}
It follows from Example~\ref{ex:exmp} and Theorem~\ref{thm:sepfcp} that
projective $(\text{FO}^{2},\text{FO})$-entity distinguishability does not  coincide with projective $(\text{FO}^{2},\text{FO}^{2})$-entity distinguishability for examples consisting of single constants.

The main idea behind the proof Theorem~\ref{thm:sepfcp} can also often
be applied to FO-fragments $\Lmc$ without the relativization property.
We illustrate this for the description logic $\mathcal{S}$ extending
$\mathcal{ALC}$ with transitive roles. $\mathcal{S}$ does not enjoy
the relativization property since one cannot express that the
restriction of a relation $R$ to a concept name is transitive. 
\begin{exmp}\label{ex:exmptrans} {\em $\mathcal{S}$ enjoys the finite
  model property but evaluating rooted CQs is not finitely
  controllable~\cite{GuIbJu-AAAI18}. To see the latter claim, let
  $\Omc=\{A\sqsubseteq \exists R.A\}$ with $R$ a transitive role,
  $\Dmc=\{A(b)\}$, and $q(x)= \exists y (R(x,y)\wedge R(y,y))$. Let
  $\Kmc=(\Omc,\Dmc)$.  Then $\Kmc\not\models q(b)$ is witnessed by the
  model $\Amf$ of $\Kmc$ consisting of an infinite $R$-chain with root
  $b$ and nodes labeled with $A$. But $q(b)$ is satisfied in every
  finite model of $\Kmc$.  This example can be used to show that
  projective $(\mathcal{S},\mathcal{S})$-separability does not
  coincide with $(\mathcal{S},\text{FO})$-separability by taking
  $\Kmc'=(\Omc,\Dmc')$ with $\Dmc'=\{R(a,c),R(c,c),A(b)\}$ and
  observing that $q$ separates $(\Kmc',\{a\},\{b\})$ but that no
  $\ALCI$-concept separates $(\Kmc',\{a\},\{b\})$ (use that
  $\mathcal{S}$ enjoys the finite model property). \finofex }
\end{exmp}
 
We note that if the ontology is empty, projective and
non-projective FO$^2$-separability coincide with FO-separability. This follows from the earlier observation that projective and
non-projective $\ALCI$-separability coincide with FO-separability for empty ontologies and the fact that $\ALCI$ is contained in FO$^{2}$.

\section{Fundamental Results on Strong Separability}\label{sec:strong}

We introduce strong separability and the special cases of
strong definability, referring expression existence, and entity distinguishability.
We observe that, in contrast to weak separability, strong projective separability and strong non-projective
separability coincide in all relevant cases. We then give a 
characterization of strong
$(\text{FO},\text{FO})$-separability that has the consequence that also in the context of strong separability UCQs have the same separating power as FO. In contrast to
Theorem~\ref{critFOwithoutUNA} for weak separability, however, it establishes a link to KB
satisfiability rather than to the evaluation of rooted UCQs. We also settle the complexity of
deciding strong separability in GNFO.

\begin{definition}
	An FO-formula $\varphi(\vec{x})$ \emph{strongly separates}
	a labeled FO-KB $(\Kmc,P,N)$ if 
	\begin{enumerate}
		
		\item 
		$\Kmc\models \varphi(\vec{a})$ for all $\vec{a}\in P$ and 
		
		\item $\Kmc\models \neg\varphi(\vec{a})$ for all $\vec{a}\in N$.
		
	\end{enumerate}
	Let $\Lmc_S$ be a fragment of FO. We say that $(\Kmc,P,N)$ is
	\emph{strongly projectively $\Lmc_S$-separable} if there is an
	$\Lmc_S$-formula $\varphi(\vec{x})$ that strongly separates
	$(\Kmc,P,N)$ and \emph{strongly (non-projectively)
		$\Lmc_S$-separable} if there is such a $\varphi(\vec{x})$ with
	$\mn{sig}(\varphi) \subseteq \mn{sig}(\Kmc)$.
\end{definition}
%
By definition, (projective) strong separability implies (projective)
weak separability, but the converse is false. 
%
\begin{exmp}
	\label{ex:strong}
{\em 	Let $\Kmc_{1}=(\emptyset,\Dmc)$ with
	$$\Dmc=\{{\sf votes}(a,c_{1}), {\sf votes}(b,c_{2}), {\sf Left}(c_{1}), {\sf Right}(c_{2})\}.$$
	Then $(\Kmc_1, \{a\},\{b\})$ is weakly separated by the
	\ALCI-concept $\exists {\sf votes}.{\sf Left}$, but it is not strongly FO-separable.
	
	Now let $\Kmc_{2}=(\Omc,\Dmc)$ with
	$$\Omc=\{\exists {\sf votes}.{\sf Left} \sqsubseteq \neg \exists {\sf
		votes}.{\sf Right}\}.
	$$
	Then $\exists {\sf votes}.{\sf Left}$ strongly separates $(\Kmc_{2},\{a\},\{b\})$.\finofex
}
\end{exmp}
As illustrated by Example~\ref{ex:strong}, `negative information'
introduced by the ontology is crucial for strong separability because
of the open world semantics and since the database cannot contain
negative information. In fact, labeled KBs with an empty ontology are
never strongly separable. In a sense, weak separability tends to be
too credulous if the data is incomplete regarding positive
information, see Example~\ref{exmp:44}, while strong
separability tends to be too sceptical if the data is incomplete
regarding negative information as shown by Example~\ref{ex:strong}.
Note also that while weak $(\text{FO},\text{FO})$-separability is
anti-monotone in the ontology, strong separability is always trivially monotone in the ontology (and also the database) in the sense that for all labeled 
KBs $(\Kmc_{i},P,N)$ with $\Kmc_{i}=(\Omc_{i},\Dmc_{i})$ for $i=1,2$, if $\Omc_{1}\subseteq \Omc_{2}$, $\Dmc_{1}\subseteq \Dmc_{2}$, and $(\Kmc_{1},P,N)$ is $\Lmc_{S}$-separable, then $(\Kmc_{2},P,N)$ is $\Lmc_{S}$-separable.

In contrast to weak separability, projective and non-projective strong
separability coincide in all cases that are relevant to this paper.
From now on, we thus omit these qualifications.
\begin{proposition}\label{allthesame}
  Let $(\Kmc,P,N)$ be an FO-KB and let $\Lmc_{S} \in \{\text{CQ},
  \text{UCQ}, \ALCI, \text{GF}, \text{openGF}, \text{GNFO},
  \text{FO}^{2}, \text{FO}\}$.  Then $(\Kmc,P,N)$ is strongly
  projectively $\Lmc_S$-separable iff it is strongly non-projectively
  $\Lmc_S$-separable.  
\end{proposition}
\begin{proof} \ Assume $\varphi(\vec{x})$ strongly separates
  $(\Kmc,P,N)$ and $R\in
  \text{sig}(\varphi)\setminus\text{sig}(\Kmc)$. Intuitively, we can
  replace every occurence of any formula of the form $R(\vec{y})$ in
  $\varphi(\vec{x})$ by a tautology (or an unsatisfiable formula)
  formulated in $\text{sig}(\Kmc)$ without affecting separation as the
  KB does not state anything about such $R$. If $\Lmc_S$ is FO, then
  we can simply replace $R(\vec y)$ by the conjunction of all $y=y$
  with $y$ in $\vec y$. In other choices of $\Lmc_S$, a bit more care
  is needed: for example, the formula obtained in this way is not
  guarded if $R$ occurs as a guard in $\varphi$. In this case we
  replace every subformula of the form $\exists
  \vec{y}(R(\vec{z},\vec{y}) \wedge \psi)$ in $\varphi(\vec{x})$ by
  (the unsatisfiable formula) $\neg (x=x)$ for some $x$ in $\vec{x}$
  and every ocurrence of $R(\vec{y})$ in $\varphi(\vec{x})$ in a
  non-guard position by $\neg (y=y)$ for some $y$ in $\vec{y}$.  For
  $\mathcal{ALCI}$, assume that the concept $C$ strongly separates
  $(\Kmc,P,N)$ and $X\in\text{sig}(C)\setminus\text{sig}(\Kmc)$. If
  $X$ is a concept name, then replace every occurrence of $X$ in $C$
  by $\bot$ and if $X$ is a role name, then replace every subconcept
  of the form $\exists X.D$ or $\exists X^{-}.D$ in $C$ by $\bot$.
  The other cases can be dealt with similarly. 
%
    %
\end{proof}

Each choice of an ontology language $\Lmc$ and a separation language
$\Lmc_{S}$ thus gives rise to a (single) strong separability problem
that we refer to as \emph{strong $(\Lmc,\Lmc_S)$-separability},
defined in the expected way.
We consider again the following special cases of $(\LmcO,\Lmc_{S})$-separability: \emph{strong $(\LmcO,\Lmc_{S})$-definability},
\emph{strong $(\LmcO,\Lmc_{S})$-referring expression existence (RE-existence)}, and
\emph{strong $(\LmcO,\Lmc_{S})$-entity distinguishability}, all defined in the obvious way. Our main complexity results for strong separability also hold for strong definability, referring expression existence, and entity distinguishability.
In fact, for FO-fragments $\Lmc_{S}$ closed under conjunction and disjunction, a labeled KB $(\Kmc,P,N)$ is strongly 
$\Lmc_S$-separable iff every KB $(\Kmc,\{\vec{a}\},\{\vec{b}\})$ with $\vec{a}\in P$ and $\vec{b}\in N$ is strongly 
$\Lmc_S$-separable. To show this, observe that if
$\varphi_{\vec{a},\vec{b}}$ strongly separates $(\Kmc,\{a\},\{b\})$ for all
$\vec{a}\in P$ and $\vec{b}\in N$, then $\bigvee_{\vec{a}\in
	P}\bigwedge_{\vec{b}\in N}\varphi_{\vec{a},\vec{b}}$ strongly separates
$(\Kmc,P,N)$. The converse direction is trivial: if $\varphi$ strongly separates $(\Kmc,P,N)$, then it strongly separates all  $(\Kmc,\{\vec{a}\},\{\vec{b}\})$ with $\vec{a}\in P$ and $\vec{b}\in N$.
\begin{observation}
	Let $\Lmc$ and $\Lmc_{S}$ be fragments of FO such that $\Lmc_{S}$ is  closed under conjunction and disjunction. Then there is a polynomial time Turing reduction of strong $(\Lmc,\Lmc_S)$-separability to strong $(\Lmc,\Lmc_S)$-entity distinguishability.
\end{observation}
We will thus mostly consider entity distinguishability when proving complexity upper bounds and RE-existence when proving complexity lower bounds. We next characterize strong $(\text{FO},\text{FO})$-separability in terms of
KB unsatisfiability and show that strong $(\text{FO},\text{FO})$-separability
coincides with strong $(\text{FO},\text{UCQ})$-separability. Let \Dmc be a database and let
$\vec{a}=(a_{1},\ldots,a_{n})$ and $\vec{b}=(b_{1},\ldots,b_{n})$ be
tuples of constants in $\Dmc$. We write $\Dmc_{\vec{a}=\vec{b}}$ to
denote the database obtained by taking $\Dmc\cup \Dmc'$, $\Dmc'$ a
copy of \Dmc with fresh constants $c'$ for $c\in \text{cons}(\Dmc)$, and then identifying $a_{i}$ and the copy $b_{i}'$ of $b_{i}$ for
$1\leq i \leq n$. For example, for $\Dmc=\{R(a,b),S(b,c),A(a),B(b)\}$ we have
$\Dmc_{a=b} = \{R(a,b),S(b,c),A(a),B(b),R(a',a),S(a,c'),A(a'),B(a)\}$, where $a',c'\in \text{Const}$ are `copies' of $a$ and $c$ respectively and we identify the copy $b'$ of $b$ with $a$. 
\begin{theorem}\label{thm:strongsep}
	
	Let $(\Kmc,P,N)$ be a labeled FO-KB,
	$\Kmc=(\Omc,\Dmc)$. 
	Then the following
	conditions are equivalent:
	\begin{enumerate}
		\item $(\Kmc,P,N)$ is strongly UCQ-separable;
		\item $(\Kmc,P,N)$ is strongly FO-separable;
		\item for all $\vec{a}\in P$ and $\vec{b}\in N$, the KB $(\Omc,\Dmc_{\vec{a}=\vec{b}})$ 
		is unsatisfiable;
		\item the UCQ $\bigvee_{\vec{a}\in 
			P}\varphi_{\Dmc_{\text{con}(\vec{a}),\vec{a}}}$
		strongly separates $(\Kmc,P,N)$.
	\end{enumerate}
\end{theorem}
\begin{proof} \ ``1. $\Rightarrow$ 2.'' is trivial. ``2. $\Rightarrow$ 3.'' Assume that $\varphi(\vec{x})$ strongly separates $(\Kmc,P,N)$ but that there are $\vec{a}\in P$ and $\vec{b}\in N$ such that $(\Omc,\Dmc_{\vec{a}=\vec{b}})$ is satisfiable. Take a model $\Amf$ of $(\Omc,\Dmc_{\vec{a}=\vec{b}})$. Then $\Amf$ is a model of $\Kmc$ and so $\Amf\models \varphi(\vec{a})$. Hence $\Amf\models \varphi(\vec{b}')$ since $\vec{a}=\vec{b}'$. Then $\Amf$ gives rise to another model of $\Kmc$ with $\Amf\models \varphi(\vec{b})$ by reinterpreting the constants in $\Dmc$ by setting $c^{\Amf}={c'}^{\Amf}$ for all $c\in \text{cons}(\Dmc)$. This contradicts the assumption that $\Kmc\models \neg \varphi(\vec{b})$.
	
``3. $\Rightarrow$ 4.'' Assume that $\bigvee_{\vec{a}\in
		P}\varphi_{\Dmc_{\text{con}(\vec{a}),\vec{a}}}$ does not strongly
	separate $(\Kmc,P,N)$. Then there are $\vec{a}\in P$, $\vec{b}\in N$, and a model $\Amf$ of $\Kmc$ such that $\Amf\models
	\varphi_{\Dmc_{\text{con}(\vec{a}),\vec{a}}}(\vec{b})$.  One
	can now interpret the constants of $\Dmc_{\vec{a}=\vec{b}}$
	in such a way that $\Amf$ becomes a model of
	$\Dmc_{\vec{a}=\vec{b}}$: assume that $v$ is a variable assignment witnessing $\Amf\models
	\varphi_{\Dmc_{\text{con}(\vec{a}),\vec{a}}}(\vec{b})$. Then set ${c'}^{\Amf} = c^{\Amf}$ for $c\in \text{cons}(\Dmc)$ and reinterpret the constants from $\text{cons}(\Dmc_{\text{con}(\vec{a})})$ in $\Amf$ by setting ${c}^{\Amf}= v(x_{c})$ for $c\in \text{cons}(\Dmc_{\text{con}(\vec{a})})$, with $x_{c}$ the variable corresponding to $c$. Then $\Amf$ is a model of $(\Omc,\Dmc_{\vec{a}=\vec{b}})$ and so the KB
	$(\Omc,\Dmc_{\vec{a}=\vec{b}})$ 
	is satisfiable.
	
``4. $\Rightarrow$ 1.'' is trivial.
\end{proof}
Note that the UCQ in Point~4 of Theorem~\ref{thm:strongsep} is a
concrete separating formula of polynomial size, and that it is
identical to the UCQ in Point~4 of Theorem~\ref{critFOwithoutUNA}.
%
%
Point~3 provides the announced link to KB unsatisfiability. 
We also obtain the following counterpart of Corollary~\ref{cor:weak1} and  \ref{cor:weak2}.

\begin{corollary}\label{cor:expr}	
	Strong $(\text{FO},\text{FO})$-separability coincides with strong $(\text{FO},\Lmc_{S})$-separability for all FO-fragments $\Lmc_{S}\supseteq \text{UCQ}$. This also holds for strong RE-existence and strong entity distinguishability for all FO-fragments $\Lmc_S \supseteq \text{CQ}$.
\end{corollary}
Recall that GNFO contains UCQ and that satisfiability of GNFO-KBs is
\TwoExpTime-complete in combined complexity and \NPclass-complete in
data complexity
\cite{DBLP:journals/jacm/BaranyCS15,DBLP:journals/pvldb/BaranyCO12}.
This directly implies the following complexity upper bounds for strong separability on GNFO-KBs.
%
\begin{corollary}
	\label{thm:gnfostrong}
       For all  FO-fragments $\Lmc_S\supseteq \text{UCQ}$,
        \begin{enumerate}

        \item strong $(\text{GNFO},\text{GNFO})$-separability coincides with
          strong $(\text{GNFO},\Lmc_S)$-separability, and the same is true for
          definability;

        \item strong $(\text{GNFO},\Lmc_{S})$-separability and
          $(\text{GNFO},\Lmc_{S})$-definability are \TwoExpTime-complete in
          combined complexity and \coNPclass-complete in data complexity.
	
        \end{enumerate}
        This holds also for strong RE-existence and strong entity
        distinguishability, for all FO-fragments $\Lmc_{S}\supseteq \text{CQ}$.
%
\end{corollary}
\begin{proof} \
It remains to prove the lower bounds. For the \coNPclass lower bound in data complexity, we give a reduction of the complement of 3-colorability to entity distinguishability and to RE-existence. \coNPclass-hardness of separability and definability follow trivially. Let $\Omc$ contain the $\mathcal{ALCI}$-CIs
$$
\top \sqsubseteq L_{1} \sqcup L_{2} \sqcup L_{3}, \quad L_{i} \sqcap \exists R.L_{i} \sqsubseteq E, \quad \exists R.E \sqsubseteq E, \quad E \sqcap A \sqcap B \sqsubseteq \bot
$$
for all $i\in\{1,2,3\}$. 
For the polynomial time reduction, consider any connected undirected graph $G$ and let $\Dmc$ be defined as the set of all $R(c,d),R(d,c)$ with $\{c,d\}$ an edge in $G$. We may assume that $G$ contains at least two nodes. For the reduction to strong $(\text{GNFO},\text{GNFO})$-entity distinguishability take two constants $a,b\in \text{cons}(\Dmc)$ and obtain $\Dmc'$ by adding $A(a)$ and $B(b)$ to $\Dmc$. Let $\Kmc=(\Omc,\Dmc')$. Then $G$ is not 3-colorable iff $(\Omc,\Dmc)\models E(c)$ for all $c\in \text{cons}(\Dmc)$ iff $(\Omc,\Dmc_{a=b}')$ is not satisfiable iff $(\Kmc,\{a\},\{b\})$ is strongly $(\text{GNFO},\text{GNFO})$-separable. For the reduction to strong $(\text{GNFO},\text{GNFO})$-referring expression existence take a constant $a\in \text{cons}(\Dmc)$ and obtain $\Dmc'$ by adding $A(a)$ and $B(c)$, $c\in N:= \text{cons}(\Dmc)\setminus\{a\}$ to $\Dmc$. Let $\Kmc=(\Omc,\Dmc')$.
Then $G$ is not 3-colorable iff $(\Omc,\Dmc)\models E(c)$ for all $c\in \text{cons}(\Dmc)$ iff $(\Omc,\Dmc_{a=c}')$ is not satisfiable for any $c\in N$ iff $(\Kmc,\{a\},N)$ is strongly $(\text{GNFO},\text{GNFO})$-separable.

For the \TwoExpTime lower bound in combined complexity, it follows from \cite{DBLP:journals/jacm/BaranyCS15,DBLP:journals/pvldb/BaranyCO12} that it is \TwoExpTime-hard to decide for a GNFO-ontology $\Omc$ and a unary relational symbol $E$ whether $\Omc\models \forall x E(x)$. Assume now that a GNFO-ontology $\Omc$ and $E$ are given. Let $\Kmc=(\Omc',\Dmc)$, where  $\Omc'= \Omc \cup \{A \sqcap B \sqcap E \sqsubseteq\bot\}$ and $\Dmc=\{A(a),B(b)\}$ with $A$ and $B$ fresh unary relation symbols. Then $\Omc\models \forall x E(x)$ iff $(\Omc',\Dmc_{a=b})$
is not satisfiable iff $(\Kmc,\{a\},\{b\})$ is strongly $(\text{GNFO},\text{GNFO})$-separable. The labeled KB $(\Kmc,\{a\},\{b\})$
also shows \TwoExpTime-hardness of definability, RE-existence and entity distinguishability.
\end{proof}
%
%
%
\section{Strong Separability for Decidable Fragments of
FO}\label{sec:strongdecidable}

We study strong $(\Lmc,\Lmc)$-separability for $\Lmc \in \{ \ALCI,
\text{GF}, \text{FO}^{2} \}$.  We show that for all these cases, strong
$(\Lmc,\Lmc)$-separability coincides with strong
$(\Lmc,\text{FO})$-separability and thus we can use the link to KB
unsatisfiability provided by Theorem~\ref{thm:strongsep} to obtain
decidability and tight complexity bounds.  


%
%
\subsection{Strong Separability of Labeled $\mathcal{ALCI}$-KBs}
\label{subsect:alci}

We show that strong $(\ALCI,\ALCI)$-separability coincides with strong
$(\ALCI,\text{FO})$-separability and also prove that strong
$(\ALCI,\ALCI)$-separability is \ExpTime-complete in combined
complexity and \coNPclass-complete in data complexity. With \Kmc-types, 
we mean the types introduced for \ALCI in Section~\ref{section:ALCI}. We identify
a type with the conjunction of concepts in it.
\begin{restatable}{theorem}{thmdlcharone}
	\label{thm:dlchar1}
	For every labeled $\ALCI$-KB $(\Kmc,P,N)$, the following conditions
	are equivalent:
	\begin{enumerate}
		\item $(\Kmc,P,N)$ is strongly $\ALCI$-separable;
		\item $(\Kmc,P,N)$ is strongly FO-separable;
		\item for all $a\in P$ and $b\in N$, there do not exist models $\Amf$ and $\Bmf$ of $\Kmc$
		such that $a^{\Amf}$ and $b^{\Bmf}$ realize the same $\Kmc$-type;
		\item the \ALCI-concept $t_1 \sqcup \cdots \sqcup t_n$ strongly
		separates $(\Kmc,P,N)$, $t_1,\dots,t_n$ the $\Kmc$-types realizable
		in $\Kmc,a$ for some $a\in P$.
	\end{enumerate}
\end{restatable}

\begin{proof} 
  \ ``$1. \Rightarrow 2.$'', ``$3.\Rightarrow 4.$'', and
  ``$4.\Rightarrow 1.$'' are straightforward.
  For ``$2. \Rightarrow 3.$'', let $\Kmc=(\Omc,\Dmc)$ and assume that
	Point~3 does not hold, that is, there exist models $\mathfrak{A}$
	and $\mathfrak{B}$ of $\Kmc$ and $a \in P$, $ b \in N$ such that
	$\text{tp}_{\mathcal{K}}(\mathfrak{A},a^{\Amf})=\text{tp}_{\mathcal{K}}(\mathfrak{B},b^{\Bmf})$. We
	prove that $(\mathcal{O},
	\mathcal{D}_{a=b})$ is satisfiable. This implies that $(\Kmc,P,N)$ is not strongly
	FO-separable by Theorem~\ref{thm:strongsep}.
	
	By reinterpreting constants, we can achieve that $\mathfrak{B}$ is a
	model of the database $\mathcal{D}'$ from the definition of
	$\mathcal{D}_{a=b}$. Let $\Amf \uplus \Bmf$ be the disjoint union of $\Amf$ and $\Bmf$. Define the structure $\mathfrak{C}$ as
	$\mathfrak{A} \uplus \mathfrak{B}$ in which $a^{\Amf}$ and
	$b'^{\Bmf}$ are identified.  There is an obvious surjection $f:
	\text{dom}(\mathfrak{A} \uplus \mathfrak{B}) \rightarrow
	\text{dom}(\mathfrak{C})$ that maps every individual to itself except
	that $f(a^{\Amf})=f({b'}^{\Bmf})=a^{\Cmf}$.  Using the fact that
	$\text{tp}_{\mathcal{K}}(\mathfrak{A},a^{\Amf})=\text{tp}_{\mathcal{K}}(\mathfrak{B},b'^{\Bmf})$
	and a simple induction on the structure of concepts $C$, we can show
	that for all $C \in \mn{cl}(\mathcal{K})$ and $d \in
	\text{dom}(\mathfrak{A} \uplus \mathfrak{B})$, $d \in
	C^{\mathfrak{A} \uplus \mathfrak{B}}$ iff $f(d) \in C^\mathfrak{C}$.
	Since \Amf and \Bmf are models of \Omc, it follows that \Cmf is a
	model of \Omc. By construction, it is also a model of $\mathcal{D}_{a=b}$.
\end{proof}
Note that Point~4 of Theorem~\ref{thm:dlchar1} provides concrete
separating concepts. These are not illuminating, but of size at most
$2^{p(||\Omc||)}$, $p$ a polynomial.  In contrast to the case of weak
separability, the length of separating concepts is thus independent of
\Dmc. Satisfiability of $\ALCI$-KBs is \ExpTime-complete~\cite{DL-Textbook} in combined complexity and NP-complete in data complexity. Hence, strong $(\ALCI,\ALCI)$-separability is in \ExpTime in combined complexity and in \coNPclass in data complexity. 
%
%
\begin{corollary}
	\label{thm:alcistrong}
	For any FO-fragment $\Lmc_{S} \supseteq \text{UCQ}$:
	\begin{enumerate}
		\item strong $(\ALCI,\ALCI)$-separability coincides with strong
	$(\ALCI,\Lmc_S)$-separability, the same is true for definability;
	   \item strong $(\ALCI,\ALCI)$-separability and strong $(\ALCI,\ALCI)$-definability are \ExpTime-complete in combined complexity and \coNPclass-complete in data complexity.
   \end{enumerate}
This also holds for strong RE-existence and entity distinguishability, for all FO-fragments $\Lmc_{S}\supseteq \text{CQ}$.
\end{corollary}
\begin{proof} \
	It remains to consider the lower bounds. The ontology $\Omc$ used in the proof of \coNPclass-hardness in data complexity for GNFO in Corollary~\ref{thm:gnfostrong} is an $\ALCI$-ontology. The lower
	bounds in data complexity therefore follow directly from that proof.
	The \ExpTime lower bound in combined complexity can also be proved in the same way as the \TwoExpTime-lower bound in Corollary~\ref{thm:gnfostrong} using the fact that it is \ExpTime-hard to decide whether $\Omc\models \top \sqsubseteq A$ for $\ALCI$-ontologies $\Omc$~\cite{DL-Textbook}.	
\end{proof}
\subsection{Strong Separability of Labeled GF-KBs}
\label{subsection:strongGF}
We show that strong $(\text{GF},\text{GF})$-separability and strong $(\text{GF},\text{openGF})$ coincide with strong
$(\text{GF},\text{FO})$-separability and that both are 2\ExpTime-complete in combined
complexity and \coNPclass-complete in data complexity. We also show that
the length of strongly separating GF-formulas is independent from the database $\Dmc$ but that this is not the case for strongly separating openGF-formulas.
With \Kmc-types, we mean the types introduced for GF in Section~\ref{section:GF}. We identify
a type with the conjunction of its formulas. We begin by formulating 
a counterpart of Theorem~\ref{thm:dlchar1} for GF.
\begin{restatable}{theorem}{thmgfcharone} \label{thm:gfchar1}
	For every labeled GF-KB $(\Kmc,P,N)$, the following conditions
	are equivalent:
	\begin{enumerate}
		
		\item $(\Kmc,P,N)$ is strongly GF-separable;
		
		\item $(\Kmc,P,N)$ is strongly FO-separable;
		
		\item for all $\vec a\in P$ and $\vec b\in N$, there do not exist models $\Amf$ and $\Bmf$ of $\Kmc$
		such that $\vec a^{\Amf}$ and $\vec b^{\Bmf}$ realize the same $\Kmc$-type;
		
		\item the GF-formula $\Phi_1(\vec x) \vee \cdots \vee
		\Phi_n(\vec x)$ strongly
		separates $(\Kmc,P,N)$, $\Phi_1(\vec x),\dots,\Phi_n(\vec x)$ the $\Kmc$-types realizable
		in $\Kmc,\vec{a}$ for some $\vec{a}\in P$.
		
	\end{enumerate}	
\end{restatable}
\begin{proof} \
	The proof is similar to the proof of Theorem~\ref{thm:dlchar1}. The  implication ``2. $\Rightarrow$ 3.'' relies on the fact that one can merge 
	models $\Amf$ and $\Bmf$ of a GF-KB $\Kmc$ which realize the same $\Kmc$-type $\Phi(\vec{x})$ at tuples $\vec{a}$ and $\vec{b}$, respectively, by identifying the interpretation of the tuples and obtain a model of $\Kmc$ realizing $\Phi(\vec{x})$.
\end{proof}
It follows that the size of strongly separating GF-formulas is at most
$2^{2^{p(||\Omc||)}}$, $p$ a polynomial, and thus does not depend on
the database. Interestingly, we can use a variation of
Example~\ref{exmp:openGFGF} to show that this is not the case for
separating openGF-formulas.	It follows also that GF and 
openGF differ in terms of the
size of strongly separating formula. 
\begin{exmp}\label{ex:gfopengf}
	{\em Define a GF-ontology $\Omc$ as follows: 
$$
\Omc = \{A_{1} \sqsubseteq \forall S.A_{1}, \quad A_{2}\sqsubseteq \forall R.A_{2}, \quad E_{2} \sqcap A_{1} \sqsubseteq \exists u.B, \quad E_{1} \sqcap A_{2} \sqsubseteq \neg \exists u.B\}.
$$
Here, $u$ denotes the universal role as discussed earlier. For example, $E_{2} \sqcap A_{1} \sqsubseteq \exists u.B$ is logically equivalent to $\forall x (E_{2}(x) \wedge A_{1}(x) \rightarrow \exists y B(y))$.
Note that the first two CIs propagate $A_{1}$ and $A_{2}$ along the
role names $S$ and $R$, respectively, and that according to the
remaining CIs $E_{2}\sqcap A_{1}$ and $E_{1} \sqcap A_{2}$ enforce
that $B$ is non-empty and empty, respectively. It follows, in
particular, that $E_{2}\sqcap A_{1}$ and $E_{1} \sqcap A_{2}$ cannot both be satisfied in a model of $\Omc$. Let 
$$
\Dmc_{n} = \{A_{1}(a_{0}),E_{1}(c_{n}),R(a_{0},c_{1}),\ldots,R(c_{n-1},c_{n})\}
		\cup\{A_{2}(b_{0}), E_{2}(d_{n}),S(b_{0},d_{1}),\ldots,S(d_{n-1},d_{n})\}.
		$$
Thus, we have an $R$-chain starting at $A_{1}(a_{0})$ and an $S$-chain
starting at $A_{2}(b_{0})$. Let $\Kmc_{n}=(\Omc,\Dmc_{n})$ and let $P=\{a_{0}\}$ and $N=\{b_{0}\}$. 
In GF (in fact in $\mathcal{ALC}$ with the universal role) the following formula strongly separates $(\Kmc_{n},P,N)$:
		$$
		(A_{1} \sqcap A_{2} \sqcap \neg\exists u.B) \sqcup (A_{1} \sqcap \neg A_{2}).
		$$
In contrast, any strongly separating formula in openGF has guarded quantifier rank at least $n$. To show this consider the models $\Amf$ and $\Bmf$ defined in Figure~\ref{fig:strongopengfdepth} with the constants $c_{i}$ and $d_{i}$
interpreted so that they are both models of $\Kmc_{n}$. We have
$\Amf,a_{0}^{\Amf}\sim_{\text{openGF},\text{sig}(\Kmc_{n})}^{n-1}\Bmf,b_{0}^{\Bmf}$
and so $\Amf,a_{0}^{\Amf}$ and $\Bmf,b_{0}^{\Bmf}$ cannot be
distinguished by any openGF-formula of guarded quantifier rank $\leq
n-1$. \finofex
}
\end{exmp}
\begin{figure*}
	
	\begin{center}

		\tikzset{every picture/.style={line width=0.5pt}} 
		
		\begin{tikzpicture}[x=0.75pt,y=0.75pt,yscale=-1,xscale=1]
			
			\draw    (100,104.5) -- (130.16,89.04) ;
			\draw [shift={(132.83,87.67)}, rotate = 512.86] [fill={rgb, 255:red, 0; green, 0; blue, 0 }  ][line width=0.08]  [draw opacity=0] (5,-3) -- (0,0) -- (5,3) -- cycle    ;
			\draw  [fill={rgb, 255:red, 0; green, 0; blue, 0 }  ,fill opacity=1 ] (139.5,85.25) .. controls (139.5,84.28) and (138.72,83.5) .. (137.75,83.5) .. controls (136.78,83.5) and (136,84.28) .. (136,85.25) .. controls (136,86.22) and (136.78,87) .. (137.75,87) .. controls (138.72,87) and (139.5,86.22) .. (139.5,85.25) -- cycle ;
			\draw    (100,104.5) -- (130.54,121.23) ;
			\draw [shift={(133.17,122.67)}, rotate = 208.71] [fill={rgb, 255:red, 0; green, 0; blue, 0 }  ][line width=0.08]  [draw opacity=0] (5,-3) -- (0,0) -- (5,3) -- cycle    ;
			\draw    (169.5,85.5) -- (196.5,85.27) ;
			\draw [shift={(199.5,85.25)}, rotate = 539.52] [fill={rgb, 255:red, 0; green, 0; blue, 0 }  ][line width=0.08]  [draw opacity=0] (5,-3) -- (0,0) -- (5,3) -- cycle    ;
			\draw  [fill={rgb, 255:red, 0; green, 0; blue, 0 }  ,fill opacity=1 ] (207.5,85.25) .. controls (207.5,84.28) and (206.72,83.5) .. (205.75,83.5) .. controls (204.78,83.5) and (204,84.28) .. (204,85.25) .. controls (204,86.22) and (204.78,87) .. (205.75,87) .. controls (206.72,87) and (207.5,86.22) .. (207.5,85.25) -- cycle ;
			\draw  [fill={rgb, 255:red, 0; green, 0; blue, 0 }  ,fill opacity=1 ] (139.5,125.25) .. controls (139.5,124.28) and (138.72,123.5) .. (137.75,123.5) .. controls (136.78,123.5) and (136,124.28) .. (136,125.25) .. controls (136,126.22) and (136.78,127) .. (137.75,127) .. controls (138.72,127) and (139.5,126.22) .. (139.5,125.25) -- cycle ;
			\draw    (169.5,125.5) -- (196.5,125.39) ;
			\draw [shift={(199.5,125.38)}, rotate = 539.76] [fill={rgb, 255:red, 0; green, 0; blue, 0 }  ][line width=0.08]  [draw opacity=0] (5,-3) -- (0,0) -- (5,3) -- cycle    ;
			\draw  [fill={rgb, 255:red, 0; green, 0; blue, 0 }  ,fill opacity=1 ] (207.5,125.25) .. controls (207.5,124.28) and (206.72,123.5) .. (205.75,123.5) .. controls (204.78,123.5) and (204,124.28) .. (204,125.25) .. controls (204,126.22) and (204.78,127) .. (205.75,127) .. controls (206.72,127) and (207.5,126.22) .. (207.5,125.25) -- cycle ;
			\draw   (77.88,104.19) .. controls (77.88,98.08) and (82.83,93.13) .. (88.94,93.13) .. controls (95.05,93.13) and (100,98.08) .. (100,104.19) .. controls (100,110.3) and (95.05,115.25) .. (88.94,115.25) .. controls (82.83,115.25) and (77.88,110.3) .. (77.88,104.19) -- cycle ;
			\draw  [dash pattern={on 0.84pt off 2.51pt}]  (145.2,85) -- (164.17,85) ;
			\draw  [dash pattern={on 0.84pt off 2.51pt}]  (145.95,125) -- (165.17,125) ;
			\draw    (100,197.5) -- (128.75,197.84) ;
			\draw [shift={(131.75,197.88)}, rotate = 180.68] [fill={rgb, 255:red, 0; green, 0; blue, 0 }  ][line width=0.08]  [draw opacity=0] (5,-3) -- (0,0) -- (5,3) -- cycle    ;
			\draw  [fill={rgb, 255:red, 0; green, 0; blue, 0 }  ,fill opacity=1 ] (139.5,198.25) .. controls (139.5,197.28) and (138.72,196.5) .. (137.75,196.5) .. controls (136.78,196.5) and (136,197.28) .. (136,198.25) .. controls (136,199.22) and (136.78,200) .. (137.75,200) .. controls (138.72,200) and (139.5,199.22) .. (139.5,198.25) -- cycle ;
			\draw    (169.5,198.5) -- (196.5,198.39) ;
			\draw [shift={(199.5,198.38)}, rotate = 539.76] [fill={rgb, 255:red, 0; green, 0; blue, 0 }  ][line width=0.08]  [draw opacity=0] (5,-3) -- (0,0) -- (5,3) -- cycle    ;
			\draw  [fill={rgb, 255:red, 0; green, 0; blue, 0 }  ,fill opacity=1 ] (207.5,198.25) .. controls (207.5,197.28) and (206.72,196.5) .. (205.75,196.5) .. controls (204.78,196.5) and (204,197.28) .. (204,198.25) .. controls (204,199.22) and (204.78,200) .. (205.75,200) .. controls (206.72,200) and (207.5,199.22) .. (207.5,198.25) -- cycle ;
			\draw   (77.88,197.19) .. controls (77.88,191.08) and (82.83,186.13) .. (88.94,186.13) .. controls (95.05,186.13) and (100,191.08) .. (100,197.19) .. controls (100,203.3) and (95.05,208.25) .. (88.94,208.25) .. controls (82.83,208.25) and (77.88,203.3) .. (77.88,197.19) -- cycle ;
			\draw  [dash pattern={on 0.84pt off 2.51pt}]  (145.95,198) -- (165.17,198) ;
			\draw   (77.88,158.19) .. controls (77.88,152.08) and (82.83,147.13) .. (88.94,147.13) .. controls (95.05,147.13) and (100,152.08) .. (100,158.19) .. controls (100,164.3) and (95.05,169.25) .. (88.94,169.25) .. controls (82.83,169.25) and (77.88,164.3) .. (77.88,158.19) -- cycle ;
			\draw    (382,104.5) -- (412.16,89.04) ;
			\draw [shift={(414.83,87.67)}, rotate = 512.86] [fill={rgb, 255:red, 0; green, 0; blue, 0 }  ][line width=0.08]  [draw opacity=0] (5,-3) -- (0,0) -- (5,3) -- cycle    ;
			\draw  [fill={rgb, 255:red, 0; green, 0; blue, 0 }  ,fill opacity=1 ] (421.5,85.25) .. controls (421.5,84.28) and (420.72,83.5) .. (419.75,83.5) .. controls (418.78,83.5) and (418,84.28) .. (418,85.25) .. controls (418,86.22) and (418.78,87) .. (419.75,87) .. controls (420.72,87) and (421.5,86.22) .. (421.5,85.25) -- cycle ;
			\draw    (382,104.5) -- (412.54,121.23) ;
			\draw [shift={(415.17,122.67)}, rotate = 208.71] [fill={rgb, 255:red, 0; green, 0; blue, 0 }  ][line width=0.08]  [draw opacity=0] (5,-3) -- (0,0) -- (5,3) -- cycle    ;
			\draw    (451.5,85.5) -- (478.5,85.27) ;
			\draw [shift={(481.5,85.25)}, rotate = 539.52] [fill={rgb, 255:red, 0; green, 0; blue, 0 }  ][line width=0.08]  [draw opacity=0] (5,-3) -- (0,0) -- (5,3) -- cycle    ;
			\draw  [fill={rgb, 255:red, 0; green, 0; blue, 0 }  ,fill opacity=1 ] (489.5,85.25) .. controls (489.5,84.28) and (488.72,83.5) .. (487.75,83.5) .. controls (486.78,83.5) and (486,84.28) .. (486,85.25) .. controls (486,86.22) and (486.78,87) .. (487.75,87) .. controls (488.72,87) and (489.5,86.22) .. (489.5,85.25) -- cycle ;
			\draw  [fill={rgb, 255:red, 0; green, 0; blue, 0 }  ,fill opacity=1 ] (421.5,125.25) .. controls (421.5,124.28) and (420.72,123.5) .. (419.75,123.5) .. controls (418.78,123.5) and (418,124.28) .. (418,125.25) .. controls (418,126.22) and (418.78,127) .. (419.75,127) .. controls (420.72,127) and (421.5,126.22) .. (421.5,125.25) -- cycle ;
			\draw    (451.5,125.5) -- (478.5,125.39) ;
			\draw [shift={(481.5,125.38)}, rotate = 539.76] [fill={rgb, 255:red, 0; green, 0; blue, 0 }  ][line width=0.08]  [draw opacity=0] (5,-3) -- (0,0) -- (5,3) -- cycle    ;
			\draw  [fill={rgb, 255:red, 0; green, 0; blue, 0 }  ,fill opacity=1 ] (489.5,125.25) .. controls (489.5,124.28) and (488.72,123.5) .. (487.75,123.5) .. controls (486.78,123.5) and (486,124.28) .. (486,125.25) .. controls (486,126.22) and (486.78,127) .. (487.75,127) .. controls (488.72,127) and (489.5,126.22) .. (489.5,125.25) -- cycle ;
			\draw   (359.88,104.19) .. controls (359.88,98.08) and (364.83,93.13) .. (370.94,93.13) .. controls (377.05,93.13) and (382,98.08) .. (382,104.19) .. controls (382,110.3) and (377.05,115.25) .. (370.94,115.25) .. controls (364.83,115.25) and (359.88,110.3) .. (359.88,104.19) -- cycle ;
			\draw  [dash pattern={on 0.84pt off 2.51pt}]  (427.2,85) -- (446.17,85) ;
			\draw  [dash pattern={on 0.84pt off 2.51pt}]  (427.95,125) -- (447.17,125) ;
			\draw    (382,197.5) -- (410.75,197.84) ;
			\draw [shift={(413.75,197.88)}, rotate = 180.68] [fill={rgb, 255:red, 0; green, 0; blue, 0 }  ][line width=0.08]  [draw opacity=0] (5,-3) -- (0,0) -- (5,3) -- cycle    ;
			\draw  [fill={rgb, 255:red, 0; green, 0; blue, 0 }  ,fill opacity=1 ] (421.5,198.25) .. controls (421.5,197.28) and (420.72,196.5) .. (419.75,196.5) .. controls (418.78,196.5) and (418,197.28) .. (418,198.25) .. controls (418,199.22) and (418.78,200) .. (419.75,200) .. controls (420.72,200) and (421.5,199.22) .. (421.5,198.25) -- cycle ;
			\draw    (451.5,198.5) -- (478.5,198.39) ;
			\draw [shift={(481.5,198.38)}, rotate = 539.76] [fill={rgb, 255:red, 0; green, 0; blue, 0 }  ][line width=0.08]  [draw opacity=0] (5,-3) -- (0,0) -- (5,3) -- cycle    ;
			\draw  [fill={rgb, 255:red, 0; green, 0; blue, 0 }  ,fill opacity=1 ] (489.5,198.25) .. controls (489.5,197.28) and (488.72,196.5) .. (487.75,196.5) .. controls (486.78,196.5) and (486,197.28) .. (486,198.25) .. controls (486,199.22) and (486.78,200) .. (487.75,200) .. controls (488.72,200) and (489.5,199.22) .. (489.5,198.25) -- cycle ;
			\draw   (359.88,197.19) .. controls (359.88,191.08) and (364.83,186.13) .. (370.94,186.13) .. controls (377.05,186.13) and (382,191.08) .. (382,197.19) .. controls (382,203.3) and (377.05,208.25) .. (370.94,208.25) .. controls (364.83,208.25) and (359.88,203.3) .. (359.88,197.19) -- cycle ;
			\draw  [dash pattern={on 0.84pt off 2.51pt}]  (427.95,198) -- (447.17,198) ;
			\draw   (359.88,158.19) .. controls (359.88,152.08) and (364.83,147.13) .. (370.94,147.13) .. controls (377.05,147.13) and (382,152.08) .. (382,158.19) .. controls (382,164.3) and (377.05,169.25) .. (370.94,169.25) .. controls (364.83,169.25) and (359.88,164.3) .. (359.88,158.19) -- cycle ;
			\draw   (22.67,39.5) -- (290.17,39.5) -- (290.17,234.83) -- (22.67,234.83) -- cycle ;
			\draw   (303.67,39.5) -- (571.17,39.5) -- (571.17,235) -- (303.67,235) -- cycle ;
			
			\draw (89.83,105.33) node  [font=\small] [align=left] {$\displaystyle a_{0}$};
			\draw (110,85.5) node  [font=\small] [align=left] {$\displaystyle R$};
			\draw (182,75) node  [font=\small] [align=left] {$\displaystyle R$};
			\draw (109.5,122) node  [font=\small] [align=left] {$\displaystyle S$};
			\draw (182,136) node  [font=\small] [align=left] {$\displaystyle S$};
			\draw (49.5,103.5) node  [font=\footnotesize] [align=left] {$\displaystyle \{ A_{1} ,A_{2}\}$};
			\draw (142,69.5) node  [font=\footnotesize] [align=left] {$\displaystyle \{ A_{1} ,A_{2}\}$};
			\draw (243.5,86) node  [font=\footnotesize] [align=left] {$\displaystyle \{ A_{1} ,A_{2}, E_1\}$};
			\draw (235,125.5) node  [font=\footnotesize] [align=left] {$\displaystyle \{ A_{1} ,A_{2}\}$};
			\draw (142,139.5) node  [font=\footnotesize] [align=left] {$\displaystyle \{ A_{1} ,A_{2}\}$};
			\draw (89.83,198.33) node  [font=\small] [align=left] {$\displaystyle b_{0}$};
			\draw (109.5,208) node  [font=\small] [align=left] {$\displaystyle S$};
			\draw (181,208) node  [font=\small] [align=left] {$\displaystyle S$};
			\draw (60.17,196.5) node  [font=\footnotesize] [align=left] {$\displaystyle \{ A_{2}\}$};
			\draw (224,198.5) node  [font=\footnotesize] [align=left] {$\displaystyle \{E_{2}\}$};
			\draw (89.83,159.33) node  [font=\small] [align=left] {$\displaystyle f$};
			\draw (33.33,51) node   [align=left] {$\displaystyle \mathfrak{A}$};
			\draw (371.83,105.33) node  [font=\small] [align=left] {$\displaystyle b_{0}$};
			\draw (392,85.5) node  [font=\small] [align=left] {$\displaystyle R$};
			\draw (464,75) node  [font=\small] [align=left] {$\displaystyle R$};
			\draw (391.5,122) node  [font=\small] [align=left] {$\displaystyle S$};
			\draw (464,136) node  [font=\small] [align=left] {$\displaystyle S$};
			\draw (331.5,103.5) node  [font=\footnotesize] [align=left] {$\displaystyle \{ A_{1} ,A_{2}\}$};
			\draw (424,69.5) node  [font=\footnotesize] [align=left] {$\displaystyle \{ A_{1} ,A_{2}\}$};
			\draw (516,86) node  [font=\footnotesize] [align=left] {$\displaystyle \{ A_{1} ,A_{2}\}$};
			\draw (525,125.5) node  [font=\footnotesize] [align=left] {$\displaystyle \{ A_{1} ,A_{2}, E_2\}$};
			\draw (424,139.5) node  [font=\footnotesize] [align=left] {$\displaystyle \{A_{1} ,A_{2}\}$};
			\draw (371.83,198.33) node  [font=\small] [align=left] {$\displaystyle a_{0}$};
			\draw (391.5,208) node  [font=\small] [align=left] {$\displaystyle R$};
			\draw (463,208) node  [font=\small] [align=left] {$\displaystyle R$};
			\draw (342.17,196.5) node  [font=\footnotesize] [align=left] {$\displaystyle \{ A_{1}\}$};
			\draw (506,198.5) node  [font=\footnotesize] [align=left] {$\displaystyle \{ E_{1}\}$};
			\draw (371.83,159.33) node  [font=\small] [align=left] {$\displaystyle f$};
			\draw (315.33,51) node   [align=left] {$\displaystyle \mathfrak{B}$};
			\draw (342.17,156.5) node  [font=\footnotesize] [align=left] {$\displaystyle \{ B \}$};
			
			\draw (139.83,96.93) node  [font=\footnotesize] [align=left] {$\displaystyle c_{1}$};
			\draw (205.43,96.93) node  [font=\footnotesize] [align=left] {$\displaystyle c_{n}$};
			\draw (425.53,112.13) node  [font=\footnotesize] [align=left] {$\displaystyle d_{1}$};
			\draw (487.93,112.13) node  [font=\footnotesize] [align=left] {$\displaystyle d_{n}$};
			\draw (423.83,186) node  [font=\footnotesize] [align=left] {$\displaystyle c_{1}$};
			\draw (486.43,186) node  [font=\footnotesize] [align=left] {$\displaystyle c_{n}$};
			\draw (140.53,186) node  [font=\footnotesize] [align=left] {$\displaystyle d_{1}$};
			\draw (207.93,186) node  [font=\footnotesize] [align=left] {$\displaystyle d_{n}$};
		\end{tikzpicture}
		\caption{Models $\Amf$ and $\Bmf$ in Example~\ref{ex:gfopengf}.} \label{fig:strongopengfdepth}
		
	\end{center}
\end{figure*}

We next show a counterpart of Theorem~\ref{thm:GFopen}.
\begin{theorem}\label{thm:strongGFopenGF}
	Strong $(\text{GF},\text{GF})$-separability coincides with strong
	$(\text{GF},\text{openGF})$-separability.
\end{theorem}
\begin{proof} \
The proof uses guarded bisimulations. It suffices to consider labeled KBs with singleton sets of positive and negative examples. Let $(\Kmc,\{\vec{a}\},\{\vec{b}\})$ be a labeled GF-KB and $\Sigma=\text{sig}(\Kmc)$. We show that the following conditions are equivalent:
\begin{enumerate}
	\item $(\Kmc,\{\vec{a}\},\vec b)$ is strongly openGF-separable;
	\item for all models \Amf and \Bmf of \Kmc: $\Amf,\vec
	a^\Amf\not\sim_{\text{openGF},\Sigma}\Bmf,\vec b^\Bmf$;
	\item for all models \Amf and \Bmf of \Kmc: $\Amf,\vec a^\Amf\not\sim_{\text{GF},\Sigma}\Bmf,\vec
	b^\Bmf$;
	\item $(\Kmc,\{\vec{a}\},\vec b)$ is strongly GF-separable.
\end{enumerate} 
The implication ``1. $\Rightarrow$ 4.'' is trivial and ``4. $\Rightarrow$ 3.'' follows directly from Lemma~\ref{lem:guardedbisim}. 

``3. $\Rightarrow$ 2.'' Suppose there are models \Amf and \Bmf of \Kmc such that $\Amf,\vec
a^\Amf\sim_{\text{openGF},\Sigma}\Bmf,\vec b^\Bmf$. Then one can construct a model $\Bmf'$ of $\Kmc$ such that $\Amf,\vec
a^\Amf\sim_{\text{GF},\Sigma}\Bmf',\vec b^\Bmf$ in exactly the same way as in the proof of the implication ``3. $\Rightarrow$ 4.'' of Theorem~\ref{thm:critGF1}. 

``2. $\Rightarrow$ 1.'' Suppose $(\Kmc,\{\vec{a}\},\vec b)$ is not strongly openGF-separable. Let $\vec{x}$ be a sequence of different variables of the same length as $\vec{a}$. Let
$$
\Gamma_{\vec{a}} = \{ \varphi(\vec x)\in \text{openGF}(\Sigma) \mid
	\Kmc\models \varphi(\vec a)\}, \quad
\Gamma_{\vec b}  =  \{ \varphi(\vec x)\in \text{openGF}(\Sigma) \mid
\,\Kmc\models \varphi(\vec b)\}
$$
Note that $\Gamma_{\vec{a}}$ and $\Gamma_{\vec b}$ are both
closed under conjunction. Say that a set $\Gamma$ of formulas in GF of the form $\varphi(\vec{x})$ is
\emph{satisfiable in} $\Kmc,\vec a$ if the set of first-order formulas $\Omc\cup \Dmc \cup
\{\varphi(\vec a) \mid \varphi(\vec x)\in \Gamma\}$
is satisfiable. We show that $\Gamma_{\vec a}\cup \Gamma_{\vec b}$ is satisfiable in $\Kmc,\vec a$ and in $\Kmc,\vec{b}$. Assume that this is not the case for $\Kmc,\vec{a}$ (the case $\Kmc,\vec{b}$ is considered in the same way). Then it follows that $\Gamma_{\vec b}$ is not satisfiable in $\Kmc,\vec a$. By compactness and closure under conjunction, there exists $\varphi(\vec
x)\in \Gamma_{\vec b}$ such that $\Kmc\models \neg \varphi(\vec
a)$. As $\Kmc\models \varphi(\vec
b)$ this contradicts the assumption
that $(\Kmc,\{\vec a\},\{\vec b\})$ is not strongly openGF-separable. 

Now, let $\Gamma_{0}=\Gamma_{\vec a}\cup \Gamma_{\vec b}$ and consider an
enumeration $\varphi_{1},\varphi_{2},\ldots$ 
of the remaining openGF$(\Sigma)$-formulas of the form $\varphi(\vec{x})$. Then we set inductively
$\Gamma_{i+1}= \Gamma_{i}\cup \{\varphi_{i+1}\}$ if $\Gamma_{i}\cup \{\varphi_{i+1}\}$ is satisfiable in both $\Kmc,\vec a$ and
$\Kmc, \vec b$ and set $\Gamma_{i+1}=\Gamma_{i}\cup \{\neg
\varphi_{i+1}\}$ otherwise.
Let $\Gamma=\bigcup_{i\geq 0}\Gamma_{i}$. One can now easily show that $\Gamma$
is satisfiable in both $\Kmc,\vec a$ and
$\Kmc, \vec b$. Hence there exist models
$\Amf$ and $\Bmf$ of $\Kmc$ such
that $\Amf\models \varphi(\vec a)$ and
$\Bmf\models \varphi(\vec b)$ for all $\varphi\in \Gamma$. Thus, by definition,
$\Amf,\vec a^{\Amf}\equiv_{\text{openGF},\Sigma}\Bmf,\vec b^{\Bmf}$.
Any structure $\Cmf$ is a substructure of an $\omega$-saturated structure $\Cmf'$ satisfying the same FO-formulas $\varphi(\vec{c})$ with $\vec{c}$ tuples
in $\text{dom}(\Cmf)$ regarded as constants~\cite{modeltheory}. Hence we may assume that $\Amf$ and $\Bmf$ are $\omega$-saturated. By Lemma~\ref{lem:guardedbisim}, we obtain 
$\Amf,\vec a^{\Amf}\sim_{\text{openGF},\Sigma}\Bmf,\vec b^{\Bmf}$, as required.
\end{proof}
Satisfiability of GF-KBs is \TwoExpTime-complete in combined
complexity and \NPclass-complete in data complexity
\cite{DBLP:journals/jsyml/Gradel99}.
We  obtain the following.
\begin{corollary}
	\label{thm:gfstrong}
	For any FO-fragment $\Lmc_{S} \supseteq \text{UCQ}$:
	\begin{enumerate}
		\item strong $(\text{GF},\text{GF})$-separability, strong $(\text{GF},\text{openGF})$, and strong
		$(GF,\Lmc_S)$-separability coincide, the same is true for definability;
		\item strong $(\text{GF},\text{GF})$-separability and strong
		$(\text{GF},\text{GF})$-definability are \TwoExpTime-complete in combined complexity and \coNPclass-complete in data complexity.
	\end{enumerate} 
This also holds for RE-existence and entity distinguishability, for all FO-fragments $\Lmc_{S}\supseteq \text{CQ}$.
\end{corollary}
\begin{proof} \
	It remains to prove the lower bounds. As $\mathcal{ALCI}$ is a fragment of GF, the ontology $\Omc$ used in the proof of \coNPclass-hardness in data complexity in Corollary~\ref{thm:gnfostrong} is a GF-ontology. The \coNPclass lower bounds in data complexity follow directly from that proof.
	The \TwoExpTime lower bound in combined complexity can also be proved in the same way as the \TwoExpTime-lower bound in Corollary~\ref{thm:gnfostrong} by using the fact that it is \TwoExpTime-hard to decide whether $\Omc\models \forall x A(x)$ for GF-ontologies $\Omc$~\cite{DBLP:journals/jsyml/Gradel99}.	
\end{proof}
\subsection{Strong Separability of Labeled FO$^2$-KBs}
\label{subsection:strongFO2}

We show that in contrast to weak separability, strong
$(\text{FO}^{2},\text{FO}^2)$-separability is decidable and coincides with strong $(\text{FO}^{2},\text{FO})$-separability. The proof
strategy is the same as for \ALCI and GF and thus we first need a
suitable notion of type for FO$^2$-KBs. Existing such notions, such as
the types defined in \cite{DBLP:journals/bsl/GradelKV97}, are not
strong enough for our purposes, and so we define and work with
a more powerful notion of type. We can then once more establish a
theorem that parallels Theorem~\ref{thm:dlchar1} and show that strong separability has the same complexity as non-satisfiability of KBs, both in
combined complexity ({\sc co}\NExpTime-complete) and in data complexity (\coNPclass-complete).

We start by introducing appropriate types for FO$^{2}$-KBs.
Recall that we assume that FO$^{2}$ uses unary and binary relation symbols only and that positive and negative examples are either constants or pairs of constants.
Assume that $\Kmc=(\Omc,\Dmc)$ is an FO$^{2}$-KB. 
Let $\cl(\Kmc)$ denote the closure under single negation
and swapping the variables $x,y$ of the set of subformulas of formulas in $\Omc$
and $\{R(x,x),R(x,y),A(x)\mid R,A \in \mn{sig}(\Kmc)\}$. 
The \emph{1-type for \Kmc}
of a pointed structure $\Amf,a$, denoted
$\text{tp}_{\mathcal{K}}(\mathfrak{A},a)$, is the set of all formulas
$\psi(x)\in \mn{cl}(\Kmc)$ such that $\Amf\models \psi(a)$.  We denote
by $T_{x}(\Kmc)$ the set of all 1-types for \Kmc. We say that
$t(x)\in T_{x}(\Kmc)$ is \emph{realized} in $\Amf,a$ if
$t(x)=\text{tp}_{\mathcal{K}}(\mathfrak{A},a)$.  Denote by $t(x)[y/x]$
the set of formulas obtained from $t(x)$ by swapping $y$ and $x$.

The \emph{$2$-type for \Kmc} of a pointed structure $\Amf,a,b$,
denoted $\text{tp}_{\mathcal{K}}(\mathfrak{A},a,b)$, is the set of all
$R(x,y)$ with $\Amf\models R(a,b)$, $R(y,x)$ with $\Amf\models
R(b,a)$, $\neg R(x,y)$ with $\Amf\not\models R(a,b)$, and $\neg
R(y,x)$ with $\Amf\not\models R(b,a)$, where $R$ is a binary relation
in $\Kmc$. In addition, $x=y\in
\text{tp}_{\mathcal{K}}(\mathfrak{A},a,b)$ if $a=b$ and $\neg (x=y)\in
\text{tp}_{\mathcal{K}}(\mathfrak{A},a,b)$ if $a\not=b$.  We denote by
$T_{x,y}(\Kmc)$ the set of all $2$-types for \Kmc. We say that
$t(x,y)\in T_{x,y}(\Kmc)$ is \emph{realized} in $\Amf,a,b$ if
$t(x,y)=\text{tp}_{\mathcal{K}}(\mathfrak{A},a,b)$.

Types as defined above are not yet sufficiently powerful to ensure
that models of $\Kmc$ can be merged at nodes satisfying the same type.
To achieve this we introduce extended types.
For $t\in T_{x}(\Kmc)$, let 
$$
\text{king}(t) = \forall y (t[y/x] \rightarrow (x=y)),
$$ 
where here and in what follows we identify a type $t$ or set $t[y/x]$ with the conjunction of the formulas they contain. Thus $\text{king}(t)$ states that $t$ is realized at most once. Now, an extended type states not only which formulas in $\mn{cl}(\Kmc)$ are satisfied but also, for example, which types are kings and which binary relations hold between realized kings.
In detail, the \emph{extended 2-type for \Kmc} of a
pointed structure $\Amf,a,b$, denoted
$\text{tp}^*_{\mathcal{K}}(\mathfrak{A},a,b)$, is the conjunction of
\begin{enumerate}
	\item $\text{tp}_{\mathcal{K}}(\mathfrak{A},a,b)\wedge \text{tp}_{\mathcal{K}}(\mathfrak{A},a) \wedge \text{tp}_{\mathcal{K}}(\mathfrak{A},b)[y/x]$ (stating which 1-types are realized in $a$ and $b$ and which relations hold between $a$ and $b$);
	\item $\exists y (\text{tp}_{\mathcal{K}}(\mathfrak{A},a,c)\wedge \text{tp}_{\mathcal{K}}(\mathfrak{A},c)[y/x])$ for any 
	$c\in \text{dom}(\Amf)\setminus\{a,b\}$ such that $\text{tp}_{\Kmc}(\Amf,c)$ is realized exactly once in $\Amf$ (stating which binary relations hold between $a$ and king types);
	\item $\exists x (\text{tp}_{\Kmc}(\Amf,c,b)\wedge \text{tp}_{\Kmc}(\Amf,c)[x/y])$ for any 
	$c\in \text{dom}(\Amf)\setminus\{a,b\}$ such that
	$\text{tp}_{\Kmc}(\Amf,c)$ is realized exactly once in $\Amf$ (stating which binary relations hold between $b$ and king types);
	\item $\neg \exists x \;t$, for any $t\in T_{x}(\Kmc)$ not realized in $\Amf$ (stating which types are not realized);
	\item $\exists x (t \wedge \text{king}(t))$ if $t\in T_{x}(\Kmc)$ is realized exactly once in $\Amf$ (stating which types are king types);
	\item $\exists x (t \wedge \neg \text{king}(t))$ if $t(x)\in T_{x}(\Kmc)$ is realized at least 
	twice in $\Amf$ (stating which types are realized at least twice);
	\item $\exists xy \text{tp}_{\Kmc}(\Amf,c,d) \wedge \text{tp}_{\Kmc}(\Amf,c) \wedge \text{tp}_{\Kmc}(\Amf,d)[y/x]$
	for any $c\not=d$ such that $\text{tp}_{\Kmc}(\Amf,c)$ and
	$\text{tp}_{\Kmc}(\Amf,d)$ are realized exactly once in $\Amf$ (stating which relations hold between king types).
\end{enumerate}
%
We denote by $T^*_{x,y}(\Kmc)$ the set of all extended $2$-types for \Kmc.  We
say that $t(x,y)\in T^*_{x,y}(\Kmc)$ is \emph{realized} in $\Amf,a,b$
if $t(x,y)=\text{tp}^{*}_{\mathcal{K}}(\mathfrak{A},a,b)$.
The \emph{extended 1-type for $\Kmc$} of a pointed structure $\Amf,a$ is
defined in the same way with $\text{tp}_{\Kmc}(\Amf,a,b),
\text{tp}_{\Kmc}(\Amf,b)[y/x]$ removed in Point~1 and with Point~3
completely removed. We also define the realization of such types by
pointed structures as expected.
\begin{theorem}
	For every labeled FO$^{2}$-KB such that the tuples in $P\cup N$ have
	length $i \in \{1,2\}$, the following are equivalent:
	\begin{enumerate}
		\item $(\Kmc,P,N)$ is strongly FO$^{2}$-separable;
		\item $(\Kmc,P,N)$ is strongly FO-separable;
		\item for all $\vec{a}\in P$ and $\vec{b}\in N$, there do not exist
		models $\Amf$ and $\Bmf$ of $\Kmc$ such that $\vec{a}^\Amf$
		and $\vec{b}^\Bmf$ realize the same extended $i$-type for \Kmc;
		\item the FO$^2$-formula $t_1 \vee \cdots \vee t_n$ strongly 
		separates $(\Kmc,P,N)$, $t_1,\dots,t_n$ the extended $i$-types
		for \Kmc realizable in $\Kmc,\vec{a}$ for some
		$\vec{a}\in P$.
	\end{enumerate}
\end{theorem}
\begin{proof} \
	Assume w.l.o.g. that the tuples in $P$ and $N$ have length two.
	Only ``2. $\Rightarrow$ 3.'' is non-trivial. Assume that Condition~3 does not hold. 
	Thus, there are $\vec{a}=(a_{1},a_{2})\in P$ and $\vec{b}=(b_{1},b_{2})\in N$ and models $\Amf$ and $\Bmf$ of $\Kmc$
	such that the extended 2-types of $\Amf,\vec{a}$ and $\Bmf,\vec{b}$ coincide.
	We show that there exists a model $\Cmf$ of $(\Omc,\Dmc_{\vec{a}=\vec{b}})$. Then, by Theorem~\ref{thm:strongsep}, 
	$(\Kmc,P,N)$ is not FO-separable.
	We construct $\Cmf$ from $\Amf$ and $\Bmf$ as follows: assume that $a_{1}^{\Amf}\not=a_{2}^{\Amf}$.
	The case $a_{1}^{\Amf}=a_{2}^{\Amf}$ can be proved similarly.
	Then, by Point~1 of the definition of extended types and since 2-types contain equality assertions,
	$b_{1}^{\Bmf}\not=b_{2}^{\Bmf}$. By Points~5 and 6, $\Amf$ and $\Bmf$ realize exactly the same 1-types once. Let $K$ denote this set of 1-types. 
        By identifying $a_{i}^{\Amf}$ and $b_{i}^{\Amf}$, $i=1,2$, and the nodes $c_{t}\in \text{dom}(\Amf)$ and $d_{t}\in \text{dom}(\Bmf)$ realizing 
        the same 1-type $t$ in $K$ we obtain a structure $\Cmf$ whose substructure induced by
	$\text{dom}(\Amf)$ coincides with $\Amf$ and whose substructure induced by $\text{dom}(\Bmf)$ coincides with $\Bmf$.
	This is well defined by Points~1, 2, 3, and 7 of the definition of extended types. Set $c^{\Cmf}=c^{\Amf}$
	for all constants $c$ in $\Dmc$ and ${c'}^{\Cmf}={c'}^{\Bmf}$ for all constants $c'$ in $\Dmc'$ (from the definition
	of $\Dmc_{\vec{a}=\vec{b}}$). It remains to define
	the 2-type realized by $(c,d)$ in $\Cmf$ for $c\in \text{dom}(\Cmf)\setminus \text{dom}(\Bmf)$
	and $d\in \text{dom}(\Cmf)\setminus \text{dom}(\Amf)$. Assume such a $(c,d)$ is given (the situation is depicted below). Then the
	type $\text{tp}_{\Kmc}(\Bmf,d)$ is realized in $\Amf$, by Points~4 and 6 of the definition of extended types. Take
	$d'\in \text{dom}(\Amf)$ such that $\text{tp}_{\Kmc}(\Amf,d')=\text{tp}_{\Kmc}(\Bmf,d)$. 
	\begin{center}
		
		\tikzset{every picture/.style={line width=0.75pt}} 
		
		\begin{tikzpicture}[x=0.75pt,y=0.75pt,yscale=-1,xscale=1]
			
			\draw   (182.69,178.05) .. controls (178.72,178.67) and (174.57,179) .. (170.3,179) .. controls (140.92,179) and (117.1,163.33) .. (117.1,144) .. controls (117.1,124.67) and (140.92,109) .. (170.3,109) .. controls (174.57,109) and (178.72,109.33) .. (182.69,109.95) -- cycle ;
			\draw    (182.5,86.33) -- (182.5,105) ;
			\draw [shift={(182.5,107)}, rotate = 269.31] [color={rgb, 255:red, 0; green, 0; blue, 0 }  ][line width=0.75]    (5.93,-5.29) .. controls (3.95,-1.4) and (3.31,-0.3) .. (0,0) .. controls (3.31,0.3) and (3.95,1.4) .. (5.93,5.29)   ;
			\draw  [fill={rgb, 255:red, 0; green, 0; blue, 0 }  ,fill opacity=1 ] (157.67,129.67) .. controls (157.67,128.75) and (158.41,128) .. (159.33,128) .. controls (160.25,128) and (161,128.75) .. (161,129.67) .. controls (161,130.59) and (160.25,131.33) .. (159.33,131.33) .. controls (158.41,131.33) and (157.67,130.59) .. (157.67,129.67) -- cycle ;
			\draw   (202,130) .. controls (202,129.08) and (202.75,128.33) .. (203.67,128.33) .. controls (204.59,128.33) and (205.33,129.08) .. (205.33,130) .. controls (205.33,130.92) and (204.59,131.67) .. (203.67,131.67) .. controls (202.75,131.67) and (202,130.92) .. (202,130) -- cycle ;
			\draw   (158.67,160) .. controls (158.67,159.08) and (159.41,158.33) .. (160.33,158.33) .. controls (161.25,158.33) and (162,159.08) .. (162,160) .. controls (162,160.92) and (161.25,161.67) .. (160.33,161.67) .. controls (159.41,161.67) and (158.67,160.92) .. (158.67,160) -- cycle ;
			\draw   (182.69,109.95) .. controls (186.67,109.33) and (190.82,109) .. (195.08,109) .. controls (224.46,109) and (248.28,124.67) .. (248.28,144) .. controls (248.28,163.33) and (224.46,179) .. (195.08,179) .. controls (190.82,179) and (186.67,178.67) .. (182.69,178.05) -- cycle ;
			
			\draw (186.33,78) node  [font=\small] [align=left] {\textit{merged points}};
			\draw (111.53,111.07) node   [align=left] {$\displaystyle \mathfrak{A}$};
			\draw (257.6,111.8) node   [align=left] {$\displaystyle \mathfrak{B}$};
			\draw (147.13,130.13) node   [align=left] {$\displaystyle c$};
			\draw (215.33,131.6) node   [align=left] {$\displaystyle d$};
			\draw (148.33,162.33) node   [align=left] {$\displaystyle d'$};

		\end{tikzpicture}
		
	\end{center}
	We may assume 
	that $d'\not=c$ as $\text{tp}_{\Kmc}(\Bmf,d)$ is realized at least twice in both $\Amf$ and in $\Bmf$. Now interpret the
	relations $R\in \mn{sig}(\Kmc)$ in $\Cmf$ in such a way that $\text{tp}_{\Kmc}(\Cmf,d,c)=\text{tp}_{\Kmc}(\Amf,d',c)$.
	One can then show by induction over the construction of formulas in $\mn{cl}(\Kmc)$ that $\Cmf$ is a model of $(\Omc,\Dmc_{\vec{a}=\vec{b}})$, as required. 
\end{proof}

As in the GF case, strongly separating formulas are of size at most
$2^{2^{p(||\Omc||)}}$, $p$ a polynomial.
As satisfiability of FO$^2$-KBs is \NExpTime-complete in combined
complexity and \NPclass-complete in data complexity
\cite{DBLP:journals/iandc/Pratt-Hartmann09} we obtain the following result.
\begin{corollary}\label{thm:fo2strong}
For any FO-fragment $\Lmc_{S} \supseteq \text{UCQ}$:
\begin{enumerate}
	\item strong $(\text{FO}^2,\text{FO}^2)$-separability and strong
	$(\text{FO}^{2},\Lmc_S)$-separability coincide, the same is true for definability;
	\item strong $(\text{FO}^{2},\text{FO}^{2})$-separability and strong
	$(\text{FO}^{2},\text{FO}^{2})$-definability are {\sc co}\NExpTime-complete in combined complexity and \coNPclass-complete in data complexity.
\end{enumerate} 
This also holds for RE-existence and entity distinguishability, for all FO-fragments $\Lmc_{S}\supseteq \text{CQ}$.	
\end{corollary}
\begin{proof} \
It remains to consider the lower bounds. As $\mathcal{ALCI}$ is a fragment of FO$^{2}$ we can again use the proof of Corollary~\ref{thm:gnfostrong}
for the lower bound in data complexity. As it is {\sc co}\NExpTime-hard 
to check whether $\Omc\models \forall x A(x)$ for an FO$^{2}$-ontology $\Omc$, we can also use the proof of Corollary~\ref{thm:gnfostrong}
for the lower bound in combined complexity.
\end{proof}
%
%

\section{Discussion and Future Work}
\label{sec:discussion}

In this article and in~\cite{DBLP:conf/ijcai/FunkJLPW19}, we have started an
investigation of the separability problem for labeled KBs, that is,
finding logical formulas that separate positive and
negative examples in knowledge bases which consist of incomplete data
and an ontology. We have
established fundamental results for several ontology and separating
languages ranging from full first-order logic to decidable fragments
thereof, e.g., expressive description logics.
In the remainder of the section, we discuss variations and extensions
of the separability problem that are not covered in this paper. In
passing, we mention several interesting directions for future work. 
%

\paragraph{Separability under the Unique Name Assumption}
Recall that we do not adopt the unique name assumption (UNA) which is essential for some of our results. Here we briefly discuss what happens if one adopts it and open problems that arise. We first note that if
$\Lmc$ and $\Lmc_{S}$ are fragments of FO without equality, then the UNA does not influence the $\Lmc_{S}$-consequences of $\Lmc$-KBs in the sense that $\Kmc\models \varphi(\vec{a})$ under UNA iff $\Kmc\models \varphi(\vec{a})$
without UNA, for any $\Lmc$-KB $\Kmc=(\Omc,\Dmc)$, $\Lmc_{S}$-formula $\varphi(\vec{x})$,
and $\vec{a}\in \text{cons}(\Dmc)^{|\vec{x}|}$~\cite{DBLP:journals/tods/BienvenuCLW14}. Clearly, in this case our results are not affected by adopting the UNA. This applies, for example, to (strong) $(\ALCI,\ALCI)$-separability and $(\ALCI,\text{UCQ})$-separability.

\newcommand{\nrightarrowtail}{\centernot\rightarrowtail}
We next consider the consequences of adopting the UNA on weak separability for languages with equality. As noted in Example~\ref{exmp:44} already, the fundamental characterization of weak $(\text{FO},\text{FO})$-separability given in Theorem~\ref{critFOwithoutUNA} does not hold under the UNA. It is, however, straightforward to obtain a characterization of $(\text{FO},\text{FO})$-separability under UNA by adjusting the conditions given in Theorem~\ref{critFOwithoutUNA} as follows. Let $\text{UCQ}^{\not=}$ denote the class of unions of conjunctive queries that also admit atoms of the form $x\not=y$ and let $\varphi^{\not=}_{\Dmc_{\text{con}(\vec{a}),\vec{a}}}$ denote the
CQ $\varphi_{\Dmc_{\text{con}(\vec{a}),\vec{a}}}$ extended by
$x_{c}\not=x_{d}$ for any constants $c\not=d$ in
$\Dmc_{\text{con}(\vec{a})}$. Let $\Dmc,\vec{a}\nrightarrowtail
\Amf,\vec{b}$ denote that there is no injective homomorphism from $\Dmc$ to $\Amf$ mapping $\vec{a}$ to $\vec{b}$. Then we obtain the following characterization using almost the same proof as before.
\begin{theorem}\label{critFOwithoutUNA3}
	Let $(\Kmc,P,\{\vec{b}\})$ be a labeled FO-KB,
	$\Kmc=(\Omc,\Dmc)$. 
	Then under the UNA the following
	conditions are equivalent:
	\begin{enumerate}
		\item $(\Kmc,P,\{\vec{b}\})$ is projectively UCQ$^{\not=}$-separable;
		\item $(\Kmc,P,\{\vec{b}\})$ is projectively FO-separable;
		\item there exists a model $\Amf$ of $\Kmc$ such that
		  for all $\vec{a}\in P$:
		  $\Dmc_{\text{con}(\vec{a})},\vec{a}\nrightarrowtail \Amf,\vec{b}^{\Amf}$;  
		\item the UCQ $\bigvee_{\vec{a}\in P}\varphi^{\not=}_{\Dmc_{\text{con}(\vec{a})},\vec{a}}$ separates
		$(\Kmc,P,\{\vec{b}\})$.
	\end{enumerate} 
\end{theorem} 
This characterization of $(\text{FO},\text{FO})$-separability under the UNA provides a very close link between rooted UCQ$^{\not=}$-evaluation on FO-KBs and FO-separability. It is known that without the restriction to rooted queries,
UCQ$^{\not=}$-evaluation on $\ALCI$-KBS is undecidable~\cite{DBLP:conf/pods/CalvaneseGL98,DBLP:conf/icdt/Rosati07}.
We obtain the following undecidability result by strengthening this result to rooted queries in a straighforward way. The proof is given in the appendix. 
\begin{restatable}{theorem}{thmundeuna}\label{thm:undecuna}
  $(\ALCI,\text{FO})$-separability is undecidable, under the UNA. 
\end{restatable}
Thus, in sharp contrast to separability without the UNA, $(\ALCI,\text{FO})$-separability behaves very differently from $(\ALCI,\ALCI)$-separability under the UNA, both from a semantic and algorithmic viewpoint. We conjecture, however, that the complexity
of projective and non-projective $(\Lmc,\Lmc)$-separability is not
affected and still \TwoExpTime-complete in combined complexity for $\Lmc\in\{\text{GF},\text{GNFO}\}$.

The UNA is also relevant for strong separability, if the ontology or query language uses equality. To illustrate this, extend the ontology $\Omc_{2}$ defined in Example~\ref{exmp:44} by stating that the binary relations ${\sf citizen\_of}$ and ${\sf born\_in}$ are both partial functions, that is, 
$
\forall x\forall y_{1}\forall y_{2} (R(x,y_{1}) \wedge R(x,y_{2}) \rightarrow y_{1}=y_{2})
$
for $R\in \{{\sf citizen\_of},{\sf born\_in}\}$. Denote the resulting ontology by $\Omc_{2}'$ and consider the KB $\Kmc_{2}'=(\Omc_{2}',\Dmc_{1})$ with $\Dmc_{1}$ introduced in Example~\ref{exmp:44}.
Then under the UNA the labeled KB $(\Kmc_{2}',\{a\},\{b\})$ is strongly separated by $\exists y ({\sf born\_in}(x,y) \wedge {\sf citizen\_of}(x,y))$ but it is
not strongly FO-separable without the UNA. 

While Theorem~\ref{thm:strongsep} does not hold under the UNA, one can
adjust it in a similar way as Theorem~\ref{critFOwithoutUNA} by
replacing UCQ by UCQ$^{\not=}$ and the CQs
$\varphi_{\Dmc_{\text{con}(\vec{a}),\vec{a}}}$ by
$\varphi^{\not=}_{\Dmc_{\text{con}(\vec{a}),\vec{a}}}$. Recall that
the database $\Dmc_{\vec{a}=\vec{b}}$ was constructed from \Dmc and a
disjoint copy $\Dmc'$ thereof, c.f.~its definition before
Theorem~\ref{thm:strongsep}. Denote for any partial injection
$S \subseteq (\text{cons}(\Dmc)\setminus [\vec{a}]) \times (\text{cons}(\Dmc')\setminus [\vec{b}])$ by $\mathcal{D}_{\vec{a}=\vec{b}}^{S}$ the database obtained from 
$\mathcal{D}_{\vec{a}=\vec{b}}$ by identifying all $c_{1},c_{2}$ with $(c_{1},c_{2})\in S$.
Then one can prove the following characterization by adapting the proof of Theorem~\ref{thm:strongsep}.
\begin{theorem}\label{critFOwithoutUNA3new}
	Let $(\Kmc,P,\{\vec{b}\})$ be a labeled FO-KB,
	$\Kmc=(\Omc,\Dmc)$. 
	Then under the UNA the following
	conditions are equivalent:
	\begin{enumerate}
		\item $(\Kmc,P,\{\vec{b}\})$ is strongly UCQ$^{\not=}$-separable;
		\item $(\Kmc,P,\{\vec{b}\})$ is strongly FO-separable;
		\item for all $\vec{a}\in P$, $\vec{b}\in N$, and all partial injections $S \subseteq (\text{cons}(\Dmc)\setminus [\vec{a}]) \times (\text{cons}(\Dmc')\setminus [\vec{b}])$,
		the KB $(\Omc,\mathcal{D}_{\vec{a}=\vec{b}}^{S})$ is unsatisfiable (under UNA);  
		\item the UCQ $\bigvee_{\vec{a}\in P}\varphi^{\not=}_{\Dmc_{\text{con}(\vec{a})},\vec{a}}$ strongly separates
		$(\Kmc,P,\{\vec{b}\})$.
	\end{enumerate} 
\end{theorem}
Recall that for $\mathcal{ALCI}$-ontologies $\Omc$ any KB $(\Omc,\Dmc)$ is satisfiable under the UNA iff it is satisfiable without UNA. Hence, in sharp contrast to weak separability, by Point~3 of Theorems~\ref{thm:strongsep} and \ref{critFOwithoutUNA3new},
strong $(\mathcal{ALCI},\text{FO})$-separability does not depend on the UNA and
coincides with strong $(\mathcal{ALCI},\mathcal{ALCI})$-separability with and without the UNA.

On the other hand, if $\Lmc\in \{\text{GNFO},\text{GF},\text{FO}^{2}\}$, then strong $(\Lmc,\text{FO})$-separability depends on the UNA. Consider, for instance,  
$$
\Omc = \{ \forall x \forall y (R(x,y) \rightarrow (R(x,x) \rightarrow (x=y)))\}
$$
and $\Dmc=\{R(a,c),R(b,b)\}$. Then $((\Omc,\Dmc),\{a\},\{b\})$ is strongly FO-separable under the UNA but not without the UNA (check Condition~3 of the characterizations).
It remains an interesting open problem to determine the relationship between strong $(\Lmc,\text{FO})$-separability and strong $(\Lmc,\Lmc)$-separability under the UNA and their complexity for $\Lmc\in \{\text{GNFO},\text{GF},\text{FO}^{2}\}$.

\paragraph{Separability for Expressive Description Logics} Projective weak separability has been studied
also for $\mathcal{ALC}$, $\mathcal{ALCQ}$, and $\mathcal{ALCQI}$, with some surprising results. As we have seen in this article, projective weak separability in $\mathcal{ALCI}$ is \NExpTime-complete in combined complexity, by mutual polynomial time reduction with the complement of rooted 
UCQ-evaluation. In $\mathcal{ALC}$, $\mathcal{ALCQ}$ and $\mathcal{ALCQI}$, projective weak separability can be mutually polynomial time reduced with the complement of natural variants of rooted UCQ-evaluation in the respective DLs~\cite{DBLP:conf/ijcai/FunkJLPW19}. Adopting the UNA, it is shown that projective $(\mathcal{ALC},\mathcal{ALC})$ and $(\mathcal{ALCQ},\mathcal{ALCQ})$-separability are also \NExpTime-complete but that projective $(\mathcal{ALCQI},\mathcal{ALCQI})$-separability is \ExpTime-complete. Strong separability
has also been studied for $\mathcal{ALC}$ and turns out to be
\ExpTime-complete~\cite{DBLP:conf/ijcai/FunkJLPW19}. As discussed
above,  in $\mathcal{ALC}$ and $\mathcal{ALCI}$ neither weak nor
strong separability  depend on the UNA, but for $\mathcal{ALCQ}$ and
$\mathcal{ALCQI}$ they do. In fact, it would be of interest to
investigate whether the complexity results for  $\mathcal{ALCQ}$ and
$\mathcal{ALCQI}$ still hold if one does not adopt the UNA. More
generally, it is a challenging open problem to systematically
investigate the semantic and algorithmic properties of separability in
the context of expressive DLs and consider description logics with
further constructors such as functional roles or number restrictions,
role hierarchies, transitive roles, expressive role inclusions, and
nominals. DLs of interest include $\mathcal{ALCO}$,
$\mathcal{ALCFIO}$, $\mathcal{S}$ (see Example~\ref{ex:exmptrans}), $\mathcal{SHIQ}$, $\mathcal{SHOIQ}$, and
$\mathcal{SROIQ}$~\cite{handbook,DL-Textbook}. Note that
$\mathcal{ALCO}$ is also discussed below when we consider in more
detail the role of constants in the separability problem.

\paragraph{Separability with Signature Restrictions}
In this article, we have investigated separability under the assumption that all relation symbols from the labeled KB can occur in the separating expression.
It is also of interest to consider a signature $\Sigma$ of relation symbols (and possibly constants) that is given as an additional input
and require separating expressions to be formulated in $\Sigma$.
This makes it possible to ‘direct’ separation towards expressions based on desired features and accordingly to exclude
features that are not supposed to be used for separation. In \cite{kr2021restr},
we have started investigating separability under
signature restrictions. Many results comparing the expressive power of different separation languages obtained in this article do not hold under signature restrictions and new separation languages combining UCQs and DLs are introduced to understand the power of projective separability under signature restrictions. Also decision problems becomes much harder.
The following table gives an overview of the complexity of separability for expressive fragments of FO with signature restrictions. For the results on $\mathcal{ALCO}$ we have admitted nominals in separating concepts but no constants are admitted in separating expressions for the remaining languages 
in the table. In the weak separability case only projective separability has been investigated (with unary relation symbols as helper symbols),  non-projective weak separability appears to be very challenging and has not yet been considered.
%
We note that while projective separability without signature restrictions is closely linked to UCQ-evaluation, projective separability with signature restrictions is closely linked to conservative extensions~\cite{DBLP:conf/kr/GhilardiLW06,DBLP:conf/icalp/JungLM0W17} and that while
strong separability without signature restrictions is closely linked to KB satisfiability, strong satisfiability with signature restrictions is closely linked to Craig interpolant existence~\cite{DBLP:conf/aaai/ArtaleJMOW21,DBLP:conf/lics/JungW21}.

\begin{center}
	\begin{tabular}{|c|c|c|}
		\hline
		& Weak Separability & Strong Separability \\
		$\mathcal{L}$ & projective restricted signature & restricted signature\\
		\hline                     
		$\mathcal{ALC}$ & \textsc{2ExpTime} & \textsc{2ExpTime}\\
		$\mathcal{ALCI}$ & \textsc{2ExpTime} & \textsc{2ExpTime}\\ 
		$\mathcal{ALCO}$ & \textsc{3ExpTime} & \textsc{2ExpTime}\\
		GF & Undecidable & \textsc{3ExpTime} \\
		FO$^2$ & Undecidable &
		$[\textsc{2Exp},\textsc{coN}2\textsc{Exp}]$  \\
		\hline
	\end{tabular}
\end{center}
Many challenging problems are open for separability under signature restrictions. For example, the complexity of deciding weak and strong separability for most of the expressive DLs mentioned in the previous paragraph and a better understanding of the separating power of expressive DLs remain to be investigated.
%

\paragraph{Separability with Constants} 
In this article we do not admit constants in ontologies nor in separating expressions.
Even without constants in separating expressions, the
admission of constants in ontologies makes a significant difference. For example, in FO and fragments such as the extension $\mathcal{ALCO}$ of $\mathcal{ALC}$ with nominals one can then state that $a\not=b$ for different constants $a,b$ in the database. Thus, one can implicitly cover the UNA with the consequences discussed above. The existence of useful model-theoretic characterization results and the complexity of separability remain to be investigated. If one admits constants in separating formulas, the situation changes even more drastically. Note that it does not make sense to admit all constants in the underpinning database
of the separability problem 
in separating formulas as in this case the formula $x=a$ would weakly separate the labeled KB $(\Kmc,\{a\},\{b\})$ whenever $a\not=b$. Thus, the 
set of constants that can occur in separating formulas has to be restricted by excluding at least some constants from the database. The general set-up in which one can restrict both the relation symbols and the constants that can occur in separating formulas has been studied in~\cite{kr2021restr} and is discussed above. We note, however, that the case in which one only restricts the use of constants but not the use of relation symbols in separating formulas has not yet been studied. We conjecture that many of the complexity results we obtained in this article still hold.

\paragraph{Separation with CQs, $\mathcal{ELI}$-Concepts, and  
$\mathcal{EL}$-Concepts} 

If one aims at finding separating expressions that
generalize from the positive examples, it is important to avoid overfitting. 
From a logical viewpoint this can be achieved by disallowing disjunction in separating formulas and admit as separating expressions only CQs (in the context of query by example) or concepts in the lightweight description logics $\mathcal{EL}$ or $\mathcal{ELI}$ (in the context of DL concept learning). Note that $\mathcal{ELI}$-concepts can be regarded as rooted tree-shaped CQs and $\mathcal{EL}$-concepts as rooted tree-shaped CQs in which children are reached by role names~\cite{DL-Textbook}.
We have seen in this article that for labeled KBs with a single positive example only, CQs often have the same separating power as UCQs. For labeled KBs with many positive example this is clearly not the case. In fact,
separability often becomes undecidable:
in~\cite{DBLP:conf/ijcai/FunkJLPW19},  it is shown by a reduction of
the undecidable CQ query inseparability problem for
$\mathcal{ALC}$-KBs~\cite{DBLP:journals/ai/BotoevaKLRWZ19} that weak
$(\mathcal{ALC},\text{CQ})$, $(\mathcal{ALC},\mathcal{ELI})$ and
$(\mathcal{ALC},\mathcal{EL})$-separability are undecidable. Note that
in this case the use of helper symbols does not make a difference.
Even if one considers ontologies given in Horn DLs, weak separability
is computationally surprisingly hard: weak $(\mathcal{EL},\mathcal{EL})$-separability is \ExpTime-complete in both combined and data complexity and weak $(\mathcal{ELI},\mathcal{ELI})$-separability is undecidable, even on labeled KBs with only two positive examples~\cite{DBLP:conf/ijcai/FunkJLPW19}. Further results for weak separability on KBs given in Horn-DLs are presented in~\cite{GuJuSa-IJCAI18,DBLP:conf/gcai/Ortiz19,aaaithis}. It is an exciting and challenging problem to find suitable subsets of the set of $\mathcal{EL}$ or $\mathcal{ELI}$ concepts for which separability becomes decidable. 
%











\bibliographystyle{elsarticle-num}

\bibliography{local}

\cleardoublepage
\appendix

\section{Proofs for Section~\ref{section:ALCI}}
\label{app:hardness}

To prove Theorem~\ref{thm:correctedresults2}, it remains to show
\NExpTime-hardness in data complexity of projective $(\ALCI,\ALCI)$-separability. This
follows from the following result.
\begin{theorem}
  \label{thm:cqhardness}
	There exists an $\ALCI$-ontology $\Omc$ such that unary rooted UCQ-evaluation on KBs with ontology $\Omc$ is \coNExpTime-hard. 
\end{theorem}
To prove this result, we adapt a \coNExpTime-hardness proof from
\cite{LutzDL07,Lutz-IJCAR08}. It works by reducing a tiling problem that asks
to tile a $2^n \times 2^n$-torus.

\medskip

A \emph{tiling system} $\Tmf$ is a triple $(T,H,V)$, where $T =
\{0,1,\dots,k-1\}$, $k \geq 0$, is a finite set of \emph{tile types}
and $H,V \subseteq T \times T$ represent the \emph{horizontal and
	vertical matching conditions}.  Let \Tmf be a tiling system and $c
= c_0\cdots c_{n-1}$ an \emph{initial condition}, i.e.\ an $n$-tuple
of tile types.  A mapping $\tau: \{0,\dots,2^{n}-1\} \times
\{0,\dots,2^{n}-1\} \to T$ is a \emph{solution} for \Tmf and $c$
if for all $x,y < 2^{n}$, the following holds where $\oplus_i$
denotes addition modulo~$i$:
\begin{enumerate}
	\item 
	if $\tau(x,y) = t$ and $\tau(x \oplus_{2^{n}} 1,y) =
	t'$, then $(t,t') \in H$;
	\item 
	if $\tau(x,y) = t$ and $\tau(x,y \oplus_{2^{n}} 1) =
	t'$, then $(t,t') \in V$;
	\item 
	$\tau(i,0) = c_i$ for $i < n$.
\end{enumerate}
It is well-known that there is a tiling system \Tmf such that it is
\NExpTime-hard to decide, given an initial condition $c$, whether
there is a solution for \Tmf and $c$. In fact, this can be easily
proved using the methods in \cite{EmdeBoas97}. For what follows, fix
such a system \Tmf. To build up intuition for the reduction, we first
describe a representation of solutions for \Tmf and $c$ in terms of
certain \emph{torus trees}, shown in Figure~\ref{fig:mod}.  Let
$\tau$ be such a solution. Then the corresponding torus tree has the following structure:
\begin{figure}[t!]
	\begin{center}
		\framebox[1\columnwidth]{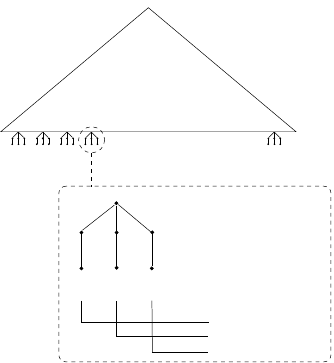}
		\caption{The structure encoding the $2^{n} \times 2^{n}$-grid.}
		\label{fig:mod}
	\end{center}
	\vspace*{-5mm}
\end{figure}
\begin{itemize}

\item The main part of the tree, shown as a large triangle in
  Figure~\ref{fig:mod}, consists of a full binary tree of depth $2n$.
  Thus, there is exactly one leaf in the main part for every
  position in the $2^n \times 2^n$-torus.

\item Every leaf in the main part has an attached gadget, which
  is itself a tree of depth two, so that the overall depth of the
  torus tree is $2n+2$;

\item The root of the gadget has three successors that we
  call $H$-nodes; each $H$-node has a single successor;
  we call these successors of $H$-nodes $G$-nodes and,
  from left to right, $G_1$-node, $G_2$-node, and $G_3$-node.

\item Each gadget represents a position in the $2^n \times 2^n$-torus
  together with the neighboring positions above and to the right.
  More precisely, the $G_1$-node represents the position $(i,j)$
  focussed on by the gadget, the $G_2$-node represents the position
  $(i+1,j)$ to the right, and the $G_3$-node represents the position
  $(i,j+1)$ above (all modulo~$2^n$). We store at these $G$-nodes the
  tile types $\tau(i,j)$, $\tau(i+1,j)$, and $\tau(i,j+1)$,
  respectively (not shown in the figure).

\end{itemize}
Note that in a torus tree, every tile type assigned to a position of
the $2^n \times 2^n$-torus is stored in three different gadgets, once
as a $G_1$-node, once as a $G_2$-node, and once as a $G_3$-node. In
principle, a torus tree may store \emph{different} tile types at these
places, but is then of course not guaranteed to correspond to a
solution of $\Tmf$ and $c$.  Such \emph{copying defects} will play an
important role in the reduction below.

In what follows, our aim is to identify a knowledge base $\Kmc=(\Omc,\Dmc)$
with selected individual name $a_0$ such that given an initial condition
$c=c_0 \cdots c_{n-1}$, we can construct in polynomial time a unary rooted UCQ
$q_c$ such that $\Kmc \not\models q_c(a_0)$ iff there is a solution for \Tmf
and $c$. In fact, $\Kmc \not\models q_c(a_0)$ implies that there is a model
\Imc of \Kmc such that $\Imc \not\models q_c(a_0)$ and we will build \Kmc and
$q_c$ such that such a model \Imc contains a (representation of a) torus tree
without copying defects. In fact, \Dmc will be of the very simple form $\{ A_0(a_0) \}$
and we will constuct $q_c$ to be of the form $q^1_c \vee q^2_c$, where $q^1_c$
and $q^2_c$ are both rooted UCQs, such that $\Imc \not\models q^1_c(a_0)$
implies the presence of a torus tree in \Imc, potentially with copying defects,
while $\Imc \not\models q^2_c(a_0)$ guarantees that no copying defects are
actually present. To achieve the latter, we actually need to enforce the
presence of a second torus tree in \Imc, as explained below.

We first construct the ontology \Omc. Note that \Omc must be 
independent of the choice of $c$, and thus also independent of the 
size of the torus. This is not the case in the mentioned reduction in 
\cite{LutzDL07,Lutz-IJCAR08} where \Omc simply enforces that models of \Kmc 
contain a torus tree, possibly with copying defects. We clearly cannot 
do this here. What we do instead is building \Omc so that it generates 
an \emph{infinite} tree in models of \Kmc that should be thought of as 
providing a `template' for torus trees of any possible depths. In this 
template, the concept names $H, G, G_1, G_2, G_3$ that we use 
to identify the eponymous nodes are \emph{not} enforced to be true 
anywhere.  We will use $q^1_c$ to make them true at the right depths 
for the initial condition $c$ given, in this way imposing a torus tree of 
the desired depth onto the template. 
We next describe the representation of torus trees in models of
\Kmc in some more detail:
\begin{itemize}

\item We use a role name $R_0$ such that each successor in the torus
  tree, both in the main part of the tree and in the gadgets, is
  reached via a role path of the form $R_0;R_0^-;R_0;R_0^-;R_0;R_0^-$;
  no branching occurs for nodes on this path.

\item We use concept names $B_1,B_2,B_3$ and make sure that the
  nodes in the torus tree are labeled with $B_1$ while  on
  each `successor path' $R_0;R_0^-;R_0;R_0^-;R_0;R_0^-$, the
  node reached via the $R_0;R_0^-$ prefix satisfies $B_2$ and the
  node reached via the $R_0;R_0^-;R_0;R_0^-$ prefix satisfies $B_3$.

\item What this achieves is that successors can be reached both by the
  role composition $R_0;R_0^-;R_0;R_0^-;R_0;R_0^-$, which should be
  thought of as a reflexive and symmetric role, and by $B_1?;R_0;R_0^-;B_2?;R_0;R_0^-;B_3?R_0;R_0^-;B_1?$,
  which is directed rather than symmetric and where each $B_i?$
  denotes a test that the node reached at this point satisfies concept
  name $B_i$.

\item We also use a role name $R_1$ to represent the torus positions
  represented by $G$-nodes. This representation is in binary, that is,
  we encode the numbers $0,\dots,2^{2n}-1$ in binary, assuming that
  the first $n$ bits describe the horizontal position and the second
  $n$
  bits the vertical position. Bit
  $i$ has value one at a domain element $d$ if
  $d \in (\exists R_1^{i+1} . T)^\Imc$ and zero if
  $d \in (\exists R_1^{i+1} . F)^\Imc$ where bit 0 is the least
  significant bit and where $\exists R_1^{i+1}$ denotes the $i+1$-fold
  nesting $\exists R_1 \cdots \exists R_1$.
  
\end{itemize}
Let us now formally define the ontology \Omc. To represent tiles, we
introduce a concept name $D_i$ for each $i \in T$. We write
$\exists R . C$ as shorthand for
$$
B_1 \sqcap \exists R_0 . \exists R_0^- . (B_2 \sqcap
\exists R_0 . \exists R_0^- . (B_3 \sqcap \exists R_0. \exists R_0^- . (B_1 \sqcap C))).
$$
Now \Omc contains the following:
$$
\begin{array}{rcl}
  A_0 &\sqsubseteq& A \\[1mm]
  A &\sqsubseteq& \exists R . (A \sqcap T) \sqcap \exists R . (A
\sqcap F) \\[1mm]
A &\sqsubseteq& \displaystyle\bigsqcap_{1 \leq i \leq 3} \exists R . (H' \sqcap
                M_1 \sqcap \exists R .(G' \sqcap G'_i \sqcap M_1))
\\[5mm]
\neg X_H \sqcap H' &\sqsubseteq& H\\[1mm]
\neg X_G \sqcap G' &\sqsubseteq& G\\[1mm]
\neg X_G \sqcap G'_i &\sqsubseteq& G_i \text{ for } 1 \leq i \leq
3\\[1mm]
G & \sqsubseteq& \displaystyle \bigsqcup_{i \in T} (D_i \sqcap 
\bigsqcap_{j \in T \setminus \{ i \}} \neg D_j) \\[1mm]
H & \sqsubseteq& \displaystyle \bigsqcup_{i \in T} (\neg D_i \sqcap 
\bigsqcap_{j \in T \setminus \{ i \}} D_j) \\[1mm]
T & \equiv& \neg F \\[1mm]
\top &\sqsubseteq&\exists R_1 . \top
\end{array}
$$
Intuitively, the first CI generates the template for the main part of the torus
tree and the second CI generates the template for the gadgets. Note that we
attach a gadget template to \emph{every} node of the main tree, independently
of its depth. As explained before, though, these gadgets are not activated:
they only satisfy the concept names $H', G', G'_1, G'_2, G'_3$ while what we
are really interested in are $H, G, G_1, G_2, G_3$. The next three CIs provide
a way to activate these concept names by making the concept names $X_H$ and
$X_G$ false. The next two lines select tile types whenever $G$ and $H$ have
been activated. In the case of $G$, exactly one tile type is made true while in
the case of $H$, exactly one tile type is made false. This is explained later
on, we speak of $G$-nodes and $H$-nodes being labeled complementarily regarding
tile types. \Omc also makes sure that the concept names $T$ and $F$, which
distinguish left successors and right successors in the template for the main
part of the torus tree, are complements of each other. Moreover, we generate
infinite $R_1$-paths to provide templates for representing torus
positions. Again, these are not yet `activated' as the concept names $T$ and
$F$ are not set anywhere on these paths. Note that all $H$-nodes and $G$-nodes
are labeled with $M_1$.

We have already mentioned above that avoiding copying defects requires us to
use a second torus tree.\footnote{The second tree is actually not used in the
  reduction in~\cite{LutzDL07,Lutz-IJCAR08}. That reduction contains a
  (somewhat minor) glitch that we fix by using the secomd tree.}  This tree is
attached via an extra $R$-edge from the root of the first tree. In contrast to
the first tree, there is no branching in the gadgets, that is, every leaf of
the torus tree proper has only a single $H$-node successor, which has a single
$G$-node successor, and the concept names $G_1,G_2,G_3$ are not used.  All $H$-
and $G$-nodes in the second tree are labeled with $M_2$. We include the
following in \Omc:
$$
\begin{array}{rcl}
A_0 &\sqsubseteq& \exists R . A' \\[1mm]
  A' &\sqsubseteq& \exists R . (A' \sqcap T) \sqcap \exists R . (A'
\sqcap F) \\[1mm]
A' &\sqsubseteq& \displaystyle \exists R . (H' \sqcap M_2 \sqcap 
                \exists R .(G' \sqcap M_2))
\end{array}
$$

Let $c=c_0 \cdots c_{n-1}$ be a given initial condition for \Tmf. We construct
the rooted UCQs $q^1_c$ and $q^2_c$, starting with $q^1_c$.  Recall that we
want to achieve that for models \Imc of \Kmc, $\Imc \not\models q^1_c(a_0)$
implies that \Imc contains (a homomorphic image of) a representation of a torus
tree, possibly with copying defects, as well as a second torus tree as
described above, and that \Omc provides an infinite template for both torus
trees in \Imc. The aim is thus to construct $q^1_c$ so that the non-existence
of a homomorphism from $q^1_c$ to \Imc `activates' in the template the two
torus trees of appropriate depth. We write $R(x,y)$ as shorthand for
$$
   B_1(x) , R_0(x,y_1) , R_0(z_1,y_1) , B_2(z_1) 
, R_0(z_1,y_2) , R_0(z_2,y_2) , B_3(z_2) 
, R_0(z_2,y_3) , R_0(y,y_3) , B_1(y),
$$
with $y_1,y_2,y_3,z_1,z_2$ fresh variables, expressing that $y$ is a
successor of $x$ in the torus tree.
%
We further use $R^i(x,y)$, for $i \geq 1$,
as a shorthand for $R(z_1,z_2),\dots,R(z_{i-1},z_i)$ where $z_1=x$,
$z_i=y$, and $z_2,\dots,z_{i-1}$ are fresh variables.  
In what follows, $x_0$ will be the (only) answer variable in each of
the constructed CQs.

Our first aim is to activate $H$-nodes on level $2n+1$ of the template as well
as $G$-nodes on level $2n+2$, both in the first ree, using the two CQs
$$
R^{2n+1}(x_0,x_1),M_1(x_1),X_H(x_1)  \quad \text{ and } \quad
R^{2n+2}(x_0,x_1),M_1(x_1),X_G(x_1).  
$$
We can do the same for the second tree using 
$$
R^{2n+2}(x_0,x_1),M_2(x_1),X_H(x_1)  \quad \text{ and } \quad
R^{2n+3}(x_0,x_1),M_2(x_1),X_G(x_1).  
$$

We next achieve the correct representation of torus positions at $G_1$-nodes in
the first tree, and at $G$-nodes in the second tree. We only treat the first
tree explicitly and leave to the reader the easy task to adapt the given CQs to
the second tree. By construction of the template, this position is already
represented by the $T$ and $F$ concept names used on the path of the tree that
leads to the $G_1$-node in question. We still need to `push this representation
down' to $R_1$-paths to achieve the representation described above.  For
$1 \leq i,j \leq 2n$ with $i+j=2n$, include the CQs
$$
R^i(x_0,x_1),T(x_1),R^{j+2}(x_1,x_2),G_1(x_2),M_1(x_2)R_1^{i}(x_2,x_3),F(x_3) 
$$
and
$$
R^i(x_0,x_1),F(x_1),R^{j+2}(x_1,x_2),G_1(x_2),M_1(x_2),R_1^{i}(x_2,x_3),T(x_3) .
$$
Note that this relies on the concept names $T$ and $F$ to be
complements of each other. For example, take a homomorphism $h$
from the upper CQ without the last atom $F(x_3)$ into a model \Imc of
\Kmc with $\Imc \not \models q^1_c$. Then $h(x_3)$ cannot make $F$
true. But since $T$ and $F$ are complementary, it must make $T$
true. 

To make the counter value unique, we should also ensure that
$\exists r^i.T$ and $\exists r^i.F$ are not both true, for any $i$
with $1 \leq i \leq 2n$. We need this for all $G$-nodes, not
 only for $G_1$-nodes.
From now on, we use $(x \text{ bit } i = V)$,
$0 \leq i < 2n$ and $V \in \{T,F\}$, to abbreviate
$R_1^{i+1}(x,y), V(y)$ for a fresh variable $y$.  For $0\leq i < 2n$,
add the CQ
%
%
$$
R^{2n+2}(x_0,x_1), G(x_1), M_1(x_1), (x_1 \text{ bit } i = T), (x_1 \text{
	bit } i = F).
$$
We further want (only for the main torus tree) that, relative to its
$G_1$-sibling, each $G_2$-node represents the horizontal neighbor position in
the grid and each $G_3$-node represents the vertical neighbor position. This
can be achieved by a couple of additional CQs that are slightly tedious. We
only give CQs which express that if a $G_1$-node represents $(i,j)$, then its
$G_2$-sibling represents $(i \oplus_{2^n} 1,j')$ for some
$j'$. 
 For $0 \leq i < n$, add the CQs
$$
R^{2n}(x_0,x_1),R^2(x_1,y_1), G_1(y_1), \!\! \bigwedge_{0 \leq j \leq i} \!\!
(y_1 \text{ bit } j = T),
R^2(x_1,y_2), G_2(y_2),  (y_2 \text{ bit } i = F) 
$$
and
$$
R^{2n}(x_0,x_1),R^2(x_1,y_1), G_1(y_1), \!\! \bigwedge_{0 \leq j < i} \!\!
(y_1 \text{ bit } j = T),
(y_1 \text{ bit } i = F),
R^2(x_1,y_2), G_2(y_2),  (y_2 \text{ bit } i = T) 
$$
and for $0 \leq j < i < n$, add
$$
R^{2n}(x_0,x_1),R^2(x_1,y_1), G_1(y_1), 
(y_1 \text{ bit } j = F), (y_1 \text{ bit } i = F), 
R^2(x_1,y_2), G_2(y_2),  (y_2 \text{ bit } i = T) 
$$
and
$$
R^{2n}(x_0,x_1),R^2(x_1,y_1), G_1(y_1), 
(y_1 \text{ bit } j = F), (y_1 \text{ bit } i = T),
R^2(x_1,y_2), G_2(y_2),  (y_2 \text{ bit } i = F).
$$
It is not difficult to construct similar CQs which express that if a
$G_1$-node represents $(i,j)$, then its $G_2$-sibling represents
$(i',j)$ for some $i'$ and if a $G_1$-node represents $(i,j)$, then
its $G_3$-sibling represents $(i,j \oplus_{2^n} 1)$.

Due to Line~6 of \Omc, every $G$-node is labeled with $D_i$ for
a unique tile type $i \in \Tmf$. 
The initial condition is now easily guaranteed. For $0 \leq i < n$,
and each $j \in \Tmf \setminus \{ c_i\}$, add the CQ
$$
R^{2n+2}(x_0,x_1),G(x_1),M_1(x_1),D_j(x_1),
(y_1 \text{ bit } 0 = V_0), \dots,   (y_1 \text{ bit } n-1 = V_{n-1})
$$
where $V_i$ is $T$ if the $i$-th bit in the binary representation of $i$ is 1
and $F$ otherwise.  We next enforce (in the first torus tree) that the
horizontal and vertical matching conditions are satisfied locally in each
gadget. For all $(i,j) \notin H$, put
%
$$
R^{2n}(x_0,x_1),R^2(x_1,y_1),G_1(y_1),D_i(y_1), 
R^2(x_1,y_2),G_2(y_2),D_j(y_2)
$$
and likewise for all $(i,j)\notin V$ and $G_3$ in place of $G_2$. To prepare
for the construction of the UCQ $q^2_c$, we also want to achieve (in both torus
trees) that $G$-nodes and their $H$-node predecessors are complementary
regarding the bits of the counter: if $d$ is an $H$-node and $e$ its $G$-node
successor, then $d$ satisfies $\exists s^i.T$ iff $e$ satisfies
$\exists s^i . F$, for $1 \leq i \leq 2n$.  We can enforce this using for
$0 \leq i < 2n$ the CQs
$$
R^{m+1}(x_0,x_1),H(x_1), (x_1 \text{ bit } i = T), 
R(x_1,x_2),G(x_2), M_1(x_2), (x_2 \text{ bit } i = T)
$$
and
$$
R^{m+1}(x_0,x_1),H(x_1), M_1(x_1), (x_1 \text{ bit } i = F), 
R(x_1,x_2),G(x_2), M_1(x_2) (x_2 \text{ bit } i = F).
$$
%
To ensure the uniqueness of bit values also at $H$-nodes, we further
add, for $0\leq i < 2n$, the CQ
$$
R^{2n+1}(x_0,x_1), H(x_1), (x_1 \text{ bit } i = T), (x_1 \text{
	bit } i = F).
$$
This finishes the construction of our first UCQ $q^1_c$.

As explained above, the purpose of the second UCQ $q^2_c$ is to make sure that
the (first and main) torus tree in models \Imc of \Kmc with
$\Imc \not \models q_c$ has no copying defects. To achieve this, it suffices
to guarantee the following:
\begin{enumerate}
	
\item[($*$)] if a $G$-node in the first tree represents the same grid position
  as a $G$-node in the second tree, then their tile types coincide.
	
\end{enumerate}
In ($*$), a $G$-node in the first tree can be any of a $G_1$-, $G_2$-, or
$G_3$-node.  The UCQ $q^2_c$ is less straightforward to construct than those in
$q^1_c$ and the ideas that we rely on were first used in
\cite{LutzDL07,Lutz-IJCAR08}. Actually, using the constructions from
\cite{LutzDL07,Lutz-IJCAR08}, we could achieve that $q^2_c$ is a CQ rather than
a UCQs. But since this makes the constructions more complex and the overall
query $q_c$ will be a UCQ anyway, we confine ourselves to $q^2_c$ being a
UCQ. Enforcing ($*$) via $q^2_c$ relies on the second torus tree and on the
complementary labeling of $H$- and $G$-nodes.

We want to construct $q^2_c$ so that it has a homomorphism into the torus trees
iff ($*$) is violated, that is, iff there is a $G$-node in the first torus tree
and one in the second tree that agree on the grid position but are labeled with
different tile types. The UCQ $q^2_c$ contains one CQ $q_{j,\ell}$ for each
choice of distinct tile types $j,\ell \in T$.  We construct $q_{j,\ell}$ from
component queries $q_0,\dots,q_{2n-1}$, which all take the form of the query
displayed on the left-hand side of Figure~\ref{fig:q1}. All edges shown there
represent subqueries of the form $S(x,y)$, which is an abbreviation for
$$
  R_0(x,y_1) , R_0(z_1,y_1) , 
 R_0(z_1,y_2) , R_0(z_2,y_2) , 
 R_0(z_2,y_3) , R_0(y,y_3)
$$
with $y_1,y_2,y_3,z_1,z_2$ fresh variables. Note that this differs
from the edges $R(x,y)$ used in $q^1_c$ in that the
concept names $B_1,B_2,B_3$ are not present, that is, here we use the
edges in the torus tree as \emph{symmetric} edges. Moreover, by
construction every node in the torus tree has a reflexive $S$-loop,
with $S$ the above role sequence. With $q_i^V(y)$, $V \in \{T,F\}$, we
abbreviate the CQ $(y \text{ bit } i = V)$ whose edges are not shown.
The only difference between the component queries $q_0,\dots,q_{2n-1}$
is that in query $q_i$, we use subqueries $q_i^T$ and $q_i^F$.
\begin{figure}[t!]
	\begin{center}

		\tikzset{every picture/.style={line width=0.75pt}} 
		
		\begin{tikzpicture}[x=0.75pt,y=0.75pt,yscale=-0.75,xscale=1]
			\draw   (90,5.6) -- (710.33,5.6) -- (710.33,524.27) -- (90,524.27) -- cycle ;
			\draw    (213,74) -- (250,98.25) ;
			\draw    (201.35,74.22) -- (165.25,98) ;
			\draw    (160,115) -- (160.16,155.14) ;
			\draw    (260,115) -- (260.16,155.14) ;
			\draw    (160,219) -- (160.16,259.14) ;
			\draw    (260,219) -- (260.16,259.14) ;
			\draw    (160.08,279) -- (160.24,319.14) ;
			\draw    (260,279) -- (260.16,319.14) ;
			\draw    (160,381) -- (160.16,421.14) ;
			\draw    (260,381) -- (260.16,421.14) ;
			\draw    (183,267.74) -- (197.16,267.83) ;
			\draw    (223.41,267.7) -- (237.57,267.79) ;
			\draw    (184,429.96) -- (239.94,430.23) ;
			\draw    (199.35,462) -- (163.25,438.22) ;
			\draw    (221.25,462.33) -- (257.35,438.56) ;
			\draw    (399.33,115.33) -- (399.49,155.47) ;
			\draw    (399.33,219.33) -- (399.49,259.47) ;
			\draw    (399.41,279.33) -- (399.57,319.47) ;
			\draw    (399.33,381.33) -- (399.49,421.47) ;
			\draw    (399.33,438.53) -- (399.44,454.2) ;
			\draw    (620.2,81.45) -- (620.33,93.47) ;
			\draw    (620,279.34) -- (620.16,319.48) ;
			\draw    (620.33,382.33) -- (620.49,422.47) ;
			\draw    (619.33,218.33) -- (619.49,258.47) ;
			\draw    (620.33,114.33) -- (620.49,154.47) ;
			
			\draw (171.4,28.2) node [anchor=north west][inner sep=0.75pt]   [align=left] {$\displaystyle G,\ M_{1} ,\ T_{j}$};
			\draw (363.67,58) node [anchor=north west][inner sep=0.75pt]   [align=left] {$\displaystyle G,\ M_{1} ,\ T_{j}$};
			\draw (587.2,20.27) node [anchor=north west][inner sep=0.75pt]   [align=left] {$\displaystyle G,\ M_{1} ,\ T_{j}$};
			\draw (202,60) node [anchor=north west][inner sep=0.75pt]   [align=left] {$\displaystyle x$};
			\draw (152.5,98) node [anchor=north west][inner sep=0.75pt]   [align=left] {$\displaystyle y_{0}$};
			\draw (252,98) node [anchor=north west][inner sep=0.75pt]   [align=left] {$\displaystyle z_{0}$};
			\draw (152,158) node [anchor=north west][inner sep=0.75pt]   [align=left] {$\displaystyle y_{1}$};
			\draw (252,158) node [anchor=north west][inner sep=0.75pt]   [align=left] {$\displaystyle z_{1}$};
			\draw (141,200) node [anchor=north west][inner sep=0.75pt]   [align=left] {$\displaystyle y_{2n+1}$};
			\draw (242,200) node [anchor=north west][inner sep=0.75pt]   [align=left] {$\displaystyle z_{2n+1}$};
			\draw (141,260) node [anchor=north west][inner sep=0.75pt]   [align=left] {$\displaystyle y_{2n+2}$};
			\draw (242,260) node [anchor=north west][inner sep=0.75pt]   [align=left] {$\displaystyle z_{2n+2}$};
			\draw (202,260) node [anchor=north west][inner sep=0.75pt]   [align=left] {$\displaystyle x_{0}$};
			\draw (141,320) node [anchor=north west][inner sep=0.75pt]   [align=left] {$\displaystyle y_{2n+3}$};
			\draw (242,320) node [anchor=north west][inner sep=0.75pt]   [align=left] {$\displaystyle z_{2n+3}$};
			\draw (141,362) node [anchor=north west][inner sep=0.75pt]   [align=left] {$\displaystyle y_{4n+3}$};
			\draw (242,362) node [anchor=north west][inner sep=0.75pt]   [align=left] {$\displaystyle z_{4n+3}$};
			\draw (141,422) node [anchor=north west][inner sep=0.75pt]   [align=left] {$\displaystyle y_{4n+4}$};
			\draw (242,422) node [anchor=north west][inner sep=0.75pt]   [align=left] {$\displaystyle z_{4n+4}$};
			\draw (205,458) node [anchor=north west][inner sep=0.75pt]   [align=left] {$\displaystyle z$};
			\draw (123.5,399) node [anchor=north west][inner sep=0.75pt]   [align=left] {$\displaystyle q_{i}^{F}$};
			\draw (279.5,399) node [anchor=north west][inner sep=0.75pt]   [align=left] {$\displaystyle q_{i}^{T}$};
			\draw (175.12,485.73) node [anchor=north west][inner sep=0.75pt]   [align=left] {$\displaystyle G,\ M_{2} ,\ T_{\ell }$};
			\draw (162,176) node [anchor=north west][inner sep=0.75pt]  [rotate=-90] [align=left] {$\displaystyle \dotsc $};
			\draw (261.5,176) node [anchor=north west][inner sep=0.75pt]  [rotate=-90] [align=left] {$\displaystyle \dotsc $};
			\draw (162,337) node [anchor=north west][inner sep=0.75pt]  [rotate=-90] [align=left] {$\displaystyle \dotsc $};
			\draw (263,337) node [anchor=north west][inner sep=0.75pt]  [rotate=-90] [align=left] {$\displaystyle \dotsc $};
			\draw (391.83,98) node [anchor=north west][inner sep=0.75pt]   [align=left] {$\displaystyle y_{0} =x$};
			\draw (391.33,158) node [anchor=north west][inner sep=0.75pt]   [align=left] {$\displaystyle y_{1} =z_{0}$};
			\draw (380.33,200) node [anchor=north west][inner sep=0.75pt]   [align=left] {$\displaystyle y_{2n+1} =z_{2n}$};
			\draw (380.33,260) node [anchor=north west][inner sep=0.75pt]   [align=left] {$\displaystyle x_{0} \ =y_{2n+2} =\ z_{2n+1}$};
			\draw (380.33,320) node [anchor=north west][inner sep=0.75pt]   [align=left] {$\displaystyle y_{2n+3} =z_{2n+2}$};
			\draw (380.33,362) node [anchor=north west][inner sep=0.75pt]   [align=left] {$\displaystyle y_{4n+3} =\ z_{4n+2}$};
			\draw (380.33,422) node [anchor=north west][inner sep=0.75pt]   [align=left] {$\displaystyle y_{4n+4} =z_{4n+3}$};
			\draw (362.83,399) node [anchor=north west][inner sep=0.75pt]   [align=left] {$\displaystyle q_{i}^{F}$};
			\draw (401.83,176.33) node [anchor=north west][inner sep=0.75pt]  [rotate=-90] [align=left] {$\displaystyle \dotsc $};
			\draw (401.83,337.33) node [anchor=north west][inner sep=0.75pt]  [rotate=-90] [align=left] {$\displaystyle \dotsc $};
			\draw (131.5,82.8) node [anchor=north west][inner sep=0.75pt]   [align=left] {$\displaystyle q_{i}^{T}$};
			\draw (272.5,82.8) node [anchor=north west][inner sep=0.75pt]   [align=left] {$\displaystyle q_{i}^{F}$};
			\draw (394,460) node [anchor=north west][inner sep=0.75pt]   [align=left] {$\displaystyle z\ =\ z_{4n+4}$};
			\draw (364.52,487.73) node [anchor=north west][inner sep=0.75pt]   [align=left] {$\displaystyle G,\ M_{2} ,\ T_{\ell }$};
			\draw (612.67,64) node [anchor=north west][inner sep=0.75pt]   [align=left] {$\displaystyle x=z_{0}$};
			\draw (612.17,98) node [anchor=north west][inner sep=0.75pt]   [align=left] {$\displaystyle y_{0} =z_{1}$};
			\draw (601,260) node [anchor=north west][inner sep=0.75pt]   [align=left] {$\displaystyle x_{0} =y_{2n+1} =z_{2n+2}$};
			\draw (601.77,322) node [anchor=north west][inner sep=0.75pt]   [align=left] {$\displaystyle y_{2n+2} =z_{2n+3}$};
			\draw (601.4,362) node [anchor=north west][inner sep=0.75pt]   [align=left] {$\displaystyle y_{4n+3} =z_{4n+4}$};
			\draw (622.9,342) node [anchor=north west][inner sep=0.75pt]  [rotate=-90] [align=left] {$\displaystyle \dotsc $};
			\draw (601.33,422) node [anchor=north west][inner sep=0.75pt]   [align=left] {$\displaystyle y_{4n+4} =z$};
			\draw (587.5,79.33) node [anchor=north west][inner sep=0.75pt]   [align=left] {$\displaystyle q_{i}^{T}$};
			\draw (588.83,48.67) node [anchor=north west][inner sep=0.75pt]   [align=left] {$\displaystyle q_{i}^{F}$};
			\draw (575.5,345.67) node [anchor=north west][inner sep=0.75pt]   [align=left] {$\displaystyle q_{i}^{T}$};
			\draw (574.17,408.33) node [anchor=north west][inner sep=0.75pt]   [align=left] {$\displaystyle q_{i}^{F}$};
			\draw (588.85,463.67) node [anchor=north west][inner sep=0.75pt]   [align=left] {$\displaystyle G,\ M_{2} ,\ T_{\ell }$};
			\draw (600.33,200) node [anchor=north west][inner sep=0.75pt]   [align=left] {$\displaystyle y_{2n} =z_{2n+1}$};
			\draw (621.83,176) node [anchor=north west][inner sep=0.75pt]  [rotate=-90] [align=left] {$\displaystyle \dotsc $};
			\draw (363.5,144.33) node [anchor=north west][inner sep=0.75pt]   [align=left] {$\displaystyle q_{i}^{F}$};
			\draw (362.17,91) node [anchor=north west][inner sep=0.75pt]   [align=left] {$\displaystyle q_{i}^{T}$};
			\draw (360.5,439) node [anchor=north west][inner sep=0.75pt]   [align=left] {$\displaystyle q_{i}^{T}$};
			\draw (611.33,155) node [anchor=north west][inner sep=0.75pt]   [align=left] {$\displaystyle y_{1} =z_{2}$};

		\end{tikzpicture}
		
		\caption{The query $q_i$ (left) and two 
			collapsings.}
		\label{fig:q1}
	\end{center}
\end{figure}

We assemble $q_0,\dots,q_{2n-1}$ into the desired CQ $q_{j,\ell}$ by taking
variable disjoint copies of $q_0,\dots,q_{2n-1}$ and then identifying (i)~the
variable $y$ of all components, (ii)~the variable $z$ of all components, and
(iii)~the variable $x_0$ of all components, which is the answer variable.  To
see why $q^2_c$ achieves~($*$), first note that any homomorphism from a CQ
$q_{j,\ell}$ in $q^2_c$ into the torus trees must map the variable $x$ to a
leaf of the first tree and $z$ to a leaf of the second tree because of their
$M_1$- and $M_2$-label. Call these leaves $a$ and~$a'$, respectively.  Since
$y_0$ and $z_0$ are connected to $x$ in the query, both must then be mapped
either to $a$ or to its predecessor $H$-node; likewise, $y_{4n+4}$ and
$z_{4n+4}$ must be mapped either to $a'$ or to its predecessor.  Since
$G$-nodes and $H$-nodes are complementary regarding the bits of the counter, we
are actually even more constrained: exactly one of $y_0$ and $z_0$ must be
mapped to $a$, and exactly one of $y_{4n+4}$ and $z_{4n+4}$ to~$a'$.  Moreover,
if $y_0$ is mapped to $a$, then $z_{4n+4}$ must be mapped to $a'$ and if $z_0$
is mapped to $a$, then $y_{4n+4}$ must be mapped to $a'$.  This is because $q$
contains a path of length $4n+4$ between $y_0$ and $y_{4n+4}$, as well as
between $z_0$ and $z_{4n+4}$, while the shortest path between $a$ and $a'$ has
$2n+5$ edges. This shows that any homomorphism $h$ from $q_{j,\ell}$ into the
torus trees gives rise to one of the two variable identifications in each query
$q_i$ shown in Figure~\ref{fig:q1}. Note that the first case implies that
$h(x)$ and $h(z)$ are both labeled with $q_i^T$ while they are both labeled
with $q_i^F$ in the second case. In summary, $h(x)$ and $h(z)$ must thus agree
on all bit values of the counter and, moreover, $h(x)$ is labeled with $T_j$
while $h(z)$ is labeld with $T_\ell$. As intended, $h$ thus identifies a
$G$-node $h(x)$ in the first torus tree and a $G$-node $h(z)$ in the second
tree that agree on the grid position but are labeled with different tile types.

Using the arguments provided above, one can now prove the following,
which finishes the proof of Theorem~\ref{thm:cqhardness}.
\begin{lemma}
  $\Kmc \not\models q_c(a_0)$ iff there is a solution for $\Tmf$ and
  $c$.
\end{lemma}

We now discuss a bit further our conjecture that projective
separability in $(\text{GF},\text{UCQ})$ and (projective or
non-projective) separability in $(\text{GNFO},\text{UCQ})$ are
\TwoExpTime-hard in data complexity, see the remarks after
Theorem~\ref{thm:gnfo} and~\ref{thm:gfo}.  For this, it suffices to
show that there is a GF ontology $\Omc$ such that unary rooted
UCQ-evaluation on KBs with ontology $\Omc$ is \TwoExpTime-hard.
Recall that we have proved Theorem~\ref{thm:correctedresults2} by
adapting a reduction from \cite{LutzDL07,Lutz-IJCAR08} that was used
there to show that evaluating unary rooted CQs on \ALCI-KBs is
\coNExpTime-hard. The main challenge of the adaptation was to deal
with the fact that the ontology \Omc has to be fixed in our case while
it may depend on the initial condition $c$ for the tiling problem in
\cite{LutzDL07,Lutz-IJCAR08}, which also required us to replace the CQ
from the original reduction with a UCQ. In~\cite{LutzDL07,Lutz-IJCAR08}, it is also
shown that evaluating Boolean CQs on \ALCI-KBs is
\TwoExpTime-hard. The reduction is from the word problem of a fixed
exponentially space bounded alternating Turing machine (ATM) rather
than from a tiling problem.  In the same way in which we have adapted
the \coNExpTime-hard proof to use a fixed ontology, it seems very well
possible to also adapt the \TwoExpTime-hardness proof from
\cite{LutzDL07,Lutz-IJCAR08} in the same way, again at the expense of
replacing the CQ used in the original reduction with a UCQ. We only
refrain from doing so because the \TwoExpTime-hardness proof in
\cite{LutzDL07,Lutz-IJCAR08} is considerably more technical than the
\coNExpTime-hardness proof from that paper, and as we have seen above
even adapting the latter to the case of a fixed ontology has
introduced rather significant additional technicalities.

It then remains to lift this \TwoExpTime-hardness to GF. In fact,
it seems straightforward to reduce the evaluation of Boolean UCQs on
\ALCI-KBs to rooted UCQ-evaluation on GF-KBs in polynomial time, on
databases that use only a single individual $a_0$. This is the case in
all the mentioned reductions and would thus clearly yield the
conjectured result. The idea of such a reduction would be to increase
the arity of all involved relation symbols by one, introducing one
extra position in each symbol. One would then rewrite the
\ALCI-ontology \Omc into a GF-ontology $\Omc'$ so that when the
ontology generates new facts, the single individual $a_0$ used in an
input database \Dmc is `passed on' to all these facts in the extra
position. The original Boolean UCQ $q$ would be modified by
introducing a fresh answer variable $x_0$ and using it in the extra
position of each atom. Clearly, the resulting UCQ $q'$ is unary and rooted.
Moreover, $(\Omc,\Dmc) \models q$ iff $(\Omc',\Dmc) \models q'(a_0)$.

We also briefly comment on obtaining lower bounds for the data
complexity of RE-existence and entity distinguishability. Here, we
would have to improve the above proof of
Theorem~\ref{thm:correctedresults2} by replacing the UCQ used there
with a CQ. However, when transitioning from the constructions in
\cite{LutzDL07,Lutz-IJCAR08} to the constructions used here (that is,
from a non-fixed ontology to a fixed ontology), many
`responsibilities' formerly taken by the ontology had to be shifted
into the query. In fact, it is this shifting that forced us to replace
the CQ from the original reduction by a UCQ. It is rather unclear
to us how and whether all the `responsibilities' taken by the query in the
reduction presented above can be encoded into a CQ in place of a UCQ.

\bigskip

We next prove Claim~2 from the proof of Theorem~\ref{thm:char12}. Let $\Kmc=(\Omc,\Dmc)$ be an $\ALCI$-KB. Recall that a sequence $\sigma=t_{0}R_{0}\ldots R_{n}t_{n+1}$
of $\Kmc$-types $t_{0},\ldots,t_{n+1}$ and $\Sigma$-roles 
$R_{0},\ldots, R_{n}$ \emph{witnesses $\ALCI$-incompleteness of a $\Kmc$-type $t$ for $\Kmc$} 
if $t=t_{0}$, $n\geq 1$, and 
\begin{itemize}
	\item $t_i\rightsquigarrow_{R_{i}} t_{i+1}$ for $i\leq n$;
	\item there exists a model $\Amf$ of $\Kmc$ and nodes 
	$d_{n-1},d_{n}\in \text{dom}(\Amf)$ with $(d_{n-1},d_{n})\in R_{n-1}^{\Amf}$ such that $d_{n-1}$ and $d_{n}$ realize $t_{n-1}$ and $t_{n}$ in $\Amf$,
	respectively, and there does not exist $d_{n+1}$ in $\Amf$ realizing $t_{n+1}$ with $(d_{n},d_{n+1})\in R_{n}^{\Amf}$.
\end{itemize}
\begin{lemma}\label{lem:critalci}
	The following conditions are equivalent, for any $\Kmc$-type $t$:
	\begin{enumerate}
		\item $t$ is not $\ALCI$-complete for $\Kmc$;
		\item there is a sequence witnessing $\ALCI$-incompleteness of $t$ for $\Kmc$;
		\item there is a sequence of length not exceeding $2^{||\Omc||}+2$ witnessing 
		$\ALCI$-incompleteness of $t$ for $\Kmc$.
	\end{enumerate}
	It is decidable in \ExpTime whether a $\Kmc$-type $t$ is $\ALCI$-complete for $\Kmc$.
\end{lemma}
\begin{proof} \
	``1. $\Rightarrow$ 2.'' 
	Let $\Sigma=\mn{sig}(\Kmc)$. Consider the tree-shaped `maximal' model $\Amf_{t}$ of
	$\Omc$ whose root $c$ realizes $t$ such that if a node $e\in \text{dom}(\Amf_{t})$ 
	realizes any $\Kmc$-type $t_{1}$ and is of depth $k\geq 0$, then for every $\Kmc$-type $t_{2}$
	with $t_{1}\rightsquigarrow_{R} t_{2}$ for some $\Sigma$-role $R$ there exists $e'$ realizing 
	$t_{2}$ of depth $k+1$ with $(e,e')\in R^{\Amf_{t}}$. (The construction of $\Amf_{t}$ is straightforward: its domain is the set of all sequence $t_{0}R_{0}\ldots R_{n-1}t_{n}$ with $n\geq 0$ such that $t_{0}=t$, all $t_{i}$ are
	$\Kmc$-types, all $R_{i}$ are $\Sigma$-roles, and $t_{i}\rightsquigarrow_{R_{i}} t_{i+1}$ for all $i<n$. We set $t_{0}R_{0}\ldots R_{n-1}t_{n}\in A^{\Amf_{t}}$ if $A\in t_{n}$ and $(w,wRt')\in R^{\Amf_{t}}$ if $wRt' \in \text{dom}(\Amf_{t})$.)
	If $t$ is not $\ALCI$-complete for $\Kmc$, 
	then there exists a model $\Bmf_{t}$ of
	$\Omc$ realizing $t$ in its root $c$ such that $\Amf_{t},c\not\equiv_{\ALCI,\Sigma}\Bmf_{t},c$.
	But then there exists a sequence $\sigma=t_{0}R_{0}\ldots R_{n}t_{n+1}$ with $n\geq 0$, $t=t_{0}$, $t_i\rightsquigarrow_{R_{i}} t_{i+1}$ for $i\leq n$, and 
        such that a sequence $c=d_{0},\ldots,d_{n}$ of nodes in $\Bmf_{t}$ realizes
	$t_{0}R_{0}\ldots R_{n-1}t_{n}$ in $\Bmf_{t}$ and there does not exist $d_{n+1}$ realizing $t_{n+1}$ with $(d_{n},d_{n+1})\in R_{n}^{\Bmf_{t}}$. Then $\sigma$ witnesses $\ALCI$-incompleteness 
	of $t$ for $\Kmc$ if $n>0$. If $n=0$, then it follows from  $\exists R_{0}.\top\in t$ that there exists $d'$
	realizing a $\Kmc$-type $t''$ in $\Bmf_{t}$ such that $(c,d')\in R_{0}^{\Bmf_{t}}$. Then
	the sequence $tR_{0}t''R_{0}^{-}tR_{0}t_{n+1}$ is as required.
	
	\medskip
	
	``2 $\Rightarrow$ 3''. This implication can be proved by a straightforward pumping argument. If $t_{0}R_{0}\ldots R_{n}t_{n+1}$ witnesses $\mathcal{ALCI}$-incompleteness of $t$ for $\Kmc$ and $n\geq 2^{||\Omc||}+2$, then there are $0\leq i<j< n$ such that $t_{i}=t_{j}$.
	Then take the sequence $t_{0}R_{0}\ldots t_{i}R_{j}t_{j+1}\ldots R_{n}t_{n+1}$ instead. 
	
	``3 $\Rightarrow$ 1'' holds by definition.
	
	\medskip
	
	To show that it is in \ExpTime to decide whether a $\Kmc$-type $t$ is $\ALCI$-complete for $\Kmc$,
	observe that one can construct a structure $\Amf$ whose domain consists of all $\Kmc$-types $t$
	and such that $t\in A^{\Amf}$ if $A\in t$ and $(t_{1},t_{2})\in R^{\Amf}$ if $t_{1}\rightsquigarrow_{R} t_{2}$.
	Then $t$ is not $\mathcal{ALCI}$-complete for $\Kmc$ iff there exists a path starting at $t$ in $\Amf$ 
	that ends with $R_{n-1}^{\Amf}t_{n}R_{n}^{\Amf}t_{n+1}$ such that the second 
	condition for sequences witnessing $\mathcal{ALCI}$-incompleteness holds. The existence of
	such a path can be decided in exponential time.
\end{proof}

\section{Proofs for Section~\ref{section:GF}}
The aim of this Section is to prove the equivalence of Points~1 to 4 and Point~5 of Theorem~\ref{thm:critGF1}. We introduce a few important notions required for the proof. 
A \emph{guarded tree decomposition} \cite{DBLP:books/daglib/p/Gradel014,tocl2020} of a structure
$\mathfrak{A}$ is a triple $(T,E,\text{bag})$ with
$(T,E)$ an undirected tree and $\text{bag}$ a function that assigns to every
$t\in T$ a guarded set $\text{bag}(t)$ in $\mathfrak{A}$ such that
\begin{enumerate}
	\item $\mathfrak{A} = \bigcup_{t\in T}\mathfrak{A}_{|\text{bag}(t)}$;
	\item $\{t \in T\mid a\in \text{bag}(t)\}$ is connected in
	$(T,E)$, for every $a\in \text{dom}(\mathfrak{A})$.
\end{enumerate}
When convenient, we assume that $(T,E)$ has a designated root $r$
which allows us to view $(T,E)$ as a directed tree. Also, it will be useful to
allow $\text{bag}(r)$ with $r$ the designated root of $T$ not to be guarded. The difference between a
classical tree decomposition 
and a
guarded one is that in the latter the elements in each bag, except the bag of the root, must be a
guarded set. While there is a classical tree decomposition of every
structure, albeit of potentially high width (that is, maximum bag
size), this is not the case for guarded tree decompositions. We say
that $\mathfrak{A}$ is \emph{guarded tree decomposable} if there
exists a guarded tree decomposition of~$\mathfrak{A}$.
Observe that for every GF-ontology $\Omc$ and GF-formula $\varphi(\vec{x})$ such that
$\Omc\not\models\varphi$ there exists a guarded tree decomposable model
$\Amf$ of $\Omc$ such that $\Amf\models\neg\varphi(\vec{a})$ for a tuple
$\vec{a}$ with $[\vec{a}]=\text{bag}(r)$, $r$ the designated rooted of the underlying tree~\cite{DBLP:books/daglib/p/Gradel014,tocl2020}.

We next provide a characterization of openGF-complete $\Kmc$-types in the same style as in Claim~2 of the proof of Theorem~\ref{thm:char12} for $\ALCI$. To this end we need some notation for paths in a structure and a way to transform
paths into strict paths without interfering with the realized openGF-types along the path. 
A \emph{path} of length $n$ from $a$ to $b$ in a structure $\Amf$ is
a sequence $R_{1}(\vec{b}_{1}),\ldots,R_{n}(\vec{b}_{n})$ with 
\begin{itemize}
	\item $\Amf\models R_{i}(\vec{b_{i}})$ and $|[\vec{b}_{i}]|\geq 2$ for all $i\leq n$; 
	\item $a\in [\vec{b}_{1}]$, $b\in [\vec{b}_{n}]$;
	\item $[\vec{b}_{i}]\cap [\vec{b}_{i+1}]\not=\emptyset$, for all $i<n$. 
\end{itemize}
Note that there is a path from $a$ to $b$ in $\Amf$ if there is a path from
$a$ to $b$ in the Gaifman graph of $\Amf$. We call a path \emph{strict} if
all $[\vec{b}_{i}]\cap [\vec{b}_{i+1}]$ are singletons
containing distinct points $c_{i}$ and there are sets $A_{1},\ldots,A_{n}\subseteq
\text{dom}(\Amf)$ covering $\text{dom}(\Amf)$ such that $[\vec{b}_{i}]\subseteq 
A_{i}$, $A_{i}\cap A_{i+1}=\{c_{i}\}$ and such that if $i < j$, then any path in the Gaifman graph of $\Amf$ from an element of $A_{i}$ to an element of $A_{j}$ contains $c_{k}$ for all $k \in \{i, \dots, j-1\}$.
We next introduce a transformation of paths into strict paths.
The \emph{partial unfolding} $\Amf_{\vec{a}}$ of a structure $\Amf$ along a tuple
$\vec{a}=(a_{1},\ldots,a_{n})$ in $\text{dom}(\Amf)$ such that 
$\text{dist}_{\Amf}(a_{i},a_{i+1})=1$ for all $i<n$
is defined as the following union of $n+1$ copies of $\Amf$. Denote the copies
by $\Amf_{1}$, $\Amf_{2},\ldots,\Amf_{n+1}$. The copies are mutually disjoint
except that $\Amf_{i}$ and $\Amf_{i+1}$ share a copy of $a_{i}$. Formally, the
domain of $\Amf_{i}$ is $\Amf\times \{i\}$ except that $(a_{i-1},i)$ is replaced by
$(a_{i-1},i-1)$, for all $i>1$. The constants are interpreted in $\Amf_{1}$ as before and we often
denote the elements $(a,1)$ of $\Amf_{1}$ simply by $a$.
We following figure illustrates this construction for a path $R_{1}(\vec{a}_{1}),\ldots,R_{n}(\vec{a}_{n})$ with $a_{i}\in [\vec{a}_{i}]\cap [\vec{a}_{i+1}]$ for $i\leq n$.

\begin{center}
\tikzset{every picture/.style={line width=0.5pt}} 

\begin{tikzpicture}[x=0.75pt,y=0.75pt,yscale=-1,xscale=1]

\draw   (277.5,71.6) -- (319.5,71.6) -- (319.5,115.64) -- (277.5,115.64) -- cycle ;
\draw  [fill={rgb, 255:red, 0; green, 0; blue, 0 }  ,fill opacity=1 ] (315.07,108.63) .. controls (316.6,108.62) and (317.86,109.76) .. (317.87,111.18) .. controls (317.89,112.6) and (316.66,113.76) .. (315.13,113.77) .. controls (313.6,113.78) and (312.34,112.64) .. (312.33,111.22) .. controls (312.31,109.8) and (313.54,108.64) .. (315.07,108.63) -- cycle ;
\draw   (244.38,106.8) -- (287.1,106.8) -- (287.1,145.2) -- (244.38,145.2) -- cycle ;
\draw  [fill={rgb, 255:red, 0; green, 0; blue, 0 }  ,fill opacity=1 ] (282.14,108.75) .. controls (283.67,108.74) and (284.92,109.88) .. (284.94,111.3) .. controls (284.95,112.72) and (283.72,113.88) .. (282.19,113.89) .. controls (280.66,113.91) and (279.41,112.77) .. (279.39,111.35) .. controls (279.38,109.93) and (280.61,108.76) .. (282.14,108.75) -- cycle ;
\draw   (343.1,72) -- (386.3,72) -- (386.3,115.64) -- (343.1,115.64) -- cycle ;
\draw   (310.3,106.4) -- (353.1,106.4) -- (353.1,144.8) -- (310.3,144.8) -- cycle ;
\draw    (282.16,111.32) -- (315.1,111.2) ;
\draw  [fill={rgb, 255:red, 0; green, 0; blue, 0 }  ,fill opacity=1 ] (380.95,108.38) .. controls (382.48,108.37) and (383.73,109.51) .. (383.75,110.93) .. controls (383.76,112.35) and (382.53,113.51) .. (381,113.53) .. controls (379.47,113.54) and (378.22,112.4) .. (378.2,110.98) .. controls (378.19,109.56) and (379.42,108.4) .. (380.95,108.38) -- cycle ;
\draw  [fill={rgb, 255:red, 0; green, 0; blue, 0 }  ,fill opacity=1 ] (348.01,108.51) .. controls (349.54,108.49) and (350.79,109.63) .. (350.81,111.05) .. controls (350.82,112.47) and (349.59,113.64) .. (348.06,113.65) .. controls (346.53,113.66) and (345.28,112.52) .. (345.26,111.1) .. controls (345.25,109.68) and (346.48,108.52) .. (348.01,108.51) -- cycle ;
\draw    (315.1,111.2) -- (348.04,111.08) ;
\draw    (348.04,111.08) -- (380.97,110.96) ;
\draw   (223,32) -- (423.5,32) -- (423.5,195.5) -- (223,195.5) -- cycle ;

\draw (258.26,147.88) node [anchor=north west][inner sep=0.75pt]   [align=left] {$\displaystyle \mathfrak{A}_{1}$};
\draw (390.67,108.13) node [anchor=north west][inner sep=0.75pt]  [font=\small] [align=left] {$\displaystyle \dots $};
\draw (321.11,148.18) node [anchor=north west][inner sep=0.75pt]   [align=left] {$\displaystyle \mathfrak{A}_{3}$};
\draw (354.43,56.08) node [anchor=north west][inner sep=0.75pt]   [align=left] {$\displaystyle \mathfrak{A}_{4}$};
\draw (338.8,116.94) node [anchor=north west][inner sep=0.75pt]  [font=\footnotesize] [align=left] {$\displaystyle a_{3}$};
\draw (311.4,116.94) node [anchor=north west][inner sep=0.75pt]  [font=\footnotesize] [align=left] {$\displaystyle a_{2}$};
\draw (273.7,116.54) node [anchor=north west][inner sep=0.75pt]  [font=\footnotesize] [align=left] {$\displaystyle a_{1}$};
\draw (290.46,56.08) node [anchor=north west][inner sep=0.75pt]   [align=left] {$\displaystyle \mathfrak{A}_{2}$};
\draw (373.4,116.74) node [anchor=north west][inner sep=0.75pt]  [font=\footnotesize] [align=left] {$\displaystyle a_{4}$};
\draw (323.3,92.94) node [anchor=north west][inner sep=0.75pt]  [font=\footnotesize] [align=left] {$\displaystyle R_{2}$};
\draw (292.1,97.74) node [anchor=north west][inner sep=0.75pt]  [font=\footnotesize] [align=left] {$\displaystyle R_{1}$};
\draw (357.9,97.74) node [anchor=north west][inner sep=0.75pt]  [font=\footnotesize] [align=left] {$\displaystyle R_{3}$};

\end{tikzpicture}
\end{center}

We use the following properties of $\Amf_{\vec{a}}$:
\begin{lemma}
\label{lem:partialunfolding}
\begin{enumerate}
	\item If $i < j$, then any path in $\mathfrak{A}_{\vec{a}}$ from an element of $\text{dom}(\mathfrak{A}_i)$ to an element of $\text{dom}(\mathfrak{A}_j)$ contains $(a_k,k)$ for all $k \in \{i, \dots, j-1\}$;
	\item Let $I$ contain for all $i$ with $1\leq i \leq n+1$ and all guarded $(b_{1},\ldots,b_{k})$ in $\Amf$ the mappings $p: (b_{1},\ldots,b_{k})\mapsto (c_{1},\ldots,c_{k})$, where $c_{j}=(b_{j},i)$ if $b_{j}\not=a_{i-1}$ and $c_{j}= (b_{j},i-1)$ if $b_{j}=a_{i-1}$. Then $I$ is a guarded bisimulation between $\Amf$ and $\Amf_{\vec{a}}$. 
	\item If $\Amf$ is a model of $\Kmc$, then $\Amf_{\vec{a}}$ is a model of $\Kmc$.
	\item The mapping $h$ from $\Amf_{\vec{a}}$ to $\Amf$ defined by setting
	$h(b,i)=b$ is a homomorphism from $\Amf_{\vec{a}}$ to $\Amf$.
\end{enumerate}
\end{lemma}
Assume that $R_{0}(\vec{a}_{0}),\ldots,R_{n}(\vec{a}_{n})$ is a path in $\Amf$
with $a_{i+1}\in [\vec{a}_{i}]\cap [\vec{a}_{i+1}]$ for $i\leq n$. Let $\vec{a}_{i}=(a_{i}^{1},\ldots,a_{i}^{n_{i}})$ and assume $a_{i}^{1}= a_{i+1}$.
Then $R_{0}(\vec{a}_{0},1),
\ldots,R_{n}(\vec{a}_{n},n+1)$ is a strict path in $\Amf_{\vec{a}}$ realizing the
same $\Kmc$-types as the original path, where 
\begin{eqnarray*}
(\vec{a}_{0},1) &:= &((a_{0}^{1},1),\ldots,(a_{0}^{n_{0}},1)) \\
(\vec{a}_{i},i+1) & := &((a_{i}^{1},i),(a_{i}^{2},i+1)\ldots,(a_{i}^{n_{i}},i+1))
\end{eqnarray*}

Let $\Kmc=(\Omc,\Dmc)$ be a GF-KB and $\Sigma=\mn{sig}(\Kmc)$.
We give a syntactic description of when a $\Kmc$-type $\Phi(x)$ is openGF-complete.
A \emph{guarded $\Kmc$-type} $\Phi(\vec{x})$ is a $\Kmc$-type
that contains an atom $R(\vec{x})$.
Call $\Kmc$-types $\Phi_{1}(\vec{x}_{1})$ and $\Phi_{2}(\vec{x}_{2})$ 
\emph{coherent} if there exists a model $\Amf$ of $\Kmc$ satisfying 
$\Phi_{1}\cup \Phi_{2}$ under an 
assignment $\mu$ for the variables in $[\vec{x}_{1}]\cup [\vec{x}_{2}]$. 
For a $\Kmc$-type $\Phi(\vec{x})$ and a subsequence $\vec{x}_{I}$ of $\vec{x}$
we denote by $\Phi_{|\vec{x}_{I}}$ the subset of $\Phi$ containing all formulas
in $\Phi$ with free variables from $\vec{x}_{I}$. $\Phi_{|\vec{x}_{I}}$ is called
the \emph{the restriction of $\Phi$ of $\vec{x}_{I}$}. Observe that $\Kmc$-types
$\Phi_{1}(\vec{x}_{1})$ and $\Phi_{2}(\vec{x}_{2})$ are coherent iff their restrictions to 
$[\vec{x}_{1}]\cap [\vec{x}_{2}]$ are logically equivalent. 
Assume a $\Kmc$-type $\Phi(x)$ is given. A sequence 
$$
\sigma = \Phi_{0}(\vec{x}_{0}),\ldots,\Phi_{n}(\vec{x}_{n}),\Phi_{n+1}(\vec{x}_{n+1})
$$
\emph{witnesses openGF-incompleteness of $\Phi$} if 
$\Phi$ is the restriction of $\Phi_{0}$ to $x$, $n\geq 0$,
and all $\Phi_{i}$, $0\leq i \leq n+1$, are guarded $\Kmc$-types each containing
the formula $\neg (x=y)$ for some variables $x,y$ (we say that the $\Phi_{i}$ are \emph{non-unary})
such that $[\vec{x}_{i}]\cap [\vec{x}_{i+1}]\not=\emptyset$,
all $\Phi_{i},\Phi_{i+1}$ are coherent, and
there exists a model $\Amf$ of $\Kmc$ and a tuple $\vec{a}_{n}$ in $\Amf$
such that $\Amf\models (\Phi_{n}\wedge \neg\exists\vec{x}_{n+1}'\Phi_{n+1})(\vec{a}_{n})$,
where $\vec{x}_{n+1}'$ is the sequence $\vec{x}_{n+1}$ without $[\vec{x}_{n}]\cap [\vec{x}_{n+1}]$.
\begin{lemma}\label{lem:GFincomplete}
	The following conditions are equivalent, for any $\Kmc$-type $\Phi(x)$:
	\begin{enumerate}
		\item $\Phi(x)$ is not openGF-complete;
		\item there is a sequence witnessing openGF-incompleteness of $\Phi(x)$;
		\item there is a sequence of length not exceeding $2^{2^{||\Omc||}}+2$ witnessing openGF 
		incompleteness of $\Phi(x)$.
	\end{enumerate}
	It is decidable in 2\ExpTime whether a $\Kmc$-type $\Phi(x)$ is openGF complete.
\end{lemma}
\begin{proof} \
	The proof is an extension of the proof of Lemma~\ref{lem:critalci} above. Let $\Sigma=\text{sig}(\Kmc)$.
	It is straightforward to construct a guarded tree decomposable model $\Amf$ of $\Omc$ with tree
	decompositon $(T,E,\text{bag})$ and root $r$ such that $\Phi(x)$ is realized in
	$\text{bag}(r)$ by $a$ and for every $\Kmc$-type $\Psi_{1}(\vec{x})$
	realized in some $\text{bag}(t)$ by $\vec{a}$ and every $\Kmc$-type $\Psi_{2}(\vec{y})$ coherent with 
	$\Psi_{1}(\vec{x})$ there exists a successor $t'$ of $t$ in $T$ such that
	$\Psi_{1}(\vec{x})\cup \Psi_{2}(\vec{y})$ is realized in $\text{bag}(t)\cup \text{bag}(t')$
	in $\Amf$ under an assignment $\mu$ of the variables $[\vec{x}]\cup [\vec{y}]$
	such that $\mu(\vec{x})=\vec{a}$. Thus, $\Amf$ satisfies 
	$\forall \vec{x}(\Psi_{1}\rightarrow \exists \vec{y}'\Psi_{2})$ for any coherent
	pair $\Psi_{1}(\vec{x}),\Psi_{2}(\vec{y})$, where $\vec{y}'$ is $\vec{y}$ without 
	$[\vec{x}]\cap [\vec{y}]$.
	
	``1 $\Rightarrow$ 2''. If $\Phi(x)$ is not openGF-complete, then there exists
	a guarded tree decomposable model $\Amf'$ of $\Kmc$ with root $r$
	which realizes $\Phi(x)$ in $\text{bag}(r)$ at $a'$ such that $\Amf,a
	\not\sim_{\text{openGF},\Sigma}\Amf',a'$. But then $\Amf,a$ realizes a sequence $\sigma$ that witnesses
	openGF-incompleteness of $\Phi(x)$, except that possibly there exists already a guarded non-unary
	$\Kmc$-type $\Phi_{0}(\vec{x}_{0})$ which is realized in some $\vec{a}_{0}$ in $\Amf$ with $a\in [\vec{a}_{0}]$
	but there is no $\vec{a}_{0}'$ in $\Amf'$ containing $a'$ and realizing $\Phi_{0}(\vec{x}_{0})$.
	Let $R_{0}(\vec{x}_{0})\in \Phi_{0}$. Then, because we included the formulas 
	$\exists\vec{y}_{1}'(R_{0}(\vec{x}_{n})\wedge x\not=y)$ with $n$ the arity of $R_{0}$, $x,y$ distinct variable in $\vec{x}_{n}$, and $\vec{y}_{1}$ defined as $\vec{y}_{n}$ without $x$, in $\text{cl}(\Kmc)$, there exists a
	non-unary guarded $\Kmc$-type $\Phi'(\vec{x}_{0}')$ containing $R_{0}(\vec{x}_{0}')$ such that
	there exists a tuple $\vec{a}_{0}'$ in $\Amf'$ containing $a'$ realizing $\Phi'$. We obtain a
	sequence $\sigma$ of any length by first taking $\Phi'(\vec{x}_{0})$ an arbitrary number of times 
	and then appending $\Phi_{0}$.
	
	\medskip
	``2 $\Rightarrow$ 3''. This can be proved by a straightforward pumping argument. This is particularly
	straightforward if one works with a sequence $\sigma$ realized by a strict path. Consider a sequence 
	$$
	\sigma = \Phi_{0}(\vec{x}_{0}),\ldots,\Phi_{n}(\vec{x}_{n}),\Phi_{n+1}(\vec{x}_{n+1})
	$$
	that witnesses openGF-incompleteness of $\Phi(x)$ and
	a model $\Amf$ of $\Kmc$ satisfying 
	$\Amf\models (\Phi_{n}\wedge \neg\exists\vec{x}_{n+1}'\Phi_{n+1})(\vec{a}_{n})$. We may assume (by possibly repeating $\Phi_{n}$
	once in the sequence) that there is a model $\Amf$ of $\Kmc$ with a path $R_{0}(\vec{a}_{0}),\ldots,R_{n}(\vec{a}_{n})$
	such that $\vec{a}_{i}$ realizes $\Phi_{i}$ and 
	$\Amf\models (\Phi_{n}\wedge \neg\exists\vec{x}_{n+1}'\Phi_{n+1})(\vec{a}_{n})$.  
	We now modify $\Amf$ in such a way that we obtain a sequence witnessing openGF-incompletensss
	of $\Phi(x)$ which is realized by a strict path. 
	Choose a sequence $\vec{a}=(a_{1},\ldots,a_{m})$ such that $a_{1}=a$ for the
	node $a$ in $\vec{a}_{0}$ realizing $\Phi(x)$, $a_{i}\not=a_{i+1}$ and
	$a_{i},a_{i+1}\in [\vec{a}_{j}]$ for some $j\leq n$, for all $i<m$, and $a_{m}\in \vec{a}_{n}$.
	Clearly one can find such a sequence for some $m\leq 2n$.
	Then take the partial unfolding $\Amf_{\vec{a}}$ of $\Amf$ along $\vec{a}$.
	In $\Amf_{\vec{a}}$ we find the required strict path (Lemma~\ref{lem:partialunfolding}). Pumping on this path is straightforward.
	
	``3 $\Rightarrow$ 1''. Straightforward.
	
	The 2\ExpTime upper bound for deciding whether a $\Kmc$-type is openGF-complete
	can now be proved similarly to the \ExpTime upper bound for deciding whether a 
	type defined by an $\ALCI$-KB is $\ALCI$-complete.
\end{proof}

\begin{lemma}\label{lem:openGF2}
	A $\Kmc$-type $\Phi(\vec{x})$ is openGF-complete iff all restrictions 
	$\Phi(x)$ of $\Phi$ to some variable $x$ in $\vec{x}$ are openGF-complete.
\end{lemma}
\begin{proof} \
	The direction from left to right is straightforward. Conversely, assume that $\Phi(\vec{x})$ is not openGF-complete. One can show similarly
	to the proof of Lemma~\ref{lem:GFincomplete} that (i) or (ii) holds:
	
	\medskip
	(i) there exists a guarded $\Kmc$-tuple $\Phi_{0}(\vec{x}_{0})$ sharing with $\vec{x}$ 
	the variables $\vec{x}_{I}$ for some nonempty $I\subseteq \{1,\ldots,n\}$ such that for $\vec{x}_{0}'$
	the variables in $\vec{x}_{0}$ without $\vec{x}_{I}$ the following hold:
	there exists a model $\Amf$ of $\Kmc$ realizing $\Phi$ in a tuple $\vec{a}$ such that 
		(a) $\Amf\models (\exists \vec{x}_{0}'\Phi_{0})(\vec{a}_{I})$ and there	also exists a model $\Bmf$ of $\Kmc$ realizing $\Phi$ in a tuple $\vec{a}$ such that (b)
		$\Bmf\not\models (\exists \vec{x}_{0}'\Phi_{0})(\vec{a}_{I})$.
	
	\medskip
	(ii) there exists a guarded $\Kmc$-tuple $\Phi_{0}(\vec{x}_{0})$ sharing with $\vec{x}$ 
	the variables $\vec{x}_{I}$ for some nonempty $I\subseteq \{1,\ldots,n\}$ and a sequence
	of guarded $\Kmc$-tuples $\Phi_{1}(\vec{x}_{1}),\ldots,\Phi_{n}(\vec{x}_{n}),\Phi_{n+1}(x_{n+1})$ 
	with $n\geq 1$ such that $\Phi(\vec{x})\cup \Phi_{0}(\vec{x}_{0})$ is satisfiable
	in a model of $\Kmc$ and 
	$\Phi_{0}(\vec{x}_{0}),\Phi_{1}(\vec{x}_{1}),\ldots,\Phi_{n}(\vec{x}_{n}),\Phi_{n+1}(x_{n+1})$ 
	satisfy the conditions of a sequence witnessing non openGF-completeness, except that no type $\Phi(x)$
	of which it witnesses non openGF-completeness is given.
	
	\medskip
	If (ii), then we are done by taking any variable $x$ in $x_{I}$ and the restriction $\Phi_{|x}$
	of $\Phi$ to $x$. Then $\Phi_{|x}$ is not openGF-complete. Now assume that (i) holds.
	We are again done if $I$ contains at most one element (we can simply take the type
	of $\text{tp}_{\Kmc}(\Amf,a_{I})$ then). Otherwise consider a relation $R_{0}$ with $R_{0}(\vec{x}_{0})\in \Phi_{0}$.
	By the closure condition on $\Kmc$-types, we have 
	for the model $\Bmf$ of $\Kmc$ satisfying (b) that $\Bmf\models \exists \vec{x}_{0}'R_{0}(\vec{x}_{0})(\vec{a}_{I})$. Take an extension $\vec{a}_{1}$
	of $\vec{a}_{I}$ such that $\Bmf\models R_{0}(\vec{a}_{1})$. Take any $a\in \vec{a}_{I}$, the unary
	$\Kmc$-type $\Phi(x)=\text{tp}_{\Kmc}(\Bmf,a)$, and the $\Kmc$-type $\Phi_{1}(\vec{x}_{1}):=\text{tp}_{\Kmc}(\Bmf,\vec{a}_{1})$. 
	Then the sequence $\Phi_{1},\Phi_{0}$ shows that $\Phi(x)$ is not openGF-complete.
\end{proof}
We next introduce \emph{guarded embeddings} as an intermediate step between
homomorphisms witnessing $\Dmc_{\text{con}(\vec{a}),\vec{a}}\rightarrow \Amf,\vec{b}^{\Amf}$ and the existence of models $\Bmf$ of $\Kmc$ such
that there exists a bounded guarded bisimulation between some $\Bmf,\vec{a}^{\Bmf}$ and $\Amf,\vec{b}^{\Amf}$. 
Let $\Dmc,\vec{a}$ be a pointed database, $\Amf,\vec{b}$ a pointed structure, $\ell\geq 0$, 
and $\Sigma \supseteq \mn{sig}(\Dmc)$ a signature. A \emph{partial embedding} is an injective partial homomorphism.
A pair $(e,H)$ is a \emph{guarded $\Sigma$ $\ell$-embedding between $\Dmc,\vec{a}$ and $\Amf,\vec{b}$}
if $e$ is a homomorphism from $\Dmc$ onto a database $\Dmc'$ and $H$
is a set of partial embeddings from $\Dmc'$ to $\Amf$ containing $h_{0}: e(\vec{a})\mapsto \vec{b}$ and 
a partial embedding $h$ from any guarded set in $\Dmc'$ to $\Amf$ such that the following condition hold:
\begin{itemize}
\item if $h_{i}:\vec{a}_{i}\mapsto \vec{b}_{i}\in H$ for $i=1,2$, then there exists a 
partial isomorphism $p:h_{1}([\vec{a}_{1}]\cap [\vec{a}_{2}]) \mapsto h_{2}([\vec{a}_{1}]\cap [\vec{a}_{2}])$
such that $p\circ h_{1}$ and $h_{2}$ coincide on $[\vec{a}_{1}]\cap [\vec{a}_{2}]$ and for any 
$\vec{c}$ with $[\vec{c}]= h_{1}([\vec{a}_{1}]\cap [\vec{a}_{2}])$,
$\Amf,\vec{c} \sim_{\text{openGF},\Sigma}^{\ell} \Amf,p(\vec{c})$.
\end{itemize}
We write $\Dmc,\vec{a} \preceq_{\text{openGF},\Sigma}^{\ell} \Amf,\vec{b}^{\Amf}$ 
if there exists a guarded $\Sigma$ $\ell$-embedding $H$
between $\Dmc,\vec{a}$ and $\Amf,\vec{b}$.

The following lemma shows that guarded $\Sigma$ $\ell$-embeddings determine a
sequence $H_{\ell},\ldots,H_{0}$ of partial embeddings satisfying the 
(forth) condition of guarded $\Sigma$ $\ell$-bisimulations.

\begin{lemma}\label{lem:guardemb}
Let $(\mathcal{D},\vec{a})$ be a pointed database and $(\mathfrak{A},\vec{b}^\mathfrak
A)$ be a pointed model such that $(\mathcal{D},\vec{a}) \preceq^\ell_{openGF, \Sigma} (\mathfrak{A},\vec{b}^\mathfrak{A})$. Then there there exist a surjective homomorphism $e : \mathcal{D} \rightarrow \mathcal{D}'$ for some database $\mathcal{D}'$ and sets $H_\ell, \dots, H_0$ of partial embeddings $\mathcal{D}' \rightarrow \mathfrak{A}$ such that
\begin{enumerate}
	\item for all $k\leq \ell$, all $h \in H_k$ and all guarded sets $\vec{c}$ in $\mathcal{D}'$ such that
	$[\vec{c}] \ \cap \ \text{dom}(h) \neq \emptyset$, there exists $h' \in H_{k-1}$ with domain $[\vec{c}]$ such that
	$h'$ coincides with $h$ on $[\vec{c}] \cap \text{dom}(h)$.  
	\item for all $k_{1},k_{2}\leq \ell$, all $h_1 \in H_{k_1}, h_2 \in H_{k_2}$, and all tuples $\vec{c}_{1},\vec{c}_2$ in $\mathcal{D}'$ such that $[\vec{c}_i]=\text{dom}(h_i)$, we have $h_1(\vec{c}) \sim^{\min(k_1,k_2)}_{openGF, \Sigma} h_2(\vec{c})$ for all $\vec{c}$ such that $[\vec{c}] = [\vec{c}_1] \cap [\vec{c}_2]$. 
\end{enumerate}	
\end{lemma}

\begin{proof} \indent 
%
%
Let $H$ be the set of partial embeddings witnessing 
$(\mathcal{D},\vec{a}) \preceq_{\text{openGF},\Sigma}^\ell (\mathfrak{A},\vec{b}^\mathfrak{A})$. 
Define $H_\ell := H$. We define $H_k$ for $k<\ell$ by induction.
Suppose $H_k$ has been defined. We define $H_{k-1}$. We assume that for all $h_1 \in H_k, h_2 \in H_\ell$ having intersecting domains $[\vec{c}_1], [\vec{c}_2]$, with $\vec{c}_2$ being guarded the following condition holds:
\begin{itemize}
	\item[(*)] for any tuple $\vec{c}$ in $\mathcal{D}'$ such that $[\vec{c}]=[\vec{c}_1] \cap [\vec{c}_2]$, there
	is a partial isomorphism $p : h_1(\vec{c})
	\mapsto h_2(\vec{c})$ witnessing $h_1(\vec{c})
	\sim^k_{\text{openGF}, \Sigma} h_2(\vec{c})$. The following figure illustrates our claim. 
\end{itemize}

\begin{center}

	\tikzset{every picture/.style={line width=0.5pt}} 
	
	\begin{tikzpicture}[x=0.75pt,y=0.75pt,yscale=-1,xscale=1]
	
	\draw   (208.7,68.16) .. controls (214.62,68.16) and (219.41,76.6) .. (219.41,87.02) .. controls (219.41,97.44) and (214.62,105.89) .. (208.7,105.89) .. controls (202.79,105.89) and (198,97.44) .. (198,87.02) .. controls (198,76.6) and (202.79,68.16) .. (208.7,68.16) -- cycle ;
	\draw   (208.7,91.74) .. controls (214.62,91.74) and (219.41,100.18) .. (219.41,110.6) .. controls (219.41,121.02) and (214.62,129.47) .. (208.7,129.47) .. controls (202.79,129.47) and (198,121.02) .. (198,110.6) .. controls (198,100.18) and (202.79,91.74) .. (208.7,91.74) -- cycle ;
	\draw   (290.77,54.01) .. controls (296.68,54.01) and (301.47,62.45) .. (301.47,72.87) .. controls (301.47,83.29) and (296.68,91.74) .. (290.77,91.74) .. controls (284.86,91.74) and (280.06,83.29) .. (280.06,72.87) .. controls (280.06,62.45) and (284.86,54.01) .. (290.77,54.01) -- cycle ;
	\draw  [draw opacity=0] (282.51,85.41) .. controls (284.82,79.44) and (287.74,75.88) .. (290.94,75.86) .. controls (294.03,75.85) and (296.9,79.13) .. (299.24,84.69) -- (291.24,121.2) -- cycle ; \draw   (282.51,85.41) .. controls (284.82,79.44) and (287.74,75.88) .. (290.94,75.86) .. controls (294.03,75.85) and (296.9,79.13) .. (299.24,84.69) ;
	\draw   (291.02,144.59) .. controls (285.1,144.53) and (280.41,136.04) .. (280.54,125.63) .. controls (280.67,115.21) and (285.56,106.81) .. (291.48,106.86) .. controls (297.39,106.92) and (302.08,115.41) .. (301.95,125.83) .. controls (301.82,136.25) and (296.93,144.65) .. (291.02,144.59) -- cycle ;
	\draw  [draw opacity=0] (299.52,113.61) .. controls (297.16,119.35) and (294.25,122.75) .. (291.11,122.73) .. controls (288.14,122.72) and (285.41,119.66) .. (283.18,114.45) -- (291.37,77.4) -- cycle ; \draw   (299.52,113.61) .. controls (297.16,119.35) and (294.25,122.75) .. (291.11,122.73) .. controls (288.14,122.72) and (285.41,119.66) .. (283.18,114.45) ;
	\draw    (208.7,87.02) -- (276.88,75.34) ;
	\draw [shift={(279.83,74.83)}, rotate = 530.28] [fill={rgb, 255:red, 0; green, 0; blue, 0 }  ][line width=0.08]  [draw opacity=0] (4,-2) -- (0,0) -- (4,2) -- cycle    ;
	\draw    (208.7,110.6) -- (280,126.8) ;
	\draw [shift={(279.17,127.5)}, rotate = 193.48] [fill={rgb, 255:red, 0; green, 0; blue, 0 }  ][line width=0.08]  [draw opacity=0] (4,-2) -- (0,0) -- (4,2) -- cycle    ;
	\draw    (291.48,85.45) -- (291.32,112.32) ;
	\draw [shift={(291.3,115.32)}, rotate = 270.34000000000003] [fill={rgb, 255:red, 0; green, 0; blue, 0 }  ][line width=0.08]  [draw opacity=0] (4,-2) -- (0,0) -- (4,2) -- cycle    ;
	\draw  [dash pattern={on 3pt off 3pt}] (305.78,58.15) .. controls (310.93,61.05) and (310.96,70.76) .. (305.85,79.84) .. controls (300.74,88.92) and (292.42,93.93) .. (287.27,91.03) .. controls (282.12,88.12) and (282.09,78.41) .. (287.2,69.33) .. controls (292.31,60.26) and (300.63,55.25) .. (305.78,58.15) -- cycle ;
	\draw    (295.5,127.83) .. controls (341.29,107.29) and (318.03,81.66) .. (309.12,72.24) ;
	\draw [shift={(307.17,70.17)}, rotate = 407.86] [fill={rgb, 255:red, 0; green, 0; blue, 0 }  ][line width=0.08]  [draw opacity=0] (4,-2) -- (0,0) -- (4,2) -- cycle    ;
	\draw   (166,43) -- (376,43) -- (376,159) -- (166,159) -- cycle ;
	
	\draw (179.8,75.6) node [anchor=north west][inner sep=0.75pt]  [font=\small] [align=left] {$\displaystyle \vec{c}_{1}$};
	\draw (179.2,100.8) node [anchor=north west][inner sep=0.75pt]  [font=\small] [align=left] {$\displaystyle \vec{c}_{2}$};
	\draw (220.62,65.1) node [anchor=north west][inner sep=0.75pt]  [font=\small,rotate=-349.35] [align=left] {$\displaystyle h_{1} \in H_{k}$};
	\draw (223.66,117.49) node [anchor=north west][inner sep=0.75pt]  [font=\small,rotate=-13.87] [align=left] {$\displaystyle h_{2} \in H_{\ell }$};
	\draw (294.33,90.67) node [anchor=north west][inner sep=0.75pt]  [font=\small] [align=left] {$\displaystyle \sim ^{k}$};
	\draw (278,94.67) node [anchor=north west][inner sep=0.75pt]  [font=\small] [align=left] {$\displaystyle p$};
	\draw (327.67,95.33) node [anchor=north west][inner sep=0.75pt]  [font=\small] [align=left] {$\displaystyle q_{h_{1} ,h_{2}}$};
	\draw (323.67,77.67) node [anchor=north west][inner sep=0.75pt]  [font=\small] [align=left] {$\displaystyle \sim ^{k-1}$};

	\end{tikzpicture}

\end{center} 		
Now assume that $h_{1},h_{2}$ satisfying the conditions above are given. As $[h_2(\vec{c}_2)]$ is guarded (by $\vec{c}_2$ being guarded and $h$ a partial homomorphism) and intersects The $h_2[[\vec{c}_1] \cap [\vec{c}_2]]$, and as $p$ witnesses a openGF $\Sigma$ $k$-bisimulation, there exists a partial isomorphism $q_{h_1,h_2}$ with domain $[h_2(\vec{c}_2)]$ witnessing
$h_2(\vec{c}_2) \sim^{k-1}_{\text{openGF}, \Sigma} q_{h_1,h_2}(h_2(\vec{c}_2))$ and that coincides with $p^{-1}$ on $h_2[[\vec{c}_1] \cap
[\vec{c}_2]]$. We then include $q_{h_1,h_2} \circ \circ h_2$ in $H_{k-1}$.

\noindent This is well-defined, as the assumption (*) holds for all $k \leq \ell$:

\begin{itemize}
	\item If $k=\ell$, then (*) is stated in the definition of $\Sigma$ $\ell$-guarded
	embeddings.
	\item If $0<k<\ell$ and (*) holds for $k$, let $h_1 \in H_{k-1}, h_2 \in
	H_\ell$ with intersecting domains $[\vec{c}_1], [\vec{c}_2]$ and 
	$\vec{c}_2$ guarded be given. Then $h_1 =
	q_{\eta_1,\eta_2} \circ \eta_2$ for some $\eta_1 \in H_k, \eta_2 \in H_\ell$, by definition of $H_{k-1}$. By
	definition of $\Sigma$ $\ell$-guarded embeddings, as
	$\eta_2$ and $h_2$ are both in
	$H_\ell$ and have intersecting domains $[\vec{c}_1]$ and $[\vec{c}_2]$,
	there exists a partial isomorphism $p'$ witnessing
	$\eta_2(\vec{c}) \sim^\ell_{\text{openGF}, \Sigma} h_2(\vec{c})$ for any $\vec{c}$ such that $[\vec{c}] = [\vec{c}_1] \cap [\vec{c}_2]$. Then, by composition of bisimulations, $p := p'_{|\eta_2[\vec{c}]} \circ
	(q_{\eta_1,\eta_2}^{-1})_{|h_1[\vec{c}]}$ is a partial isomorphism  witnessing $h_1(\vec{c}) \sim^{k-1}_{\text{openGF}, \Sigma} h_2(\vec{c})$ i.e. (*) holds for $k-1$. The situation is illustrated by the following figure.
\end{itemize}
\begin{center}

	\tikzset{every picture/.style={line width=0.5pt}} 
	
	\begin{tikzpicture}[x=0.75pt,y=0.75pt,yscale=-1,xscale=1]
	
	\draw   (385.7,86.16) .. controls (391.62,86.16) and (396.41,94.6) .. (396.41,105.02) .. controls (396.41,115.44) and (391.62,123.89) .. (385.7,123.89) .. controls (379.79,123.89) and (375,115.44) .. (375,105.02) .. controls (375,94.6) and (379.79,86.16) .. (385.7,86.16) -- cycle ;
	\draw   (385.7,109.74) .. controls (391.62,109.74) and (396.41,118.18) .. (396.41,128.6) .. controls (396.41,139.02) and (391.62,147.47) .. (385.7,147.47) .. controls (379.79,147.47) and (375,139.02) .. (375,128.6) .. controls (375,118.18) and (379.79,109.74) .. (385.7,109.74) -- cycle ;
	\draw   (469.77,87.34) .. controls (475.68,87.34) and (480.47,95.79) .. (480.47,106.21) .. controls (480.47,116.62) and (475.68,125.07) .. (469.77,125.07) .. controls (463.86,125.07) and (459.06,116.62) .. (459.06,106.21) .. controls (459.06,95.79) and (463.86,87.34) .. (469.77,87.34) -- cycle ;
	\draw  [draw opacity=0] (461.51,117.74) .. controls (463.82,111.77) and (466.74,108.21) .. (469.94,108.2) .. controls (473.03,108.18) and (475.9,111.46) .. (478.24,117.02) -- (470.24,153.53) -- cycle ; \draw   (461.51,117.74) .. controls (463.82,111.77) and (466.74,108.21) .. (469.94,108.2) .. controls (473.03,108.18) and (475.9,111.46) .. (478.24,117.02) ;
	\draw   (469.54,176.17) .. controls (463.63,176.11) and (458.94,167.62) .. (459.06,157.21) .. controls (459.19,146.79) and (464.09,138.39) .. (470,138.44) .. controls (475.91,138.5) and (480.6,146.99) .. (480.47,157.41) .. controls (480.34,167.83) and (475.45,176.23) .. (469.54,176.17) -- cycle ;
	\draw  [draw opacity=0] (477.85,145.28) .. controls (475.5,151.02) and (472.58,154.41) .. (469.44,154.4) .. controls (466.47,154.39) and (463.74,151.32) .. (461.51,146.12) -- (469.7,109.07) -- cycle ; \draw   (477.85,145.28) .. controls (475.5,151.02) and (472.58,154.41) .. (469.44,154.4) .. controls (466.47,154.39) and (463.74,151.32) .. (461.51,146.12) ;
	\draw    (386.04,100.35) -- (461.5,100.49) ;
	\draw [shift={(464.5,100.5)}, rotate = 180.11] [fill={rgb, 255:red, 0; green, 0; blue, 0 }  ][line width=0.08]  [draw opacity=0] (4,-2) -- (0,0) -- (4,2) -- cycle    ;
	\draw    (385.7,128.6) -- (460.36,157.1) ;
	\draw [shift={(463.17,158.17)}, rotate = 200.89] [fill={rgb, 255:red, 0; green, 0; blue, 0 }  ][line width=0.08]  [draw opacity=0] (4,-2) -- (0,0) -- (4,2) -- cycle    ;
	\draw   (556.77,87.34) .. controls (562.68,87.34) and (567.47,95.79) .. (567.47,106.21) .. controls (567.47,116.62) and (562.68,125.07) .. (556.77,125.07) .. controls (550.86,125.07) and (546.06,116.62) .. (546.06,106.21) .. controls (546.06,95.79) and (550.86,87.34) .. (556.77,87.34) -- cycle ;
	\draw  [draw opacity=0] (548.51,117.74) .. controls (550.82,111.77) and (553.74,108.21) .. (556.94,108.2) .. controls (560.03,108.18) and (562.9,111.46) .. (565.24,117.02) -- (557.24,153.53) -- cycle ; \draw   (548.51,117.74) .. controls (550.82,111.77) and (553.74,108.21) .. (556.94,108.2) .. controls (560.03,108.18) and (562.9,111.46) .. (565.24,117.02) ;
	\draw    (475.04,100.35) -- (550.5,100.49) ;
	\draw [shift={(553.5,100.5)}, rotate = 180.11] [fill={rgb, 255:red, 0; green, 0; blue, 0 }  ][line width=0.08]  [draw opacity=0] (4,-2) -- (0,0) -- (4,2) -- cycle    ;
	\draw    (386.04,100.35) .. controls (403,64.2) and (505.1,41.29) .. (552.09,98.73) ;
	\draw [shift={(553.5,100.5)}, rotate = 232.17000000000002] [fill={rgb, 255:red, 0; green, 0; blue, 0 }  ][line width=0.08]  [draw opacity=0] (4,-2) -- (0,0) -- (4,2) -- cycle    ;
	\draw    (470.37,119.35) -- (470.49,143.83) ;
	\draw [shift={(470.5,146.83)}, rotate = 269.73] [fill={rgb, 255:red, 0; green, 0; blue, 0 }  ][line width=0.08]  [draw opacity=0] (4,-2) -- (0,0) -- (4,2) -- cycle    ;
	\draw    (470.5,146.83) .. controls (533.17,164.69) and (549.36,134.63) .. (556.61,120.16) ;
	\draw [shift={(470.5,146.83)}, rotate = 15] [fill={rgb, 255:red, 0; green, 0; blue, 0 }  ][line width=0.08]  [draw opacity=0] (4,-2) -- (0,0) -- (4,2) -- cycle    ;
	\draw   (343,41.52) -- (593,41.52) -- (593,192.52) -- (343,192.52) -- cycle ;
	
	\draw (356.8,93.6) node [anchor=north west][inner sep=0.75pt]  [font=\small] [align=left] {$\displaystyle \vec{c}_{1}$};
	\draw (356.2,118.8) node [anchor=north west][inner sep=0.75pt]  [font=\small] [align=left] {$\displaystyle \vec{c}_{2}$};
	\draw (402.08,102.83) node [anchor=north west][inner sep=0.75pt]  [font=\small,rotate=0] [align=left] {$\displaystyle \eta _{2} \in H_{\ell }$};
	\draw (402.67,138.31) node [anchor=north west][inner sep=0.75pt]  [font=\small,rotate=-21.32] [align=left] {$\displaystyle h_{2} \in H_{\ell }$};
	\draw (498.58,103.71) node [anchor=north west][inner sep=0.75pt]  [font=\small,rotate=-359.52] [align=left] {$\displaystyle q_{\eta _{1} ,\eta _{2}}$};
	\draw (494.91,84.38) node [anchor=north west][inner sep=0.75pt]  [font=\small,rotate=-359.52] [align=left] {$\displaystyle \sim ^{k-1}$};
	\draw (439.58,47.05) node [anchor=north west][inner sep=0.75pt]  [font=\small,rotate=-359.52] [align=left] {$\displaystyle h_{1} \in H_{k-1}$};
	\draw (473.25,121.71) node [anchor=north west][inner sep=0.75pt]  [font=\small,rotate=-359.52] [align=left] {$\displaystyle \sim ^{\ell }$};
	\draw (454.25,123.71) node [anchor=north west][inner sep=0.75pt]  [font=\small,rotate=-359.52] [align=left] {$\displaystyle p'$};
	\draw (493.51,154.3) node [anchor=north west][inner sep=0.75pt]  [font=\small,rotate=-353.27] [align=left] {$\displaystyle p' \circ q_{\eta _{1} ,\eta _{2}}^{-1}$};
	\draw (492.45,134.45) node [anchor=north west][inner sep=0.75pt]  [font=\small,rotate=-359.52] [align=left] {$\displaystyle \sim ^{k-1}$};

	\end{tikzpicture}

\end{center}

Elements of $H_{k-1}$ are partial embeddings, as compositions of partial isomorphisms with partial embeddings.
We thus have a homomorphism $e: \mathcal{D} \rightarrow \mathcal{D}'$ and sets $H_\ell, \dots, H_0$ of partial embeddings $\mathcal{D}' \rightarrow \mathfrak{A}$. We now prove that Conditions~1 and 2 hold. 
\begin{enumerate}
	\item Let $0\leq k \leq \ell$ and $h_1 \in H_k$ with domain $[\vec{c}_1]$. Let $\vec{c}_2$ be guarded in $\mathcal{D}'$ such that $[\vec{c}_1] \cap [\vec{c}_2] \neq \emptyset$. 
	By definition of $\ell$-guarded embeddings, every
	guarded tuple is the domain of some embedding in $H = H_\ell$. In
	particular there exists $h_2 \in H_\ell$ with domain $[\vec{c}_2]$. 
	Then Condition~(*) holds, with matching notation.
	Consider $q_{h_1,h_2}$ and $p$ as defined above. A witnessing partial
	homomorphism $h'$ can be defined as $h' := q_{h_1,h_2} \circ h_2 \in H_{k-1}$.
	Since $p^{-1} \circ h_2$ coincides with $h_1$ on $[\vec{c}_1] \cap
	[\vec{c}_2]$, and $q_{h_1,h_2}$ coincides with $p^{-1}$ on $h_2[[\vec{c}_1]
	\cap [\vec{c}_2]]$, it follows that $h'$ coincides with $h_1$ on
	$[\vec{c}_1] \cap [\vec{c}_2]$.
	\item Let $h_1 \in H_{k_1}, h_2 \in H_{k_2}$ with intersecting domains $[\vec{c}_1], [\vec{c}_2]$. By definition of $\ell$-guarded embeddings, there exists $h'_2$ in $H=H_\ell$ with domain $[\vec{c}_2]$. By (*), for every $\vec{c}$ in $\mathcal{D}'$ such that $[\vec{c}] = [\vec{c}_1] \cap [\vec{c}_2]$ we have $h_1(\vec{c}) \sim^{k_1}_{\text{openGF}, \Sigma} h'_2(\vec{c})$ and $h_2(\vec{c}) \sim^{k_2}_{\text{openGF}, \Sigma} h'_2(\vec{c})$, thus $h_1(\vec{c}) \sim^{\min(k_1,k_2)}_{\text{openGF}, \Sigma} h_2(\vec{c})$ by composition of bisimulations.
\end{enumerate}
\end{proof}

Observe that if $H_\ell, \dots, H_0$ satisfying the conditions of Lemma~\ref{lem:guardemb} exist, then $H_\ell \subseteq \dots \subseteq H_0$: let $k \leq \ell$ and $\vec{c} \mapsto \vec{d} \in H_k$. By condition (1), since $[\vec{c}] \cap [\vec{c}] \neq \emptyset$ there exists $\vec{c} \mapsto \vec{d}' \in H_{k-1}$ that coincides with $\vec{c} \mapsto \vec{d}$ on $[\vec{c}]$, i.e. $\vec{c} \mapsto \vec{d} \in H_{k-1}$.

We are now in a position to show the first main lemma linking characterizations of separability using bounded guarded bisimulations, bounded guarded  embeddings, and bounded connected guarded bisimulations. 
\begin{lemma}\label{thm:gfcrit1}
Let $(\Kmc,P,\{\vec{b}\})$ be a labeled GF-KB and $\Sigma= \text{sig}(\Kmc)$.
Then the following conditions are equivalent:
\begin{enumerate}
	\item $(\Kmc,P,\{\vec{b}\})$ is openGF-separable.
	\item $(\Kmc,P,\{\vec{b}\})$ is GF-separable.
	\item there exists a (finite) model $\Amf$ of $\Kmc$ and $\ell \geq 0$ such that for all 
	models 
	$\Bmf$ of $\Kmc$ and $\vec{a}\in P$: $\Bmf,\vec{a}^{\Bmf}\not\sim_{\text{GF},\Sigma}^{\ell}\Amf,\vec{b}^{\Amf}$.
	\item there exists a (finite) model $\Amf$ of $\Kmc$ and $\ell \geq 0$ such that for 
	all $\vec{a}\in P$: $\Dmc_{\text{con}(\vec{a})},\vec{a}\not\preceq_{\text{openGF},\Sigma}^{\ell}\Amf,\vec{b}^{\Amf}$.
	\item there exists a (finite) model $\Amf$ of $\Kmc$ and $\ell \geq 0$ such that for all models 
	$\Bmf$ of $\Kmc$ and all $\vec{a}\in P$: $\Bmf,\vec{a}^{\Bmf}\not\sim_{\text{openGF},\Sigma}^{\ell}\Amf,\vec{b}^{\Amf}$.
\end{enumerate}
\end{lemma}
\begin{proof} \
The implications ``1. $\Rightarrow$ 2.'', ``2. $\Rightarrow$ 3.'', and ``5. $\Rightarrow$ 1.'' are straightforward. We prove
``3. $\Rightarrow$ 4.'' and ``4. $\Rightarrow$ 5.''. 

``3. $\Rightarrow$ 4.'' Take a model $\Amf$ of $\Kmc$ and $\ell\geq 0$ witnessing  Condition~3. We may assume that $\ell$ exceeds the maximum guarded quantifier rank of formulas in $\Kmc$.
We show that Condition~4 holds for $\Amf$ and $\ell$.
Assume for a proof by contradiction that there exists 
$\vec{a}_{0}\in P$ such that there exists
a guarded $\Sigma$ $\ell$-embedding $(e,H)$ from $\Dmc_{\text{con}(\vec{a}_{0})},\vec{a}_{0}$ to 
$\Amf,\vec{b}^{\Amf}$. 
Assume $e:\Dmc_{\text{con}(\vec{a}_{0})}\mapsto \Dmc'$ and that $e(\vec{a}_{0})=\vec{a}_{0}'$.
We construct a model $\Bmf$ as follows: first take a copy $\Bmf'$ of $\Amf$. 
For the constants $c\in \text{cons}(\Dmc)\setminus
\text{cons}(\Dmc_{\text{cons}(\vec{a_{0}})})$, we define $c^{\Bmf'}$ as the copy of $c^{\Amf}$
in $\Bmf'$. The interpretation of the constants in $\Dmc_{\text{cons}(\vec{a}_{0})}$ will be defined later. 
We define $\Bmf$ as the disjoint union of $\Bmf'$ and $\Bmf''$, where
$\Bmf''$ is defined next. 
We denote by $H'$ the set obtained from $H$ with $\vec{a}_{0}'\mapsto \vec{b}^{\Amf}$ removed if 
$\vec{a}_{0}$ is not guarded. Now let 
$$
\text{dom}(\Bmf'')= (H'\times \text{dom}(\Amf))/_{\sim},
$$
where $\sim$ identifies all $(h,d), (h',d')$ such that $(h,d)=(h',d')$ or there exists 
$c\in \text{dom}(h)\cap \text{dom}(h')$ such that $h(c)=d$ and $h'(c)=d'$. 
Denote the equivalence class of 
$(h,d)$ w.r.t.~$\sim$ by $[h,d]$. For any constant $c$ in $\Dmc_{\text{con}(\vec{a}_{0})}$, 
we set $c^{\Bmf''}= [h,h(e(c))]$, where $h\in H'$ is such that $e(c)\in \text{dom}(h)$. 
Observe that this is well defined as $(h',h'(e(c)))\sim (h,h(e(c)))$ for any $h'\in H'$ with 
$e(c)\in \text{dom}(h')$. We define the interpretation $R^{\Bmf''}$
of the relation symbol $R$ by setting for $e_1, \dots, e_n \in \mathrm{dom}(\mathfrak{B}'')$, 
$\mathfrak{B}'' \vDash R(e_1, \dots, e_n)$ if there exists $h \in H'$ and $c_1, \dots, c_n \in \mathrm{dom}(\mathfrak{A})$ 
such that $e_i = [h,c_i]$ and $\mathfrak{A} \vDash R(c_1, \dots c_n)$. Then, the map 
\begin{align*}
f_h : \ \   \mathrm{dom}(\mathfrak{A})&\rightarrow (H' \times \mathrm{dom}(\mathfrak{A}))/_{\sim} \\
c  &\mapsto [h,c]
\end{align*}
is an embedding from $\mathfrak{A}$ to $\mathfrak{B}''$, by definition.

We show that $\Bmf,\vec{a}_{0}^{\Bmf} \sim_{\text{GF},\Sigma}^{\ell} \Amf,\vec{b}^{\Amf}$. By construction and the assumption that
$\ell$ exceeds the guarded quantifier rank of $\Kmc$ it also follows that 
$\Bmf$ is a model of $\Kmc$. It thus follows that we have derived a 
contradiction to the assumption that $\Amf$ and $\ell$ witness Condition~3.

To define a guarded $\Sigma$ $\ell$-bisimulation $\hat{H}_{\ell},\ldots,\hat{H}_{0}$, 
let $S_{i}$ be the set of $p:\vec{c}\mapsto \vec{d}$
witnessing that $\Amf,\vec{c} \sim_{\text{openGF},\Sigma}^{i} \Amf,\vec{d}$, where $\vec{c}$ is guarded. 
Then include in $\hat{H}_{i}$
\begin{itemize}
	\item all $\vec{c}' \mapsto \vec{c}$, where $\vec{c}'$ is the copy in $\Bmf'$ of the guarded tuple $\vec{c}$ in $\Amf$;
	\item all compositions $p\circ (f^{-1}_{h})_{|[\vec{d}]}$ for any guarded tuple $\vec{d}$ in the range of $f_{h}$ and $p\in S_{i}$;
\end{itemize}
In addition, include in $\vec{a}_{0}^{\Bmf}\mapsto \vec{b}^{\Amf}$ in all $\hat{H}_{i}$, $0\leq i \leq \ell$. We show that $\hat{H}_{\ell},\ldots,\hat{H}_{0}$ is a guarded $\Sigma$ $\ell$-bisimulation. 

For any $i\leq \ell$ any $g\in \hat{H}_i$ is clearly a partial
$\Sigma$-isomorphism, either trivially if $\text{dom}(g) \subseteq \mathfrak{B}'$ or by composition of partial $\Sigma$-isomorphisms if $\text{dom}(g) \subseteq \mathfrak{B}''$. By definition, $\hat{H}_\ell$ contains $\vec{a}_0^\mathfrak{B}
\mapsto \vec{b}^\mathfrak{A}$. We thus only need to check the ``Forth" and
``Back" conditions for guarded $\ell$-bisimulations. Let $g \in \hat{H}_k$ 
for some $k$ with $1\leq k \leq \ell$.
By definition of $\hat{H}_k$, we have either $\text{dom}(g) \subseteq \mathfrak{B}'$ or $\text{dom}(g) \subseteq \mathfrak{B}''$. In each case, we show that for any guarded $\vec{c}$ in $\mathfrak{B}$ and guarded $\vec{d}$
in $\mathfrak{A}$, there exists $g'_0 \in \hat{H}_{k-1}$ with domain $[\vec{c}]$ that
coincides with $g$ on $[\vec{c}] \cap \text{dom}(g)$ (Forth) and there exists $g'_1 \in
\hat{H}_{k-1}$ such that $\text{dom}((g'_1)^{-1})=[\vec{d}]$ and $(g'_1)^{-1}$ coincides with $g^{-1}$ on
$[\vec{d}] \cap \text{im}(g)$ (Back). 

First assume that $[\vec{c}] \cap \text{dom}(g) = \emptyset$. Then, as $\vec{c}$ is guarded in $\mathfrak{B}$, it is either included in $\mathfrak{B}'$ or included in $\mathfrak{B}''$. If $\vec{c}$ is included in $\mathfrak{B}'$, then the partial isomorphism mapping $\vec{c}$ to its copy in $\mathfrak{A}$ is in $\hat{H}_{k-1}$, as required. If $\vec{c}$ is in $\mathfrak{B}''$, then $\vec{c}$ can be written $([h,c_1], \dots, [h,c_n])$ for some $h \in H'$ and $c_1, \dots, c_n \in \text{dom}(\mathfrak{A})$ as it is guarded. But then $(f_h)^{-1}_{|[\vec{c}]} \in \hat{H}_{k-1}$ is as required.

The case $[\vec{d}] \cap \text{im}(g) = \emptyset$ is similar. Assume $\vec{d}= (d_1, \dots, d_n)$ and let $[h,\vec{d}] := ([h,d_1], \dots, [h,d_n]) \in
\mathfrak{B}''$ for any $h \in H'$. 
Then $(f_h)^{-1}_{|[h,\vec{d}]} \in \hat{H}_{k-1}$ is as
required, for any $h \in H'$. We now focus on proving (Forth) and (Back) assuming intersections are not empty.

\medskip

(1) Suppose $\text{dom}(g) \subseteq \mathfrak{B}'$. 

(Forth) Let $\vec{c}$ be guarded in $\mathfrak{B}$ such that
$[\vec{c}] \cap \text{dom}(g) \neq \emptyset$. 
We show there exists $g' \in \hat{H}_{k-1}$ that coincides with $g$ on $\text{dom}(g) \cap
\text{dom}(g')$, such that $[\vec{c}] = \text{dom}(g')$.
By construction of $\mathfrak{B}$,  $[\vec{c}] \cap \text{dom}(g) \neq \emptyset$ and $\text{dom}(g) \subseteq \mathfrak{B}'$ imply $[\vec{c}] \subseteq \mathfrak{B}'$.
By definition of $\hat{H}_k$, $\text{dom}(g)$ is the copy of $\text{im}(g)$
in $\mathfrak{B}'$. Therefore simply take $g'$ to be the partial isomorphism $\vec{c} \mapsto \vec{d}$ such that $\vec{c}$ is the copy in $\mathfrak{B}'$ of $\vec{d}$; it clearly coincides with $g$ on the intersection of their domains, and is in $\hat{H}_{k-1}$ which contains every ``copying" function, by definition.

(Back) Let $\vec{d}$ be guarded in $\mathfrak{A}$ such that
$[\vec{d}] \cap \text{im}(g) \neq \emptyset$. Take $\vec{c}$ to be the copy in $\mathfrak{B}'$ of $\vec{d}$. Then,
the partial isomorphism $g' := \vec{c} \mapsto \vec{d}$ is in $\hat{H}_{k-1}$ by definition, and is such that $(g')^{-1}$ coincides with $g^{-1}$ on $\text{im}(g) \cap \text{im}(g')$.

\medskip

(2) Suppose $\text{dom}(g) \subseteq \mathfrak{B}''$. 

(Forth) 
Write $\text{dom}(g)$ as $([h_1,c_1], \dots, [h_n,c_n])$ with $h_1, \dots, h_n \in H'$ and $(c_1, \dots, c_n) =: \vec{c}$ a tuple in $\mathfrak{A}$. We
want to prove that for any $([h'_1,c'_1], \dots, [h'_m,c'_m])$ guarded in $\mathfrak{B}$ that intersects $\text{dom}(g)$ 
there exists $g' \in \hat{H}_{k-1}$ with domain $([h',c'_1], \dots, [h',c'_m])$ that coincides 
with $g$ on $\text{dom}(g) \cap \text{dom}(g')$.
As $([h'_1,c'_1], \dots, [h'_m,c'_m])$ is guarded in $\mathfrak{B}$ and intersects $\text{dom}(g)$ which is
in $\mathfrak{B}''$, it also has to be contained in $\mathfrak{B}''$, 
by construction of
$\mathfrak{B}$. The fact it is guarded implies we can write it as 
$([h', c'_1], \dots, [h', c'_m])$ for some $h' \in H'$, with $(c'_1, \dots, c'_m)$ being guarded in $\mathfrak{A}$, again by construction of
$\mathfrak{B}$.
As for $([h_1,c_1], \dots, [h_n,c_n])$, we can write it as $([h,c_1], \dots, [h,c_n])$ for some $h \in H'$,
either because it is guarded or because it is equal to $\vec{a}^\mathfrak{B}$, i.e. $([h, h(a_1)], \dots,
[h,h(a_n)])$ for some $h \in H'$. By definition of $\hat{H}_k$, we can write
$g$ as $p \circ (f_h^{-1})_{|\text{dom}(g)}$ for some $p \in S_k$, and we know $g'$ has to be in the form $p'
\circ (f_{h'}^{-1})_{|\text{dom}(g')}$ for some $p' \in S_{k-1}$. 
For notation purposes, we write $[h,\vec{c}]=([h,c_1], \dots, [h,c_n])$
and $[h',\vec{c}']=([h',c'_1], \dots, [h',c'_m])$. 

\medskip
\noindent
Case 1. $h = h'$. Then $[h,\vec{c}] \cap [h,\vec{c}'] \neq \emptyset$ implies
$[\vec{c}] \cap [\vec{c}'] \neq \emptyset$. Then, since $p \in S_k$ and $\text{dom}(p) =
[\vec{c}]$, by definition of guarded $k$-bisimulations there exists $p' \in S_{k-1}$ with
domain $[\vec{c}']$ that coincides with $p$ on $[\vec{c}] \cap [\vec{c}']$. Then $g' := p'
\circ (f_{h}^{-1})_{|[h,\vec{c}']} \in \hat{H}_{k-1}$ is as required.

\medskip
\noindent
Case 2. $h \neq h'$. Figure~\ref{fig:3implies4a} illustrates the following construction.
For all $[h,c_i] \in [[h,\vec{c}]] \cap [[h',\vec{c}']]$ we have that $c_i = h(d_i)$ and
$c'_i = h'(d_i)$ for some $d_i \in
\text{dom}(h) \cap \text{dom}(h')$. 
For any tuple $\vec{d}$ in $\mathcal{D}'$ such that
$[\vec{d}] = \text{dom}(h) \cap \text{dom}(h')$, we have $\mathfrak{A},h'(\vec{d}) \sim^\ell_{\text{GF}, \Sigma}
\mathfrak{A},h(\vec{d})$, witnessed by some partial isomorphism $q : [h'(\vec{d})] \rightarrow
[h(\vec{d})]$. Also, via $p$, we have $\mathfrak{A}, \vec{d}' \sim^k_{\text{GF}, \Sigma} \mathfrak{A},
p(\vec{d}')$ for any $\vec{d}'$ such that $[\vec{d}'] = [h(\vec{d})] \cap [\vec{c}]$. By composition, for any $\vec{d}''$ such that $[\vec{d}''] = [h'(\vec{d})] \cap [\vec{c}']$ we have
$\mathfrak{A}, \vec{d}'' \sim^k_{\text{GF}, \Sigma} \mathfrak{A},p(q(\vec{d}''))$. Because $p \circ q$ is in $S_k$ (by definition of $S_k$) and because $[\vec{c}']$ 
trivially intersects $[h'(\vec{d})] \cap [\vec{c}']$, there exists, by definition of guarded $k$-bisimulations, a partial
isomorphism $p' \in S_{k-1}$ of domain $[\vec{c}']$ that coincides with $p \circ q$ on $ [h'(\vec{d})] \cap [\vec{c}']$. Then, $g' := p' \circ (f_{h'}^{-1})_{|[h',\vec{c}']}$ is the desired partial isomorphism in $\hat{H}_{k-1}$.

\begin{figure*}
	\begin{center}
		\includegraphics[scale=0.335]{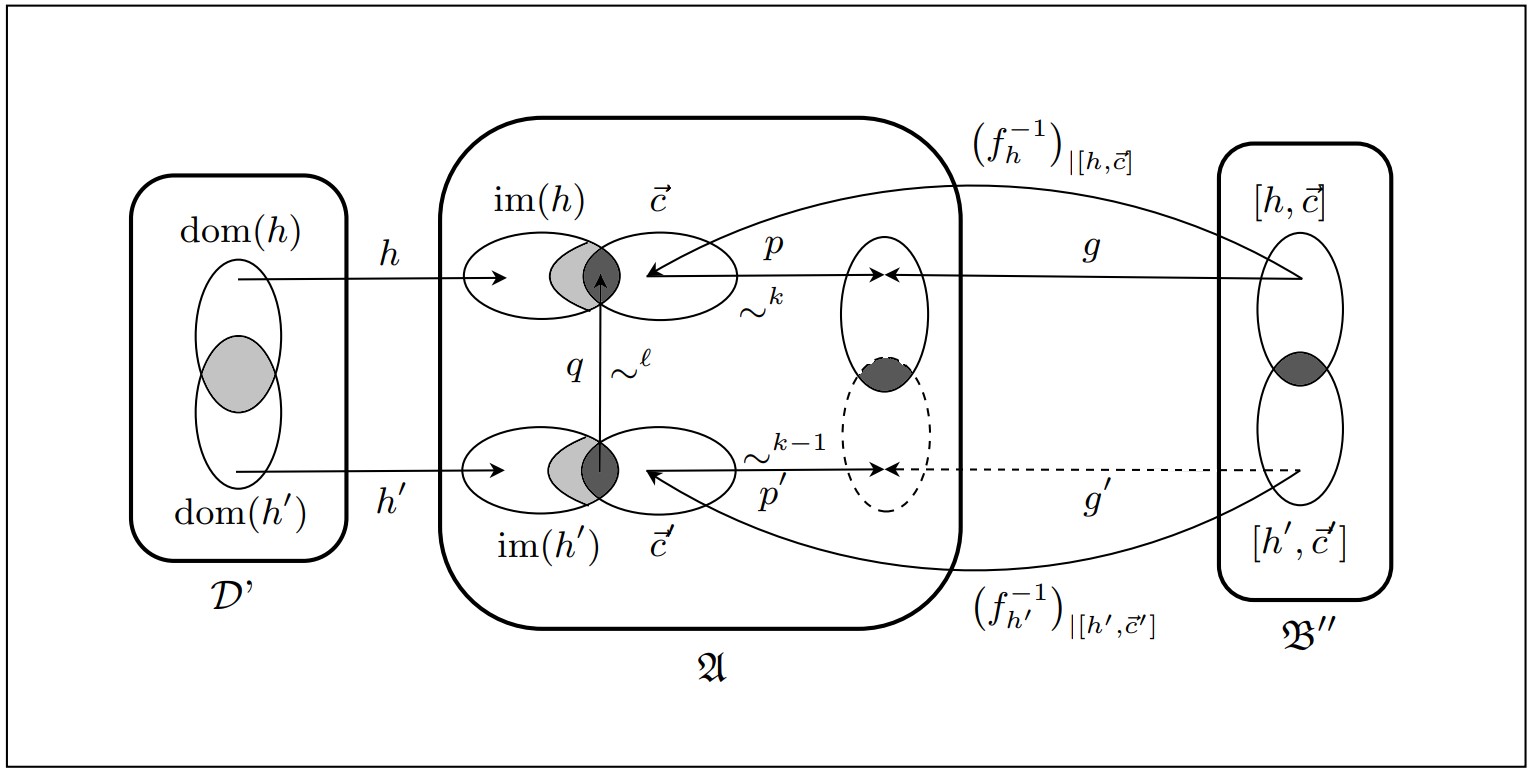}
		\caption{Illustration of proof of (Forth) condition for  $\hat{H}_{\ell},\ldots,\hat{H}_{0}$.} \label{fig:3implies4a}
	\end{center}
\end{figure*}

\medskip

(Back) The construction is illustrated in Figure~\ref{fig:3implies4b}.
Let $\vec{d}$ be guarded in $\mathfrak{A}$ such that $[\vec{d}] \cap
\text{im}(g)$. We show there exists $g' \in \hat{H}_{k-1}$ with image $[\vec{d}]$ such
that $(g')^{-1}$ coincides with $g^{-1}$ on $[\vec{d}] \cap \text{im}(g)$. By
definition of $\hat{H}_k$ we can write $g = p \circ (f_h^{-1})_{|\text{dom}(g)}$ for
some $h \in H'$ and $p \in S_k$, and we know $g'$ has to be of the form $p' \circ
(f_{h'}^{-1})_{|[\vec{d}]}$ for some $h' \in H'$. By definition of guarded
$k$-bisimulations there exists $p' \in S_{k-1}$ such that $\text{im}(p') = [\vec{d}]$
and $p'^{-1}$ coincides with $p^{-1}$ on $\text{im}(p) \cap \text{im}(p')$. Given
that $\text{im}(g) = \text{im}(p)$, if we
write $\vec{d} = (d_1, \dots, d_n)$ and $[h,p'^{-1}(\vec{d})] := ([h,p'^{-1}(d_1)], \dots, [h,p'^{-1}(d_n)])
\in
\mathfrak{B}''$, then $g' = p' \circ (f_h)^{-1}_{|[h,p'^{-1}(\vec{d})]} \in \hat{H}_{k-1}$ is as required.

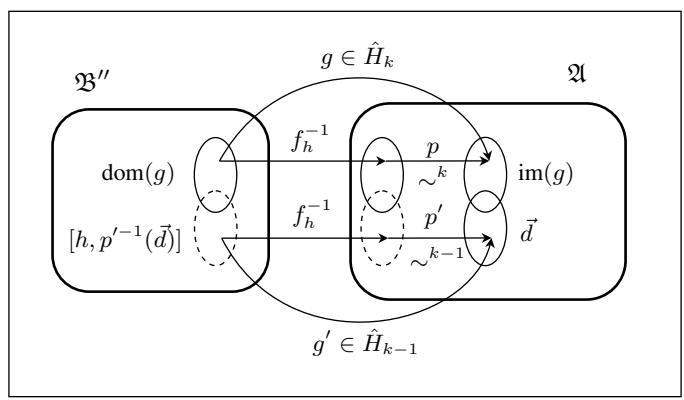
\begin{figure*}

\begin{center}
	
	\tikzset{every picture/.style={line width=0.5pt}} 
	
	\begin{tikzpicture}[x=0.75pt,y=0.75pt,yscale=-1,xscale=1]
	
	\draw   (405.7,129.74) .. controls (411.62,129.74) and (416.41,138.18) .. (416.41,148.6) .. controls (416.41,159.02) and (411.62,167.47) .. (405.7,167.47) .. controls (399.79,167.47) and (395,159.02) .. (395,148.6) .. controls (395,138.18) and (399.79,129.74) .. (405.7,129.74) -- cycle ;
	\draw   (489.7,129.74) .. controls (495.62,129.74) and (500.41,138.18) .. (500.41,148.6) .. controls (500.41,159.02) and (495.62,167.47) .. (489.7,167.47) .. controls (483.79,167.47) and (479,159.02) .. (479,148.6) .. controls (479,138.18) and (483.79,129.74) .. (489.7,129.74) -- cycle ;
	\draw    (407.7,141.6) -- (488.7,141.6) ;
	\draw [shift={(491.7,141.6)}, rotate = 180] [fill={rgb, 255:red, 0; green, 0; blue, 0 }  ][line width=0.08]  [draw opacity=0] (5,-2.5) -- (0,0) -- (5,2.5) -- (3.5,0) -- cycle    ;
	\draw    (491.7,141.6) -- (540.8,141.51) ;
	\draw [shift={(543.8,141.5)}, rotate = 539.89] [fill={rgb, 255:red, 0; green, 0; blue, 0 }  ][line width=0.08]  [draw opacity=0] (5,-2.5) -- (0,0) -- (5,2.5) -- (3.5,0) -- cycle    ;
	\draw   (541.8,129.63) .. controls (547.71,129.63) and (552.5,138.08) .. (552.5,148.5) .. controls (552.5,158.92) and (547.71,167.37) .. (541.8,167.37) .. controls (535.88,167.37) and (531.09,158.92) .. (531.09,148.5) .. controls (531.09,138.08) and (535.88,129.63) .. (541.8,129.63) -- cycle ;
	\draw   (541.8,156.63) .. controls (547.71,156.63) and (552.5,165.08) .. (552.5,175.5) .. controls (552.5,185.92) and (547.71,194.37) .. (541.8,194.37) .. controls (535.88,194.37) and (531.09,185.92) .. (531.09,175.5) .. controls (531.09,165.08) and (535.88,156.63) .. (541.8,156.63) -- cycle ;
	\draw  [dash pattern={on 2pt off 2pt}] (489.8,156.63) .. controls (495.71,156.63) and (500.5,165.08) .. (500.5,175.5) .. controls (500.5,185.92) and (495.71,194.37) .. (489.8,194.37) .. controls (483.88,194.37) and (479.09,185.92) .. (479.09,175.5) .. controls (479.09,165.08) and (483.88,156.63) .. (489.8,156.63) -- cycle ;
	\draw  [dash pattern={on 2pt off 2pt}] (405.8,156.63) .. controls (411.71,156.63) and (416.5,165.08) .. (416.5,175.5) .. controls (416.5,185.92) and (411.71,194.37) .. (405.8,194.37) .. controls (399.88,194.37) and (395.09,185.92) .. (395.09,175.5) .. controls (395.09,165.08) and (399.88,156.63) .. (405.8,156.63) -- cycle ;
	\draw    (408.8,180.5) -- (489.8,180.5) ;
	\draw [shift={(492.8,180.5)}, rotate = 180] [fill={rgb, 255:red, 0; green, 0; blue, 0 }  ][line width=0.08]  [draw opacity=0] (5,-2.5) -- (0,0) -- (5,2.5) -- (3.5,0) -- cycle    ;
	\draw    (492.8,180.5) -- (541.89,180.4) ;
	\draw [shift={(544.89,180.4)}, rotate = 539.89] [fill={rgb, 255:red, 0; green, 0; blue, 0 }  ][line width=0.08]  [draw opacity=0] (5,-2.5) -- (0,0) -- (5,2.5) -- (3.5,0) -- cycle    ;
	\draw    (407.7,141.6) .. controls (442.78,77.8) and (529.92,94.77) .. (543.08,138.77) ;
	\draw [shift={(543.8,141.5)}, rotate = 257.38] [fill={rgb, 255:red, 0; green, 0; blue, 0 }  ][line width=0.08]  [draw opacity=0] (5,-2.5) -- (0,0) -- (5,2.5) -- (3.5,0) -- cycle    ;
	\draw    (408.8,180.5) .. controls (442.15,246.16) and (534.06,226.34) .. (544.36,183.07) ;
	\draw [shift={(544.89,180.4)}, rotate = 458.89] [fill={rgb, 255:red, 0; green, 0; blue, 0 }  ][line width=0.08]  [draw opacity=0] (5,-2.5) -- (0,0) -- (5,2.5) -- (3.5,0) -- cycle    ;
	\draw  [line width=1]  (473.8,132.3) .. controls (473.8,121.36) and (482.66,112.5) .. (493.6,112.5) -- (592.7,112.5) .. controls (603.64,112.5) and (612.5,121.36) .. (612.5,132.3) -- (612.5,191.7) .. controls (612.5,202.64) and (603.64,211.5) .. (592.7,211.5) -- (493.6,211.5) .. controls (482.66,211.5) and (473.8,202.64) .. (473.8,191.7) -- cycle ;
	\draw  [line width=1]  (323.5,133.9) .. controls (323.5,123.74) and (331.74,115.5) .. (341.9,115.5) -- (414.1,115.5) .. controls (424.26,115.5) and (432.5,123.74) .. (432.5,133.9) -- (432.5,189.1) .. controls (432.5,199.26) and (424.26,207.5) .. (414.1,207.5) -- (341.9,207.5) .. controls (331.74,207.5) and (323.5,199.26) .. (323.5,189.1) -- cycle ;
	\draw   (301.5,66) -- (640.5,66) -- (640.5,260.5) -- (301.5,260.5) -- cycle ;
	
	\draw (347,141) node [anchor=north west][inner sep=0.75pt]  [font=\small] [align=left] {$\displaystyle \text{dom}( g)$};
	\draw (442,122) node [anchor=north west][inner sep=0.75pt]  [font=\small] [align=left] {$\displaystyle f^{-1}_{h}$};
	\draw (511,131) node [anchor=north west][inner sep=0.75pt]  [font=\small] [align=left] {$\displaystyle p$};
	\draw (506.41,142.6) node [anchor=north west][inner sep=0.75pt]  [font=\small] [align=left] {$\displaystyle \sim ^{k}$};
	\draw (510,162) node [anchor=north west][inner sep=0.75pt]  [font=\small] [align=left] {$\displaystyle p'$};
	\draw (502.41,183.6) node [anchor=north west][inner sep=0.75pt]  [font=\small] [align=left] {$\displaystyle \sim ^{k-1}$};
	\draw (558,169) node [anchor=north west][inner sep=0.75pt]  [font=\small] [align=left] {$\displaystyle \vec{d}$};
	\draw (557,141) node [anchor=north west][inner sep=0.75pt]  [font=\small] [align=left] {$\displaystyle \text{im}( g)$};
	\draw (330,172) node [anchor=north west][inner sep=0.75pt]  [font=\small] [align=left] {$\displaystyle [h, p'^{-1}(\vec{d})]$};
	\draw (457.45,80) node [anchor=north west][inner sep=0.75pt]  [font=\small] [align=left] {$\displaystyle g\in \hat{H}_{k}$};
	\draw (453.33,224.67) node [anchor=north west][inner sep=0.75pt]  [font=\small] [align=left] {$\displaystyle g'\in \hat{H}_{k-1}$};
	\draw (333,94) node [anchor=north west][inner sep=0.75pt]   [align=left] {$\displaystyle \mathfrak{B} ''$};
	\draw (582,91) node [anchor=north west][inner sep=0.75pt]   [align=left] {$\displaystyle \mathfrak{A}$};
	\draw (443,160) node [anchor=north west][inner sep=0.75pt]  [font=\small] [align=left] {$\displaystyle f^{-1}_{h}$};

	\end{tikzpicture}
	
\end{center}
\caption{Illustration of proof of (Back) condition for  $\hat{H}_{\ell},\ldots,\hat{H}_{0}$.} \label{fig:3implies4b}            
\end{figure*}




\medskip
\noindent
``4. $\Rightarrow$ 5.'' For an indirect proof, suppose $I_{2\ell},\ldots,I_{0}$ is a guarded 
$\Sigma$ $2\ell$-bisimulation between 
$\Bmf,\vec{a}^{\Bmf}$ and $\Amf,\vec{b}^{\Amf}$ for a model $\Bmf$ of $\Kmc$, where $\ell\geq |\Dmc|$. 
We may assume that $I_{i+1}\subseteq I_{i}$ for $i< 2\ell$.
Let $\Dmc'$ be the restriction of $\Bmf$ to $\{ c^{\Bmf} \mid c\in \text{cons}(\Dmc_{\text{con}(\vec{a})})\}$,
where we regard the elements $c^{\Bmf}$ as constants. Define $e:\Dmc_{\text{con}(\vec{a})} \rightarrow \Dmc'$
by setting $e(c)=c^{\Bmf}$. Let $H$ contain $h_{0}:\vec{a}^{\Bmf}\mapsto \vec{b}^{\Amf}$ and, for 
every guarded tuple $\vec{d}$ in $\Dmc'$ any $h:\vec{d}\mapsto \vec{c}\in I_{\ell}$.
It is easy to show that $(e,H)$ is a guarded $\Sigma$ $\ell$-embedding:
assume that $h_{i}:\vec{c}_{i}\mapsto \vec{d}_{i}\in H$ for $i=1,2$.
Let $X_{1},X_{2}$ be the images of $[\vec{c}_{1}]\cap [\vec{c}_{2}]$ under $h_{i}$
and $\vec{d}$ such that $[\vec{d}]=X_{1}$. Then we have 
$h_{i}:\vec{c}_{i}\mapsto \vec{d}_{i}\in I_{\ell}$, for $i=1,2$.
Let $p$ be the restriction of $h_{2}\circ h_{1}^{-1}$ to $X_{1}$. By definition 
$p$ is a partial isomorphism from $X_{1}$ to $X_{2}$. It is as required as 
$$
\Amf,\vec{d} \sim_{\text{openGF},\Sigma}^{\ell}\Bmf,h_{1}^{-1}(\vec{d})
\sim_{\text{openGF},\Sigma}^{\ell}\Amf,h_{2}(h^{-1}(\vec{d})).
$$
\end{proof}
We finally observe the following link between bounded guarded embeddings and homomorphisms. 
\begin{lemma}\label{lem:homguarded}
Let $\mathcal{D}, \vec{a}$ be a pointed database, $\mathfrak{A}, \vec{b}$ be a pointed structure, and $\ell \geq |\mathcal{D}|$. If $\mathcal{D}_{con(\vec{a})}, \vec{a} \preceq^\ell_{\text{openGF}, \Sigma} \mathfrak{A}, \vec{b}$ and $\mathcal{D}_{con(\vec{a})}, \vec{a} \not\rightarrow \mathfrak{A}, \vec{b}$, then there exist $d,d'$ with $d \neq d'$ in $\mathfrak{A}^{\leq |\mathcal{D}|}_{\vec{b}}$ such that $\mathfrak{A},d \sim^{\ell - |\mathcal{D}|}_{\text{openGF}, \Sigma} \mathfrak{A}, d'$.
\end{lemma}
\begin{proof} \
Let $e, H_\ell, \dots, H_0$ witness $\mathcal{D}_{con(\vec{a})}, \vec{a} \preceq^\ell_{\text{openGF}, \Sigma} \mathfrak{A}, \vec{b}$ as in Lemma \ref{lem:guardemb}.
We define a sequence of mappings $S_{0},\ldots,S_{\ell}$ with $S_{k}\subseteq H_{\ell-k}$ for $k\leq \ell$ as follows.
Define $S_0 = \{e(\vec{a}) \mapsto \vec{b}\}$ and assume that $S_{k}$ has been defined for some $k<\ell$. To define $S_{k+1}$, choose for every $h \in S_k$ 
and all guarded $\vec{c}$ intersecting $\text{dom}(h)$ an $h' \in
H_{\ell - k -1}$ with domain $[\vec{c}]$ that coincides with $h$ on $[\vec{c}] \cap \text{dom}(h)$ (this is possible by Condition~1 of Lemma \ref{lem:guardemb})
and add it to $S_{k+1}$. Define $\overline{h} = \bigcup(\bigcup_{k \leq
|\mathcal{D}|}S_k)$. 
We can see $\overline{h}$ as a set of pairs from
$$
\text{dom}(\mathcal{D}'_{\text{con}(e(\vec{a}))}) \times \text{dom}(\mathfrak{A}),
$$
which may or may not be functional.
We know $h$ is not a homomorphism
from $\mathcal{D}'_{\text{con}(e(\vec{a}))}, e(\vec{a})$ to $\mathfrak{A}, \vec{b}^{\Amf}$ because otherwise $h \circ e$ would witness $\mathcal{D}_{\text{con}(\vec{a})},
\vec{a} \rightarrow \mathfrak{A}, \vec{b}^{\Amf}$. However,
\begin{itemize}[nolistsep]
\item[*] $\mathcal{D}' \vDash R(c_1,\dots,c_n)$ implies $\mathfrak{A} \vDash R(d_1, \dots,
d_n)$ for every $(c_1,d_1), \dots, (c_n,d_n) \in \overline{h}$ and every $n$-ary $R \in \Sigma$, since $\overline{h}$ is a
union of partial homomorphisms.
\item[*] for every $c \in \mathcal{D}'_{\text{con}(e(\vec{a}))}$ there exists $h \in
\bigcup_{k \leq \text{dist}_{\mathcal{D}'}(c, e(\vec{a}))}S_{k}$ such that $c \in \text{dom}(h)$, so $\overline{h}$ is defined on the
entire underlying set of $\mathcal{D}'_{\text{con}(e(\vec{a}))}$, as $\text{dist}_{\mathcal{D}'}(c,
e(\vec{a})) \leq |\mathcal{D}'|$ trivially and $|\mathcal{D}'| \leq |\mathcal{D}|$ as $e$ is surjective.
\item[*] $e(\vec{a}) \mapsto \vec{b}$ is included in $\overline{h}$.
\end{itemize}
Therefore the only possible reason as to why $\overline{h}$ is not a homomorphism witnessing $\mathcal{D}_{\text{con}(\vec{a})}, \vec{a} \rightarrow \mathfrak{A}, \vec{b}$ is that $\overline{h}$ is not functional, i.e. there exist $c \in \text{dom}(\mathcal{D}'_{\text{con}(e(\vec{a}))})$ and $d,d' \in \text{dom}(\mathfrak{A})$ such that $d
\neq d'$ and $(c,d), (c,d') \in \overline{h}$. As every $h$
included in $\overline{h}$ is functional, that implies there
exist $h,h' \in \bigcup_{k \leq |\mathcal{D}|} S_k$ such that $h(c) = d$ and
$h'(c) = d'$. There exist $k_1, k_2 \geq \ell -
|\mathcal{D}|$ such that $h \in H_{k_1}$ and $h' \in
H_{k_2}$. By condition (2) of Lemma \ref{lem:guardemb} we
get $\mathfrak{A},d
\sim^{\min(k_1,k_2)}_{\text{openGF},\Sigma}
\mathfrak{A},d'$, hence $\mathfrak{A},d \sim^{\ell -
|\mathcal{D}|}_{\text{openGF},\Sigma} \mathfrak{A},d'$.
Finally, $d,d' \in \mathfrak{A}^{\leq
|\mathcal{D}|}_{\vec{b}}$ follows from the fact that 
$\text{dist}_{\mathfrak{A}}(h(c), \vec{b}) \leq k$
for any $c \in \text{dom}(h)$ such that $h \in S_k$. This can be proved by
induction on $k$.
\end{proof}
We are in a position now to prove the equivalence of Points~1 to 4 and Point~5 of Theorem~\ref{thm:critGF1}. In fact, the equivalence follows from the following lemma since Point~1 of Lemma~\ref{lem:fin} is equivalent to Point~1 of Theorem~\ref{thm:critGF1} (by Lemma~\ref{thm:gfcrit1}) and Point~2 of Lemma~\ref{lem:fin} coincides with Point~5 of Theorem~\ref{thm:critGF1}.
\begin{lemma}\label{lem:fin}
   Assume that $(\Kmc,P,\{\vec{b}\})$ is a labeled GF-KB with $\Kmc=(\Omc,\Dmc)$, $\Sigma=\text{sig}(\Kmc)$, and $\vec{b}=(b_{1},\ldots,b_{n})$. Then the following conditions are equivalent:
   \begin{enumerate}
   	\item there exists a model $\Amf$ of $\Kmc$ and $\ell\geq 0$ such for all $\vec{a}\in P$: 
   	$\Dmc_{\text{con}(\vec{a})},\vec{a} \not\preceq_{\text{openGF},\Sigma}^{\ell}\Amf,\vec{b}^{\Amf}$;
   	\item there exists a model $\Amf$ of $\Kmc$ such that for all $\vec{a}\in P$:
   	%
   	\begin{enumerate}
   		\item $\Dmc_{\text{con}(\vec{a})},\vec{a}\not\rightarrow \Amf,\vec{b}^{\Amf}$ and
   		\item if the set $I$ of all $i$ such that
   		$\text{tp}_{\Kmc}(\Amf,b_{i}^{\Amf})$ is connected and openGF-complete is not empty,
   		then 
   		\begin{enumerate}
   			
   			\item $J=\{1,\dots,n\}\setminus I \neq \emptyset$
   			and
   			$\Dmc_{\text{con}(\vec{a}_{J})},\vec{a}_{J}\not\rightarrow \Amf,\vec{b}_{J}^{\Amf}$ or
   			
   			\item $\text{tp}_{\Kmc}(\Amf,\vec{b}^{\Amf})$ is not realizable in 
   			$\Kmc,\vec{a}$.
   			
   		\end{enumerate}
   	\end{enumerate}
   	\end{enumerate}
\end{lemma}

\begin{proof} \
``1. $\Rightarrow$ 2.'' Assume that Point~1 holds for $\Amf$ and $\ell_{0}\geq 0$. Assume $\vec{a}\in P$ is given.
As $\Dmc_{\text{con}(\vec{a})},\vec{a} \rightarrow \Amf,\vec{b}^{\Amf}$ implies
$\Dmc_{\text{con}(\vec{a})},\vec{a} \preceq_{\text{openGF},\Sigma}^{\ell}\Amf,\vec{b}^{\Amf}$
for all $\ell\geq 0$, we obtain that (a) holds.
To show that (b) holds for $\Amf$ and $\vec{a}$, 
assume that $I$ as defined in the lemma is not empty
and that $\text{tp}_{\Kmc}(\Amf,\vec{b}^{\Amf})$ is realizable in 
$\Kmc,\vec{a}$. Take a model $\Bmf$ witnessing this.  
Consider the maximal sets $I_{1},\ldots,I_{k}\subseteq \{1,\ldots,n\}$ such that 
$\vec{b}_{I_{j}}^{\Bmf}$ is in a connected component $\Bmf_{j}$ of $\Bmf$. 
Then there exists at least one $j$ such that 
$\text{tp}_{\Kmc}(\Amf,\vec{b}_{I_{j}}^{\Amf})$ is not openGF-complete: otherwise
$\Bmf,\vec{a}_{I_{j}}^{\Bmf}\sim_{\text{openGF},\Sigma}\Amf,\vec{b}_{I_{j}}^{\Amf}$ for all $j$
and so $\Dmc_{\text{con}(\vec{a}_{I_{j}})},\vec{a}_{I_{j}} \preceq_{\text{openGF},\Sigma}^{\ell}\Amf,\vec{b}_{I_{j}}^{\Amf}$
for all $\ell\geq 0$, thus $\Dmc_{\text{con}(\vec{a})},\vec{a} \preceq_{\text{openGF},\Sigma}^{\ell}\Amf,\vec{b}^{\Amf}$,
for all $\ell\geq 0$, a contradiction.

For any $j$ such that $\text{tp}_{\Kmc}(\Amf,\vec{b}_{I_{j}}^{\Amf})$ is not openGF-complete 
we have by Lemma~\ref{lem:openGF2}
that $\text{tp}_{\Kmc}(\Amf,\vec{b}_{i}^{\Amf})$ is not openGF-complete for some $i\in I_{j}$.
Therefore $J\not=\emptyset$.
Assume now for a proof by contradiction that 
$\Dmc_{\text{con}(\vec{a}_{J})},\vec{a}_{J}\rightarrow \Amf,\vec{b}_{J}^{\Amf}$.
Then $\Dmc_{\text{con}(\vec{a}_{J})},\vec{a}_{J}\preceq_{\text{openGF},\Sigma}^{\ell} \Amf,\vec{b}_{J}^{\Amf}$,
for any $\ell\geq 0$. By Lemma~\ref{lem:openGF2}, $\text{tp}_{\Kmc}(\Amf,\vec{b}_{I}^{\Amf})$ is openGF-complete
and so $\Bmf,\vec{a}_{I}^{\Bmf}\sim_{\text{openGF},\Sigma}\Amf,\vec{b}_{I}^{\Amf}$,
and therefore
$\Dmc_{\text{con}(\vec{a}_{I})},\vec{a}_{I} \preceq_{\text{openGF},\Sigma}^{\ell}\Amf,\vec{b}_{I}^{\Amf}$, for every 
$\ell\geq 0$.
As $\Dmc_{\text{con}(\vec{a}_{I})}$ and $\Dmc_{\text{con}(\vec{a}_{J})}$ must be disjoint, it follows
that $\Dmc_{\text{con}(\vec{a})},\vec{a} \preceq_{\text{openGF},\Sigma}^{\ell}\Amf,\vec{b}^{\Amf}$ for any $\ell\geq 0$,
and we have derived a contradiction.

\medskip

``2. $\Rightarrow$ 1.'' Assume Conditions~(a) and (b) hold for some model $\Amf$ of $\Kmc$ for all $\vec{a}\in P$. 
As GF is finitely controllable there exists a finite such model $\Amf$. 
Assume that the set $I$ defined in the lemma is empty. The case $I\not=\emptyset$ is considered in exactly the same way and therefore omitted.
Our aim is to show that $\Dmc_{\text{con}(\vec{a})},\vec{a} \not\preceq_{\text{openGF},\Sigma}^{\ell}\Cmf,\vec{b}^{\Amf}$ for a variant 
$\Cmf$ of $\Amf$ and for sufficiently
large $\ell$, where $\Sigma=\text{sig}(\Kmc)$. 

Let $X$ be the set of $i\leq n$ such that $\Phi_{i}(x)=\text{tp}_{\Kmc}(\Amf,b_{i}^{\Amf})$ is disconnected. If $X=\{1,\ldots,n\}$,
then $\neg \bigwedge_{i\in X}\Phi_{i}(x_{i})$ separates $(\Kmc,P,\{\vec{b}\})$ and we are done. Otherwise, let $\Amf_{i}$, $i\in X$, 
be the maximal connected components of $\Amf$ containing the singleton $b_{i}^{\Amf}$. Since $\text{tp}_{\Kmc}(\Amf,b_{i}^{\Amf})$ is disconnected we have $\text{dom}(\Amf_{i})=\{b_{i}^{\Amf}\}$, for all $i\in X$.
We partition the remaining part of $\Amf$ without $\Amf_{i}$, $i\in X$, into components as follows. 
Define a binary relation
$E$ on the class of $\Kmc$-types $\Phi(x)$ with one free variable $x$ such that $(\Phi(x),\Psi(x))\in E$ iff there exists
a model $\Amf$ of $\Kmc$ and nodes $a,b$ in $\text{dom}(\Amf)$ such that $a,b$ are in the same connected component
in $\Amf$ and $a$ and $b$ realize $\Phi$ and $\Psi$, respectively. It is easy to see that $E$ is an equivalence relation. Let $\Amf'$ and $\{\Emf \mid \Emf\in K\}$ be the maximal components of $\Amf$ 
without $\{b_{i}^{\Amf}\mid i\in X\}$ such that:
\begin{itemize}
\item all nodes in any $\Emf$ are connected to a node in $\{c^{\Amf}\mid c\in \text{cons}(\Dmc)\}$ and all $\Kmc$-types $\Phi(x)$ 
realized in $\Emf$ are $E$-equivalent;
\item no node in $\Amf'$ is connected to a node in $\{c^{\Amf}\mid c\in \text{cons}(\Dmc)\}$.
\end{itemize}
Observe that $\Amf$ is the disjoint union of $\Amf_{i}$, $i\in X$, $\Amf'$, and the structures in $K$.
Let $\Emf\in K$ and let $\Dmc'$ be the restriction of $\Dmc$ to the constants $c\in \text{cons}(\Dmc)$ such that $c^{\Emf}\in \text{dom}(\Emf)$.
Let $I_{0}$ be the set of $i$ with $b_{i}^{\Amf}\in \text{dom}(\Emf)$. 
We aim to construct a model $\Cmf$ of $(\Omc,\Dmc')$ such that 

\begin{figure*}
\begin{center}

	\tikzset{every picture/.style={line width=0.5pt}} 
	

	\caption{Construction of $\mathfrak{C}$.} \label{fig:GFthm7}
\end{center}
\end{figure*}

\medskip
\noindent
($\ast$) if $\Dmc_{\text{con}(\vec{a}_{I_{0}})},\vec{a}_{I_{0}} \not\rightarrow \Amf,\vec{b}_{I_{0}}^{\Amf}$, then there exists $\ell\geq 0$ such that
$\Dmc_{\text{con}(\vec{a}_{I_{0}})},\vec{a}_{I_{0}} \not\preceq_{\text{openGF},\Sigma}^{\ell} \Cmf,\vec{b}_{I_{0}}^{\Amf}$.

\medskip
For any model $\Cmf$ of $\Dmc'$ and $d\in \text{dom}(\Cmf)$ we let the distance $\text{dist}_{\Cmf}(\Dmc',d)=\ell$ iff
$\ell$ is minimal such $\text{dist}(c^{\Cmf},d)\leq \ell$ for at least one $c\in \text{cons}(\Dmc')$. We denote by 
$\Cmf^{\leq \ell}_{\Dmc'}$ the substructure of $\Cmf$ induced by the set of 
nodes $d$ in $\Cmf$ with $\text{dist}_{\Cmf}(\Dmc',d)\leq \ell$.
We construct for any $\ell\geq 0$ a model $\Cmf$ of $\Omc$ that coincides with $\Emf$ on $\{c^{\Emf}\mid c\in \text{cons}(\Dmc')\}$
such that $\Cmf^{\leq \ell}_{\Dmc'}$ is finite and there exists $\ell'\geq \ell$ with
\begin{itemize}
\item[(a)] $\Cmf^{\leq \ell}_{\Dmc'},\vec{b}_{I_{0}}^{\Amf}\rightarrow \Amf,\vec{b}_{I_{0}}^{\Amf}$;
\item[(b)] for any two distinct $d_{1},d_{2}\in \text{dom}(\Cmf^{\leq \ell}_{\Dmc'})$, 
$\Cmf,d_{1}\not\sim_{\text{openGF},\Sigma}^{\ell'}\Cmf,d_{2}$.
\end{itemize} 
We first show that ($\ast$) follows. Assume $\ell'$ is such that (b) holds. Let $\ell''=\ell'+|\Dmc|$ and $\ell\geq |\Dmc|$. 
Assume that $\Dmc_{\text{con}(\vec{a}_{I_{0}})},\vec{a}_{I_{0}}\preceq_{\text{openGF},\Sigma}^{\ell''} \Cmf,\vec{b}_{I_{0}}^{\Amf}$
but $\Dmc_{\text{con}(\vec{a}_{I_{0}})},\vec{a}_{I_{0}}\not\rightarrow \Amf,\vec{b}_{I_{0}}^{\Amf}$.
By Condition~(a),  
$$
\Dmc_{\text{con}(\vec{a}_{I_{0}})},\vec{a}_{I_{0}}\not\rightarrow \Cmf^{\leq \ell''}_{\Dmc'},\vec{b}_{I_{0}}^{\Cmf}
$$
By Lemma~\ref{lem:homguarded}, there exist $d,d'$ with $d\not=d'$ in $\Cmf^{\leq |\Dmc|}_{\Dmc'}$
such that $\Cmf,d \sim_{\text{openGF},\Sigma}^{\ell'} \Cmf,d'$ and we have 
derived a contradiction to Condition~(b).

\medskip

We come to the construction of $\Cmf$. It is helpful to consider in parallel the simpler construction of $\Cmf$ in the proof of Theorem~\ref{thm:char12}. Note that the tasks are almost the same as Condition (a) corresponds to Condition (i) in Theorem~\ref{thm:char12} and Condition (b) is a bounded version of Condition (ii) in Theorem~\ref{thm:char12}. Assume $\ell\geq 0$ is given. To construct $\Cmf$, let $T_{\Emf}$ be the set of $\Kmc$-types $\Phi(x)$
that are $E$-equivalent to some $\Kmc$-type realized in $\Emf$ with $\Emf\in K$. Observe that $T_{\Emf}$ is an equivalence class for the relation $E$, by construction of $\Emf$.
No $\Kmc$-type in $T_{\Emf}$ is openGF-complete and, by construction, we find a
sequence
$$
\sigma = \Phi_{0}^{\sigma},\Phi_{1}^{\sigma},\Phi_{2}^{\sigma}
$$
such that for any $\Kmc$-type $\Phi(x)\in T_{\Emf}$ we find a sequence witnessing
openGF-incompleteness of $\Phi(x)$ that ends with $\sigma$.
As a first step of the construction of $\Cmf$, 
we define a model $\Bmf$ of $\Kmc$ by repeatedly 
forming the partial unfolding of $\Emf$ so that 

\medskip
\noindent
$(\text{path})$ from any $f_{0}\in \Bmf^{\leq \ell}_{\Dmc'}$ there exists
a strict path $R_{1}^{f_{0}}(\vec{d}_{1}),\ldots,R_{k}^{f_{0}}(\vec{d}_{k})$ from $f_{0}$ to some $f_{1}$
such that $\text{dist}_{\Bmf}(\Dmc',f_{1}))=\ell$.

\medskip

For the inductive construction of $\Bmf$, let $\Bmf_{0}=\Amf$ and include all 
$d\in \text{dom}(\Amf^{\leq \ell}_{\Dmc'})$ into the frontier $F_{0}$.
Assume $\Bmf_{i}$ and frontier $F_{i}$ have been constructed. If $F_{i}$ is empty, we are done
and set $\Bmf=\Bmf_{i}$. Otherwise take $d \in F_{i}$ and let $d'\not=d$ be any element contained
in a joint guarded set with $d$ in $\Bmf_{i}$. Assume $k=\text{dist}_{\Bmf_{i}}(\Dmc',d)$.
Then let $\Bmf_{i+1}$ be the partial unfolding $(\Bmf_{i})_{\vec{d}}$ 
of $\Bmf_{i}$ for the tuple $\vec{d}=(d,d',d,d',\ldots)$ of length $\ell-k$, and obtain $F_{i+1}$ by removing $d$ from
$F_{i}$ and adding all new nodes in 
$\text{dom}((\Bmf_{i+1})^{\leq \ell}_{\Dmc'})$.
Clearly this construction terminates after finitely many steps and $(\text{path})$ holds, see Lemma~\ref{lem:partialunfolding}.
\medskip

Let $L$ denote the set of all $d$ in $\Bmf$ with $\text{dist}_{\Bmf}(\Dmc',d)=\ell$
and let $L'$ denote the set of all $\vec{d}$ of arity $\geq 2$
in $\Bmf$ such that there exist $R$ with $\Bmf\models R(\vec{d})$
and $d\in [\vec{d}]$ with $\text{dist}_{\Bmf}(\Dmc',d)=\ell$.
We obtain $\Cmf$ by keeping $\Bmf^{\leq \ell}_{\Dmc'}$ and the guarded sets that intersect with
it and attaching to every $d\in L$ and $\vec{d}\in L'$ 
guarded tree decomposable $\Fmf_{d}$ and $\Fmf_{\vec{d}}'$
such that in the resulting model no $d$ in $L$ is guarded $\Sigma$ $\ell'$-bisimilar to any other 
$d'$ in $\Bmf^{\leq \ell}_{\Dmc'}$ for a sufficiently large $\ell'$. It then directly follows that $\Cmf$
satisfies Conditions~(a) and (b). 

The construction of $\mathfrak{F}_{\vec{d}}'$ is straightforward. Fix $\vec{d}\in L'$. Let $\Phi_{0}':=\text{tp}_{\Kmc}(\Bmf,\vec{d})$.
Then $\mathfrak{F}_{\vec{d}}'$ is defined as the tree decomposable model $\Amf_{r}'$ of
$\Omc$ with tree decompositon $(T',E',\text{bag}')$ and root $r$ such that
$\Amf_{r}'\models \Phi_{0}'(\vec{d})$ and $\text{bag}(r)=[\vec{d}]$ and for every $\Kmc$-type 
$\Psi_{1}(\vec{x}_{1})$ realized by some $\vec{c}$ with $[\vec{c}]=\text{bag}(t)$ 
and $\Kmc$-type $\Psi_{2}(\vec{x}_{2})$ coherent with $\Psi_{1}(\vec{x}_{1})$ there exists a successor 
$t'$ of $t$ in $T$ such that $\Psi_{1}(\vec{x}_{1})\cup \Psi_{2}(\vec{x}_{2})$ is 
realized in $\text{bag}(t) \cup \text{bag}(t')$ under an assignment $\mu$ of the variables 
$[\vec{x}_{1}]\cup [\vec{x}_{2}]$ such that $\mu(\vec{x}_{1})=\vec{c}_{1}$.
The only property of $\mathfrak{F}_{\vec{d}}'$ we need is that 
$$
\mathfrak{F}_{\vec{d}}'\models \forall \vec{x}_{1}(\Phi_{1}^{\sigma}\rightarrow \exists \vec{x}'\Phi_{2}^{\sigma})
$$
where here and in what follows $\vec{x}_{1}$ are the variables in $\Phi_{1}^{\sigma}$ and
$\vec{x}'$ are the variables in $\Phi_{2}^{\sigma}$ that are not in $\Phi_{1}^{\sigma}$.

The construction of $\mathfrak{F}_{d}$ is more involved, but follows closely the construction of $\mathfrak{F}_{d}$ in the proof of Theorem~\ref{thm:char12}. Let $L_{\Omc}=2^{2^{||\Omc||}}+1$ and take for any $d\in L$ a number 
$$
N_{d}>|\Bmf^{\leq \ell+1}_{\Dmc'}|+2(L_{\Omc}+1)
$$
such that $|N_{d}-N_{d'}|>2(L_{\Omc}+1)$ for $d\not=d'$. 
Fix $d\in L$ and let $\Phi_{0}(x)=\text{tp}_{\Kmc}(\Amf,d)$. Then $\Phi_{0}(x)\in T_{\Emf}$ and we find a sequence
$\Phi_{0}(\vec{x}_{0}),\ldots,\Phi_{n_{d}}(\vec{x}_{n_{d}}),\Phi_{n_{d}+1}(\vec{x}_{n_{d}+1})$ 
that witnesses openGF incompleteness of $\Phi_{0}(x_{0})$
and ends with $\Phi_{0}^{\sigma}\Phi_{1}^{\sigma}\Phi_{2}^{\sigma}$. 
By Lemma~\ref{lem:critalci} we may assume that $1\leq n_{d}\leq L_{\Omc}+1$. 
Let
$$
\Psi(x)= \exists \Sigma^{L_{\Omc}+1}.(\Phi_{1}^{\sigma}\wedge \neg \exists \vec{x}'\Phi_{2}^{\sigma}),
$$
where $\exists \Sigma^{k}.\chi$ stands for the disjunction of all openGF formulas stating 
that the exists a path from $x$ along relations in $\Sigma$ of length at most $k$
to a tuple where $\chi$ holds.
To construct $\Fmf_{d}$ consider the guarded tree decomposable model $\Amf_{r}$ of
$\Omc$ with guarded tree decomposition $(T,E,\text{bag})$ and root $r$ such that
$\Amf_{r}\models \Phi_{0}(c_{0})$ for some 
constant $c_{0}$ with $\text{bag}(r)=\{c_{0}\}$ and for every $\Kmc$-type 
$\Psi_{1}(\vec{x}_{1})$ realized by some $\vec{c}$ with $[\vec{c}]=\text{bag}(t)$ 
and $\Kmc$ type $\Psi_{2}(\vec{x}_{2})$ coherent with $\Psi_{1}(\vec{x}_{1})$ there exists a successor 
$t'$ of $t$ in $T$ such that $\Psi_{1}(\vec{x}_{1})\cup \Psi_{2}(\vec{x}_{2})$ is 
realized in $\text{bag}(t) \cup \text{bag}(t')$ under an assignment $\mu$ of the variables 
$[\vec{x}_{1}]\cup [\vec{x}_{2}]$ such that $\mu(\vec{x}_{1})=\vec{c}_{1}$, except if 
$\Psi_{1}\wedge \neg \exists \vec{x}'\Psi_{2}$ (with $\vec{x}'$ the sequence of variables
in $\vec{x}_{2}$ which are not in $\vec{x}_{1}$) is equivalent to 
$\Phi_{1}^{\sigma}\wedge \neg \exists \vec{x}'\Phi_{2}^{\sigma}$ and 
$\text{dist}_{\Amf_{r}}(\text{bag}(t),\text{bag}(r))\leq N_{d}+ L_{\Omc}+1$.
Observe that 
\begin{itemize}
	\item $\Amf_{r}\models \Psi(e)$ for all $e$ with $\text{dist}_{\Amf_{r}}(c_{0},e)\leq N_{d}$;
	\item $\Amf_{r}\models \neg \Psi(e)$ for all $e$ with $\text{dist}_{\Amf_{r}}(c_{0},e)> 
	N_{d}+2(L_{\Omc}+1)$.
\end{itemize}
Moreover, $\Amf_{r}$ contains a strict path 
$$
R_{1}(\vec{e}_{1}),\ldots,R_{n_{d}}(\vec{e}_{n_{d}}),\ldots ,R_{n_{d}}(\vec{e}_{n_{d}+2N_{d}})
$$
from $e_{0}\in [\vec{e}_{1}]$ to $c_{0}\in [\vec{e}_{n_{d}+2N_{d}}]$ such that $\Phi_{0}(x)$
is realized in $e_{0}$. Then $\Fmf_{d}$ is obtained from $\Amf_{r}$ by renaming $e_{0}$ to $d$. 
Finally $\Cmf$ is obtained by hooking $\Fmf_{d}$ at $d$ to $\Bmf^{\leq \ell}_{\Dmc'}$ for all $d\in L$.
$\Cmf$ is a model of $\Kmc$ since $\Phi_{0}(x)$ is realized in $e_{0}$ and $d$. 
Moreover, it clearly satisfies Condition~(a). For Condition~(b) assume $d\in L$ is 
as above. Let
$$
\varphi_{d}(x)=\forall \Sigma^{N_{d}}.\Psi 
$$
where $\forall \Sigma^{k}.\chi$ stands for $\neg \exists \Sigma^{k}.\neg \chi$.
Then $\Cmf\models \varphi_{d}(c_{0})$ and by construction no node that is not in $\text{dom}(\Fmf_{d})$ satisfies $\varphi_{d}$.
Condition~(b) now follows from the fact that there exists a path from $d$ to a node
satisfying $\varphi_{d}$ that is shorter than any such path in $\Cmf$ from any other node in 
$\Bmf^{\leq \ell}_{\Dmc'}$ to a node satisfying $\varphi_{d}$.

\medskip
We have proved ($\ast$). We now aim to extend ($\ast$) and show that $\Dmc_{\text{con}(\vec{a})},\vec{a}
\not\preceq_{\text{openGF},\Sigma}^{\ell} \Cmf,\vec{b}^{\Amf}$, for appropriately defined $\Cmf$, all $\vec{a}\in P$, and
for sufficiently large $\ell$.
Let $\vec{a}\in P$ be fixed.
If $\Dmc_{\text{con}(\vec{a}_{I_{0}})},\vec{a}_{I_{0}} \not\rightarrow \Amf,\vec{b}_{I_{0}}^{\Amf}$
for some $I_{0}$ associated to some $\Emf\in K$, then, by ($\ast$),
$\Dmc_{\text{con}(\vec{a}_{I_{0}})},\vec{a}_{I_{0}} \not\preceq_{\text{openGF},\Sigma}^{\ell} \Cmf,\vec{b}_{I_{0}}^{\Cmf}$,
for some $\ell$, and therefore $\Dmc_{\text{con}(\vec{a})},\vec{a}
\not\preceq_{\text{openGF},\Sigma}^{\ell} \Cmf,\vec{b}^{\Amf}$, for some $\ell$, and we are done.
Now assume that $\Dmc_{\text{con}(\vec{a}_{I_{0}})},\vec{a}_{I_{0}} \rightarrow \Amf,\vec{b}_{I_{0}}^{\Amf}$
for all $I_{0}$ associated with any $\Emf\in K$. We know that 
$\Dmc_{\text{con}(\vec{a})},\vec{a} \not\rightarrow \Amf,\vec{b}^{\Amf}$. But then either
(i) $\Dmc_{\text{con}(a_{i})},a_{i}\not\preceq_{\text{openGF},\Sigma}^\ell \Cmf,b_{i}^{\Cmf}$ for some $i\in X$
(and we are done) or (ii) some $a_{i},a_{j}$ with $i\not=j$ and $i,j\in X$ are connected in $\Dmc$,
or (iii) some $a_{i}$, $i\in X$ and $a\in [\vec{a}_{I_{0}}]$ with $I_{0}$ linked to some $\Emf\in K$
are connected in $\Dmc$ or (iv) some $a\in [\vec{a}_{I_{0}}]$ and $a'\in [\vec{a}_{I_{0}'}]$ with $I_{0}$ 
and $I_{0}'$ linked to distinct $\Emf\in K$
are connected in $\Dmc$. In all these cases it follows that $\Dmc_{\text{con}(\vec{a})},\vec{a}
\not\preceq_{\text{openGF},\Sigma}^{\ell} \Cmf,\vec{b}^{\Amf}$, for sufficiently large $\ell$.
\end{proof}

\section{Proofs for Section~\ref{sec:discussion}}

\thmundeuna*

\begin{proof}\ The proof is by reduction from the infinite tiling
  problem. Recall the definition of a tiling system $(T,H,V)$ and a
  solution to $(T,H,V)$ from the proof of Theorem~\ref{thm:fo2fo}.
  Given a tiling system $(T,H,V)$, we construct a labeled
  $\mathcal{ALCI}$-KB $(\Kmc,P,\{b\})$ with $\Kmc=(\Omc,\Dmc)$ as
  follows:
\begin{align}
  \Omc & = \{B \sqsubseteq \exists U.P_{0},\label{eq:succ1}\\
  & \phantom{ {}={} } P_{0} \sqsubseteq \exists U^{-}.\top \sqcap
  \exists R_{h}. P_{1} \sqcap \exists R_{v}. P_{2}, \label{eq:succ2}\\
  & \phantom{ {}={} } P_{1} \sqsubseteq \exists U^{-}.\top \sqcap
  \exists R_{h}. P_{0} \sqcap \exists R_{v}. P_{3},{}
  \label{eq:succ3}\\
  & \phantom{ {}={} } P_{2} \sqsubseteq \exists U^{-}.\top \sqcap
  \exists R_{h}. P_{3} \sqcap \exists R_{v}. P_{0},{}
  \label{eq:succ4}\\
  & \phantom{ {}={} } P_{3} \sqsubseteq \exists U^{-}.\top \sqcap
  \exists R_{h}. P_{2} \sqcap \exists R_{v}.
  P_{1},{}\label{eq:succ5}\\
  & \phantom{ {}={} } P_{i} \sqcap P_{j}\sqsubseteq \bot, & \text{for
  $0\leq i < j \leq 3$} \label{eq:succ6}\\
  & \phantom{ {}={} } B \sqcap P_{i} \sqsubseteq \bot, & \text{for
  $0\leq i \leq 3$}\label{eq:succ7}\\ & \phantom{ {}={} } P_{0} \sqcup
  P_{1} \sqcup P_{2} \sqcup P_{3} \sqsubseteq \bigsqcup_{t\in T}
  (A_{t} \sqcap \bigsqcap_{t'\in T\setminus\{t\}} \neg
  A_{t'}),{}\label{eq:succ8}\\ & \phantom{ {}={} } A_t \sqsubseteq
  \forall R_v.\bigsqcup_{(t,t')\in V} A_{t'}, & \text{for all $t\in
  T$} \label{eq:succ9}\\
  & \phantom{ {}={} } A_t \sqsubseteq \forall R_h.\bigsqcup_{(t,t')\in
  H} A_{t'} & \text{for all $t\in T$} \label{eq:succ10}\\ &\phantom{
    {}={} }\! \} \notag \end{align}
Let $\Dmc=\Dmc_{a_{0}}\cup \Dmc_{a_{1}}\cup \Dmc_{a_{2}} \cup
\Dmc_{b}$, where 
\begin{eqnarray*}
  \Dmc_{a_{0}} & = &
  \{U(a_{0},c_{1}),R_{h}(c_{1},c_{2}),R_{v}(c_{1},c_{3}),R_{v}(c_{2},c_{4}),R_{h}(c_{3},c_{5})\}
  \\
  \Dmc_{a_{1}} & = &
  \{U(a_{1},d_{1}),R_{h}(d_{1},d_{2}),U(d_{3},d_{2})\}\\
  \Dmc_{a_{2}} & = &
  \{U(a_{2},e_{1}),R_{v}(e_{1},e_{2}),U(e_{3},e_{2})\}\\
  \Dmc_{b}     & = & \{B(b)\}
\end{eqnarray*}
and let $P=\{a_{0},a_{1},a_{2}\}$. The connected components of the positive examples $P$ can be depicted as follows:
\begin{center}
	
\tikzset{every picture/.style={line width=0.5pt}} 

\begin{tikzpicture}[x=0.75pt,y=0.75pt,yscale=-1,xscale=1]
	
	\draw    (271.22,275.24) -- (271.55,247.87) ;
	\draw [shift={(271.58,244.88)}, rotate = 450.68] [fill={rgb, 255:red, 0; green, 0; blue, 0 }  ][line width=0.08]  [draw opacity=0] (5,-2.5) -- (0,0) -- (5,2.5) -- (3.5,0) -- cycle    ;
	\draw    (277.01,279.78) -- (303.58,279.86) ;
	\draw [shift={(306.58,279.88)}, rotate = 180.19] [fill={rgb, 255:red, 0; green, 0; blue, 0 }  ][line width=0.08]  [draw opacity=0] (5,-2.5) -- (0,0) -- (5,2.5) -- (3.5,0) -- cycle    ;
	\draw    (276.01,239.79) -- (296.36,239.67) ;
	\draw [shift={(299.36,239.65)}, rotate = 539.6700000000001] [fill={rgb, 255:red, 0; green, 0; blue, 0 }  ][line width=0.08]  [draw opacity=0] (5,-2.5) -- (0,0) -- (5,2.5) -- (3.5,0) -- cycle    ;
	\draw    (311.16,275.19) -- (311.5,256.65) ;
	\draw [shift={(311.56,253.65)}, rotate = 451.07] [fill={rgb, 255:red, 0; green, 0; blue, 0 }  ][line width=0.08]  [draw opacity=0] (5,-2.5) -- (0,0) -- (5,2.5) -- (3.5,0) -- cycle    ;
	\draw    (386.34,278.62) -- (406.69,278.5) ;
	\draw [shift={(409.69,278.49)}, rotate = 539.6700000000001] [fill={rgb, 255:red, 0; green, 0; blue, 0 }  ][line width=0.08]  [draw opacity=0] (5,-2.5) -- (0,0) -- (5,2.5) -- (3.5,0) -- cycle    ;
	\draw    (443.29,257.06) .. controls (455.49,266.09) and (464.45,279.91) .. (421.7,278.09) ;
	\draw [shift={(419.01,277.96)}, rotate = 363.14] [fill={rgb, 255:red, 0; green, 0; blue, 0 }  ][line width=0.08]  [draw opacity=0] (5,-2.5) -- (0,0) -- (5,2.5) -- (3.5,0) -- cycle    ;
	\draw    (481.34,278.31) -- (501.69,278.2) ;
	\draw [shift={(504.69,278.18)}, rotate = 539.6700000000001] [fill={rgb, 255:red, 0; green, 0; blue, 0 }  ][line width=0.08]  [draw opacity=0] (5,-2.5) -- (0,0) -- (5,2.5) -- (3.5,0) -- cycle    ;
	\draw    (538.29,256.75) .. controls (550.49,265.79) and (559.45,279.6) .. (516.7,277.79) ;
	\draw [shift={(514.01,277.65)}, rotate = 363.14] [fill={rgb, 255:red, 0; green, 0; blue, 0 }  ][line width=0.08]  [draw opacity=0] (5,-2.5) -- (0,0) -- (5,2.5) -- (3.5,0) -- cycle    ;
	\draw  [fill={rgb, 255:red, 0; green, 0; blue, 0 }  ,fill opacity=1 ] (273.3,279.05) .. controls (273.3,278.08) and (272.52,277.3) .. (271.55,277.3) .. controls (270.58,277.3) and (269.8,278.08) .. (269.8,279.05) .. controls (269.8,280.02) and (270.58,280.8) .. (271.55,280.8) .. controls (272.52,280.8) and (273.3,280.02) .. (273.3,279.05) -- cycle ;
	\draw  [fill={rgb, 255:red, 0; green, 0; blue, 0 }  ,fill opacity=1 ] (272.9,240.05) .. controls (272.9,239.08) and (272.12,238.3) .. (271.15,238.3) .. controls (270.18,238.3) and (269.4,239.08) .. (269.4,240.05) .. controls (269.4,241.02) and (270.18,241.8) .. (271.15,241.8) .. controls (272.12,241.8) and (272.9,241.02) .. (272.9,240.05) -- cycle ;
	\draw  [fill={rgb, 255:red, 0; green, 0; blue, 0 }  ,fill opacity=1 ] (304.9,239.8) .. controls (304.9,238.83) and (304.12,238.05) .. (303.15,238.05) .. controls (302.18,238.05) and (301.4,238.83) .. (301.4,239.8) .. controls (301.4,240.77) and (302.18,241.55) .. (303.15,241.55) .. controls (304.12,241.55) and (304.9,240.77) .. (304.9,239.8) -- cycle ;
	\draw  [fill={rgb, 255:red, 0; green, 0; blue, 0 }  ,fill opacity=1 ] (313.4,249.3) .. controls (313.4,248.33) and (312.62,247.55) .. (311.65,247.55) .. controls (310.68,247.55) and (309.9,248.33) .. (309.9,249.3) .. controls (309.9,250.27) and (310.68,251.05) .. (311.65,251.05) .. controls (312.62,251.05) and (313.4,250.27) .. (313.4,249.3) -- cycle ;
	\draw  [fill={rgb, 255:red, 0; green, 0; blue, 0 }  ,fill opacity=1 ] (312.9,279.05) .. controls (312.9,278.08) and (312.12,277.3) .. (311.15,277.3) .. controls (310.18,277.3) and (309.4,278.08) .. (309.4,279.05) .. controls (309.4,280.02) and (310.18,280.8) .. (311.15,280.8) .. controls (312.12,280.8) and (312.9,280.02) .. (312.9,279.05) -- cycle ;
	\draw  [fill={rgb, 255:red, 0; green, 0; blue, 0 }  ,fill opacity=1 ] (383.75,278.55) .. controls (383.75,277.58) and (382.97,276.8) .. (382,276.8) .. controls (381.03,276.8) and (380.25,277.58) .. (380.25,278.55) .. controls (380.25,279.52) and (381.03,280.3) .. (382,280.3) .. controls (382.97,280.3) and (383.75,279.52) .. (383.75,278.55) -- cycle ;
	\draw  [fill={rgb, 255:red, 0; green, 0; blue, 0 }  ,fill opacity=1 ] (415.5,278.05) .. controls (415.5,277.08) and (414.72,276.3) .. (413.75,276.3) .. controls (412.78,276.3) and (412,277.08) .. (412,278.05) .. controls (412,279.02) and (412.78,279.8) .. (413.75,279.8) .. controls (414.72,279.8) and (415.5,279.02) .. (415.5,278.05) -- cycle ;
	\draw  [fill={rgb, 255:red, 0; green, 0; blue, 0 }  ,fill opacity=1 ] (479,278.05) .. controls (479,277.08) and (478.22,276.3) .. (477.25,276.3) .. controls (476.28,276.3) and (475.5,277.08) .. (475.5,278.05) .. controls (475.5,279.02) and (476.28,279.8) .. (477.25,279.8) .. controls (478.22,279.8) and (479,279.02) .. (479,278.05) -- cycle ;
	\draw  [fill={rgb, 255:red, 0; green, 0; blue, 0 }  ,fill opacity=1 ] (511,278.05) .. controls (511,277.08) and (510.22,276.3) .. (509.25,276.3) .. controls (508.28,276.3) and (507.5,277.08) .. (507.5,278.05) .. controls (507.5,279.02) and (508.28,279.8) .. (509.25,279.8) .. controls (510.22,279.8) and (511,279.02) .. (511,278.05) -- cycle ;
	\draw    (358.17,328.75) .. controls (297.23,329.73) and (272.03,319.35) .. (271.27,286.39) ;
	\draw [shift={(271.26,283.81)}, rotate = 450.82] [fill={rgb, 255:red, 0; green, 0; blue, 0 }  ][line width=0.08]  [draw opacity=0] (5,-2.5) -- (0,0) -- (5,2.5) -- (3.5,0) -- cycle    ;
	\draw    (381.92,323.04) -- (382.05,285.69) ;
	\draw [shift={(382.06,282.69)}, rotate = 450.2] [fill={rgb, 255:red, 0; green, 0; blue, 0 }  ][line width=0.08]  [draw opacity=0] (5,-2.5) -- (0,0) -- (5,2.5) -- (3.5,0) -- cycle    ;
	\draw    (404.42,328.29) .. controls (478.91,325.9) and (478.48,305.49) .. (477.42,285.8) ;
	\draw [shift={(477.28,283.03)}, rotate = 447.21] [fill={rgb, 255:red, 0; green, 0; blue, 0 }  ][line width=0.08]  [draw opacity=0] (5,-2.5) -- (0,0) -- (5,2.5) -- (3.5,0) -- cycle    ;
	\draw  [fill={rgb, 255:red, 0; green, 0; blue, 0 }  ,fill opacity=1 ] (366.3,328.05) .. controls (366.3,327.08) and (365.52,326.3) .. (364.55,326.3) .. controls (363.58,326.3) and (362.8,327.08) .. (362.8,328.05) .. controls (362.8,329.02) and (363.58,329.8) .. (364.55,329.8) .. controls (365.52,329.8) and (366.3,329.02) .. (366.3,328.05) -- cycle ;
	\draw  [fill={rgb, 255:red, 0; green, 0; blue, 0 }  ,fill opacity=1 ] (383.47,328.05) .. controls (383.47,327.08) and (382.68,326.3) .. (381.72,326.3) .. controls (380.75,326.3) and (379.97,327.08) .. (379.97,328.05) .. controls (379.97,329.02) and (380.75,329.8) .. (381.72,329.8) .. controls (382.68,329.8) and (383.47,329.02) .. (383.47,328.05) -- cycle ;
	\draw  [fill={rgb, 255:red, 0; green, 0; blue, 0 }  ,fill opacity=1 ] (400.3,328.05) .. controls (400.3,327.08) and (399.52,326.3) .. (398.55,326.3) .. controls (397.58,326.3) and (396.8,327.08) .. (396.8,328.05) .. controls (396.8,329.02) and (397.58,329.8) .. (398.55,329.8) .. controls (399.52,329.8) and (400.3,329.02) .. (400.3,328.05) -- cycle ;
	\draw  [fill={rgb, 255:red, 0; green, 0; blue, 0 }  ,fill opacity=1 ] (441.92,254.52) .. controls (441.92,253.55) and (441.14,252.77) .. (440.17,252.77) .. controls (439.2,252.77) and (438.42,253.55) .. (438.42,254.52) .. controls (438.42,255.48) and (439.2,256.27) .. (440.17,256.27) .. controls (441.14,256.27) and (441.92,255.48) .. (441.92,254.52) -- cycle ;
	\draw  [fill={rgb, 255:red, 0; green, 0; blue, 0 }  ,fill opacity=1 ] (537.17,254.07) .. controls (537.17,253.1) and (536.38,252.32) .. (535.42,252.32) .. controls (534.45,252.32) and (533.67,253.1) .. (533.67,254.07) .. controls (533.67,255.04) and (534.45,255.82) .. (535.42,255.82) .. controls (536.38,255.82) and (537.17,255.04) .. (537.17,254.07) -- cycle ;
\draw   (350.53,336.67) .. controls (350.53,341.34) and (352.86,343.67) .. (357.53,343.67) -- (372.73,343.67) .. controls (379.4,343.67) and (382.73,346) .. (382.73,350.67) .. controls (382.73,346) and (386.06,343.67) .. (392.73,343.67)(389.73,343.67) -- (407.93,343.67) .. controls (412.6,343.67) and (414.93,341.34) .. (414.93,336.67) ;
	
	\draw (357.61,332) node [anchor=north west][inner sep=0.75pt]  [font=\footnotesize] [align=left] {$\displaystyle a_{0}$};
	\draw (436.94,262.89) node [anchor=north west][inner sep=0.75pt]  [font=\footnotesize] [align=left] {$\displaystyle U$};
	\draw (256.19,253.47) node [anchor=north west][inner sep=0.75pt]  [font=\footnotesize] [align=left] {$\displaystyle R_{v}$};
	\draw (286.94,281.89) node [anchor=north west][inner sep=0.75pt]  [font=\footnotesize] [align=left] {$\displaystyle R_{h}$};
	\draw (314.44,259.89) node [anchor=north west][inner sep=0.75pt]  [font=\footnotesize] [align=left] {$\displaystyle R_{v}$};
	\draw (277.78,225.22) node [anchor=north west][inner sep=0.75pt]  [font=\footnotesize] [align=left] {$\displaystyle R_{h}$};
	\draw (255.28,271.09) node [anchor=north west][inner sep=0.75pt]  [font=\footnotesize] [align=left] {$\displaystyle c_{1}$};
	\draw (315.82,278.19) node [anchor=north west][inner sep=0.75pt]  [font=\footnotesize] [align=left] {$\displaystyle c_{2}$};
	\draw (256.09,230.19) node [anchor=north west][inner sep=0.75pt]  [font=\footnotesize] [align=left] {$\displaystyle c_{3}$};
	\draw (316.59,238.22) node [anchor=north west][inner sep=0.75pt]  [font=\footnotesize] [align=left] {$\displaystyle c_{4}$};
	\draw (303.76,228.56) node [anchor=north west][inner sep=0.75pt]  [font=\footnotesize] [align=left] {$\displaystyle c_{5}$};
	\draw (388.34,263.56) node [anchor=north west][inner sep=0.75pt]  [font=\footnotesize] [align=left] {$\displaystyle R_{h}$};
	\draw (368.09,265.52) node [anchor=north west][inner sep=0.75pt]  [font=\footnotesize] [align=left] {$\displaystyle d_{1}$};
	\draw (409.09,263.06) node [anchor=north west][inner sep=0.75pt]  [font=\footnotesize] [align=left] {$\displaystyle d_{2}$};
	\draw (376.94,332) node [anchor=north west][inner sep=0.75pt]  [font=\footnotesize] [align=left] {$\displaystyle a_{1}$};
	\draw (299.94,328.72) node [anchor=north west][inner sep=0.75pt]  [font=\footnotesize] [align=left] {$\displaystyle U$};
	\draw (433.36,238.3) node [anchor=north west][inner sep=0.75pt]  [font=\footnotesize] [align=left] {$\displaystyle d_{3}$};
	\draw (484.84,263.25) node [anchor=north west][inner sep=0.75pt]  [font=\footnotesize] [align=left] {$\displaystyle R_{h}$};
	\draw (468.59,264.22) node [anchor=north west][inner sep=0.75pt]  [font=\footnotesize] [align=left] {$\displaystyle e_{1}$};
	\draw (503.09,264.25) node [anchor=north west][inner sep=0.75pt]  [font=\footnotesize] [align=left] {$\displaystyle e_{2}$};
	\draw (394.94,332) node [anchor=north west][inner sep=0.75pt]  [font=\footnotesize] [align=left] {$\displaystyle a_{2}$};
	\draw (531.94,262.92) node [anchor=north west][inner sep=0.75pt]  [font=\footnotesize] [align=left] {$\displaystyle U$};
	\draw (528.61,239) node [anchor=north west][inner sep=0.75pt]  [font=\footnotesize] [align=left] {$\displaystyle e_{3}$};
	\draw (449.44,326.72) node [anchor=north west][inner sep=0.75pt]  [font=\footnotesize] [align=left] {$\displaystyle U$};
	\draw (385.94,299.22) node [anchor=north west][inner sep=0.75pt]  [font=\footnotesize] [align=left] {$\displaystyle U$};
	\draw (377.97,354.6) node [anchor=north west][inner sep=0.75pt]   [align=left] {$\displaystyle P$};

\end{tikzpicture}

\end{center}	
 We aim to show that $(T,H,V)$ admits a
solution iff $\Kmc$ is $\text{FO}$-separable. By Point~3 of
Theorem~\ref{critFOwithoutUNA3} it suffices to show the following.

\smallskip\noindent{\textit{Claim.}} $(T,H,V)$
admits a solution iff there exists a model $\Amf$ of $\Kmc$ such that
$\Dmc_{\text{con}(a)},a\nrightarrowtail \Amf,b^{\Amf}$, for all
$a\in P$.   

\smallskip\noindent{\textit{Proof of the Claim.}} 
For $(\Rightarrow)$, assume there exists
a solution $\tau$ to $(T,H,V)$. Then define a model $\Amf$ of $\Kmc$
by taking $\Dmc$ viewed as a structure and connecting $b$ via $U$ to
all pairs in $\mathbbm{N}\times \mathbbm{N}$. On $\mathbbm{N}\times
\mathbbm{N}$ we replicate the solution $\tau$ and make sure that the
concept names $P_{i}$ are interpreted in a suitable way. In detail,
let
\begin{align*}
	B^\Amf & = \{b\} \\
	R_v^{\Amf} & = \{(c,c') \mid R_{v}(c,c') \in \Dmc\} \cup \{(
	  (i,j), (i,j+1))\mid i,j\in \mathbbm{N}\} \\
	R_h^{\Amf} & = \{(c,c') \mid R_{h}(c,c') \in \Dmc\} \cup \{(
	  (i,j), (i+1,j))\mid i,j\in \mathbbm{N}\}\\
	A_t^{\Amf} & = \{(i,j)\in \mathbbm{N} \times \mathbbm{N}\mid
      \tau(i,j)=t\} \hspace{5cm}\text{for all $t\in T$}\\
	U^{\Amf} & = \{(c,c') \mid U(c,c') \in \Dmc\} \cup \{
	  (b,(i,j))\mid i,j\in \mathbbm{N}\} \\
	P_{0}^{\Amf} & = \{(2i,2j) \mid i,j\in \mathbbm{N}\}\\
	P_{1}^{\Amf} & = \{(2i+1,2j) \mid i,j\in \mathbbm{N}\}\\
	P_{2}^{\Amf} & = \{(2i,2j+1) \mid i,j\in \mathbbm{N}\}\\
	P_{3}^{\Amf} & = \{(2i+1,2j+1) \mid i,j\in \mathbbm{N}\}
\end{align*} 
and interpret all constants in $\Dmc$ by themselves. It
is easy to see that $\Amf$ is a model of $\Kmc$ and
$\Dmc_{\text{con}(a)},a\nrightarrowtail \Amf,b^{\Amf}$, for all
$a\in P$. 

For $(\Leftarrow)$, assume that $\Amf$ is a model of $\Kmc$ such
that $\Dmc_{\text{con}(a)},a\nrightarrowtail \Amf,b^{\Amf}$, for
all $a\in P$. We can then inductively show that $\Amf$ contains an
infinite grid by using the CIs~\eqref{eq:succ1} to~\eqref{eq:succ7}.
In more detail, by~\eqref{eq:succ1}, there is a $U$-successor $b_0$ of
$b$ that satisfies $P_0$. Moreover,
by~\eqref{eq:succ2}-\eqref{eq:succ5}, there are elements
$b_1,b_2,b_3,b_3'$ such that 
\begin{itemize}

  \item $b_i$ satisfies $P_i$, for each $i\in\{1,2,3\}$ and $b_3'$
    satisfies $P_3$;

  \item $b_1$ is an $R_h$-successor of $b_0$ and $b_2$ is an
    $R_v$-successor of $b_0$;
    
  \item $b_3$ is an $R_h$-successor of $b_2$ and $b_3'$ is an
    $R_v$-successor of $b_1$.
\end{itemize}
Moreover, by~\eqref{eq:succ6} and~\eqref{eq:succ7}, $b,b_0,b_1,b_2$
are pairwise distinct and also different from $b_3,b_3'$.  Note that
all elements $b_0,b_1,b_2,b_3,b_3'$ have an $U$-predecessor,
by~\eqref{eq:succ2}-\eqref{eq:succ5}. Due to
$\Dmc_{\text{con}(a)},a_i\nrightarrowtail \Amf,b^{\Amf}$, for $i\in
\{1,2\}$, each such $U$-predecessor is in fact $b$.
Moreover, $\Dmc_{\text{con}(a)},a_0\nrightarrowtail
\Amf,b^{\Amf}$ entails that $b_3=b_3'$, closing the grid cell. We can
continue the argument in this way to obtain the full infinite grid. 
The remaining CIs of $\Omc$ ensure that the grid is labeled with a
solution to $(T,H,V)$. This finishes the proof of the Claim.
\end{proof} 

\end{document}